\pdfoutput=1
\documentclass[12pt,letterpaper,titlepage]{article}

\makeatletter
\def\input@path{{./}}
\makeatother

\usepackage[margin=1in]{geometry}

\usepackage{amsmath, amssymb, amsthm, bm}

\usepackage[justification=justified, labelfont={color=black, bf}, labelsep=period, singlelinecheck=false, font=footnotesize]{caption}
\usepackage[justification=justified, labelfont={color=black, bf}, font=footnotesize]{subcaption}
\usepackage{graphicx, booktabs, multicol, multirow, adjustbox, tabularx, array, threeparttable}
\usepackage{siunitx}       %
\newcommand{\sym}[1]{\rlap{\ensuremath{^{#1}}}}
\usepackage{makecell}
\usepackage{float}
\usepackage{pdflscape}     %

\usepackage{setspace}
\usepackage{enumitem}

\usepackage{xcolor,colortbl}
\definecolor{prussianblue}{RGB}{0, 49, 82}

\usepackage[colorlinks=true, linkcolor=prussianblue, citecolor=prussianblue, urlcolor=prussianblue]{hyperref}

\usepackage[T1]{fontenc}
\usepackage{txfonts}

\usepackage{appendix}

\usepackage[stable]{footmisc}  %
\usepackage{titling}           %
\usepackage[backend=biber, maxnames=5, style=authoryear, natbib=true, doi=false, isbn=false, url=false, eprint=false, uniquelist=false, uniquename=false]{biblatex}
\usepackage{xspace}
\newcommand{\NTM}{\text{NTM}\xspace}
\newcommand{\FFM}{\text{FFM}\xspace}

\newcommand{\poly}{\text{poly}\xspace}
\newcommand{\mn}{\text{m}\xspace}

\providecommand{\orig}[1]{}
\providecommand{\propdel}[1]{\ignorespaces}
\usepackage[capitalise,noabbrev]{cleveref}
\crefname{appendix}{Appendix}{Appendices}
\crefname{subappendix}{Appendix}{Appendices}
\crefname{subsubappendix}{Appendix}{Appendices}
\crefname{assumption}{Assumption}{Assumptions}
\AddToHook{cmd/appendix/after}{%
  \crefalias{section}{appendix}%
  \crefalias{subsection}{subappendix}%
  \crefalias{subsubsection}{subsubappendix}%
}
\usepackage{mathtools}
\let\coloneq\coloneqq %
\usepackage{tikz}
\usepackage[most]{tcolorbox}
\usepackage{microtype}

\newcommand{\E}{\mathbb{E}}
\newcommand{\Prb}{\mathbb{P}}
\newcommand{\ppost}{\rho}
\newcommand{\rvar}{\nu}
\newcommand{\one}{\mathbf{1}}

\newcommand{\PM}{\mathrm{P}}

\newcommand{\hbuy}{\mathtt{+}}
\newcommand{\hzero}{\mathtt{0}}
\newcommand{\hsell}{\mathtt{-}}

\DeclareMathOperator*{\argmax}{arg\,max}

\theoremstyle{plain}
\newtheorem{theorem}{Theorem}%
\newtheorem{proposition}{Proposition}
\newtheorem{lemma}{Lemma}
\newtheorem{corollary}{Corollary}
\theoremstyle{definition}
\newtheorem{definition}{Definition}
\newtheorem{assumption}{Assumption}

\graphicspath{{exhibits/}}

\newcommand{\titletext}{Settlement Manipulation in Prediction Markets}
\title{\titletext}

\newif\ifanonymous
\anonymousfalse

\ifanonymous
\author{}
\date{}
\else
\author{David Dai\thanks{Department of Management Science and Engineering, Stanford University; htdai@stanford.edu.} \and Ruizhe Jia\thanks{Department of Management Science and Engineering, Stanford University; ruizhe@stanford.edu.} \and Shihao Yu\thanks{Lee Kong Chian School of Business, Singapore Management University; shihaoyu@smu.edu.sg.}}
\date{Current version: November 23, 2026}
\fi

\begin{document}

\ifanonymous
\else
\begingroup
\thanksmarkseries{fnsymbol}
\maketitle
\endgroup
\setcounter{footnote}{0}
\fi

\vspace*{2.5cm}
\begin{center}
  {\LARGE\titletext\par}
\end{center}

\vspace{2em}
{\centering\bfseries\abstractname\par}
\vspace{0.75em}
\begin{spacing}{1}
\noindent Prediction markets increasingly list contracts settling on an asset price that holders can move by trading the underlying. We build a model showing that such contracts transfer wealth from prediction-market liquidity traders to manipulators and the prediction-market maker and harm price discovery in the underlying, even as it becomes more liquid. After the launch of Polymarket's five-minute Bitcoin contract, settlement-time spot order flow spikes, causing large price reversals after settlement. Manipulators capture a large amount of profit, mostly from retail. Manipulation is largely absent in the fifteen-minute contracts, and it disappeared from the five-minute contract once Polymarket began settling on a 60-second average price: lengthening the horizon or averaging the settlement price removes it.

\end{spacing}

\clearpage

\section{Introduction}
\label{sec:intro}

Prediction markets have grown with remarkable speed: combined monthly volume on Kalshi and Polymarket rose nearly fivefold in seven months, from under \$5 billion in September 2025 to about \$24 billion by April 2026.\footnote{Pew Research Center, ``Trading Volume on Prediction Markets Has Soared in Recent Months'' (2026).} Surprisingly, the fastest-growing of these contracts are not written on elections or sports, but on asset prices. Each is a binary contract: it pays \$1 if the asset is above a chosen level at the contract's close, and \$0 otherwise. These began in cryptocurrencies, where they are already large: on Kalshi alone, volume passed \$1 billion a month for the first time in March 2026.\footnote{CF Benchmarks, ``Kalshi Leads Surging Crypto Event Contract Market'' (2026).} They have since spread to other assets. Kalshi now lists the same up-or-down contract on gold and silver,\footnote{Kalshi, ``Gold and Silver 15-minute markets. Live on Kalshi,'' post on X, August 3, 2026, \url{https://x.com/Kalshi/status/2084315859192057996}.} and the SEC has approved Nasdaq's binary options on the Nasdaq-100 index and Cboe's extension of its own.\footnote{Nasdaq MRX filed on March 2, 2026, and the SEC approved its binary broad-based index options on April 30, 2026 (Exchange Act Release No.~34-105342). Cboe filed on April 2, 2026 (Exchange Act Release No.~34-105247), and the SEC approved its proposal on July 17, 2026 (Exchange Act Release No.~34-105936).} The venues advertise these contracts heavily to retail traders on social media.\footnote{For example, Kalshi has run a paid video advertisement for its fifteen-minute silver market on Instagram since August 2026 (Meta Ad Library, \url{https://www.facebook.com/ads/library/?id=1047326131471208}). Olivia Little, ``Kalshi is aggressively advertising on TikTok, seemingly targeting young people while promising a `money hack','' Media Matters for America, March 30, 2026, documents a broader campaign across TikTok, Instagram, Reddit, and YouTube.} Notably, the contract horizon keeps shrinking: from a day, to an hour, to fifteen minutes, and on Polymarket to just five minutes.

The exchanges and their regulators also make a public case for these products: they benefit retail investors and support price discovery. Cboe argues that extending its binary options ``will benefit investors, particularly retail investors and other investors who prefer simplicity.''\footnote{Cboe Exchange filing, Exchange Act Release No.~34-105247, p.~18. The proposed extension covers non-broad-based indexes.} Nasdaq sees another benefit in listing them: ``a centralized forum for price discovery.''\footnote{Nasdaq MRX filing as described in Exchange Act Release No.~34-105342, pp.~4 and 19.} The regulators add their own endorsements. The SEC has described binary options on broad-based indexes as ``a potentially useful investment choice,''\footnote{Exchange Act Release No.~34-105342, p.~27, recalling the Commission's earlier assessment.} and the CFTC tells investors that prediction markets ``can sometimes forecast event outcomes better than polling or other forms of forecasting.''\footnote{Commodity Futures Trading Commission, ``Understanding Prediction Markets and Event Contracts'' (investor advisory), \texttt{cftc.gov}.} And the academic record on prediction markets gives little reason to worry about manipulation: a manipulator can increase price accuracy \citep{hanson2009manipulator}, and documented manipulation attempts reverted quickly \citep{rhode2008manipulating}.

This paper shows the opposite on each count for asset-price contracts: the contracts are manipulated, the price of the underlying asset becomes less informative, and retail traders lose money. The vulnerability is structural: an asset-price contract settles on a financial price, and that price can be moved by trading the underlying market itself, the cash-settlement manipulation of \citet{kumar1992futures}, now inside a retail prediction market with horizons as short as five minutes. U.S.\ law already requires every listed contract to be ``not readily susceptible to manipulation'';\footnote{Commodity Exchange Act, Designated Contract Market Core Principle~3; reaffirmed for event contracts in CFTC Staff Advisory~26-08 (2026).} we show that asset-price contracts raise exactly the concern that standard is meant to prevent. The footprint is in the data (Figure~\ref{fig:headline}): after the five-minute contract went live on Polymarket, Binance spot order flow spikes in the final seconds before settlement, and the resulting price movements partly revert shortly after. This is a transitory push to manipulate the spot price, not trading on information.

\begin{figure}[!t]
    \centering
    \includegraphics[width=\linewidth]{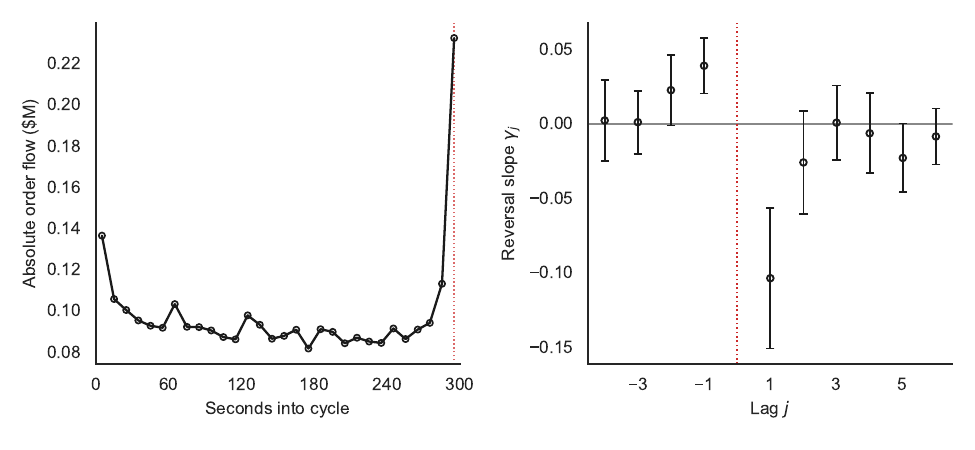}
    \caption{\textbf{The footprint of settlement manipulation.} This figure plots Binance BTC spot activity around the five-minute contract's close, in the post-launch period. \emph{Left}: absolute order flow per five-minute cycle, folded onto the rolling clock; the dotted vertical line marks the final ten seconds before the close. \emph{Right}: the post-settlement reversal slope $\gamma_j$ at each lag $j$ around the close, the slope of the lag-$j$ ten-second return on the final-ten-second return. Section~\ref{sec:empirics} details both constructions and establishes the patterns against pre-launch and across-moneyness benchmarks.}
    \label{fig:headline}
\end{figure}

We develop a sequential game-theoretic model of the two markets. The spot market has its usual players (an informed trader who knows the asset's value and liquidity traders who do not) and now a manipulator, who has bought the contract and can trade the spot just before settlement. Her order pools with the rest, and because she trades on no information, a late buy is now less likely to be the informed trader's. It therefore carries less information about the asset, so the settlement price reveals less of it and price discovery falls. Yet, by lowering the adverse selection the spot maker faces, it leads the spot maker to supply more liquidity, so the manipulation makes the spot market more liquid even as it makes the price less informative. The manipulator does not enter every time: the more often a late buy is expected to be hers, the more the spot maker discounts it, and the less her push moves the closing price. She therefore chooses the entry rate that maximizes her total profit.

One would expect manipulation to hurt the prediction-market maker, who would protect itself with a wider bid-ask spread. We show the opposite. The maker charges no spread and gains from odds tilted in its favor, and it is the ordinary trader who pays, by buying a contract whose odds are worse than fair. The reason lies in the spot market. A contract buy is public, but the spot maker cannot tell the manipulator from an ordinary contract buyer. So after anyone buys the contract, the spot maker treats a late spot buy as a possible push, not as news, and moves its price by less. The closing price is then tilted against the buyer's side, whoever she is. The manipulator can offset the tilt with her push; the ordinary buyer cannot. Both pay the same price for the contract, but the ordinary buyer collects less on average, and her loss is a transfer to the prediction-market maker and the manipulator.

These harms are not inevitable: both the model and the evidence point to the contract's horizon as a design lever. %
 A short horizon leaves little room for the asset's own price discovery before settlement, so a fixed spot push is more likely to determine the contract's resolution. A longer horizon lets more ordinary spot trading reach the price before the close, so the cases in which a push is pivotal become rarer and the loss of price discovery shrinks. The remedy is direct: lengthen the horizon. A second lever is to settle the contract on a price averaged over a window, not on the price at a single instant, so that a late push must move the whole average. 

Empirically, we study Polymarket's Bitcoin five-minute up/down contract: a binary claim that pays \$1 if Bitcoin is higher at the close of a five-minute window than at its open, and \$0 otherwise. Every trade is recorded on a public blockchain, so each trade and trader is traceable, a transparency centralized exchanges and options venues do not offer. The contract resolves not on Polymarket but against a Chainlink oracle that averages the price across major spot exchanges. The price on Binance, the largest crypto exchange by volume, tracks that benchmark closely: its mid sits about two and a half basis points from the oracle and finishes on the same side of the strike as the resolution about $85\%$ of the time. We therefore measure the settlement-time flow on Binance. The contract's launch on February 12, 2026 offers a natural experiment: no five-minute contract existed before that date, so we compare three periods: before the fifteen-minute and four-hour contracts launched, with only those two contracts trading, and with the five-minute contract live.

One would not expect a market this small to move one this large: in the median five-minute cycle, about 105~thousand one-dollar contracts trade on Polymarket, against \$3.15~million on Binance. Yet as soon as the contract launched, large spot orders began arriving in the final ten seconds before each close: net order flow on Binance jumps by about 50\% relative to the pre-launch level, and volatility rises with it. Moreover, the spike follows the incentive to push. It is sharpest where a push is pivotal: in the roughly 6\% of cycles priced near even just before the close, the jump is about $3.9$ times that in the rest. The order-flow spike is not caused by information arriving late: within ten seconds the price reverts, by about a quarter in the near-even cycles and a tenth in the others. Real information would persist; the price impact of a manipulative push reverts. Finally, consistent with the model, these near-even closes add depth on Binance, where an uninformed push lowers the adverse selection the spot maker faces.

We then identify the cycles most likely to contain a manipulative push. For each cycle we measure the size of the final-ten-second net order flow on Binance relative to the cycle's own typical flow, a measure that uses neither the contract's price nor its outcome. We classify the cycles in the top decile of this measure, about $1{,}600$, as manipulated. Away from the close these cycles are quieter than average. Moreover, they fall disproportionately in thin-liquidity hours, when a given dollar of near-settlement flow moves the spot price the most: $56\%$ in the overnight window and $44\%$ on weekends, against $40\%$ and $27\%$ of normal cycles.

Within these manipulated cycles, the pushes do more than cluster at the close: the side they back wins far more often than it otherwise would. In cycles where the contract was priced near even, when the push ran against the favored side, the underdog won 65\% of the time, against 41\% in normal cycles. The stronger test is the cycles that were all but decided. One would expect a side given a 90-to-100\% chance just before the close to win almost every time, and in normal cycles it loses only 1\% of the time. Yet when a push ran against it, it lost 34\% of the time. A bet the market treated as near-certain was overturned one time in three.

The model predicts who gains from manipulation and who pays for it, and the blockchain lets us check, because every wallet's profit is visible. We identify as manipulators the traders who gain in the manipulated cycles, an identification based on realized profits rather than a direct observation of intent. Just 821 wallets fit, about one in three hundred of the roughly 243{,}000 that traded the contract. They take \$8.2~million in the manipulated cycles and break even in the rest: the edge lives only where the manipulation does. Retail traders, the liquidity traders of the model, pay for it: they lose \$7.6~million in these cycles, whereas market makers, who earn \$3.1~million in normal cycles, lose only \$0.6~million. All of it comes from a single contract, on a single platform, over two months.

The natural innocent explanation is hedging: a prediction-market maker who has sold the binary contract offsets its risk by trading the spot. But hedging does not fit the data, for two reasons. The first is the state of the contract when the pushing happens. A binary contract's sensitivity to the spot price is large only when the price sits near the strike; once one side is almost certain to win, the contract barely moves with the spot, and the prediction-market maker is left with almost nothing to hedge. Yet those nearly decided cycles (where the market already gave one side a 90-to-100\% chance) are exactly the ones in which a push still overturns the outcome, one time in three. A trade placed when there is nothing to hedge cannot be a hedge. The second reason is the timing. Even when there is something to hedge, a prediction-market maker would build its spot position gradually, as that exposure accumulates over the cycle. Instead, the flow arrives in a single burst at the very end. A position assembled only in the closing seconds is not a hedge accumulated over the cycle; it is a trade placed to move the settlement price.

Both design levers can be examined in the data. The first is the horizon. Polymarket's fifteen-minute contract settles on the same rule as the five-minute contract and was the most heavily traded of the up/down contracts before the five-minute launch. Yet at the closes where it settles, the final-ten-second increase in Binance order flow is about $7\%$, compared to $50\%$ for the five-minute contract, with no detectable price effect and no reversal.

The second design lever, an averaged settlement price, was put to the test after our main sample ended. On August 7, 2026, Polymarket stopped settling the contract on a single oracle print and began settling it on a time-weighted average,\footnote{The first version of this paper was posted on June 30, 2026, and was reported by Bloomberg on July 15, 2026 (Cameron Fozi, ``Stanford Study Says Polymarket Crypto Bets Are Being Manipulated''). Polymarket announced the change on July 30, 2026 (\url{https://x.com/devjoshstevens/status/2082820700876746877}).} over the final 60 seconds from August 14. A push confined to the final ten seconds then moves the settlement price by only a sixth as much, so the incentive to push should weaken. It does. After the switch, the manipulation footprint largely disappears. Moreover, the pushes that remain no longer decide outcomes: against the favored side they win no more often than in normal cycles. The profits disappear with them. In the 79 days before the switch, the rule we use in the main sample identifies 921 wallets as manipulators, who together made \$7.4~million on manipulated cycles; over an equally long period after the switch, the same wallets made \$8{,}000. Nor do manipulators appear to adapt by spreading the push over the averaging window: order flow over the window does not concentrate in near-even cycles.

The paper proceeds as follows. Section~\ref{sec:literature} positions the contribution. Section~\ref{sec:institutional} describes the contracts and how they settle. Sections~\ref{sec:theory_setup} and~\ref{sec:theory_results} develop the model and its predictions. Section~\ref{sec:empirics} presents the testable predictions, the data, the empirical results, and the effect of Polymarket's switch to an averaged settlement price, and Section~\ref{sec:conclusion} discusses the market-design implications.

\section{Related Literature}
\label{sec:literature}

This paper joins three strands of literature that have largely developed apart: the economics of prediction markets, the manipulation of cash-settled derivatives, and the market microstructure of decentralized and crypto venues.

\paragraph{Prediction markets.}
Prediction markets are valued as information aggregators whose prices forecast events as accurately as expert judgment \citep{wolfers2004prediction}, a view recently reaffirmed for macroeconomic contracts on Kalshi \citep{diercks2026kalshi}. A long-standing concern is that such prices are vulnerable to manipulation, yet the literature has largely concluded that manipulation need not impair price accuracy. In a Kyle-style model, a manipulator's order flow raises the returns to informed trading and can even improve price accuracy \citep{hanson2009manipulator}; in laboratory markets, other traders anticipate and discount manipulators' biased orders, leaving accuracy intact \citep{hanson2006information}; and documented manipulation attempts in the field have produced only transient, mean-reverting price effects \citep{rhode2008manipulating}. A growing empirical literature exploits transaction-level data from Polymarket and Kalshi: one strand examines who profits and who is informed, finding that a small group of sophisticated traders captures most of the gains while retail participants fund them \citep{akey2026who,gomezcram2026accuracy,dellavedova2026profits}; a parallel strand measures adverse selection and order flow \citep{bartlett2026adverse,tsang2026anatomy,dubach2026anatomy} and screens for suspicious accounts \citep{mitts2026informed}. Our results on who gains and who loses are related to \citet{bartlett2026adverse}, who show on Kalshi that market makers earn more where adverse selection is greater, although spreads are only modestly wider, because traders systematically overbet one side. We identify a separate source of maker gains and trader losses that does not depend on how traders bet: settlement manipulation, which lowers what even a trader with symmetric demands can expect to collect.

{To our knowledge this is the first study of ultra-short, oracle-settled prediction markets on an asset's price, and the first to identify settlement manipulation in them: a trader who holds the binary contract trades the underlying asset near settlement to move the price on which the contract pays. We show that this manipulation not only transfers wealth from liquidity traders to manipulators in the prediction market, but also harms price discovery in the underlying spot market. The paper has also been followed by a change in market design. Weeks after its first version was reported in the press,\footnote{Cameron Fozi, ``Stanford Study Says Polymarket Crypto Bets Are Being Manipulated,'' Bloomberg, July 15, 2026; Wu Blockchain, ``Study Finds Signs of Manipulation in Bitcoin Bets on World's Largest Prediction Market,'' July 16, 2026. Polymarket announced the settlement change on July 30, 2026.} Polymarket began settling these contracts on a time-weighted average price, a change that we show sharply reduces the manipulation.}

\paragraph{Manipulation of cash-settled derivatives.}
{At its core, our mechanism is the cash-settlement manipulation of \citet{kumar1992futures}: a trader takes a derivative position and then trades the underlying spot to move the cash-settlement price, profiting when the gain on the position exceeds the cost of the spot trade. Another theory paper studying manipulation in derivatives markets is \citet{zhang2022competition}, in which manipulators can take endogenous derivative contract positions and can leave all agents worse off. On the empirical side, the literature documents settlement and expiration distortions in equities and volatility benchmarks \citep{ni2005stock,comertonforde2011measuring,griffin2018vix,baltussen2023derivative}. \citet{spatt2014security} surveys a variety of manipulation strategies. Finally, a design literature asks how to build benchmarks that resist manipulation \citep{duffie2021robust}.

    Our contribution relative to this work is twofold. First, on the theory side, the economics of who gains and who loses differ. In \citet{kumar1992futures}, as in our model, manipulation makes the spot market more liquid and uninformed derivative traders bear the loss; there, however, market makers break even and the informativeness of prices is unchanged. In our model, the derivative market maker earns a rent at a zero bid-ask spread, so part of the uninformed traders' loss accrues to the market maker, and the closing spot price becomes less informative. In \citet{zhang2022competition}, the cost of manipulation is basis risk borne by hedgers in a market without private information; here it arises from adverse selection in the spot market. We further show that a longer contract horizon lowers both the manipulator's profit and the price-discovery loss. Second, on the empirical side, because this market settles on-chain---unlike the equity, futures, and volatility-benchmark venues where cash-settlement manipulation is otherwise studied---we observe trading at the level of pseudonymous wallets, which allows a more granular analysis: we identify the likely manipulators and measure the wealth they extract. }

\paragraph{Decentralized and crypto market microstructure.}
{Finally, the paper contributes to the microstructure of decentralized and crypto markets by studying a new on-chain venue: a binary, oracle-settled prediction market. One strand studies the microstructure, price discovery, and quality of crypto trading venues, both centralized and on-chain \citep{makarov2020trading,lehar2022fragmentation,klein2023discovery,aoyagi2025coexisting,easleyohara2024crypto,barbon2024quality}. A second studies liquidity provision and price discovery on decentralized exchanges and automated market makers \citep{hasbrouck2024dex,capponi2025liquidity,lehar2025decentralized,capponi2026price,hasbrouck2026fees}. A third studies value extraction through transaction ordering and front-running, and the design flaws and responses it has prompted \citep{daian2020flashboys,park2023conceptual,capponi2025mev}. A fourth concerns on-chain settlement and the price oracles that feed external data to smart contracts \citep{chiu2019settlement,cong2025oracle}. Our mechanism is nonetheless distinct from these on-chain frictions: Polymarket runs a central limit order book, and the manipulation operates on the centralized spot exchanges whose prices determine the contract's settlement, working through trading into the settlement window rather than transaction reordering or a compromised oracle.}

\section{Institutional Background}
\label{sec:institutional}

\subsection{Prediction markets and asset-price contracts}

A prediction market lets participants trade contracts that pay off on the realized outcome of a future event, so the contract's price can be read directly as the market's estimated probability that the event occurs. The established contracts settle on real-world events (elections, policy decisions, sporting results) and are valued as forecasters whose prices have matched or beaten conventional methods \citep{wolfers2004prediction,arrow2008promise}. One desirable feature stands out: the settlement event lies outside the traders' control (it is \emph{exogenous} to them). 

The fastest-growing asset-price contracts are the cryptocurrency ``up/down'' markets: a yes-or-no bet on the direction of a reference asset (most heavily Bitcoin) over a fixed horizon, paying off only if the asset ends the horizon above where it began. The horizons are short, from minutes to hours, and the contracts open and settle automatically on a rolling schedule, so a single asset generates a continuous stream of fresh binary markets through the day.\footnote{Within months of launch, Polymarket's five- and fifteen-minute crypto up/down markets traded over \$4~billion cumulatively and roughly tripled platform daily volume; sources: Polymarket and Chainlink launch announcements, 2026.} The contract is no longer confined to crypto. In August 2026 Kalshi listed fifteen-minute up/down markets on gold and silver, and CME Group lists hourly event contracts on futures traded across its exchanges.\footnote{Kalshi, ``Gold and Silver 15-minute markets. Live on Kalshi,'' post on X, August 3, 2026, \url{https://x.com/Kalshi/status/2084315859192057996}. CME Group self-certified its hourly event contracts with the CFTC on November 19, 2025; they pay \$1 and settle on a volume-weighted average price over the final 30 seconds or the final minute, depending on the exchange.}

The same structure has now reached mainstream finance. In April 2026 the U.S.\ Securities and Exchange Commission approved Nasdaq's binary options on the Nasdaq-100 index. In June 2026 Cboe began listing binary options on the S\&P~500, with a new expiration every trading day, and in July 2026 the Commission approved Cboe's proposal to extend binary options to further indexes.\footnote{Nasdaq MRX, approved April 30, 2026 (Exchange Act Release No.~34-105342); the options pay \$100 and settle on the Nasdaq Closing Cross. Cboe, ``Cboe Introduces Cboe Predicts,'' press release, June 23, 2026 (S\&P~500 binaries expiring the same or the next trading day, paying \$100 or \$0). Cboe's extension was filed on April 2, 2026 (Release No.~34-105247) and approved on July 17, 2026 (Release No.~34-105936), limited after amendment to the Cboe Magnificent~10 index.} Like the crypto up/down markets, each pays a fixed amount if a stated price condition is met at expiry and nothing otherwise. These products carry the asset-price contract onto the most heavily traded U.S.\ equity indices. Their horizons are a day or longer, not minutes; whether the settlement vulnerability we document in crypto carries over depends on the contract's horizon and its settlement rule, the two design levers we study in Section~\ref{sec:results-horizon}.

\subsection{Venues}

These contracts trade on venues that differ sharply in design. \emph{Polymarket} is a decentralized, blockchain-based market: its contracts are created and settled automatically by code that runs on a public blockchain (on-chain smart contracts), every trade is recorded on that blockchain, and positions are held in a token pegged to the U.S.\ dollar (a dollar stablecoin) \citep{tsang2026anatomy,dubach2026anatomy}. It lists the widest menu of horizons, including five-minute, fifteen-minute, and four-hour up/down contracts on Bitcoin and several other coins. \emph{Kalshi} is a centralized, federally regulated U.S.\ exchange (trades route through one operator rather than a blockchain) \citep{diercks2026kalshi,cftc2026advisory}; over our sample it lists only a fifteen-minute crypto up/down contract, with no five-minute or four-hour equivalent; its hourly and daily crypto series are price-range markets, not binary up/down. 

\subsection{Polymarket}

Each Polymarket crypto up/down contract pays \$1 if the reference asset's price at the close is at or above its price at the open, and \$0 otherwise. The contracts are created and settled programmatically on a rolling schedule, so a fresh five-minute Bitcoin contract opens every five minutes around the clock. The three horizons trade simultaneously and overlap on the clock (every fifteen-minute contract spans three consecutive five-minute contracts), a feature we exploit below.

\paragraph{Resolution.} The five-minute, fifteen-minute, and four-hour contracts settle against a price produced by a decentralized oracle (an outside service that reports a reference price by combining data from many exchanges) rather than against any single exchange. The oracle, Chainlink, aggregates trade data from many centralized and decentralized venues (including the largest spot exchanges) and reduces them to one consensus print through a multi-stage procedure~\citep{cong2025oracle,aspembitova2022oracle}: several professional aggregators each compute a volume-weighted average price (a price that weights each trade by its size) across the venues they cover, independent node operators take the middle value (the median) across aggregators, and the network publishes the middle value across nodes.\footnote{For the fast markets Polymarket reads Chainlink Data Streams, a low-latency pull-based feed, rather than the standard on-chain push feed, so the settlement value can be sampled at the exact close; see Chainlink Data Feeds documentation, \url{https://docs.chain.link/data-feeds}.} The strike against which a contract resolves is the oracle price sampled at the contract's open, and the settlement value is the oracle price at its close; the ``up'' side wins if and only if the close print is at or above the strike.\footnote{Through August 6, 2026, Polymarket set the strike and the settlement value using single snapshots of Chainlink's spot-price feeds; all results in Sections~\ref{sec:headline} and~\ref{sec:wallet_cohorts} come from this regime. Beginning August 7, 2026, it switched to Chainlink-computed TWAPs at those same boundaries. The initial lookback was 30 seconds for five-minute markets and 60 seconds for fifteen-minute and four-hour markets. On August 14, the five-minute lookback also became 60 seconds~\citep{polymarket2026twap}. Section~\ref{sec:twap} studies the effect of this change.} (The one-hour and daily contracts, which fall outside our sample, instead reference a single exchange's spot price directly.)

This aggregated construction is central to the manipulation question we study. Because the resolution price combines many exchanges into one agreed figure, moving it appears to require moving the price on many exchanges at once rather than on just one. Yet within a five-minute horizon, arbitrage (traders profiting from price gaps across exchanges) keeps the largest single venue and the combined price within a few basis points (hundredths of a percent) of each other. A spot-price move concentrated on one venue can therefore still push the combined price across the strike if the move is large enough relative to that narrow margin. Appendix~\ref{app:chainlink_validation} measures this gap and the residual disagreement directly.

\paragraph{Product rollout.} The three Bitcoin contracts launched in sequence, and that staggered introduction is the backbone of our between-period comparison. The fifteen-minute contract began trading on October 9, 2025 and the four-hour contract on October 15, 2025; the five-minute contract was created on-chain in December 2025 but sat dormant, with no trading, until its public launch on February 12, 2026.\footnote{Launch dates are identified from the on-chain trade record; see Section~\ref{sec:data}.} We therefore partition the sample into three periods: a pre-launch baseline before any of these contracts traded, an intermediate period with the fifteen-minute and four-hour contracts live, and a final period with all three (including the five-minute contract) trading together.

\section{Theoretical Setup}
\label{sec:theory_setup}
\label{sec:the-model}

Two markets operate in parallel. On a spot market a risky asset
changes hands; on a separate venue a binary contract settles on that asset's
price at the close of trading, with some noise. The settlement rule is what links them: the
contract settles on a price the spot market itself discovers, and a trader
who holds the contract can move that price by trading the asset. An uninformed
manipulator turns this into a strategy. She takes a position in the contract,
then trades the spot in the final moments before settlement to push its price
toward the side she holds.

\subsection{The Two Markets and the Contract}
\label{sec:setup-markets}

The first asset is the spot asset itself. The fundamental
value of the spot asset is a random variable $V$ with equal
probability of being \(-\sigma\) or \(+\sigma\), drawn by nature at stage~0.
$\sigma>0$ is the volatility of the fundamental value and exogenous.
One informed spot trader observes the realization.
Everyone else holds the symmetric prior. Spot positions ultimately
liquidate at the value of \(V\).

The second asset is the binary contract, a claim on the asset's direction. Crucially, it
does not settle on the realized value $V$ but on the spot market's \emph{posterior
mean} of $V$ at the end of the prediction horizon, that is, the value the spot market's own trading has revealed, perturbed
by a small independent noise. This is the link the manipulator exploits: by moving the
closing spot price, she moves the very quantity her contract settles on. Write $h$ for
the order history, the record of buys, sells and no-trades on the spot market, and $\bar V_h=\E\left[V\mid h\right]$ for the spot market's posterior mean of $V$ given
$h$ (both made precise in \cref{sec:setup-order-flow}).

For a final order history \(h\), the binary contract pays
\begin{equation}
\label{eq:contract-payoff}
    D_\varepsilon\left(h\right)
    =
    \begin{cases}
        0, & \bar V_h+\varepsilon\xi<0,\\[0.3em]
        \frac{1}{2}, & \bar V_h+\varepsilon\xi=0,\\[0.3em]
        1, & \bar V_h+\varepsilon\xi>0,
    \end{cases}
\end{equation}
where $\xi$ is a random noise term uniformly distributed on \(\left[-1,1\right]\), independent of \(V\) and all orders, and
$\varepsilon>0$ is an exogenous parameter that controls the noise level.\footnote{
In the no-noise limit \(\varepsilon \downarrow 0\),
the expected payoff jumps when $\bar V_h$ crosses zero and we might not have an
equilibrium. The noise smooths the payoff and guarantees the existence of an equilibrium.
In the large-noise limit \(\varepsilon \uparrow \infty\), the binary contract decouples
from the spot market and its expected payoff converges to \(1/2\) regardless of the
order history, which makes manipulation unprofitable.}

The payoff in~\eqref{eq:contract-payoff} makes the contract a \emph{binary} (or digital) option on the spot price: it pays out only when the settlement value \(\bar V_h+\varepsilon\xi\) clears a fixed threshold, here zero. We therefore adopt standard option terminology throughout. We call that threshold the \emph{strike}; the contract finishes \emph{in-the-money} when the settlement value is above the strike (it pays~\(1\)), \emph{out-of-the-money} when below it (it pays~\(0\)), and \emph{at-the-money} when exactly at the strike (it pays $1/2$).

Why settle on a noisy price rather than on the spot price directly? In
practice the settlement reference is not the manipulator's to dictate. The oracle
that resolves the contract reads many exchanges at once, not just the one venue
she can trade, and her push lands amid a stream of other traders' orders on those venues that she does
not control. A push therefore carries the settlement value across the strike only with
some probability, never for sure. The
term \(\varepsilon\xi\) captures exactly this residual uncertainty, and
\(\varepsilon\) sets how much of it remains.

Let $H$ denote the cumulative distribution function of \(\xi\). The expected payoff
of the binary contract given \(h\) is
\begin{equation}
\label{eq:contract-expected-payoff}
    \E\left[D_\varepsilon\left(h\right)\mid h\right]
    =
    H\left(\frac{\bar V_h}{\varepsilon}\right)
    =
    \begin{cases}
        0, & \bar V_h\le -\varepsilon,\\[0.3em]
        \frac{1}{2}+\frac{\bar V_h}{2\varepsilon},
        & -\varepsilon<\bar V_h<\varepsilon,\\[0.5em]
        1, & \bar V_h\ge \varepsilon.
    \end{cases}
\end{equation}
That is, its expected payoff is \(1\) if the final posterior mean is weakly above \(\varepsilon\),
\(0\) weakly below \(-\varepsilon\), and interpolates linearly on the band of
width \(2\varepsilon\) around zero. The noise therefore matters only near the strike: when the closing price ends more than \(\varepsilon\) away from it, the noise cannot change which side is paid. A push thus pays only when it carries the closing price into or across this band.

\subsection{Timeline}
\label{sec:setup-timeline}

The model has five types of participants (market makers on the spot and
prediction-market venues, an informed spot trader, a manipulator, and liquidity
traders) who act over five stages (\cref{fig:timeline}). Each order below is
publicly observed. At each trading stage the venue's market maker has posted bid
and ask quotes, and the arriving order executes against them. {Later, \cref{sec:setup-players}
specifies how each participant behaves strategically; here we only lay out the sequence.}

\noindent\textbf{Stage 0 (Nature):}
Draws \(V\in\left\{-\sigma,+\sigma\right\}\), observed only by the informed spot trader.

\noindent\textbf{Stage 1 (Prediction market):}
One order \(Y_1^\PM\) arrives on the prediction market, from the manipulator or a liquidity trader.

\noindent\textbf{Stage 2 (Spot market):}
One order \(Y_2\) arrives on the spot market, from the informed trader or a liquidity trader.

\noindent\textbf{Stage 3 (Spot market during settlement window):}
One order \(Y_3\) arrives on the spot market just before the close: the manipulator
if she entered at stage~1, otherwise an ordinary spot order like stage~2's.

\noindent\textbf{Stage 4 (Settlement):}
The binary contract settles on the spot market's posterior mean of $V$ plus a small noise
after stage~3 ends.
All spot positions liquidate at the realized value of \(V\).

\begin{figure}[ht]
\centering
\resizebox{\textwidth}{!}{%
\begin{tikzpicture}[font=\small]
  \draw[->,thick] (-0.6,0) -- (15.4,0);
  \node[anchor=west] at (15.4,0) {\footnotesize time};
  \foreach \x in {0,7,14}{\draw[thick] (\x,0.12) -- (\x,-0.12);}
  \foreach \x in {3.5,10.5}{\filldraw (\x,0) circle (2.4pt);}
  \node[align=center,anchor=south] at (0,0.22){\textbf{Stage 0}\\[-2pt]\footnotesize Nature};
  \node[align=center,anchor=south] at (3.5,0.22){\textbf{Stage 1}\\[-2pt]\footnotesize Prediction market};
  \node[align=center,anchor=south] at (7,0.22){\textbf{Stage 2}\\[-2pt]\footnotesize Spot};
  \node[align=center,anchor=south] at (10.5,0.22){\textbf{Stage 3}\\[-2pt]\footnotesize Spot: settlement window};
  \node[align=center,anchor=south] at (14,0.22){\textbf{Stage 4}\\[-2pt]\footnotesize Settlement};
  \node[align=center,anchor=north,text width=2.6cm] at (0,-0.28){\footnotesize draw $V\in\{-\sigma,+\sigma\}$; informed observes it};
  \node[align=center,anchor=north,text width=2.9cm] at (3.5,-0.28){\footnotesize order $Y_1^{\PM}$; manipulator or liquidity};
  \node[align=center,anchor=north,text width=2.6cm] at (7,-0.28){\footnotesize order $Y_2$; informed or liquidity};
  \node[align=center,anchor=north,text width=2.9cm] at (10.5,-0.28){\footnotesize order $Y_3$; manipulator if entered, else ordinary};
  \node[align=center,anchor=north,text width=2.6cm] at (14,-0.28){\footnotesize contract pays on $\bar V_h+\varepsilon\xi$; spot liquidates at $V$};
\end{tikzpicture}%
}
\caption{The five-stage timeline. Filled dots mark the manipulator's two decision points: entry on the prediction market (Stage~1) and the settlement-window spot trade (Stage~3).}
\label{fig:timeline}
\end{figure}

\subsection{Players}
\label{sec:setup-players}

\paragraph{Market makers.}
A competitive, risk-neutral market maker stands on each venue. Every order, on either venue, is public, so both makers see the same information, and the spot maker also sees the stage-1 contract order. A maker cannot distinguish informed or strategic orders from liquidity trades, so its quote conditions on the order alone. Competition drives its expected profit to zero: it quotes a one-unit bid and a one-unit ask equal to the asset's expected value given the order it faces. One qualification is required. The book cannot be crossed, so the ask never falls below the bid; where the break-even quotes would cross, the maker posts the tightest feasible book, a zero spread, and earns a rent.

The spot maker's stage-2 quotes are \(\left(B_2,A_2\right)\). Its stage-3 quotes depend on the stage-2 order history \(h_2\) they follow, so we write \(\left(B_{h_2},A_{h_2}\right)\); \(A_{h_2}\) is thus the stage-3 ask after \(h_2\), not a stage-2 quote. The prediction-market maker's stage-1 quotes are \(\left(B_1^\PM,A_1^\PM\right)\).

\paragraph{Informed trader.}
The informed trader knows the realized \(V\) and arrives with probability \(\pi\in\left(0,1\right)\) at each \emph{ordinary} spot stage (stage~2 always, and stage~3 whenever the manipulator has not entered). As in \citet{glosten1985bid}, she buys when the asset is worth more than the ask and sells when it is worth less:
\begin{equation}
\label{eq:informed-rule}
    \text{buy if }V>A,
    \qquad
    \text{sell if }V<B,
    \qquad
    \text{stay out otherwise},
\end{equation}
where \(\left(B,A\right)\) are the spot quotes she faces: \(\left(B_2,A_2\right)\) at stage~2 and \(\left(B_{h_2},A_{h_2}\right)\) at stage~3.\footnote{Because break-even quotes lie inside \(\left[-\sigma,\sigma\right]\), one side always pays, so she never stays out.}

\paragraph{Liquidity traders.}
Liquidity traders are uninformed and trade for reasons outside the model, buying or selling at random, each with probability one half. One sits on each venue: on the prediction market a liquidity trader trades at stage~1 whenever the manipulator does not enter, and on the spot market a liquidity trader trades at any ordinary spot stage at which the informed trader does not arrive. Either way exactly one order arrives at stage~1 and at every ordinary spot stage, so the stage-2 outcome is a buy or a sell.

\paragraph{Manipulator.}
The manipulator never observes \(V\). She profits not from information but from the contract itself: by trading the spot in the settlement window she moves the price the contract pays on. She makes two choices.

The first is whether to enter. At stage~1 she buys or sells the contract with probability \(p_{\mathrm M}\in\left[0,1\right]\); otherwise a liquidity trader trades at stage~1. The prediction-market maker sees only the order, not who sent it. After a contract buy (the sell case is symmetric), both market makers therefore believe that the buyer is the manipulator with probability \(p_{\mathrm M}\). She chooses \(p_{\mathrm M}\) before trading begins and both market makers observe it. Their quotes respond to the entry probability she picks, and she sets it anticipating that response: the entry probability is a commitment.\footnote{A fresh five-minute Bitcoin contract opens every five minutes, around the clock (\cref{sec:institutional}), and every trade is recorded on a public blockchain. Market makers therefore see the same game play out hundreds of times a day, can learn how often a manipulator stands behind a contract buy, and would eventually detect a manipulator who departed from her usual frequency. We do not model the repetition; letting the market makers observe \(p_{\mathrm M}\) stands in for it.}

The second is whether to push. If she entered, she is the only trader at stage~3.\footnote{Giving her the whole settlement window is a simplification that keeps one order per stage. Its consequence is that her entry displaces the informed trader at stage~3, a channel we return to in \cref{sec:results-price-discovery}.} Having seen the stage-2 order \(h_2\), she buys the spot with probability \(\alpha_{h_2}\in[0,1]\), moving the close toward the side of the contract she holds, and otherwise does nothing. Unlike the entry probability, the push is not a commitment: she decides it in the settlement window, when she already holds the contract, so it must be her best choice given the quotes she faces. If she did not enter, stage~3 is an ordinary spot stage, exactly like stage~2.

Her payoff is the contract's payoff minus the price she paid for it, plus, if she pushes, the value \(V\) of the spot unit she buys minus the ask she pays for it. Every order is for one unit: one contract on the prediction market, one unit of the asset on the spot market.

\subsection{Equilibrium}
\label{sec:setup-equilibrium}

\begin{definition}[Equilibrium]\label{def:full-equilibrium}
An \emph{equilibrium} is a manipulator strategy 
(her entry probability \(p_{\mathrm M}\) and, for each value of \(p_{\mathrm M}\), her push probability \(\alpha_{h_2}\)), a quoting
rule for each market maker, and a system of beliefs, taking the informed trader's rule in~\eqref{eq:informed-rule} and liquidity-trader behavior as given, such that:
\begin{enumerate}[leftmargin=2em,itemsep=0.15em]
    \item \emph{optimal entry}: she chooses \(p_{\mathrm M}\) to maximize her ex-ante expected profit, anticipating the quotes, beliefs and stage-3 play that each value of \(p_{\mathrm M}\) induces;
    \item \emph{sequential rationality at stage~3}: at each \(p_{\mathrm M}\), her stage-3 trade maximizes her expected profit given the makers' quotes and the beliefs;
    \item \emph{zero-profit quoting}: each maker's bid and ask equal the expected value
    conditional on the order it faces, subject to a non-crossing book (the ask at least the
    bid), so that where the unconstrained break-even quotes would cross, the maker quotes a zero
    spread and earns the resulting rent;
    \item \emph{belief consistency}: beliefs are formed from the strategies by Bayes' rule on
    every information set reached with positive probability.
\end{enumerate}
Quoting and belief consistency hold at every value of \(p_{\mathrm M}\), not only at the one she chooses. Where the stage-3 condition leaves more than one push probability \(\alpha_{h_2}\) consistent with
the quotes, we select the smallest one that is stable to small perturbations of the conjecture,
following \citet{zhu2014dark}.
\end{definition}

\phantomsection\label{sec:setup-canonical}
We study the equilibrium under the following assumption.
\begin{assumption}[Small-noise region]\label{ass:canonical}
$\varepsilon<\sigma\pi<\tfrac12$.
\end{assumption}
The first inequality says that settlement noise is small relative to the price move from one order. The second says that a push is worth its cost.

\subsection{Order Histories and Notation}
\label{sec:setup-order-flow}

We describe what has happened on the spot market by its \emph{order history}. Each spot order has one of three outcomes: a buy, written \(\hbuy\), a sell, written \(\hsell\), or no trade, written \(\hzero\). An order history \(h\) lists these outcomes in sequence. After stage~2 the history is a single symbol \(h_2\), either \(\hbuy\) or \(\hsell\), because some order always arrives at an ordinary spot stage. After the settlement window it is a pair \(h_2h_3\). For example, \(\hsell\hbuy\) is a sell at stage~2 followed by a buy in the settlement window, and \(\hsell\hzero\) is a sell followed by no trade, which occurs only when the manipulator has entered and chooses not to push.

Each history carries a belief about the asset's value. Let \(\ppost_h=\Prb(V=+\sigma\mid h)\) be the probability that the value is high given the history \(h\). Because \(V\) takes only the two values \(\pm\sigma\), this one number fixes the posterior mean and the posterior variance:
\[
    \bar V_h=\sigma\left(2\ppost_h-1\right),
    \qquad
    \rvar_h=\sigma^2\bigl(1-(2\ppost_h-1)^2\bigr)=4\sigma^2\ppost_h\left(1-\ppost_h\right).
\]
For example, after a single buy the posterior mean is \(\bar V_{\hbuy}=\sigma\pi\), and after a single sell it is \(\bar V_{\hsell}=-\sigma\pi\) (\cref{lem:stage2-quotes}).

Finally, two probabilities recur. We write \(q_{h_2}\) for the probability that the stage-2 history is \(h_2\), and \(q_{h_2h_3}\) for the probability that the settlement window then produces \(h_3\), given \(h_2\). Their product \(q_{h_2}q_{h_2h_3}\) is the probability of the full history \(h_2h_3\).

\section{Theoretical Results}
\label{sec:theory_results}

This section characterizes the equilibrium. \Cref{sec:results-stage-quotes} solves the manipulator's settlement-window push for a fixed entry probability. \Cref{sec:results-price-discovery} shows that manipulator participation reduces price discovery in the spot market at settlement. \Cref{sec:results-canonical} endogenizes entry and characterizes the resulting transfer and its effect on spot liquidity. The remainder studies the contract horizon as a design lever, first with a second ordinary round and then with $n$ rounds (\cref{sec:results-horizon-n}). All proofs are deferred to the appendix.

\subsection{The Manipulator's Settlement-Window Push}
\label{sec:results-stage-quotes}
\label{sec:stage3-settlement-window-continuation-equilibrium}

This subsection characterizes the manipulator's behavior in the settlement window. At stage~3 she holds the contract she bought at stage~1; we take her position to be long, the short case being symmetric. Her only decision is whether to \emph{push}, that is, to buy the spot asset so as to raise the close, and with what probability (\cref{fig:stage3-tree}). A push has both a payoff and a price: it raises the closing posterior mean, and with it her contract's expected payoff, but she must pay a spread.

\begin{figure}[ht]
    \centering
    \begin{tikzpicture}[font=\small,
            state/.style={rectangle,rounded corners,draw,align=center,inner sep=3pt},
            leaf/.style={rectangle,rounded corners,draw,align=center,inner sep=3pt,fill=gray!6},
            elab/.style={font=\footnotesize,inner sep=1.5pt}]
        \node[state] (root) at (0,3) {Stage-2 spot order $Y_2$};
        \node[state] (buy)  at (-3.7,1.4) {spot \textbf{buy} ($h_2=\hbuy$)\\[-1pt]$\bar V_{\hbuy}=+\sigma\pi$};
        \node[state] (sell) at ( 3.7,1.4) {spot \textbf{sell} ($h_2=\hsell$)\\[-1pt]$\bar V_{\hsell}=-\sigma\pi$};
        \node[leaf] (bp) at (-5.2,-0.7) {push\\[-1pt]$\alpha_{\hbuy}$};
        \node[leaf] (bo) at (-2.4,-0.7) {stay out\\[-1pt]$1-\alpha_{\hbuy}$};
        \node[leaf] (sp) at ( 2.4,-0.7) {push\\[-1pt]$\alpha_{\hsell}$};
        \node[leaf] (so) at ( 5.2,-0.7) {stay out\\[-1pt]$1-\alpha_{\hsell}$};
        \draw[->] (root) -- (buy)  node[elab,midway,above left] {$\tfrac12$};
        \draw[->] (root) -- (sell) node[elab,midway,above right]{$\tfrac12$};
        \draw[->] (buy)  -- (bp);
        \draw[->] (buy)  -- (bo);
        \draw[->] (sell) -- (sp);
        \draw[->] (sell) -- (so);
    \end{tikzpicture}
    \caption{The stage-3 continuation game, conditional on a stage-1 buy on the prediction market. The stage-2
        spot order produces an up posterior mean ($h_2=\hbuy$) or a down posterior mean ($h_2=\hsell$), each
        with probability $\tfrac12$. If the manipulator entered at stage~1 (probability
        $p_{\mathrm M}$), she then either pushes the spot with the push probability $\alpha_{h_2}$
        derived below or stays out; otherwise the settlement-window order is an ordinary spot order
        (an informed trader or a liquidity trader), exactly as at stage~2.}
    \label{fig:stage3-tree}
\end{figure}

We hold her entry probability \(p_{\mathrm M}\in\left(0,1\right)\) fixed\footnote{If \(p_{\mathrm M}=0\), she is absent from stage~3 and there is no push. If \(p_{\mathrm M}=1\), the only stage-3 trader is the manipulator, who is uninformed, so her buy reveals nothing and the ask equals the posterior mean \(\bar V_{h_2}\) whatever her push probability.} and solve for a \emph{continuation equilibrium} of this stage: stage-3 spot quotes and a push probability \(\alpha\) such that the quotes break even given the push the spot maker expects, the push is optimal for the manipulator at those quotes, and the maker's beliefs follow Bayes' rule. \Cref{sec:results-canonical} then lets her choose \(p_{\mathrm M}\).

The spot maker's quotes and the manipulator's push depend on each other. We therefore first derive the quotes when the spot maker expects a given push probability, then the manipulator's best response to those quotes, and finally the push probability at which her choice matches what the spot maker expects. To anticipate \cref{prop:push-policy}: she never pushes after a stage-2 buy; after a stage-2 sell she pushes with certainty when her entry probability is low and with a probability strictly between zero and one when it is high.

The starting point is the spot maker's beliefs at the start of the settlement window. They are determined by stage~2, an ordinary Glosten--Milgrom round in which the order comes from the informed trader with probability~\(\pi\).

\begin{lemma}[Stage-2 quotes and posteriors]\label{lem:stage2-quotes}
    The stage-2 break-even quotes are the ask $A_2=\sigma\pi$ and the bid
    $B_2=-\sigma\pi$. Each outcome $h_2\in\left\{\hbuy,\hsell\right\}$ occurs with probability
    $q_{h_2}=\tfrac12$ and leaves the spot posterior
    \[
        \ppost_{h_2}=\frac{1\pm\pi}{2},\qquad
        \bar V_{h_2}=\pm\,\sigma\pi,\qquad
        \rvar_{h_2}=\sigma^2\left(1-\pi^2\right),
    \]
    with the plus sign after a buy and the minus sign after a sell. The posterior mean
    $\bar V_{h_2}$ is the spot price after the stage-2 order, and the residual variance
    $\rvar_{h_2}$ is the uncertainty about the asset that remains.
\end{lemma}

The spot maker enters stage~3 with the belief in \cref{lem:stage2-quotes}. Given this belief, consider the manipulator's strategy. She never sells the spot, because a sale would move the close downward, against the long contract she holds. Her stage-3 strategy therefore reduces to a single number: the probability \(\alpha_{h_2}\) with which she pushes, that is, buys the spot, after the stage-2 order \(h_2\).

Now consider the spot maker's stage-3 quotes when it expects the manipulator to push with probability \(\alpha_{h_2}\). The maker cannot see who sends the order. A buy may come from the informed trader, from a liquidity trader, or from the manipulator, and only the informed trader's buy carries information about the asset. The break-even ask therefore averages the asset's value over these possible buyers.

\begin{lemma}[Stage-3 quotes]\label{prop:stage3-closed-form-quotes}
    Fix the stage-2 outcome \(h_2\in\left\{\hbuy,\hsell\right\}\), an interior entry probability \(p_{\mathrm M}\in\left(0,1\right)\), and a conjectured stage-3 push probability \(\alpha_{h_2}\in\left[0,1\right]\). The unique break-even stage-3 quotes are
    \[
        A_{\hbuy}\left(\alpha_{\hbuy};p_{\mathrm M}\right)
        =
        \frac{\sigma\pi\left[p_{\mathrm M}\alpha_{\hbuy}+\left(1-p_{\mathrm M}\right)\right]}
        {p_{\mathrm M}\alpha_{\hbuy}+\left(1-p_{\mathrm M}\right)\frac{1+\pi^2}{2}},
        \qquad
        A_{\hsell}\left(\alpha_{\hsell};p_{\mathrm M}\right)
        =
        \frac{-\sigma\pi\, p_{\mathrm M}\alpha_{\hsell}}
        {p_{\mathrm M}\alpha_{\hsell}+\left(1-p_{\mathrm M}\right)\frac{1-\pi^2}{2}} ,
    \]
    \[
        B_{\hbuy}=0,
        \qquad
        B_{\hsell}=\frac{-2\sigma\pi}{1+\pi^2} .
    \]
\end{lemma}

When no push is expected (\(\alpha_{h_2}=0\)), the ask is the ordinary Glosten--Milgrom ask: \(2\sigma\pi/\left(1+\pi^2\right)\) after a stage-2 buy and \(0\), the strike, after a stage-2 sell. The bids do not depend on \(\alpha_{h_2}\), since the manipulator never sells. Because each quote equals the asset's expected value given the order, the posterior mean after a stage-3 buy is the ask. A push therefore moves the posterior mean on which the contract settles from \(\bar V_{h_2}\) to the ask. The lower the ask, the less effective the push: the spot market is more liquid, so her buy moves the price by less. How low the ask is depends on how often a push is expected.

\begin{corollary}[Stage-3 quote monotonicity]\label{cor:stage3-pm-buy-quote-monotonicity}
    After either stage-2 order, the stage-3 ask is strictly decreasing in the push probability \(\alpha\) that the spot maker expects, and weakly decreasing in the entry probability \(p_{\mathrm M}\), strictly so when \(\alpha>0\). The bid and the posterior mean after no trade do not depend on either.
\end{corollary}

The more often the manipulator is expected to push, the less likely a given buy is to come from the informed trader. A buy then carries less adverse selection, and the spot maker lowers its ask. For the manipulator this cuts both ways: an expected push is cheaper, but it also moves the price by less. Whether she pushes depends on the balance of the two, to which we now turn.

When does a push pay? Buying the spot raises her contract's expected payoff from \(H\left(\bar V_{h_2}/\varepsilon\right)\) to \(H\left(A_{h_2}/\varepsilon\right)\), and it costs her the \emph{push cost} \(A_{h_2}-\bar V_{h_2}\), the premium of the ask over the posterior mean. The difference is the \emph{best-response gap}:
\[
    \begin{aligned}
        \Gamma_{h_2}\left(\alpha_{h_2}\right)
         & =
        \underbrace{H\left(\frac{A_{h_2}\left(\alpha_{h_2};p_{\mathrm M}\right)}{\varepsilon}\right)
            -
            \left[A_{h_2}\left(\alpha_{h_2};p_{\mathrm M}\right)-\bar V_{h_2}\right]}_{\substack{\text{expected payoff} \\ \text{from buying}}}
        -
        \underbrace{H\left(\frac{\bar V_{h_2}}{\varepsilon}\right)}_{\substack{\text{expected payoff} \\ \text{from staying out}}} \\
         & =
        \underbrace{H\left(\frac{A_{h_2}\left(\alpha_{h_2};p_{\mathrm M}\right)}{\varepsilon}\right)
            -
            H\left(\frac{\bar V_{h_2}}{\varepsilon}\right)}_{\substack{\text{increase in expected} \\ \text{contract payoff from buying}}}
        -
        \underbrace{\vphantom{H\left(\frac{A_{h_2}\left(\alpha_{h_2};p_{\mathrm M}\right)}{\varepsilon}\right)
                -
                H\left(\frac{\bar V_{h_2}}{\varepsilon}\right)}\left[A_{h_2}\left(\alpha_{h_2};p_{\mathrm M}\right)-\bar V_{h_2}\right]}_{\text{push cost}}.
    \end{aligned}
\]
She pushes when the gap is positive and stays out when it is negative. The gap depends on the push probability the spot maker expects, because that probability determines the ask. The price she paid for the contract at stage~1 is sunk and does not enter.

In equilibrium the spot maker's expectation must be correct: the push probability the manipulator chooses must equal the one the spot maker priced against. Because an expected push lowers both the gain and the cost, more than one push probability can satisfy this condition; following \citet{zhu2014dark}, we select the smallest stable one.\footnote{A push probability is stable if a small error in the spot maker's expectation corrects itself. \Cref{app:selected-continuation-existence} gives the formal definition and the proof of existence, and \cref{app:selected-continuation-equilibrium-by-path} gives the closed-form expression for the cutoff \(p_0\).} We write \(\alpha_{h_2}^{*}\) for the selected push probability and \(A_{h_2}^{*}=A_{h_2}(\alpha_{h_2}^{*};p_{\mathrm M})\) for the stage-3 ask at that probability.

\begin{proposition}[The manipulator's push policy]\label{prop:push-policy}
    Under \cref{ass:canonical}, a selected continuation equilibrium exists for every entry probability \(p_{\mathrm M}\in\left(0,1\right)\). In it, the manipulator never pushes after a stage-2 buy. After a stage-2 sell, there is a cutoff \(p_0\in\left(0,1\right)\) such that she pushes with certainty when \(p_{\mathrm M}\le p_0\) and with probability strictly between zero and one when \(p_{\mathrm M}>p_0\):
    \[
        \alpha_{\hbuy}^{*}=0,
        \qquad
        \alpha_{\hsell}^{*}\left(p_{\mathrm M}\right)
        =
        \begin{cases}
            1, & 0<p_{\mathrm M}\le p_0,\\[0.5em]
            \displaystyle
            \frac{1-p_{\mathrm M}}{p_{\mathrm M}}
            \frac{\varepsilon\left(1-\pi^2\right)\left(1-2\sigma\pi\right)}
            {2\left(\sigma\pi-\varepsilon\right)},
            & p_0<p_{\mathrm M}<1.
        \end{cases}
    \]
\end{proposition}

The intuition is simple. If the price has already moved her way, she has no reason to push: the contract pays in any case, and a push would only cost her. If the price has moved against her, she has an incentive to push. She does not, however, push without limit: the more often a push is expected, the less effective it is, because the spot maker discounts a late buy. When she enters rarely, a push remains effective and she always pushes; when she enters often, she pushes only part of the time, up to the point at which a further push no longer pays.

\subsection{Price Discovery Falls}
\label{sec:results-price-discovery}
\label{sec:price-discovery}

Does the manipulator's presence change how much the closing spot price reveals about the asset? We measure what the price reveals by the uncertainty about \(V\) that remains once stage~3 closes, the \emph{expected residual variance}:
\[
    \E\left[\rvar_{h_2 h_3}\right]=
    \sum_{h_2}q_{h_2}\sum_{h_3}q_{h_2h_3}\,\rvar_{h_2h_3},
\]
where \(q_{h_2h_3}\) is the probability of \(h_3\) given \(h_2\). A price that tracks \(V\) closely leaves little residual variance, so a higher residual variance means a less informative price.\footnote{By the law of total variance, the residual variance and the variance of the closing posterior mean sum to the asset's prior variance, \(\sigma^2=\E\left[\rvar_{h_2 h_3}\right]+\operatorname{Var}\left(\bar V_{h_2 h_3}\right)\).} The benchmark is the market without the manipulator, \(p_{\mathrm M}=0\).\footnote{From the spot market's perspective, setting \(p_{\mathrm M}=0\) is equivalent to shutting down the prediction market, since it is the only interaction channel between the two.}

\begin{proposition}[Manipulator participation strictly harms price discovery]\label{prop:price-discovery-sign}
    With informed traders present (\(\pi>0\)), for every interior entry probability \(p_{\mathrm M}\in\left(0,1\right)\) and every conjectured profile of stage-3 push probabilities \(\left(\alpha_{h_2}\right)_{h_2\in\left\{\hbuy,\hsell\right\}}\),
    \[
        \E\left[\rvar_{h_2 h_3}\right]\left(p_{\mathrm M}\right)>
        \E\left[\rvar_{h_2 h_3}\right]\left(0\right) .
    \]
\end{proposition}

The loss runs through two channels. The first is displacement: with probability $p_{\mathrm M}$ the manipulator is the only trader in the settlement window, which reduces the informed order flow by the factor $1-p_{\mathrm M}$, and this alone makes the loss strict even when she never pushes ($\alpha_{h_2}=0$). The second is dilution: once she does push ($\alpha_{h_2}>0$), her uninformed buys pool with the informed trader's, so a buy reveals less about the asset. The benchmark $p_{\mathrm M}=0$ shuts down both.

This is the mirror image of Zhu's dark-pool result. In \citet{zhu2014dark}, the
second venue draws relatively uninformed flow away from the exchange, concentrating informed
traders there and sharpening the exchange price. Here the linked contract feeds uninformed,
settlement-motivated flow into the price-setting market, so the closing posterior mean
becomes less informative. A derivative that settles on a market's price can therefore
degrade the informativeness of the price it references.

\subsection{Equilibrium and the Liquidity Trader's Loss}
\label{sec:results-canonical}

Who pays for manipulation? One would expect the answer to show in the spread: a maker who may be selling to a manipulator widens its spread, and liquidity traders pay it. Surprisingly, no such spread appears. The prediction market posts a \emph{zero} bid-ask spread, and the spot market becomes more liquid, not less. Judged by liquidity, manipulation harms no one. Yet the liquidity trader who buys the contract is harmed, and precisely because of this liquidity: the spot maker, suspecting a push, responds less to a late buy, so the closing price recovers less after an early sell and the contract pays less often. More surprisingly still, the prediction-market maker gains, earning a rent while charging no spread. This subsection solves for the equilibrium and shows why.

What is the contract worth to each of its two possible buyers? We consider a buyer of the ``up'' contract; the ``down'' contract is the mirror case. Fix the entry probability \(p_{\mathrm M}\). After the stage-2 order \(h_2\), the manipulator's expected contract payoff is
\begin{equation}\label{eq:payoff-manipulator}
    \Pi_{\mathrm M}\left(p_{\mathrm M},h_2\right)
    =
    \alpha_{h_2}^{*}
    H\left(\frac{A_{h_2}^{*}}{\varepsilon}\right)
    +
    \left[1-\alpha_{h_2}^{*}\right]
    H\left(\frac{\bar V_{h_2}}{\varepsilon}\right) .
\end{equation}
The manipulator is paid on the ask when she pushes and on the stage-2 posterior mean when she does not. The liquidity trader cannot push. The liquidity trader's contract is paid on the ask if a buy arrives in the settlement window, and on the bid if a sell arrives:
\begin{equation}\label{eq:payoff-liquidity}
    \Pi_{\mathrm L}\left(p_{\mathrm M},h_2\right)
    =
    \left[\frac{1-\pi}{2}+\pi\,\ppost_{h_2}\right]
    H\left(\frac{A_{h_2}^{*}}{\varepsilon}\right)
    +
    \left[\frac{1-\pi}{2}+\pi\left(1-\ppost_{h_2}\right)\right]
    H\left(\frac{B_{h_2}}{\varepsilon}\right),
\end{equation}
where \(\ppost_{h_2}\) is the stage-2 posterior probability of the up state. The two buyers are thus paid on the same closing prices. The only difference is that the manipulator can make the buy happen, whereas for the liquidity trader a buy arrives only by chance. We write \(\Pi_{\mathrm M}\left(p_{\mathrm M}\right)\) and \(\Pi_{\mathrm L}\left(p_{\mathrm M}\right)\) for the two payoffs averaged over the stage-2 order.

The prediction-market maker cannot tell the two buyers apart, so it prices the contract to break even across them:
\[
    A_1^{\PM,\mathrm{BE}}\left(p_{\mathrm M}\right)
    =
    p_{\mathrm M}\Pi_{\mathrm M}\left(p_{\mathrm M}\right)
    +
    \left(1-p_{\mathrm M}\right)\Pi_{\mathrm L}\left(p_{\mathrm M}\right).
\]
If this price is below one half, the maker posts one half. An ask below one half would lie below the bid, and a limit order book cannot be crossed:
\[
    A_1^{\PM *}\left(p_{\mathrm M}\right)
    =
    \max\left\{A_1^{\PM,\mathrm{BE}}\left(p_{\mathrm M}\right),\ \tfrac{1}{2}\right\},
    \qquad
    B_1^{\PM *}\left(p_{\mathrm M}\right)=1-A_1^{\PM *}\left(p_{\mathrm M}\right).
\]
Both the break-even price and the two payoffs depend on how often the manipulator enters, to which we now turn.

The manipulator chooses how often to enter. Both market makers observe her entry probability (\cref{sec:setup-players}), so her choice changes the prices she then faces: the more often she enters, the lower the spot ask after a stage-2 sell, which makes each push cheaper but also less effective. Her profit per entry is her expected contract payoff net of the expected push cost and the contract price,
\[
    S_{\mathrm M}\left(p_{\mathrm M}\right)
    =
    \left[\Pi_{\mathrm M}\left(p_{\mathrm M}\right)-K\left(p_{\mathrm M}\right)\right]
    -
    A_1^{\PM *}\left(p_{\mathrm M}\right) ,
\]
where \(K\left(p_{\mathrm M}\right)\) is the expected push cost, the premium of the ask over the posterior mean weighted by the probability that she pushes. She chooses the entry probability that maximizes her expected profit, the entry probability times the profit per entry:\footnote{When the maximizer is not unique, we select the smallest. Her observable entry probability is what allows her to profit: a manipulator who took the quotes as given would enter until her profit per entry was zero (\cref{app:optimal-manipulator-entry-probability}).}
\[
    \bar{S}_{\mathrm M}\left(p_{\mathrm M}\right)
    =
    p_{\mathrm M}\,S_{\mathrm M}\left(p_{\mathrm M}\right),
    \qquad
    p_{\mathrm M}^*
    \in
    \argmax_{p_{\mathrm M}\in\left[0,1\right]}
    \bar{S}_{\mathrm M}\left(p_{\mathrm M}\right).
\]
The equilibrium thus determines together how often the manipulator enters, when she pushes, and the price of the contract. The following theorem characterizes it, where \(p_0\) is the cutoff of \cref{prop:push-policy}.

\begin{theorem}[Equilibrium]\label{prop:canonical-equilibrium}
    Fix \(\sigma>0\), \(\pi\in\left(0,1\right)\), and \(\varepsilon>0\) satisfying \cref{ass:canonical}. The equilibrium of \cref{def:full-equilibrium} has:
    \begin{enumerate}[leftmargin=2em,itemsep=0.15em]
        \item optimal entry probability
              \[
                  p_{\mathrm M}^{*}
                  =
                  \sqrt{\left(\frac{1-\pi^2}{1+\pi^2}\right)^2
                      +\frac{1-\pi^2}{1+\pi^2}p_0}
                  -\frac{1-\pi^2}{1+\pi^2}
                  \in \left(0, p_0\right);
              \]
        \item push policy
              \[
                  \alpha_{\hbuy}^{*}=0,
                  \qquad
                  \alpha_{\hsell}^{*}=1;
              \]
        \item strictly positive expected profit
              \[
                  \bar{S}_{\mathrm M}\left(p_{\mathrm M}^{*}\right)
                  =
                  \frac{\left(1-2\sigma\pi\right)\left(p_{\mathrm M}^{*}\right)^{2}}{4p_0}
                  >0;
              \]
        \item zero-spread quotes for the prediction market:
              \[
                  A_1^{\PM{*}}=B_1^{\PM{*}}=\frac{1}{2}.
              \]
    \end{enumerate}
\end{theorem}

The manipulator profits, and yet she does not enter every time. The reason is that her own presence weakens her push. The more often she enters, the more the spot maker expects a late buy to be hers, and the lower it sets the ask after a stage-2 sell. A lower ask makes each push cheaper, but it also lifts the closing price by less, and the second effect dominates. She therefore holds back. She enters with a probability below the cutoff \(p_0\), low enough that a push after a stage-2 sell always pays, which is why she then pushes with certainty.

The prediction market, by contrast, shows a zero spread. The prediction-market maker sells to a manipulator, charges nothing for it, and does not lose. Who, then, pays?

\emph{Why the liquidity trader loses.} The liquidity trader who bought the contract pays, and the payoffs \eqref{eq:payoff-manipulator} and \eqref{eq:payoff-liquidity} show why. The manipulator decides where the contract settles: after a stage-2 sell she pushes, so it settles on the ask for certain. The liquidity trader decides nothing: the contract settles on the ask only if a buy happens to arrive, and on the bid otherwise. This difference makes the contract worth more to the manipulator. By itself, however, it would not make the liquidity trader lose. The liquidity trader loses because the manipulator's presence lowers the ask.

\begin{corollary}[The liquidity trader's expected payoff]\label{cor:liquidity-payoff}
    In the equilibrium of \cref{prop:canonical-equilibrium}, the liquidity trader's expected contract payoff is
    \[
        \Pi_{\mathrm L}\left(p_{\mathrm M}^{*}\right)
        =
        \frac{3+\pi^2}{8}
        +
        \frac{1-\pi^2}{4}\,
        H\left(\frac{A_{\hsell}^{*}\left(p_{\mathrm M}^{*}\right)}{\varepsilon}\right)
        <\frac{1}{2},
    \]
    where \(A_{\hsell}^{*}\left(p_{\mathrm M}^{*}\right)<0\) is the stage-3 ask after a stage-2 sell. Without the manipulator, this ask equals the strike and the payoff equals one half.
\end{corollary}

The corollary isolates the source of the loss. The first term comes from a stage-2 buy, where the manipulator has no effect. The second comes from a stage-2 sell: the contract then pays only if a buy arrives, and it settles on the ask. The manipulator enters the liquidity trader's payoff only through that ask. The price of the contract is one half in either case, so the whole difference is the liquidity trader's loss.

\Cref{fig:cross-market-channel} follows the one order history on which this happens: an early spot sell, then an ordinary buy at the close. Without the manipulator (panel a), the sell takes the price below the strike and the late buy brings it back. With the manipulator (panel b), the spot maker cannot tell this late buy from a push, so it sells at a lower ask, and the price ends below the strike. The other three histories are the same in both panels. Yet this one change is enough: the liquidity trader still pays one half, but in the example of the figure the contract now pays with probability $0.47$, not $0.50$.

\begin{figure}[ht]
    \centering
    \includegraphics[width=\linewidth]{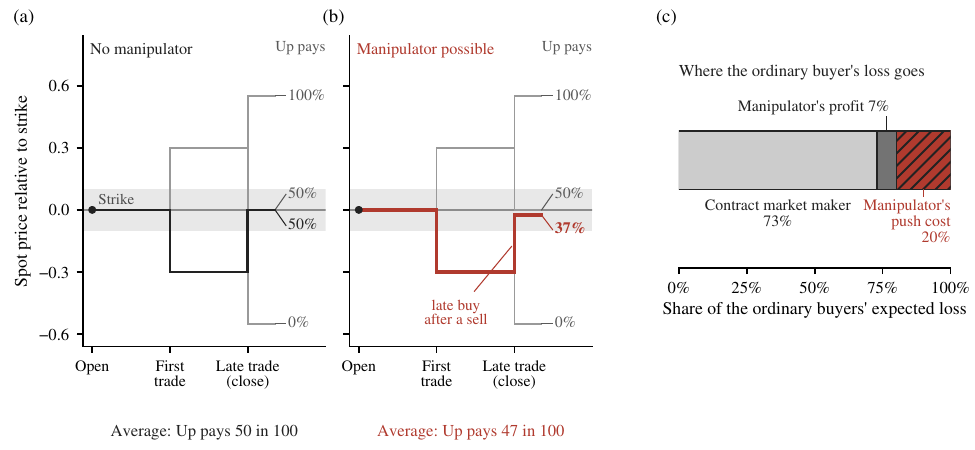}
    \caption{\textbf{How an ordinary contract buyer loses, and who receives her loss.} Panels (a) and (b) plot the spot price relative to the strike along the four possible order histories, with the probability that the ``up'' contract pays at the end of each; the shaded band is the region in which settlement noise leaves the outcome uncertain. In panel (a) there is no manipulator. In panel (b) any contract buyer may be the manipulator, so the spot maker discounts a late buy after a sell (\cref{cor:stage3-pm-buy-quote-monotonicity}) and that history ends below the strike; the other three histories are unchanged. The buyer pays one half for the contract, but her probability of being paid falls from $0.50$ to $0.47$. Panel (c) splits the liquidity traders' expected loss into the prediction-market maker's rent, the manipulator's profit, and the manipulator's expected push cost (\cref{prop:pm-maker-rent}). Parameters: $\sigma=1$, $\pi=0.3$, $\varepsilon=0.1$.}
    \label{fig:cross-market-channel}
\end{figure}

The spot buyer thus gets a better price on precisely the trade that leaves the contract less likely to pay. The lower ask that makes the spot market more liquid is the ask in \cref{cor:liquidity-payoff}: the more liquid the spot market, the less the contract is worth to the liquidity trader. The liquidity trader need not trade with the manipulator, hold any view about the asset, or even know that the spot market exists. Holding the contract is enough.

This also explains the zero spread. The manipulator pays the same one half, and because she can make the late buy herself, the contract is worth more than one half to her. The liquidity trader's shortfall is larger than the manipulator's advantage, however, so across the two buyers the contract is worth less than one half. The prediction-market maker would sell it for less if it could. It cannot quote below one half, and so it keeps the difference. The following proposition gives the accounting.

\begin{proposition}[Liquidity-trader loss decomposition]\label{prop:pm-maker-rent}\label{cor:contract-payoff-wedge}
    At the equilibrium of \cref{prop:canonical-equilibrium}, evaluated at the interior entry probability \(p_{\mathrm M}^{*}\in\left(0,p_0\right)\), the liquidity trader's ex-ante loss splits exactly into three components:
    \[
        \left(1-p_{\mathrm M}^{*}\right)L\left(p_{\mathrm M}^{*}\right)
        =
        R^{*}
        +
        p_{\mathrm M}^{*}S_{\mathrm M}\left(p_{\mathrm M}^{*}\right)
        +
        p_{\mathrm M}^{*}K\left(p_{\mathrm M}^{*}\right),
    \]
    where \(R^{*}\coloneq A_1^{\PM *}-A_1^{\PM,\mathrm{BE}}\left(p_{\mathrm M}^{*}\right)>0\) is the prediction-market maker's rent per contract, \(L\left(p_{\mathrm M}^{*}\right)\coloneq \tfrac{1}{2}-\Pi_{\mathrm L}\left(p_{\mathrm M}^{*}\right)>0\) is the liquidity trader's loss per contract, and \(S_{\mathrm M}\) and \(K\) are the manipulator's profit per entry and expected push cost.
\end{proposition}

The identity shows where the liquidity trader's loss goes (\cref{fig:cross-market-channel}, panel c). Part of it is the prediction-market maker's rent, the difference between one half and the break-even price. Part is the manipulator's profit. The remainder is the manipulator's push cost, which is paid in the spot market, where it helps the spot maker cover its losses to the informed trader.\footnote{The prediction-market maker's share exists only when the break-even price is below one half, as in the baseline. With an even number of ordinary rounds it vanishes (\cref{sec:results-horizon}).}

When does manipulation pay most? The following corollary gives the answer.

\begin{corollary}[Comparative statics]\label{lem:freq-magnitude}
    The signs of the comparative statics are summarized in the table below.
    \begin{center}
        \renewcommand{\arraystretch}{1.4}
        \begin{tabular}{@{}llccc@{}}
            \toprule
            Quantity                                        & Symbol                                                           & \(\partial/\partial\varepsilon\) & \(\partial/\partial\sigma\) & \(\partial/\partial\pi\) \\
            \midrule
            Pure-push cutoff                                & \(p_0\)                                                          & \(>0\)                           & \(<0\)                      & \(<0\)                   \\
            Optimal entry probability                       & \(p_{\mathrm M}^{*}\)                                            & \(>0\)                           & \(<0\)                      & \(<0\)                   \\
            Stage-3 ask                                     & \(A_{\hsell}^{*}\left(p_{\mathrm M}^{*}\right)\)                 & \(<0\)                           & \(>0\)                      & \(>0\)                   \\
            Push cost at \(p_{\mathrm M}^{*}\) & \(A_{\hsell}^{*}\left(p_{\mathrm M}^{*}\right)-\bar V_{\hsell}\) & \(<0\)                           & \(>0\)                      & \(>0\)                   \\
            Ex-ante profit                     & \(\bar{S}_{\mathrm M}\left(p_{\mathrm M}^{*}\right)\)        & \(>0\)                           & \(<0\)                      & \(<0\)                   \\
            \bottomrule
        \end{tabular}
    \end{center}
\end{corollary}

Manipulation pays most when a push can still decide the contract after an early sell. Greater settlement noise $\varepsilon$ widens the range of prices over which the outcome is uncertain, so the manipulator enters more often. A late buy after a sell is then more likely to be hers, which lowers the ask and the cost of a push, and her expected profit rises. Greater volatility $\sigma$ or a higher probability $\pi$ of an informed order works the other way. Each takes the price further below the strike after an early sell, so a push costs more, the manipulator enters less often, and her expected profit falls.

\subsection{The Contract Horizon and Market Design}

\label{sec:results-horizon}
\label{sec:results-horizon-two}
\label{sec:results-horizon-n}

The harm of the previous subsection is not inevitable. It arises because the contract settles after a single round of trading, when the price is still close enough to the strike for one push to decide the outcome. A venue controls this directly, through the contract's \emph{horizon}, the span of trading before the contract settles. With a longer horizon, more orders reach the price before the close and carry it away from the strike, so a push is less likely to matter. We show that every additional round of trading makes the closing price more informative, and that as the horizon grows, the probability that a push is pivotal falls to zero, and the manipulator's profit with it. The horizon is thus a design lever against manipulation.

We lengthen the contract's horizon by letting stage~2 consist of $n\ge 1$ ordinary trading rounds, between the manipulator's entry and the settlement window, which stays a single stage. Each round is an ordinary Glosten--Milgrom round: an informed trader arrives with probability $\pi$ and trades on her knowledge of $V$ as in~\eqref{eq:informed-rule}, and otherwise a liquidity trader buys or sells at random, so each round adds one buy $\hbuy$ or one sell $\hsell$ to the public record. The manipulator still trades only in the settlement window. The baseline is the one-round case, $n=1$. A superscript $(n)$ marks objects of the $n$-round game, so $p_{\mathrm M}^{*(1)}$ is the $p_{\mathrm M}^{*}$ of \cref{prop:canonical-equilibrium}.

A stage-2 order history with $n$ rounds is denoted by a string of $n$ buys or sells, $h_2=h_{2.1}\cdots h_{2.n}$. What the market learns from it is summarized by a single number, the \emph{order imbalance} $d\left(h_2\right)$, the number of buys minus the number of sells. The imbalance is a sufficient statistic for the posterior entering settlement: two order histories with the same imbalance leave the market at the same posterior, so a stage-2 order history can be summarized by \(d\) alone.

A push moves the closing price by a bounded amount, so it can change which side wins only when the price is already near the strike, that is, when the imbalance is near balance. Once the order flow has carried the imbalance far to one side, the outcome is decided and the manipulator does not push. The horizon works by making a near-balance imbalance rarer. \Cref{fig:n-stage-timeline} lays out the timeline.

\begin{figure}[ht]
    \centering
    \begin{tikzpicture}[font=\small,>=stealth,x=1cm,y=1cm]
        \draw[->,thick] (0,0) -- (12.6,0) node[right]{time};
        \foreach \x in {1,3.5,4.5,5.5,7.5,10.5} \filldraw (\x,0) circle (1.7pt);
        \node at (6.5,0) {$\cdots$};
        \node[above=3pt,align=center] at (1,0) {\textbf{Stage 1}};
        \node[above=2pt] at (3.5,0) {\footnotesize $2.1$};
        \node[above=2pt] at (4.5,0) {\footnotesize $2.2$};
        \node[above=2pt] at (5.5,0) {\footnotesize $2.3$};
        \node[above=2pt] at (7.5,0) {\footnotesize $2.n$};
        \node[above=3pt,align=center] at (10.5,0) {\textbf{Stage 3}};
        \node[below=4pt,align=center,text width=2.7cm] at (1,0) {\footnotesize entry: manipulator enters w.p.\ $p_{\mathrm M}$, else a liquidity trader};
        \node[below=4pt,align=center,text width=4.4cm] at (5.5,0) {\footnotesize $n$ ordinary Glosten--Milgrom rounds: each a buy $\hbuy$ or sell $\hsell$};
        \node[below=4pt,align=center,text width=2.9cm] at (10.5,0) {\footnotesize settlement window: push w.p.\ $\alpha$; contract settles on the posterior mean plus $\varepsilon\xi$};
    \end{tikzpicture}
    \caption{The extended game with $n$ ordinary spot rounds between entry and settlement. The
        manipulator decides whether to enter at stage~1, $n$ ordinary Glosten--Milgrom rounds then reveal
        order flow, and at stage~3 the manipulator (if she entered) pushes the spot price before the
        contract settles. The baseline (\cref{sec:results-canonical}) is the case $n=1$.}
    \label{fig:n-stage-timeline}
\end{figure}

We solve the two-round case first, where the order histories can be listed and the mechanism is clearest, and then the general case. With two ordinary rounds there are four order histories (\cref{fig:two-stage-tree}). When the two orders agree, both buys or both sells, they carry the price far enough from the strike that the outcome is already decided. When they disagree, they cancel and leave the price at the strike, where a push can still change the outcome. A push therefore pays only on the two histories in which the orders disagree.

\begin{figure}[ht]
    \centering
    \begin{tikzpicture}[font=\small,
            state/.style={rectangle,rounded corners,draw,align=center,inner sep=3pt},
            leaf/.style={rectangle,rounded corners,draw,align=center,inner sep=3pt,fill=gray!6},
            push/.style={rectangle,rounded corners,draw,align=center,inner sep=3pt,fill=gray!20},
            elab/.style={font=\footnotesize,inner sep=1.5pt}]
        \node[state] (root) at (0,2.8) {Stage 2.1 spot order};
        \node[state] (a) at (-3.4,1.4) {2.1: \textbf{buy}};
        \node[state] (b) at ( 3.4,1.4) {2.1: \textbf{sell}};
        \node[leaf] (uu) at (-5.2,-1.0) {$\hbuy\hbuy$\\[-1pt]above strike\\[-1pt]stay out, $\alpha^{*}=0$};
        \node[push] (ud) at (-1.75,-1.0) {$\hbuy\hsell$\\[-1pt]at strike\\[-1pt]\textbf{push}, $\alpha^{*}=1$};
        \node[push] (du) at ( 1.75,-1.0) {$\hsell\hbuy$\\[-1pt]at strike\\[-1pt]\textbf{push}, $\alpha^{*}=1$};
        \node[leaf] (dd) at ( 5.2,-1.0) {$\hsell\hsell$\\[-1pt]below strike\\[-1pt]stay out, $\alpha^{*}=0$};
        \draw[->] (root) -- (a);
        \draw[->] (root) -- (b);
        \draw[->] (a) -- (uu) node[elab,midway,above left] {2.2 buy};
        \draw[->] (a) -- (ud) node[elab,midway,above right]{2.2 sell};
        \draw[->] (b) -- (du) node[elab,midway,above left] {2.2 buy};
        \draw[->] (b) -- (dd) node[elab,midway,above right]{2.2 sell};
    \end{tikzpicture}
    \caption{The lengthened game with two ordinary spot rounds before the settlement
        window, conditional on a stage-1 buy. The two ordinary orders (stage~2.1 then stage~2.2) produce
        four stage-2 order histories. The agreeing order histories $\hbuy\hbuy$ and $\hsell\hsell$ move
        the posterior mean far enough from the strike that a push cannot change the resolution, so the
        manipulator does not push; the disagreeing order histories $\hbuy\hsell$ and $\hsell\hbuy$ leave it at the
        strike, where a push is pivotal and she pushes.}
    \label{fig:two-stage-tree}
\end{figure}

The equilibrium of the two-round game is as follows.

\begin{proposition}[Equilibrium with two ordinary spot rounds]
    \label{prop:horizon-two-equilibrium}
    Fix \(\sigma>0\), \(\pi\in\left(0,1\right)\), and \(\varepsilon>0\) satisfying \cref{ass:canonical}. The two-round game has a unique equilibrium (\cref{def:full-equilibrium}), with a unique optimal entry probability \(p_{\mathrm M}^{*\left(2\right)}\in\left(0,1\right)\). In this equilibrium:
    \begin{enumerate}[leftmargin=2em,itemsep=0.2em]
        \item the manipulator pushes only after the two ordinary orders disagree, and never after they agree:
              \[
                  \alpha_{\hbuy\hbuy}^{*}
                  =
                  \alpha_{\hsell\hsell}^{*}
                  =
                  0,
                  \qquad
                  \alpha_{\hbuy\hsell}^{*}
                  =
                  \alpha_{\hsell\hbuy}^{*}
                  =
                  1;
              \]

        \item the closing posterior mean is strictly more informative than in the baseline:
              \[
                  \E\left[\rvar_{h_2 h_3}^{\left(2\right)}\right]
                  \left(p_{\mathrm M}^{*\left(2\right)}\right)
                  <
                  \E\left[\rvar_{h_2 h_3}^{\left(1\right)}\right]
                  \left(p_{\mathrm M}^{*\left(1\right)}\right).
              \]
    \end{enumerate}
\end{proposition}

With two rounds, the manipulator pushes only when the two orders disagree, and the closing price is more informative than in the baseline even though she still pushes. The same pattern holds for any number of rounds.

\begin{proposition}[Equilibrium with $n>2$ ordinary spot rounds]
    \label{prop:horizon-n-equilibrium}
    Fix \(n > 2\), \(\sigma>0\), \(\pi\in\left(0,1\right)\), and \(\varepsilon>0\) satisfying \cref{ass:canonical}. The $n$-round game has a unique equilibrium (\cref{def:full-equilibrium}), with a unique optimal entry probability \(p_{\mathrm M}^{*\left(n\right)} \in \left(0, 1\right)\). In this equilibrium:
    \begin{enumerate}[leftmargin=2em,itemsep=0.2em]
        \item the manipulator pushes only when the order imbalance is near balance, leaving the price close to the strike:
              \[
                  \alpha_{h_2}^{*\left(n\right)}
                  =
                  \begin{cases}
                      1, & n \text{ is odd and } d\left(h_2\right)=-1, \\
                      1, & n \text{ is even and } d\left(h_2\right)=0, \\
                      0, & \text{otherwise;}
                  \end{cases}
              \]

        \item the closing posterior mean becomes strictly more informative with every added round:
              \[
                  \E\left[\rvar_{h_2 h_3}^{\left(n\right)}\right]
                  \left(p_{\mathrm M}^{*\left(n\right)}\right)
                  <
                  \E\left[\rvar_{h_2 h_3}^{\left(n-1\right)}\right]
                  \left(p_{\mathrm M}^{*\left(n-1\right)}\right).
              \]

        \item the entry probability differs by the parity of $n$:
              \[
                  p_{\mathrm M}^{*\left(n\right)}
                  =
                  \begin{cases}
                      p_{\mathrm M}^{*\left(1\right)}, & n \text{ is odd},  \\
                      p_{\mathrm M}^{*\left(2\right)}, & n \text{ is even}.
                  \end{cases}
              \]
    \end{enumerate}
\end{proposition}

However long the horizon, a push pays at exactly one imbalance: balance when $n$ is even, and one sell in excess when $n$ is odd. These are the order histories that leave the price at the strike or one order below it.

The closing price becomes more informative with every round. Each additional round gives the market one more observation of the order flow, so the closing price tracks the asset's value more closely and the expected residual variance of \cref{sec:price-discovery} falls. Part of this improvement would occur without a manipulator, since each round adds informed trading. A sharper comparison also holds: with $n$ rounds and manipulation, the closing price is more informative than with $n-1$ rounds and no manipulator at all (\cref{app:horizon-two-price-discovery,app:horizon-n-price-discovery}).

In the baseline the liquidity trader pays no spread. The whole cost comes from the lower closing price described in \cref{sec:results-canonical}. This is not so at every horizon. With an even number of rounds the prediction-market maker posts a positive spread (\cref{app:horizon-two-entry,app:horizon-n-entry}), so the liquidity trader pays the spread and also holds a contract that pays less often because of the lower closing price. The two costs can therefore coexist. The baseline separates them: it shows that the second cost is present even when no spread is charged.

A push pays at one imbalance only, and this is why pivotal pushes, and the profit from them, disappear as the horizon grows. The imbalance moves by one with each order, and it drifts toward the side of the asset's true value, because informed orders always trade in that direction (\cref{fig:imbalance-walk}). Over a long horizon the imbalance therefore ends far from balance, on one side or the other, while the imbalance at which a push pays stays where it is. The probability of ending there falls to zero, and the manipulator's expected profit with it.

\begin{figure}[ht]
    \centering
    \begin{tikzpicture}[font=\footnotesize,>=stealth,x=0.5cm,y=0.42cm]
        \begin{scope}
            \fill[gray!18] (0,-1) rectangle (5,0);
            \draw[->] (0,-4.4) -- (0,4.4) node[above]{$d$};
            \draw[->] (0,0) -- (5.6,0) node[right]{round};
            \draw[very thick,gray!35!black] (0,0)--(1,1)--(2,0)--(3,1)--(4,0);
            \filldraw (4,0) circle (2.2pt);
            \node[anchor=north,align=center,text width=4.2cm] at (2.5,-5.0)
            {\textbf{short horizon}: the walk ends in the band ($d\in\{-1,0\}$)---a push can still change the resolution};
        \end{scope}
        \begin{scope}[xshift=8.2cm]
            \fill[gray!18] (0,-1) rectangle (11,0);
            \draw[->] (0,-4.4) -- (0,4.4) node[above]{$d$};
            \draw[->] (0,0) -- (11.6,0) node[right]{round};
            \draw[very thick,gray!35!black] (0,0)--(1,1)--(2,2)--(3,1)--(4,2)--(5,3)--(6,2)--(7,3)--(8,4)--(9,3)--(10,4);
            \filldraw (10,4) circle (2.2pt);
            \node[anchor=north,align=center,text width=5.2cm] at (5.5,-5.0)
            {\textbf{long horizon}: the walk is likely to end far from balance, outside the band, so a push no longer affects the resolution};
        \end{scope}
    \end{tikzpicture}
    \caption{The pre-settlement order imbalance $d$ as a $\pm 1$ random walk over the ordinary
        rounds. The shaded strip is the pivotal band: a settlement-window push can change the resolution
        only when the walk ends at $d=0$ when $n$
        is even and $d=-1$ when $n$ is odd. As the horizon $n$ grows the walk spreads out (the imbalance has
        standard deviation of order $n$, because informed orders give it a drift) while the band stays fixed, so the chance of ending inside
        it---the pivotal probability $m^{\left(n\right)}$---falls to zero.}
    \label{fig:imbalance-walk}
\end{figure}

\begin{proposition}[Pivotal pushes vanish as the horizon grows]
    \label{cor:horizon-n-manipulation-prone-mass}
    As the horizon grows, a push is almost never pivotal and manipulation becomes worthless. Write \(m^{\left(n\right)}\) for the chance a push is still pivotal, that is, the probability that the stage-2 order history ends at \(d=0\) when \(n\) is even or at \(d=-1\) when \(n\) is odd (so \(m^{(1)}=1/2\) and \(m^{(2)}=(1-\pi^2)/2\)), \(S_{\mathrm M}^{\left(n\right)}\) for the manipulator's profit conditional on entry, and \(\bar S_{\mathrm M}^{\left(n\right)}\) for her unconditional expected profit at the optimal entry. All three vanish:
    \[
        \lim_{n \uparrow\infty} m^{\left(n\right)}
        =
        0,
        \qquad
        \lim_{n \uparrow\infty} S_{\mathrm M}^{\left(n\right)}
        \left(p_{\mathrm M}^{*\left(n\right)}\right)
        =
        0,
        \qquad
        \lim_{n \uparrow\infty} \bar S_{\mathrm M}^{\left(n\right)}
        \left(p_{\mathrm M}^{*\left(n\right)}\right)
        =
        0.
    \]
\end{proposition}

The manipulator still enters, as often as before. What disappears is the occasion to push. Lengthening the horizon does not make a push impossible, since a large enough order always moves the closing price. It makes a push of given size almost never pivotal. A contract is safe not when a manipulator cannot act, but when enough price discovery precedes settlement that a single trader is unlikely to decide the outcome.

A robust settlement reference is a second, equally powerful lever. Our contract settles on the spot price at a single instant---the close---which is exactly what a final-seconds push can move. A reference that aggregates instead---averaged over a window, or struck in a closing auction that pools many participants' orders---forces the manipulator to move the whole aggregate rather than one print; a time-weighted oracle, for one, raises the manipulation cost roughly in proportion to the averaging window \citep{aspembitova2022oracle}. Polymarket's switch to Chainlink TWAP feeds in August 2026 is a real-world instance of this lever~\citep{polymarket2026twap}. Our propositions hold the settlement reference fixed and so do not cover the switch formally, but \cref{sec:twap} shows that it removed the final-seconds push: afterward the push no longer moves the spot price, no longer reverts, and no longer flips outcomes or earns manipulators money. The two levers address manipulation from opposite sides: a longer horizon adds price discovery \emph{before} settlement, while an aggregating reference spreads the settlement price \emph{across} orders or time. The binary Nasdaq-100 index options the SEC approved in April 2026 (SR-MRX-2026-05) sit on the safe side of both: their expirations run from weekly to multi-year, and they settle on the index value struck in the Nasdaq Closing Cross, the exchange's closing auction, which pools the day's closing interest into one print and is correspondingly hard to push.\footnote{Nasdaq MRX, LLC, \emph{Order Granting Accelerated Approval of a Proposed Rule Change to Adopt New Options Rule~3B to List and Trade Binary Broad-Based Index Options}, Exchange Act Release No.~34-105342, File No.~SR-MRX-2026-05 (Apr.~30, 2026). The settlement value is the Nasdaq-100 Index as reported at the conclusion of the Nasdaq Closing Cross (Nasdaq Equity~4, Rule~4757); the contracts are cash-settled, European-style binary options listed in weekly, monthly, quarterly, end-of-month, and long-term (12--60 month) expirations.}

\section{Empirics}
\label{sec:empirics}

This section takes the model's predictions to the data. We first state what settlement manipulation should leave in the data according to the theory (Section~\ref{sec:predictions}). We then describe the sample (Section~\ref{sec:data}) and present the empirical evidence in three parts: Section~\ref{sec:headline} documents a near-settlement \emph{push} in the spot market, defined as a burst of one-sided trading in the final seconds before a contract settles, and shows it bears the predicted signature: it appears only for the five-minute contract, concentrates in the cycles the prediction market has left undecided, causes price reversals within seconds of settlement, and improves spot liquidity; Section~\ref{sec:wallet_cohorts} then shows the push is deliberate manipulation: it overturns resolutions even in near-certain cycles, transfers wealth from retail holders to a small cohort of manipulators, and cannot be explained solely by prediction-market makers' hedging; Section~\ref{sec:design_levers} then turns to what removes the push: it is strongly attenuated in the fifteen-minute contract, as the model's horizon prediction implies, and it largely disappeared, along with the profits from it, once Polymarket began settling on a time-weighted average price in August 2026.

\subsection{Model Predictions}
\label{sec:predictions}

The first set of predictions concerns the footprint manipulation leaves in the market. If settlement manipulation exists, it must appear as a directional push in the final seconds before the close: a burst of one-sided order flow on the spot venue the oracle samples, large enough to move the spot price (Theorem~\ref{prop:canonical-equilibrium}). Because that flow carries no information about fundamentals, its price impact should be transitory: the move should revert once settlement has passed, because the settlement-window price is less informative (Proposition~\ref{prop:price-discovery-sign}), in contrast to the permanent impact of informed trading.

The same lack of information content shapes liquidity on the spot venue. The manipulator's uninformed order pools with informed flow and lowers the adverse selection any individual trade carries, so spot market makers supply \emph{more} liquidity into the close even as activity surges (Corollary~\ref{cor:stage3-pm-buy-quote-monotonicity}). Finally, the entire footprint should be governed by whether a push can change anything: it concentrates in cycles the prediction market has left undecided near the close, because the manipulator pushes only when a push can change the outcome (Proposition~\ref{prop:push-policy}).

The second prediction concerns who pays. A successful push can overturn the resolution the market had priced, even in near-certain cycles we later identify in the data (Theorem~\ref{prop:canonical-equilibrium}). The contract is then mispriced relative to the manipulated resolution, and the loss falls on whoever holds the losing side at the close: liquidity traders, as predicted in the model. Manipulator profits are, to first order, liquidity traders' losses, according to the model's accounting identity (Proposition~\ref{prop:pm-maker-rent}).

The third prediction concerns the contract horizon. The footprint should be strongest in the shortest-horizon contract, where the asset's own price discovery has had the least time to unfold and a fixed-size push is most likely to be pivotal, and should weaken as the horizon lengthens (Proposition~\ref{prop:horizon-n-equilibrium} and Proposition~\ref{cor:horizon-n-manipulation-prone-mass}).

\subsection{Data}
\label{sec:data}
\label{app:data_quality}

The analysis combines three data sources: on-chain Polymarket trade records for the crypto up/down contracts, off-chain Polymarket order-book snapshots, and trades and quotes for BTC spot on Binance. Binance is the reference spot market throughout: it carries the largest BTC spot trading volume and arguably has the largest impact on the Chainlink reference price against which the Polymarket contracts resolve, so it is the natural venue on which a settlement manipulator would trade. Appendix~\ref{app:chainlink_validation} validates this choice directly, showing that the Binance midquote tracks the Chainlink resolution price to within a few basis points.

\begin{itemize}
    \item \emph{Polymarket on-chain trades.} Every order fill executed on Polymarket's exchange contracts settles on the Polygon blockchain, from which we extract the crypto up/down markets: over 60 million fills on the BTC five-minute contract alone. Each fill carries the block timestamp, the wallet identifiers on both sides, the asset on each side (one side is USDC, the other an Up- or Down-token), and the amounts deposited and received, which together recover the trade's price (in dollars per share, between 0 and 1), size, direction, and aggressor.\footnote{The on-chain log produces multiple records for each economic Polymarket trade: one aggregate record summarizing the aggressor's fill, plus one record for each resting order the aggressor consumes. We keep the single aggregate record per trade, so volume is never double-counted; on that record, the aggressor's deposited and received assets jointly identify both the contract side traded and the dollars exchanged. For wallet-level position and PnL tracking (Section~\ref{subsec:wallet_cohorts_id}), we instead sum the resting-side counterparty across both record classes, so each wallet's net position is recovered exactly once.} We link each fill to its market's cycle, strike, and realized resolution. Volume is measured two ways: \emph{share volume}, the number of shares changing hands (each share pays \$1 if its side wins, so the share count is the natural unit of exposure), and \emph{dollar volume}, the cash transferred. We sign each fill by the aggressor's direction, expressed in the Up token (an aggressive Up buy or Down sell is bullish-Up), so net order flow over a window measures directional pressure on the contract.
    \item \emph{Polymarket order book.} Polymarket's market-data feed delivers both full limit-order-book snapshots (the prevailing best bid and ask alongside several dozen price levels of resting depth on each side) and single-level updates whenever an existing level's resting size changes, for both tokens of each market. The feed covers roughly the final seven weeks of the sample. From it we measure the Up-token midquote and resting depth near the mid.
    \item \emph{Binance BTC spot trades and order book.} The full tick-by-tick trade records and order-book snapshots for the reference asset. Each trade carries price, size, and aggressor side, from which we build signed spot order flow; each book snapshot carries the five best bid and ask levels with their resting quantities, timestamped at 100-millisecond precision, from which we build time-weighted spreads and depth. Both are aggregated to ten-second buckets aligned with the five-minute contract clock.
\end{itemize}

\paragraph{Sample window and rollout periods.} The main sample covers July 1, 2025 through April 8, 2026, starting about three months before the launch of the oracle-resolved Polymarket up/down contracts we study and ending about two months after the five-minute contract launched. We divide the sample into three periods based on the contract rollout schedule: the pre-launch baseline (P1), July 1 through October 8, 2025 (100 days); the intermediate period with the fifteen-minute and four-hour contracts live (P2), October 9, 2025 through February 11, 2026 (126 days); and the post-launch period with the five-minute contract also live (P3), February 12 through April 8, 2026 (56 days).\footnote{Four narrow windows within the post-launch period require care: in three (late February 2026) the oracle-resolved Polymarket crypto contracts stop accepting trades while the rest of the platform trades normally (consistent with brief oracle-integration interruptions), and in a fourth (March 2026) the platform halts entirely. Cycles in these windows carry no on-chain Polymarket trades, so they drop out of every analysis that requires Polymarket data; the Binance-side panels retain them. On the spot side, we exclude October 10, 2025 (a market-wide flash crash) from the intra-cycle Binance figures, because a transient hundred-dollar spread during that episode would dominate a time-weighted average.} Section~\ref{sec:twap} adds two later 79-day windows for the BTC five-minute contract on either side of Polymarket's switch to an averaged settlement price: the pre-TWAP window, May 20 through August 6, 2026, and the TWAP window, August 14 through October 31, 2026, leaving out the transition week of August 7--13.

\paragraph{Summary statistics.} Table~\ref{tab:data_summary} reports summary statistics for the per-cycle trading activity over the five-minute contract's live window.\footnote{Each five-minute contract is listed for trading several hours before its live resolution window opens, but about 93\% of five-minute Bitcoin trades fall inside the live window (the 300 seconds between open and close), with the remainder split roughly evenly between pre-window positioning and the brief post-close interval before settlement. All cycle-level aggregates restrict to the live window.}

\begin{table}[!t]
    \centering
    \footnotesize
    \caption{\textbf{Per-cycle summary statistics, BTC five-minute contract.} This table reports per-cycle trading activity on the BTC five-minute contract and the underlying Binance spot market, one observation per cycle in the contract's live window (February 12 -- April 8, 2026); $N$ is the number of cycles and the remaining columns report the cross-cycle distribution. The \emph{Polymarket} rows aggregate the contract's own trades over the 300-second live window: \emph{Trades} is the per-cycle taker trade count, \emph{Volume} the number of shares traded (each carrying \$1 of face value), and \emph{Distinct traders} the number of unique wallets with a fill. The \emph{Binance} rows aggregate the reference asset's spot tape over the same cycles: \emph{Trades} and \emph{Volume} (dollar volume) are per-cycle sums, and \emph{Absolute return} is the absolute open-to-close change in the log midquote over the five-minute window (basis points).}
    \label{tab:data_summary}
    \begin{tabular}{lrrrrrrrr}
    \toprule
     & $N$ & Mean & SD & Min & P25 & Median & P75 & Max \\
    \midrule
    \multicolumn{9}{l}{\textit{Polymarket BTC 5-minute contract}} \\
    \quad Trades & 16{,}073 & 3{,}912 & 1{,}228 & 2 & 3{,}084 & 3{,}811 & 4{,}642 & 12{,}622 \\
    \quad Volume (1000 shares) & 16{,}073 & 108.9 & 36.4 & 0.0 & 84.3 & 104.9 & 129.8 & 323.9 \\
    \quad Distinct traders & 16{,}073 & 1{,}689 & 374 & 4 & 1{,}467 & 1{,}685 & 1{,}918 & 3{,}332 \\
    \addlinespace
    \multicolumn{9}{l}{\textit{Binance BTC spot}} \\
    \quad Trades & 16{,}128 & 15{,}793 & 14{,}039 & 765 & 7{,}131 & 11{,}676 & 19{,}290 & 223{,}213 \\
    \quad Volume (\$M) & 16{,}128 & 5.08 & 7.01 & 0.17 & 1.79 & 3.15 & 5.86 & 284.69 \\
    \quad Absolute return (bps) & 16{,}128 & 10.16 & 11.81 & 0.00 & 2.88 & 6.58 & 13.51 & 376.29 \\
    \bottomrule
\end{tabular}

\end{table}

Two observations stand out. First, the prediction market is small relative to its underlying spot: the median cycle trades about 105 thousand \$1-face shares on Polymarket against \$3.15 million of Binance volume. At first glance this asymmetry cuts against a manipulation story: if the prize pool is small relative to the market one must move, why bother? Two considerations restore the incentive. The gap is a factor of thirty, not several orders of magnitude, so the pool at stake is economically meaningful. More importantly, the five-minute-volume comparison is not the relevant one: the manipulator does not need to match five minutes of spot trading, only to deliver the final push, a burst of one-sided trading in the last seconds that carries the close across the strike. Second, ordinary five-minute spot moves are modest: the median absolute open-to-close return is 6.6 bps, with the inter-quartile range from 2.9 to 13.5 bps. A near-settlement spot push of a few basis points can therefore be large enough to clear the strike in a near-the-money cycle.

\subsection{Does manipulation leave a footprint?}
\label{sec:headline}

Section~\ref{sec:predictions} tells us what to look for. If manipulation exists, the close of each cycle should carry a footprint with four parts: a burst of directional \emph{order flow} (the push itself), an outsized \emph{absolute return} (the push moves the price), improved spot \emph{depth} (the push is uninformed, so spot makers supply more liquidity, not less), and a post-settlement \emph{reversal} (for the same reason, the price impact does not last). We look for this footprint in two steps. We first show it emerges with the launch of the five-minute contract. A launch coincidence alone cannot establish manipulation, so we then turn to the identification: within the post-launch period, the footprint concentrates in the cycles still undecided near the close, where a push can flip the resolution. The footprint is also specific to the five-minute contract: the fifteen-minute contract, live earlier and settling on the same rule, exhibits only a strongly attenuated version of it, a comparison we take up as a test of the model's horizon prediction in Section~\ref{sec:fifteen_min}.

\subsubsection{The footprint emerges with the five-minute launch}
\label{subsec:launch_impact}
\label{subsec:activity_did}
\label{app:method}

\paragraph{Near-settlement activity spike.}

We divide each five-minute cycle into 30 ten-second bins, indexed 0 through 29, and compute within each bin three of the footprint's four parts: absolute order flow, absolute return, and depth within \$0.10 of the midquote. For each bin we then average across all cycles in a period and compare the pre-launch baseline (P1) against the five-minute-live period (P3); the intermediate period (P2), in which only the fifteen-minute and four-hour contracts trade, returns when we study the fifteen-minute contract on its own (Section~\ref{sec:fifteen_min}).

\begin{figure}[t]
    \centering
    \includegraphics[width=\textwidth]{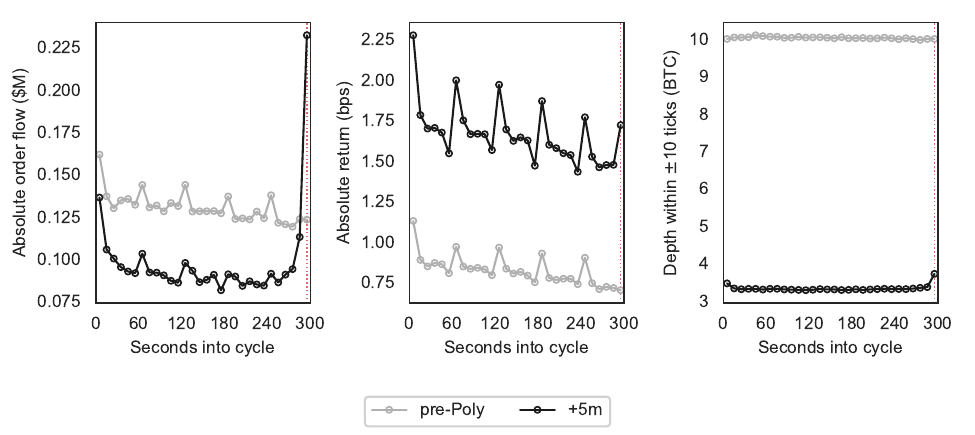}
    \caption{\textbf{Binance BTC intra-cycle metrics per 5-min cycle, by period.} This figure plots the intra-cycle profile of three Binance BTC metrics, by period. Each panel folds every 10s Binance BTC bucket onto its position in the rolling 5-minute clock cycle and overlays, for each period, the metric pooled over all 5-minute cycles in that period (see Section~\ref{app:method}). \emph{Left}: absolute order flow per cycle (\$M), the per-cycle $|$net signed taker dollar flow$|$ in the bin. \emph{Center}: absolute 10s return ($|\Delta \log \mathrm{mid}|$, bps). \emph{Right}: time-weighted depth within $\pm$10 ticks of the mid, cumulative resting size (bid + ask) at price levels within 10 ticks ($\$0.10$) of the mid. The dotted vertical line marks the final 10s before close. Two periods are overlaid: pre-Polymarket (P1, Jul--Oct 2025) and the five-minute-live period (P3, Feb 12 -- Apr 8, 2026); the intermediate period with only the fifteen-minute and four-hour contracts live (P2) is omitted.}
    \label{fig:byperiod}
\end{figure}

Figure~\ref{fig:byperiod} plots the resulting intra-cycle profiles. The pattern is striking. In bin 29, the last ten seconds of the cycle, all three metrics rise sharply once the five-minute contract is live (the $+5$m series); in the pre-launch baseline, bin 29 looks no different from the first 25 bins of the cycle, which we call the \emph{body}. The increase is confined to the final bin, so the contrast at the close is not a broad shift in Binance trading.\footnote{There are minor spikes at the minute-ends. They are likely driven by trading algorithms that react to the minute mark.}

We formalize the visual pattern with a pre/post comparison of the bin-29 activity, comparing the pre-launch baseline (P1) against the period in which the five-minute contract is live (P3):
\begin{equation}
    \log y_{c,b} \;=\; \gamma_b + \alpha_d + \mu_h + \delta_{P3}\,\mathbf{1}\{b{=}29\}{\cdot}\mathbf{1}\{P3\} + \varepsilon_{c,b},
    \label{eq:activity_did}
\end{equation}
where $\log y_{c,b}$ is the log of the metric in cycle $c$ and bin $b$, restricted to the body bins and the final-10s bin, $b \in \{0, \ldots, 24, 29\}$, so that the near-settlement bin is compared with the body alone. $\alpha_d$ and $\mu_h$ are date and hour-of-day fixed effects; $\gamma_b$ is a bin fixed effect that absorbs the pre-existing intra-cycle profile. The specification pools all cycles: we do not differentiate overlapping from non-overlapping closes here; the overlap split returns in the fifteen-minute analysis (Section~\ref{sec:fifteen_min}) and, interacted with moneyness, in the within-P3 refinements of Appendix~\ref{app:boundary_atm_zooms}. The coefficient $\delta_{P3}$ is the average bin-29-versus-body activity increase once the five-minute contract is live, relative to the pre-launch baseline.

Table~\ref{tab:reg_uncond} reports the estimates. Once the five-minute contract is live, the typical cycle's bin-29 absolute order flow is 50\% higher relative to the body than before the launch ($\hat{\delta}_{P3} = 0.406$ log points), its absolute return is 15\% higher ($0.140$), and its near-mid depth 5\% higher ($0.048$), all significant at the 1\% level. Three parts of the footprint thus appear at the close once the five-minute contract begins to settle there.

\begin{table}[!t]
    \centering
    \footnotesize
    \caption{\textbf{Near-settlement activity: pre-Polymarket vs five-minute period.} This table reports cycle-level OLS estimates on the (cycle, bin) panel restricted to bins $\{0\text{--}24, 29\}$ and to the pre-Polymarket period (P1, Jul--Oct 2025) and the five-minute-live period (P3, Feb 12 -- Apr 8, 2026); the intermediate period in which only the fifteen-minute and four-hour contracts are live (P2) is excluded. The dependent variable is $\log y_{c,b}$, the log of the metric named in the column header: \emph{Order flow}, per-cycle absolute net signed taker dollar flow ($|\sum_i d_i\,p_i a_i|$, $d=+1$ buy / $-1$ sell); \emph{Abs.\ return}, absolute 10\,s log return ($|\Delta \log \mathrm{mid}|$, bps); \emph{Depth}, time-weighted resting size within $\pm$10 ticks of the mid. $\delta_{P3}$ is the coefficient on $\mathbf{1}\{b{=}29\}\cdot\mathbf{1}\{P3\}$: the bin-29-versus-body activity increase once the five-minute contract is live, relative to the pre-Polymarket baseline. Fixed effects are date, hour-of-day, and bin; the specification pools all cycles and does not differentiate overlap from non-overlap closes. Coefficients are log differences. One observation per (cycle, bin) cell. HC1 standard errors in parentheses are clustered at the cycle level. Significance: $^{*}\,p<0.10$, $^{**}\,p<0.05$, $^{***}\,p<0.01$.}
    \label{tab:reg_uncond}
    \begin{tabular}{l S[table-format=-1.3] S[table-format=-1.3] S[table-format=-1.3]}
    \toprule
     & {(1)} & {(2)} & {(3)} \\
     & {Order flow} & {Abs. return} & {Depth} \\
    \midrule
    $\delta_{P3}$ & 0.406{\sym{***}} & 0.140{\sym{***}} & 0.048{\sym{***}} \\
     & (0.019) & (0.018) & (0.003) \\
    \midrule
    Date FE & {Yes} & {Yes} & {Yes} \\
    Hour-of-day FE & {Yes} & {Yes} & {Yes} \\
    Bin FE & {Yes} & {Yes} & {Yes} \\
    Observations & {1{,}168{,}118} & {738{,}313} & {1{,}168{,}114} \\
    Clusters (cycle) & {44{,}928} & {44{,}696} & {44{,}928} \\
    \bottomrule
\end{tabular}

\end{table}

\paragraph{Post-settlement reversal.}
\label{sec:reversal}
\label{subsec:reversal_did}

The fourth part of the footprint is the reversal. Because manipulation flow is uninformative, its price impact should not last. We test this by regressing each lagged ten-second return around the settlement on the final-10s return $r_{29}$ at the cycle level.

Separately at each lag $j$, we estimate
\begin{equation}
    r_{j,c} \;=\; \alpha_d + \mu_h + \beta\,r_{29,c} + \delta_{P3}\,r_{29,c}\,\mathbf{1}\{P3\} + \varepsilon_{j,c},
    \label{eq:reversal_did}
\end{equation}
on the pre-Polymarket (P1) and five-minute-live (P3) cycles. The dependent variable $r_{j,c}$ is the lag-$j$ ten-second return of cycle $c$ measured relative to its close: positive lags map to bins 0--5 of the next cycle (post-settlement), negative lags to bins 25--28 of the same cycle (the run-up into the close), and $j=0$ is the bin-29 reference itself, omitted because its slope is $1$ trivially. The regressor $r_{29,c}$ is the final-10s (bin-29) return. $\beta$ is the pre-Polymarket reversal slope and $\delta_{P3}$ the change in that slope once the five-minute contract is live; $\alpha_d$ and $\mu_h$ are date and hour-of-day fixed effects, and standard errors are clustered by date. A reverting price is a negative slope; the key lag is $j = +1$, the first ten seconds after the close, where a transitory manipulation push predicts reversal.

\begin{figure}[t]
    \centering
    \includegraphics[width=0.57\textwidth]{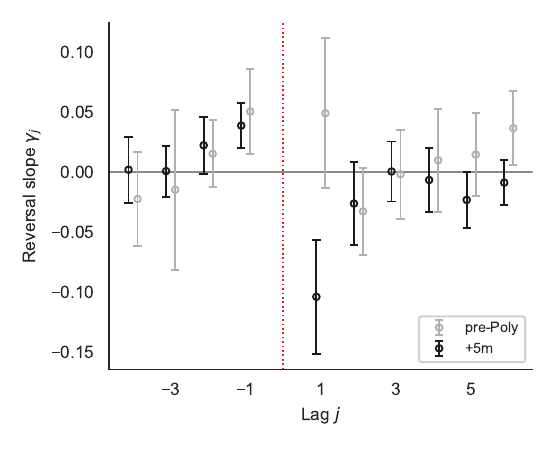}
    \caption{\textbf{Post-settlement reversal slope $\gamma_j$: pre-Polymarket vs five-minute period.} This figure plots, at each lag $j$, the slope $\gamma_j$ from a per-period pooled OLS of the lag-$j$ 10s return on the final-10s return $r_{29}$ (date and hour-of-day fixed effects, standard errors clustered by date), overlaid for the pre-Polymarket period (P1) and the five-minute-live period (P3); the intermediate +15m/4h period (P2) is omitted. Markers are $\gamma_j$; error bars are $\pm 1.96 \times$ cluster-robust SE; $j = 0$ is omitted. The dotted vertical marks the close. Negative $j$ is bin $29+j$ of the same cycle; positive $j$ is bin $j-1$ of the next cycle.}
    \label{fig:reversal_uncond}
\end{figure}

Figure~\ref{fig:reversal_uncond} overlays, for the two periods, the reversal slope at each lag estimated separately in each period, denoted $\gamma_j$. A sharp reversal appears in the first ten seconds after the close, and only with the five-minute contract: $\gamma_{+1}$ is strongly negative once the contract is live and indistinguishable from zero in the pre-launch baseline, while at $j \ge +2$ the slopes are flat in both periods.

Table~\ref{tab:reversal_uncond} confirms it. Against a pre-launch baseline with no reversal at the close ($\beta = 0.05$ at $j=+1$, indistinguishable from a random walk), the five-minute contract shifts the reversal slope by $\delta_{P3} = -0.15$, so the post-close slope under the five-minute contract is $\beta + \delta_{P3} = -0.10$ at $j=+1$: about a tenth of the final-10s spot move reverses within the next ten seconds. The shift is specific to the first post-close lag: $\delta_{P3}$ is flat at $j \ge +2$. Thus, all four parts of the footprint appear with the five-minute contract and are absent before its launch.

\begin{table}[!t]
    \centering
    \footnotesize
    \caption{\textbf{Post-settlement reversal: pre-Polymarket vs five-minute period.} This table reports the cycle-level OLS $r_j = \alpha_d + \mu_h + \beta\,r_{29} + \delta_{P3}\,r_{29}\mathbf{1}\{P3\} + \varepsilon$, estimated separately at each lag $j \in \{-3,-2,-1,+1,+2,+3\}$, on the pre-Polymarket period (P1, Jul 1 -- Oct 8, 2025) and the five-minute-live period (P3, Feb 12 -- Apr 8, 2026); the intermediate period in which only the fifteen-minute and four-hour contracts are live (P2) is excluded. $\beta$ is the pre-Polymarket reversal slope on $r_{29}$ and $\delta_{P3}$ its change once the five-minute contract is live (a negative slope is reversal). Date and hour-of-day fixed effects; HC1 standard errors in parentheses are clustered by date. Negative $j$ = same-cycle bin $(29 + j)$; positive $j$ = bin $(j-1)$ of the next cycle. Significance: $^{*}\,p<0.10$, $^{**}\,p<0.05$, $^{***}\,p<0.01$.}
    \label{tab:reversal_uncond}
    \begin{tabular}{l S[table-format=-1.3] S[table-format=-1.3] S[table-format=-1.3] S[table-format=-1.3] S[table-format=-1.3] S[table-format=-1.3]}
    \toprule
     & \multicolumn{3}{c}{Pre-close} & \multicolumn{3}{c}{Post-close} \\
    \cmidrule(lr){2-4} \cmidrule(lr){5-7}
     & {$j = -3$} & {$j = -2$} & {$j = -1$} & {$j = +1$} & {$j = +2$} & {$j = +3$} \\
    \midrule
    $\beta$ & -0.014 & 0.015 & 0.051{\sym{***}} & 0.049 & -0.033{\sym{*}} & -0.002 \\
     & (0.034) & (0.014) & (0.018) & (0.032) & (0.019) & (0.019) \\[2pt]
    $\delta_{P3}$ & 0.015 & 0.007 & -0.012 & -0.152{\sym{***}} & 0.007 & 0.002 \\
     & (0.036) & (0.019) & (0.020) & (0.040) & (0.026) & (0.023) \\
    \midrule
    Date FE & {Yes} & {Yes} & {Yes} & {Yes} & {Yes} & {Yes} \\
    Hour-of-day FE & {Yes} & {Yes} & {Yes} & {Yes} & {Yes} & {Yes} \\
    Observations & \multicolumn{6}{c}{44{,}772} \\
    Clusters (UTC date) & \multicolumn{6}{c}{156} \\
    \bottomrule
\end{tabular}

\end{table}

\subsubsection{The footprint concentrates in NTM cycles}
\label{subsec:atm_concentration}
\label{sec:atm}
\label{sec:mechanism}
\label{subsec:activity_atm}
\label{subsec:reversal_atm}

The before/after evidence so far shows a four-part footprint that coincides with the five-minute launch, but a launch coincidence cannot by itself separate manipulation from any change in close-time trading. The identification comes from a cross-sectional cut \emph{within} the five-minute period. As the model predicts, if the push is manipulation, all four parts of the footprint should be largest in cycles where the outcome is still undecided near the close, where a small spot push can flip the resolution and the manipulator's payoff is greatest, and muted in cycles already decided, where a push changes nothing. Because this contrast is taken within a single period, it differences out any time trend common to the launch.

\paragraph{Near-the-money classification.} We classify each post-launch five-minute cycle by its prediction-market price near settlement. Specifically, for each cycle, we take the Up-token price 10 seconds before the settlement and label the cycle \emph{near-the-money} (NTM) if the Up-token price is between 0.40 and 0.60 (the outcome is still close to even odds, so a push across the strike can flip it) and \emph{far-from-the-money} (FFM) otherwise. About 6\% of cycles are NTM.\footnote{The $\sim$4\% of P3 cycles with no qualifying trade in the 10--12\,s window are dropped as unclassified; these are systematically the sparse-trading cycles. The tight window trades coverage for freshness: a wider window would retain them but allow the classifying price to be stale by up to several minutes, mislabeling a cycle when the spot moves in the interim. Pre-launch cycles (P1, P2) have no on-chain 5m trades and are unclassified by construction.}

\paragraph{Near-settlement spike in NTM cycles.} Figure~\ref{fig:cond} plots three intra-cycle metrics --- absolute order flow, absolute return, and depth --- on both venues, separately for NTM and FFM cycles. The figure shows two results. First, the Binance spike concentrates in NTM cycles: order flow and absolute return rise at the close in both groups, but markedly more where the contract is still undecided and a push can flip the resolution. Second, liquidity diverges across the venues into the NTM close: Binance makers \emph{add} near-mid depth while Polymarket makers \emph{withdraw} it.

Two further observations concern the prediction market itself. The Polymarket price changes more in NTM cycles for a mechanical reason: a price near $0.5$ must move to $0$ or $1$ as the contract resolves, while the price of a decided contract is already there. Polymarket order flow, by contrast, increases \emph{less} at the close in NTM cycles. Close-time order flow on Polymarket is dominated by one-sided buying of the winning side of an already-decided contract, which is largest when the outcome is no longer in doubt; in undecided cycles the directional order flow appears on Binance instead.

\begin{figure}[t]
    \centering
    \begin{subfigure}{\textwidth}
        \centering
        \includegraphics[width=\textwidth]{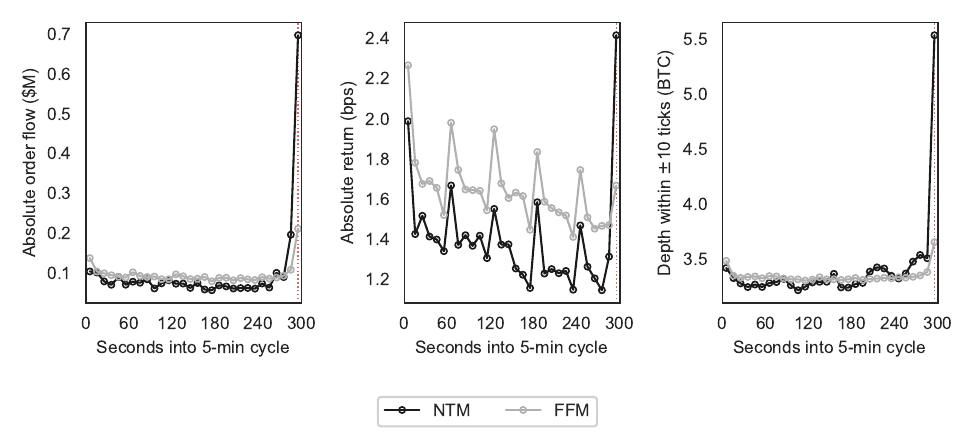}
        \caption{Binance BTC spot market.}
        \label{fig:cond_binance}
    \end{subfigure}

    \begin{subfigure}{\textwidth}
        \centering
        \includegraphics[width=\textwidth]{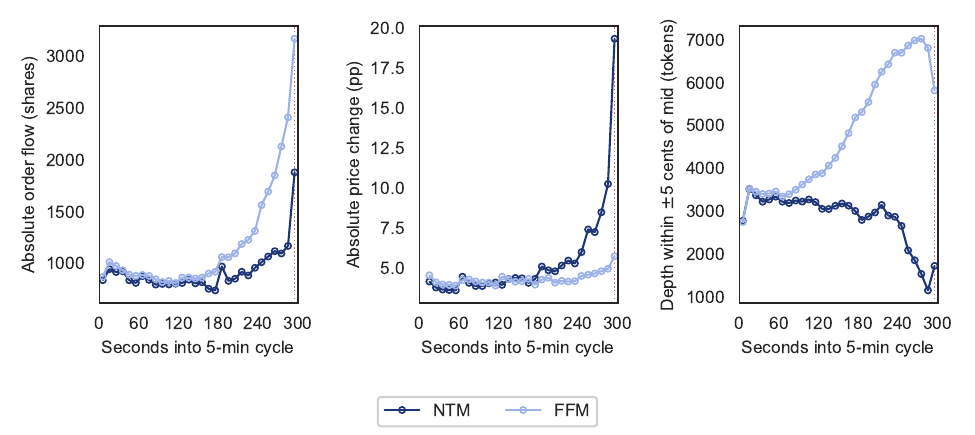}
        \caption{Polymarket BTC 5-minute up/down contract.}
        \label{fig:cond_poly}
    \end{subfigure}
    \caption{\textbf{Intra-cycle metrics per 5-min cycle, \NTM vs.\ \FFM\ (P3 only).} This figure plots intra-cycle profiles on both venues, separately for \NTM\ and \FFM\ cycles. Each P3 5-minute cycle is classified \NTM\ ($|p_{\text{Up}} - 0.5| < 0.10$) or \FFM\ from the latest taker-focused on-chain Polymarket Up-token trade at time-to-maturity TTM~$\in[10,12]$\,s; each metric is pooled over all classified P3 5-minute cycles. Each $1\times3$ grid plots one metric's intra-cycle profile against seconds into the 5-min cycle; the dotted vertical marks the final 10s before close. \textbf{(a)~Binance BTC spot}: absolute order flow per cycle (\$M, $|$net signed taker dollar flow$|$), absolute 10s return ($|\Delta \log \mathrm{mid}|$, bps), time-weighted depth within $\pm$10 ticks of the mid. \textbf{(b)~Polymarket BTC 5-minute up/down contract}: absolute order flow per cycle ($|$net signed Up-token share flow$|$), absolute 10s probability change ($|\Delta p|$, pp), time-weighted depth within $\pm$5 cents of the mid (cumulative Up-token resting size at price levels within \$0.05 of the mid).}
    \label{fig:cond}
\end{figure}

To quantify the NTM-vs-FFM contrast on both venues, we estimate the following regression:
\[
    \log y_{c,b} \;=\; \alpha_d + \gamma_b + \mu_h
    + \beta_{\NTM}\,\mathbf{1}\{\NTM_c\}
    + \delta_{\NTM}\,\mathbf{1}\{b{=}29\}{\cdot}\mathbf{1}\{\NTM_c\}
    + \varepsilon_{c,b},
\]
where $\alpha_d$, $\mu_h$, $\gamma_b$ capture the date, hour-of-day, and bin fixed effects. The standard errors are clustered at the cycle level. The \FFM\ bin-29 increase is absorbed into the bin fixed effect $\gamma_{29}$, so $\delta_{\NTM}$ measures the additional bin-29-vs-body activity increase in \NTM\ cycles over the \FFM\ baseline.

Table~\ref{tab:reg_cond} reports the regression results, which confirm both results. First, on Binance, the bin-29 order flow is $\hat\delta_{\NTM} = 1.37$ log points larger in NTM than in FFM cycles (about $e^{1.37} \approx 3.9$ times the FFM one on the typical cycle), and the same NTM concentration holds for absolute return ($0.56$), both significant at the $1\%$ level. On Polymarket, by contrast, the NTM order-flow effect is \emph{negative} ($-0.18$): the final-bin order-flow increase is smaller in NTM cycles than in FFM cycles, where one-sided buying of the winning side dominates. The Polymarket price-change effect is large and positive ($2.62$) but largely mechanical, because resolution moves a price near $0.5$ to $0$ or $1$ while the price of a decided contract is already there. In cycles where the outcome is still undecided, the additional directional order flow therefore appears on Binance, not on Polymarket.

Second, Binance and Polymarket see opposite patterns in liquidity: into the NTM close Binance depth is about 19\% higher ($\hat\delta_{\NTM} = 0.17$) while Polymarket depth is about 64\% lower ($\hat\delta_{\NTM} = -1.02$). The higher Binance depth is what the model predicts (Section~\ref{sec:predictions}).

A maker quoting the Polymarket binary near the close becomes the counterparty to anyone about to push the spot across the strike: the manipulator buys the side she is about to make win, so the prediction-market maker who sold it to her is left on the losing side of the resolution. Polymarket makers therefore withdraw depth as the close approaches, and withdraw the most in NTM cycles, where a push that flips the resolution is most likely. On Binance the same push is uninformed and therefore generates only a \emph{transitory} price impact. The push flow is therefore benign for Binance makers, who supply \emph{more} depth into the NTM close. Appendix~\ref{app:activity_overlap} refines the NTM cut by fifteen-minute contract overlap.

\begin{table}[!t]
    \centering
    \footnotesize
    \caption{\textbf{P3 near-settlement activity, NTM vs.\ FFM cycles, Binance and Polymarket.} This table reports cycle-level OLS estimates on the (cycle, bin) panel restricted to bins $\{0\text{--}24, 29\}$ and P3 cycles, $\log y_{c,b} = \alpha_d + \gamma_b + \mu_h + \beta_{\NTM}\mathbf{1}\{\NTM_c\} + \delta_{\NTM}\,\mathbf{1}\{b{=}29\}\mathbf{1}\{\NTM_c\} + \varepsilon_{c,b}$, with date, hour-of-day, and bin fixed effects. The first three columns report Binance spot, the last three the Polymarket 5-minute contract; within each venue the dependent variable $\log y_{c,b}$ is the log of the column metric. The activity column is \emph{Order flow} on both venues: the per-cycle absolute net signed flow, in taker dollars on Binance ($|\sum_i d_i\,p_i a_i|$) and in Up-token shares on Polymarket. The second column per venue is \emph{Abs.\ return} on Binance (absolute 10\,s log return, $|\Delta \log \mathrm{mid}|$, bps) and \emph{Abs.\ price change} on Polymarket (absolute 10\,s probability change, $|\Delta p|$, pp); the third is \emph{Depth} (time-weighted resting size within a mid-anchored band: $\pm$10 ticks of the mid on Binance, $\pm$5 cents on Polymarket). A P3 cycle is NTM iff the 5m-contract Polymarket Up-token price is within $0.10$ of $0.5$ at TTM~$\in[10,12]$\,s, FFM otherwise. $\delta_{\NTM}$ is the bin-29-vs-body activity increase in NTM cycles over the FFM baseline (the latter absorbed in the bin FE); coefficients are log differences, so $\delta$ is a multiplicative effect of $e^{\delta}$ on the typical (geometric-mean) cycle. Appendix~\ref{app:activity_overlap} reports the same regression with the NTM cut refined into the full 15-min-overlap $\times$ per-contract-NTM split. One observation per (cycle, bin) cell. Polymarket order-book metrics (absolute return, depth) have smaller samples than order flow because the order-book archive begins later than the trade archive and outage days are excluded. HC1 standard errors in parentheses are clustered at the cycle level. Significance: $^{*}\,p<0.10$, $^{**}\,p<0.05$, $^{***}\,p<0.01$.}
    \label{tab:reg_cond}
    \begin{tabular}{l S[table-format=-1.3] S[table-format=-1.3] S[table-format=-1.3] S[table-format=-1.3] S[table-format=-1.3] S[table-format=-1.3]}
    \toprule
     & \multicolumn{3}{c}{Binance} & \multicolumn{3}{c}{Polymarket} \\
    \cmidrule(lr){2-4}\cmidrule(lr){5-7}
     & {Order flow} & {Abs. return} & {Depth} & {Order flow} & {Abs. price change} & {Depth} \\
    \midrule
    $\delta_{\NTM}$ & 1.366{\sym{***}} & 0.555{\sym{***}} & 0.173{\sym{***}} & -0.182{\sym{***}} & 2.618{\sym{***}} & -1.016{\sym{***}} \\
     & (0.086) & (0.052) & (0.022) & (0.044) & (0.089) & (0.052) \\
    \midrule
    Date FE & {Yes} & {Yes} & {Yes} & {Yes} & {Yes} & {Yes} \\
    Hour-of-day FE & {Yes} & {Yes} & {Yes} & {Yes} & {Yes} & {Yes} \\
    Bin FE & {Yes} & {Yes} & {Yes} & {Yes} & {Yes} & {Yes} \\
    Observations & {395{,}790} & {307{,}108} & {395{,}790} & {395{,}154} & {138{,}405} & {142{,}468} \\
    Clusters (cycle) & {15{,}223} & {15{,}221} & {15{,}223} & {15{,}223} & {5{,}631} & {5{,}631} \\
    \bottomrule
\end{tabular}

\end{table}

\paragraph{Post-settlement price reversal in NTM cycles.}
In addition to the activity spike, NTM cycles see a more pronounced post-settlement price reversal. Using the same moneyness classifier, at each lag we estimate the following regression:
\[
    r_{j,c} \;=\; \alpha_d + \mu_h + \nu_g + \gamma_{j}\,r_{29,c} + \delta_{j,\NTM}\,\mathbf{1}\{\NTM_c\}\,r_{29,c} + \varepsilon_{j,c},
\]
where $\gamma_{j}$ is the FFM baseline reversal slope on $r_{29}$ and $\delta_{j,\NTM}$ is the NTM increment over it, so the NTM-cycle slope is $\gamma_{j} + \delta_{j,\NTM}$ and $\delta_{j,\NTM}$ is directly the NTM-minus-FFM reversal gap. $\alpha_d$, $\mu_h$, and $\nu_g$ are date, hour-of-day, and moneyness-group fixed effects; standard errors are clustered by date.

We report the regression results in Table~\ref{tab:reversal_cond} and plot the slopes in Figure~\ref{fig:reversal_cond}. FFM cycles revert by about 10\% ($\gamma_{+1} = -0.10$), and NTM cycles revert significantly more: the NTM increment is $\delta_{+1,\NTM} = -0.15$, so NTM cycles revert by about 25\% in total. A transitory push is profitable only if it moves a price that the prediction market has not already decided, so the cycles that revert most are the NTM ones, where the push can flip the resolution. As before, in Appendix~\ref{app:reversal_overlap} we further partition the cycles by 15-minute overlap. We find that the reversal is deeper in cycles where the five-minute and fifteen-minute contracts settle at the same close.

\begin{figure}[t]
    \centering
    \includegraphics[width=0.57\textwidth]{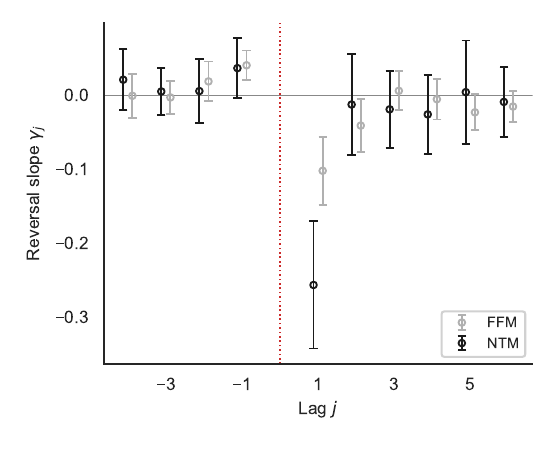}
    \caption{\textbf{Within-P3 post-settlement reversal slope $\gamma_j$, by 5m-contract moneyness.} This figure plots $\gamma_j$ for each 5m-contract moneyness group (the FFM baseline and the NTM total $\gamma_j + \delta_{j,\NTM}$ of Table~\ref{tab:reversal_cond}), on all P3 cycles split by 5m-contract Polymarket NTM (Up-token price within $0.10$ of $0.5$ near TTM~$=10$\,s). Markers are $\gamma_j$; error bars are $\pm 1.96 \times$ cluster-robust SE (clustered by date); $j = 0$ is omitted.}
    \label{fig:reversal_cond}
\end{figure}

\begin{table}[!t]
    \centering
    \footnotesize
    \caption{\textbf{Within-P3 post-settlement reversal slope by 5m-contract moneyness.} This table reports, at each lag $j \in \{-3,-2,-1,+1,+2,+3\}$, the two coefficients from a single cycle-level OLS on all P3 cycles (Feb 12 -- Apr 8, 2026), $r_j = \alpha_d + \mu_h + \nu_g + \gamma_j\,r_{29} + \delta_{j,\NTM}\,\mathbf{1}\{\NTM\}\,r_{29} + \varepsilon$: $\gamma_j$ is the FFM baseline reversal slope and $\delta_{j,\NTM}$ the NTM increment over it, so the NTM-cycle slope is $\gamma_j + \delta_{j,\NTM}$ and $\delta_{j,\NTM}$ is the NTM-minus-FFM gap; the stars test each coefficient against zero. A cycle is $\NTM_{5\mn}$ iff the 5m-contract Polymarket Up-token price is within $0.10$ of $0.5$ near TTM~$\in[10,12]$\,s, FFM otherwise. Date, hour-of-day, and moneyness-group fixed effects; standard errors in parentheses are clustered by date. Negative $j$ = same-cycle bin $(29 + j)$; positive $j$ = bin $(j-1)$ of the next cycle. Appendix~\ref{app:reversal_overlap} dissects these groups further by 15-minute overlap. Significance: $^{*}\,p<0.10$, $^{**}\,p<0.05$, $^{***}\,p<0.01$.}
    \label{tab:reversal_cond}
    \begin{tabular}{l S[table-format=-1.3] S[table-format=-1.3] S[table-format=-1.3] S[table-format=-1.3] S[table-format=-1.3] S[table-format=-1.3]}
    \toprule
     & \multicolumn{3}{c}{Pre-close} & \multicolumn{3}{c}{Post-close} \\
    \cmidrule(lr){2-4} \cmidrule(lr){5-7}
     & {$j = -3$} & {$j = -2$} & {$j = -1$} & {$j = +1$} & {$j = +2$} & {$j = +3$} \\
    \midrule
    FFM$_{5\mn}$ & -0.002 & 0.019 & 0.041{\sym{***}} & -0.101{\sym{***}} & -0.040{\sym{**}} & 0.007 \\
     & (0.011) & (0.014) & (0.010) & (0.023) & (0.018) & (0.013) \\[2pt]
    $\delta_{\NTM_{5\mn}}$ & 0.008 & -0.013 & -0.004 & -0.154{\sym{***}} & 0.028 & -0.025 \\
     & (0.017) & (0.024) & (0.021) & (0.038) & (0.037) & (0.026) \\
    \midrule
    Date FE & {Yes} & {Yes} & {Yes} & {Yes} & {Yes} & {Yes} \\
    Hour-of-day FE & {Yes} & {Yes} & {Yes} & {Yes} & {Yes} & {Yes} \\
    Cycle-group FE & {Yes} & {Yes} & {Yes} & {Yes} & {Yes} & {Yes} \\
    Observations & \multicolumn{6}{c}{15{,}168} \\
    Clusters (UTC date) & \multicolumn{6}{c}{56} \\
    \bottomrule
\end{tabular}

\end{table}

\clearpage
\subsection{Who wins and loses in manipulated cycles?}
\label{sec:wallet_cohorts}

Section~\ref{sec:headline} shows that the five-minute launch coincides with a transitory final-10s spot push concentrated in the cycles the prediction market has left undecided: the manipulation footprint. This section tests the model's second prediction (Section~\ref{sec:predictions}): that the push transfers wealth, from the liquidity traders holding the losing side to the manipulator. We first identify the cycles with the largest push, describe when they occur, and test whether the push succeeds; we then classify the traders active in them into types and measure who profits and who pays. Last, we rule out the alternative that the push is prediction-market makers hedging their binary exposure. Throughout this section, the sample is the five-minute-live period (P3), the same window as the within-period identification above.

\subsubsection{Identification of manipulated cycles}
\label{subsec:wallet_cohorts_push_id}
\label{subsec:wallet_cohorts_setup}

We identify the cycles with the largest push using a within-cycle, outcome-independent statistic:
\[
    \textit{PushIntensity}_c \;=\; \frac{|\,\text{order flow}_{29,c}\,|}{\operatorname{median}_{b\,\in\,\{0,\ldots,24\}} |\,\text{order flow}_{b,c}\,|},
\]
where order flow$_{b,c}$ is the Binance taker buy minus sell dollar volume in bin $b$ of cycle $c$. The numerator is the magnitude of the final-10s order flow; the denominator is the typical magnitude of order flow per body bin in the \emph{same} cycle. This normalization has three useful properties: (i) it is dimensionless and bounded below by zero; (ii) it is invariant to the overall scale of activity in the cycle, so a quiet and an active cycle with the same push-vs-body asymmetry score equally; and (iii) it is constructed entirely from the cycle's own intra-cycle Binance order flow, with no input from the resolution price, the Polymarket Up-token price, or the strike, so a trader's PnL is not mechanically related to its score.

We label a cycle a \emph{manipulated cycle} if its $\textit{PushIntensity}$ is at or above the 90th percentile of the cycle distribution, $\textit{PushIntensity} \geq 16.11$. Of the 16{,}073 cycles with on-chain Polymarket trading, this gives 1{,}613 manipulated cycles and 14{,}460 \emph{normal cycles}.\footnote{A further 55 cycles in the window, Polymarket outage periods and ultra-thin overnight cycles, carry a $\textit{PushIntensity}$ score from Binance order flow but have no Polymarket trades, and drop out of the trader-level analysis.} The top-decile cut is a convention, chosen large enough that the manipulated-cycle sample carries meaningful statistical power (1{,}613 cycles $\times$ on average $\sim$1{,}500 traders per cycle) and small enough that the cycles selected are visibly distinct from the body of the distribution (the 90th-percentile $\textit{PushIntensity}$ of $16$ is an order of magnitude larger than the median $\textit{PushIntensity}$ of $\sim 0.9$, the latter close to what would be expected under no asymmetry between the final 10s and the body).

Manipulated cycles differ from normal cycles in when and where they occur (\cref{app:manipulated_cycle_characteristics}, Table~\ref{tab:push_cycle_summary}). They are quieter away from the close, not busier: the typical body-bin order flow on Binance is \$18k, against \$43k in normal cycles. They fall disproportionately in thin-liquidity hours: 56\% fall in the Asia/overnight hours (vs 40\% of normal cycles) and 44\% on weekends (vs 27\%). And they concentrate where the outcome is still undecided: 18.7\% are \NTM, against 4.4\% of normal cycles, although $\textit{PushIntensity}$ uses no information from the Polymarket price.

\subsubsection{Does manipulation succeed?}

Does the push change the resolution? We classify each cycle by how confident Polymarket is just before the close: the favored-side price at TTM $=10$\,s, defined as $p_{\mathrm{fav}} = \max(p_{\mathrm{Up}},\, 1 - p_{\mathrm{Up}})$, ranges from $0.5$ (a coin flip) to $1$ (a settled outcome). Within manipulated cycles, the push can go in the \emph{same} direction as the favored side or in the \emph{opposite} direction, toward the underdog. Table~\ref{tab:push_characteristics} reports, by confidence bin and push direction, how often the side the push backs wins the resolution; the same side's win rate in normal cycles, which carry no push, is the baseline. Because manipulated cycles are selected on the size of the bin-29 flow, not on its direction or the outcome, the comparison is not mechanical.

The main pattern is that opposite-direction pushes succeed: a push against the favored side raises the underdog's win rate well above its normal-cycle rate in every confidence bin, most of all where Polymarket is most confident. In the $[0.90, 1.00)$ bin the underdog wins just 1.0\% of normal cycles but 34.2\% of manipulated ones: a push against a side the market priced at $\geq 90\%$ reverses the resolution roughly a third of the time, versus once in a hundred absent a push. Same-direction pushes, by contrast, matter only near the money: they add $\sim 18$\,pp to the favored side's win rate in the lowest confidence bin but are flat, even marginally negative, once the favored side already wins the large majority of normal cycles, where a push has little to add.

The second pattern concerns the direction of the pushes: the split between same- and opposite-direction pushes is close to even in \emph{every} confidence bin, including the most lopsided one, where Polymarket prices the favored side at $\geq 90\%$ yet half the pushes still go against it. Manipulators do not merely reinforce near-certain outcomes; they frequently trade to overturn them, and the expected return is attractive: an underdog priced below ten cents that wins a third of the time returns several times its stake in expectation. We quantify the realized profits below.

\begin{table}[!t]
    \centering
    \footnotesize
    \caption{\textbf{Manipulated-cycle outcome rates by Polymarket-favored confidence, P3 BTC 5\mn.} This table reports, by confidence bin and push direction, how often the side the push backs wins the resolution. Cycles are binned by the Polymarket-favored confidence at TTM~$=10$\,s, $p_{\mathrm{fav}} = \max(p_{\mathrm{Up}}, 1 - p_{\mathrm{Up}})$ (width 0.10, folded over Up/Down since the empirical Up-vs-Down asymmetry is negligible). Each bin has two rows for the two manipulation directions, defined relative to the Polymarket-favored side: \emph{Same-direction}, where the bin-29 Binance push agrees with the favored side, and \emph{Opposite-direction}, where it pushes against the favored side (toward the underdog). \emph{\# cycles} and \emph{\% cycles} give the number and share of the bin's manipulated cycles in each direction. \emph{Win rate} is the rate at which the side the push targets, the favored side for same-direction, the underdog for opposite-direction, resolves as the winner among the manipulated cycles of that direction; \emph{Win rate (normal cycles)} is the same side's unconditional win rate in the bin's normal cycles (which carry no push, hence no same/opposite split), serving as the no-manipulation baseline. Manipulated cycles are the top decile by $\textit{PushIntensity}$ ($\textit{PushIntensity} \geq 16.11$); normal cycles are the remaining 90\%.}
    \label{tab:push_characteristics}
    \begin{tabular}{clrrrr}
    \toprule
    \makecell{$p_{\mathrm{fav}}$\\range} & \makecell{Manipulation\\direction} & \# cycles & \% cycles & Win rate & \makecell{Win rate\\(normal cycles)} \\
    \midrule
    \multirow{2}{*}{{[0.50, 0.60)}} & Same-direction & 149 & 52\% & 76.5\% & 58.6\% \\
     & Opposite-direction & 136 & 48\% & 65.4\% & 41.4\% \\
    \addlinespace
    \multirow{2}{*}{{[0.60, 0.70)}} & Same-direction & 125 & 46\% & 76.0\% & 70.9\% \\
     & Opposite-direction & 145 & 54\% & 68.3\% & 29.1\% \\
    \addlinespace
    \multirow{2}{*}{{[0.70, 0.80)}} & Same-direction & 111 & 54\% & 82.9\% & 82.3\% \\
     & Opposite-direction & 94 & 46\% & 54.3\% & 17.7\% \\
    \addlinespace
    \multirow{2}{*}{{[0.80, 0.90)}} & Same-direction & 95 & 43\% & 92.6\% & 91.3\% \\
     & Opposite-direction & 125 & 57\% & 54.4\% & 8.7\% \\
    \addlinespace
    \multirow{2}{*}{{[0.90, 1.00)}} & Same-direction & 300 & 50\% & 96.7\% & 99.0\% \\
     & Opposite-direction & 304 & 50\% & 34.2\% & 1.0\% \\
    \bottomrule
\end{tabular}

\end{table}

\subsubsection{Trader types and PnL}
\label{subsec:wallet_cohorts_id}

\paragraph{PnL calculation.} To attribute profits and losses across trader types, we compute realized PnL at the wallet-cycle level. Each on-chain Polymarket fill identifies the two counterparties by Polygon wallet address; for each (wallet, cycle) we accumulate the wallet's signed cash flow and net Up/Down token positions from its fills, then settle the token positions at the Chainlink resolution price to obtain the cycle's realized PnL:
\[
    \text{PnL}_{w,c}
    \;=\; \Delta\text{cash}_{w,c}
    + N_{\text{Up},w,c}\,\mathbf{1}\{\text{Up wins}\}
    + N_{\text{Down},w,c}\,\mathbf{1}\{\text{Down wins}\},
\]
where $\Delta\text{cash}_{w,c}$ is the wallet's USDC inflow during the cycle (positive when sells exceed buys, negative otherwise) and each token settles at \$1 for the winning side and \$0 for the losing side. Because each 5\mn market has its own token IDs, a wallet's pre-cycle position is identically zero and each cycle is a closed loop. This is \emph{gross} PnL: exchange fees are not subtracted; the fee schedule changes mid-sample, and net-of-fee PnL is left to future work.\footnote{Fee-enabled crypto contracts charge only the aggressor: resting liquidity pays no exchange fee and is eligible for rebates funded out of taker-fee revenue. The taker fee is proportional to $p\,(1-p)$ in the contract price $p$, so it peaks at even odds, where short-horizon up/down contracts spend much of their lives. The schedule moved several times within our sample; the five-minute contracts were fee-enabled from launch, with a peak effective rate near 1.56\% at even odds, raised to roughly 1.80\% at the end of March 2026.} Settlement uses the Chainlink resolution price where the on-chain archive has it; for the small fraction of cycles with a missing Chainlink print, we fall back to the last trade in the 30 seconds after the close, by which time the Polymarket price has converged to the resolved outcome (convergence takes about ten seconds). The panel covers 16{,}073 cycles with on-chain Polymarket trading and 243{,}155 distinct wallets, which we treat as 243{,}155 traders (multi-wallet identity by a single trader is unobserved).

\paragraph{Trader-type identification.} We partition traders into three mutually exclusive types. A trader is labeled a \emph{manipulator} if it satisfies two criteria across the sample: it participated in at least five manipulated cycles, and its aggregate net PnL across those cycles is at least \$2{,}000. The filter selects on participation and profitability in manipulated cycles, not on manipulation itself: a trader could in principle profit in those cycles by luck or by unrelated information.

A trader is labeled a \emph{market maker} (MM) on the prediction-market contract if it does not pass the manipulator filter and satisfies three behavioral criteria, all defined on its trading across the sample. First, at least 85\% of its fills have it on the resting (maker) side rather than crossing the spread. Second, conditional on being active in a cycle, it averages at least 50 fills per cycle, roughly one fill every six seconds. Third, its per-cycle median directional inventory at the close, defined as the absolute difference between its Up- and Down-token positions as a fraction of the gross shares it traded in that cycle, is at most 20\%. The first two criteria capture passive execution and continuous quoting; the third separates liquidity provision from position-taking: an MM closes each cycle near delta-neutral rather than carrying a directional bet across resolution.

Every other trader is \emph{retail/other}, the residual category. Under this classification the trader population breaks down to 821 manipulators (0.34\% of the population), 227 MMs (0.09\%), and 242{,}107 retail/other traders (99.57\%).

\subsubsection{Profits and losses in manipulated cycles}
\label{subsec:wallet_cohorts_aggregate}

Who profits and who pays in the manipulated cycles? Table~\ref{tab:cohort_pnl} answers with four metrics per trader type, each split across manipulated and normal cycles: total net PnL, net PnL per cycle, the fraction of cycles in which the type loses money, and its share of the losses within a cycle. All four point to the same answer: the push transfers wealth from retail traders to manipulators, with market makers largely unaffected.

The PnL levels give the size of the transfer. Manipulators capture \$8.2\,M in aggregate across manipulated cycles, \$5{,}096 per manipulated cycle. Retail traders bear nearly all of it: they lose \$7.6\,M, \$4{,}715 per manipulated cycle. MMs largely avoid the transfer (a loss of \$0.6\,M, \$381 per cycle), consistent with the Section~\ref{sec:atm} finding that MMs withdraw their book into the NTM close: they are rarely the counterparty because they do not quote in the cycles where the push is most likely.

The remaining two metrics show that the transfer is broad-based, not driven by a few large-loss cycles. Retail is net negative in 65\% of manipulated cycles, against 48\% of normal ones. The last metric works at the cycle level: the gross loss pool, defined as the sum across trader types of their net losses within a cycle, equals the sum of their gains by zero-sum, and a type's loss share is its loss divided by that pool. Weighted across cycles by pool size, retail bears 70.5\% of manipulated-cycle losses versus 46.0\% in normal cycles. The manipulator cohort still carries 20.2\% of the manipulated-cycle pool even as the dominant net winner: that is the loss it books in the minority of cycles where its push fails, outweighed by its wins (in normal cycles, where it is more often a net loser, its share rises to 43.2\%). The MM share is small in both types of cycles, 9.3\% and 10.8\%, so losses pass from retail traders to manipulators directly.

Normal cycles show the reverse and confirm that the transfer is specific to manipulation. There, MMs collect the spread (\$3.1\,M, \$215 per cycle), retail pays the ordinary cost of trading (a loss of \$3.2\,M, \$221 per cycle), and manipulators break even (\$0.1\,M, \$6 per cycle): the cohort that earns \$5{,}096 per manipulated cycle has no advantage in normal cycles. Retail's per-cycle loss is roughly twenty times larger when a push is present, and its counterparty switches from the MMs earning the spread to the manipulators. The full trader-type attribution is robust to the manipulator thresholds (Appendix~\ref{app:cohort_sensitivity}).

\begin{table}[!t]
    \centering
    \footnotesize
    \setlength{\tabcolsep}{4.5pt}
    \caption{\textbf{Trader-type PnL in manipulated vs normal cycles, P3 BTC 5\mn.} This table reports four PnL metrics per trader type, split across manipulated and normal cycles. Trader types are mutually exclusive and exhaustive. \emph{Manipulator}: traders with $\geq$5 manipulated cycles and $\geq$\$2{,}000 aggregate manipulated-cycle PnL. \emph{MM}: residual of that filter satisfying passive-rate $\geq$85\%, $\geq$50 fills per cycle when active, and per-cycle median directional inventory $\leq$20\% of gross shares traded. \emph{Retail/other}: residual. Columns report four metrics, each split across manipulated (top-decile $\textit{PushIntensity}$) and normal (other) cycles. \emph{Total PnL} is the cohort's aggregate net PnL over all its (trader, cycle) observations of that type, in millions of dollars; within each cycle type it sums to zero across cohorts by construction. \emph{PnL per cycle} is Total PnL divided by the number of cycles of that type (1{,}613 manipulated, 14{,}460 normal), in dollars. \emph{Losing cycles} is the percentage of cycles in which the cohort is net negative, where the (cohort, cycle) cell aggregates a cohort's PnL across its member traders within a cycle. \emph{Loss share} is the cohort's percentage share of the \emph{gross} loss pool (the sum across cohorts of their losses within a cycle), aggregated across cycles with pool-size weights, so it sums to 100\% down each regime, and a net-winning cohort still carries a positive share from the cycles in which it happens to lose.}
    \label{tab:cohort_pnl}
    \begin{tabular}{lrrrrrrrr}
    \toprule
     & \multicolumn{2}{c}{Total PnL (\$ million)} & \multicolumn{2}{c}{PnL per cycle (\$)} & \multicolumn{2}{c}{Losing cycles (\%)} & \multicolumn{2}{c}{Loss share (\%)} \\
    \cmidrule(lr){2-3}\cmidrule(lr){4-5}\cmidrule(lr){6-7}\cmidrule(lr){8-9}
     & Manip. & Normal & Manip. & Normal & Manip. & Normal & Manip. & Normal \\
    \midrule
    Manipulator & 8.22 & 0.09 & 5{,}096 & 6 & 38.0 & 62.4 & 20.2 & 43.2 \\
    Market maker & $-$0.62 & 3.11 & $-$381 & 215 & 58.6 & 37.7 & 9.3 & 10.8 \\
    Retail / other & $-$7.61 & $-$3.20 & $-$4{,}715 & $-$221 & 65.3 & 47.7 & 70.5 & 46.0 \\
    \bottomrule
\end{tabular}

\end{table}

\subsubsection{Ruling out the gamma-hedging channel}
\label{subsec:wallet_cohorts_hedging}

The leading alternative explanation for the bin-29 push is hedging: a Polymarket maker who has sold Up tokens through the body offsets the resulting short by trading BTC, and that hedge, not manipulation, is what moves the spot across the strike. The channel requires no intent. Retail buys Up tokens through the body; makers absorb the demand by selling Up; the accumulated short must be hedged in the spot; the hedge pushes the close across the strike, and the body-period Up buyers profit. Two facts rule this out. The first, and the sharper, is the state of the contract when the push occurs: in the cycles a push most often flips, the maker has essentially nothing to hedge. The second is the timing of the spot flow.

\paragraph{There is nothing to hedge in the cycles a push flips.} The delta of a binary contract, its sensitivity to the spot price, falls toward zero away from the strike. When the favored side is priced near $0.5$ the contract behaves almost like the spot itself, and a maker who is short it carries a real directional exposure; but once the price moves toward $0$ or $1$, the contract is all but settled, its value barely responds to a further spot move, and the short position it represents is close to delta-flat. A maker holding such a position has almost nothing to hedge.

Yet these near-certain cycles are the ones the push overturns. In the $[0.90, 1.00)$ confidence bin --- where Polymarket priced the favored side at $\geq 90\%$ ten seconds before the close, so the binary's delta is close to zero --- a push against the favored side reverses the resolution $34.2\%$ of the time, against $1.0\%$ in cycles with no push (Table~\ref{tab:push_characteristics}). A spot trade placed when the maker has almost no exposure to hedge is not a hedge. Whatever hedging response the spot flow contains can only come \emph{after} the push, once the push has moved the price back toward the strike and restored the exposure that hedging requires; it cannot be the push's cause. The hedging channel is coherent only in the small near-the-money slice where the binary's delta is still material (about $6\%$ of cycles); it does not apply to the near-certain cycles that the push flips one time in three.

\paragraph{The timing of the spot flow is wrong for a hedge.} The second fact concerns the cycles where there \emph{is} something to hedge. A maker who hedges dynamically buys the spot as its short accumulates. The MM cohort's short builds steadily through the cycle, to roughly $1{,}350$ shares as the final ten seconds begin, but the matching spot buying does not appear: cumulative Binance order flow is essentially zero at second $215$ of the $300$-second cycle, and the final ten seconds alone carry roughly \$1.7\,M. A spot position built only in the final seconds is inconsistent with a hedge of an exposure that accumulated over the cycle, and consistent with a trade timed to move the settlement price. \Cref{app:pushdown_flow} reports the position trajectories and considers a hedge deferred to the very close.

\subsection{Can contract design remove manipulation?}
\label{sec:design_levers}

In the previous section, we document the manipulation: a final-seconds push on the spot market that flips resolutions and transfers wealth from retail traders to manipulators. This section asks whether a different contract design can remove or mitigate it. We study two design levers: the contract horizon, by comparing the five-minute contract with the fifteen-minute contract, and the settlement price, using Polymarket's switch from a single print to a time-weighted average in August 2026.

\subsubsection{A longer horizon: the fifteen-minute contract}
\label{sec:fifteen_min}

The model predicts that the manipulation footprint weakens as the settlement horizon lengthens, because a longer window gives the asset's own price discovery more time to unfold, so a fixed-size push is less likely to be pivotal and manipulation opportunities become rarer ex ante (Section~\ref{sec:predictions}). The fifteen-minute contract provides the test. It went live earlier (in period P2), resolves based on the same Chainlink oracle price, and, during P2, was the most heavily traded of the crypto up/down contracts (Appendix~\ref{app:volume}), so a weak footprint cannot be attributed to a small prize.

\paragraph{Identification.} Within the five-minute clock, the fifteen-minute (and four-hour) contract settles only at the closes that coincide with a 15-minute mark: the \emph{overlapping} cycles, one in three; the other closes carry no fifteen-minute settlement. The fifteen-minute effect is therefore identified by contrasting overlapping with non-overlapping closes. One complication forces a richer design: 15-minute marks might be intrinsically special on the spot market (round-time execution algorithms cluster there), so a raw overlap contrast would confound the fifteen-minute settlement with this clock effect.

We difference it against the same contrast in the pre-Polymarket baseline: on the (cycle, bin) panel over P1 and P2, restricted to bins $\{0\text{--}24, 29\}$ as before, we estimate
\[
    \log y_{c,b} \;=\; \theta_{t_c,\,b} + \alpha_d + \mu_h + \delta_{P2}\,\mathbf{1}\{b{=}29\}{\cdot}\mathbf{1}\{P2\} + \delta_{P2,\,5\mn\cap15\mn}\,\mathbf{1}\{b{=}29\}{\cdot}\mathbf{1}\{P2\}{\cdot}D_{5\mn\,\cap\,15\mn,c} + \varepsilon_{c,b},
\]
where $D_{5\mn\,\cap\,15\mn,c}$ indicates an overlapping cycle. The specification modifies the activity regression of Section~\ref{subsec:launch_impact} in two ways. First, the bin fixed effect is replaced by a cycle-type$\times$bin fixed effect $\theta_{t_c,\,b}$, which gives overlapping and non-overlapping cycles each their own intra-cycle profile and thereby absorbs the clock effect, whatever its shape. Second, the bin-29 shift is split by overlap status: $\delta_{P2}$ is the bin-29-versus-body increase in P2 at non-overlapping closes, a pure period effect, and $\delta_{P2,\,5\mn\cap15\mn}$ is the \emph{additional} increase at overlapping closes, the difference-in-differences $(\text{P2}-\text{P1})\times(\text{overlap}-\text{non-overlap})$ that isolates the fifteen-minute settlement effect, netting out both the period and the clock.

The reversal regression of Section~\ref{subsec:launch_impact} is augmented in the same way. At each lag $j$, estimated separately on P1 and P2 cycles,
\[
    \begin{aligned}
        r_{j,c} \;=\;\; & \alpha_d + \mu_h + \beta\,r_{29,c} + \beta_{5\mn\cap15\mn}\,r_{29,c}\,D_{5\mn\,\cap\,15\mn,c}                                                                   \\
                        & + \delta_{P2}\,r_{29,c}\,\mathbf{1}\{P2\} + \delta_{P2\,\times\,(5\mn\cap15\mn)}\,r_{29,c}\,\mathbf{1}\{P2\}{\cdot}D_{5\mn\,\cap\,15\mn,c} + \varepsilon_{j,c},
    \end{aligned}
\]
with the lower-order level terms $D_{5\mn\,\cap\,15\mn,c}$ and $\mathbf{1}\{P2\}{\cdot}D_{5\mn\,\cap\,15\mn,c}$ included; the $\mathbf{1}\{P2\}$ level is absorbed by the date fixed effects. The four slopes mirror the activity design: $\beta$ is the pre-Polymarket reversal slope at non-overlapping closes and $\beta_{5\mn\cap15\mn}$ its overlap increment, the clock effect on the slope; $\delta_{P2}$ is the P2 slope change at non-overlapping closes, the period effect; and $\delta_{P2\,\times\,(5\mn\cap15\mn)}$ is the additional P2 change at overlapping closes, the difference-in-differences that gives the fifteen-minute settlement effect on the reversal slope.

\paragraph{Activity and reversal.} Table~\ref{tab:p2_activity} reports the activity difference-in-differences over P1 and P2, and Table~\ref{tab:p2_reversal} the reversal one. The fifteen-minute settlement increment $\delta_{P2,\,5\mn\cap15\mn}$ is small on every metric: an order-flow increase of 7\% ($\hat\delta = 0.068$) against 50\% for the five-minute contract, a depth increase of 2\% against 5\%, no detectable absolute-return effect, and no post-settlement reversal at $j=+1$. Smaller but not exactly zero is what the horizon mechanism predicts: a longer window shrinks the incentive rather than eliminating it.

\begin{table}[!t]
    \centering
    \footnotesize
    \caption{\textbf{Fifteen-minute settlement: near-settlement activity, pre-Polymarket vs P2 by overlap.} This table reports cycle-level OLS estimates on the (cycle, bin) panel restricted to bins $\{0\text{--}24,29\}$ and to the pre-Polymarket period (P1, Jul--Oct 2025) and the +15m/4h period (P2, Oct 2025 -- Feb 2026). $\delta_{P2}$ is the bin-29-versus-body activity increase in P2 non-overlapping cycles relative to P1; $\delta_{P2,\,5\mn\cap15\mn}$ is the \emph{additional} increase at overlapping closes (those landing on a 15-min mark, where the fifteen-minute contract settles), the difference-in-differences $(\text{P2}-\text{P1})\times(\text{overlap}-\text{non-overlap})$ isolating the fifteen-minute settlement effect. Metrics and standard errors as in Table~\ref{tab:reg_uncond}; fixed effects are date, hour-of-day, and cycle-type$\times$bin. Coefficients are log differences. Significance: $^{*}\,p<0.10$, $^{**}\,p<0.05$, $^{***}\,p<0.01$.}
    \label{tab:p2_activity}
    \begin{tabular}{l S[table-format=-1.3] S[table-format=-1.3] S[table-format=-1.3]}
    \toprule
     & {(1)} & {(2)} & {(3)} \\
     & {Order flow} & {Abs. return} & {Depth} \\
    \midrule
    $\delta_{P2}$ & -0.062{\sym{***}} & -0.091{\sym{***}} & 0.000 \\
     & (0.017) & (0.019) & (0.003) \\[4pt]
    $\delta_{P2,\,5\mn\cap15\mn}$ & 0.068{\sym{**}} & 0.042 & 0.021{\sym{***}} \\
     & (0.032) & (0.035) & (0.005) \\
    \midrule
    Date FE & {Yes} & {Yes} & {Yes} \\
    Hour-of-day FE & {Yes} & {Yes} & {Yes} \\
    Cycle-type $\times$ Bin FE & {Yes} & {Yes} & {Yes} \\
    Observations & {1{,}684{,}789} & {1{,}071{,}666} & {1{,}684{,}782} \\
    Clusters (cycle) & {64{,}800} & {64{,}391} & {64{,}800} \\
    \bottomrule
\end{tabular}

\end{table}

\begin{table}[!t]
    \centering
    \footnotesize
    \caption{\textbf{Fifteen-minute settlement: post-settlement reversal, pre-Polymarket vs P2 by overlap.} This table reports, at each lag $j \in \{-3,\ldots,+3\}$, a cycle-level OLS over P1 and P2 cycles. $\beta$ is the pre-Polymarket reversal slope on $r_{29}$ at non-overlapping closes, $\beta_{5\mn\cap15\mn}$ its pre-Polymarket overlap increment (the clock effect at 15-min marks), $\delta_{P2}$ the P2 slope shift at non-overlapping closes, and $\delta_{P2\,\times\,(5\mn\cap15\mn)}$ the P2 overlap increment: the fifteen-minute settlement effect, a difference-in-differences. A negative slope is reversal. Date and hour-of-day fixed effects; HC1 standard errors in parentheses clustered by date. Negative $j$ = same-cycle bin $(29+j)$; positive $j$ = bin $(j-1)$ of the next cycle. Significance: $^{*}\,p<0.10$, $^{**}\,p<0.05$, $^{***}\,p<0.01$.}
    \label{tab:p2_reversal}
    \begin{tabular}{l S[table-format=-1.3] S[table-format=-1.3] S[table-format=-1.3] S[table-format=-1.3] S[table-format=-1.3] S[table-format=-1.3]}
    \toprule
     & \multicolumn{3}{c}{Pre-close} & \multicolumn{3}{c}{Post-close} \\
    \cmidrule(lr){2-4} \cmidrule(lr){5-7}
     & {$j = -3$} & {$j = -2$} & {$j = -1$} & {$j = +1$} & {$j = +2$} & {$j = +3$} \\
    \midrule
    $\beta$ & -0.034 & 0.000 & 0.036 & 0.073{\sym{***}} & -0.000 & 0.010 \\
     & (0.047) & (0.018) & (0.022) & (0.028) & (0.019) & (0.018) \\[2pt]
    $\beta_{5\mn\cap15\mn}$ & 0.066 & 0.053{\sym{*}} & 0.050 & -0.081 & -0.110{\sym{**}} & -0.039 \\
     & (0.047) & (0.030) & (0.034) & (0.077) & (0.045) & (0.047) \\[2pt]
    $\delta_{P2}$ & 0.033 & -0.017 & 0.015 & -0.116{\sym{***}} & 0.014 & 0.011 \\
     & (0.052) & (0.033) & (0.030) & (0.036) & (0.024) & (0.025) \\[2pt]
    $\delta_{P2\,\times\,(5\mn\cap15\mn)}$ & 0.002 & -0.058 & -0.081{\sym{*}} & 0.024 & 0.066 & -0.005 \\
     & (0.085) & (0.046) & (0.043) & (0.087) & (0.060) & (0.059) \\
    \midrule
    Date FE & {Yes} & {Yes} & {Yes} & {Yes} & {Yes} & {Yes} \\
    Hour-of-day FE & {Yes} & {Yes} & {Yes} & {Yes} & {Yes} & {Yes} \\
    Observations & \multicolumn{6}{c}{64{,}575} \\
    Clusters (UTC date) & \multicolumn{6}{c}{225} \\
    \bottomrule
\end{tabular}

\end{table}

Conditioning on moneyness does not change the conclusion: even the fifteen-minute NTM cycles show at most a modest activity increase and no reversal, far below the five-minute footprint (Appendix~\ref{app:p2_placebo}). The horizon prediction is borne out: the footprint that emerges with the five-minute contract, the shortest-horizon contract, appears at the fifteen-minute horizon only in strongly attenuated form.

\subsubsection{An averaged settlement price: the TWAP switch}
\label{sec:twap}

In August 2026, Polymarket changed the settlement rule of its crypto up/down contracts: the strike and the settlement price are no longer single Chainlink prints at the open and the close but time-weighted average prices (TWAPs) over the preceding 60 seconds, which we call the \emph{settlement window}.\footnote{The switch took effect on August 7, 2026, with a 30-second average for the five-minute contract and a 60-second average for the fifteen-minute and four-hour contracts. From August 14, 2026, the five-minute contract also settles on a 60-second average.} Under a 60-second average, a push in the final ten seconds moves the settlement price by only a sixth as much as it moves a single print. The switch therefore tests the second lever: does averaging the settlement price remove the manipulation?

We compare two 79-day windows on either side of the switch: the \emph{pre-TWAP window}, May 20 to August 6, 2026, and the \emph{TWAP window}, August 14 to October 31, 2026.\footnote{The TWAP window starts on August 14, the first day of the 60-second rule for the five-minute contract, and the pre-TWAP window covers the 79 days before the switch, so that the two have the same length. We leave out the transition week of August 7 to 13, in which the five-minute contract settled on a 30-second average; April 9 to May 19, 2026 is not used.} We report P3 alongside as a reference; like the pre-TWAP window, it settles on single prints. The contract throughout is the BTC five-minute contract. We first ask whether the footprint survives the switch, then whether the final-ten-second pushes that remain still decide outcomes and are still profitable for the wallets that profited from them. In Appendix~\ref{app:twap_adapt}, we further analyze whether manipulators adapt by moving the push from the final ten seconds into the 60-second settlement window.

\paragraph{The footprint after the switch.}
\label{sec:twap_footprint}

Figure~\ref{fig:twap_fold} plots the intra-cycle profiles of the three Binance metrics: absolute order flow, absolute return, and depth, for P3, the pre-TWAP window, and the TWAP window. As it shows, the manipulation footprint largely disappears with the switch: the final-bin spike of P3 remains in the pre-TWAP window and is absent in the TWAP window for all three metrics.

To quantify the change, we estimate equation~\eqref{eq:activity_did} once per window $w \in \{\text{P3}, \text{pre-TWAP}, \text{TWAP}\}$, each time on the cycles of P1 and $w$ and with $\mathbf{1}\{w\}$ in place of $\mathbf{1}\{P3\}$. The coefficient $\delta_w$ is the bin-29-versus-body activity increase in window $w$ relative to the pre-launch baseline; for P3 it is the estimate of Table~\ref{tab:reg_uncond}.  Table~\ref{tab:twap_activity} reports the estimates. Relative to the body, the typical cycle's bin-29 order flow is 50\% higher than in the baseline in P3, 30\% higher in the pre-TWAP window, and only 13\% higher in the TWAP window. The absolute return is 15\% and 7\% higher in the two point-settled windows and indistinguishable from the baseline under the TWAP. The depth increase is 5\% in both point-settled windows and under 1\% under the TWAP. Thus, after the switch the final-bin order-flow increase is less than half as large as before, and the order flow that still arrives in the last ten seconds no longer moves the spot price.

\begin{figure}[!t]
    \centering
    \includegraphics[width=\textwidth]{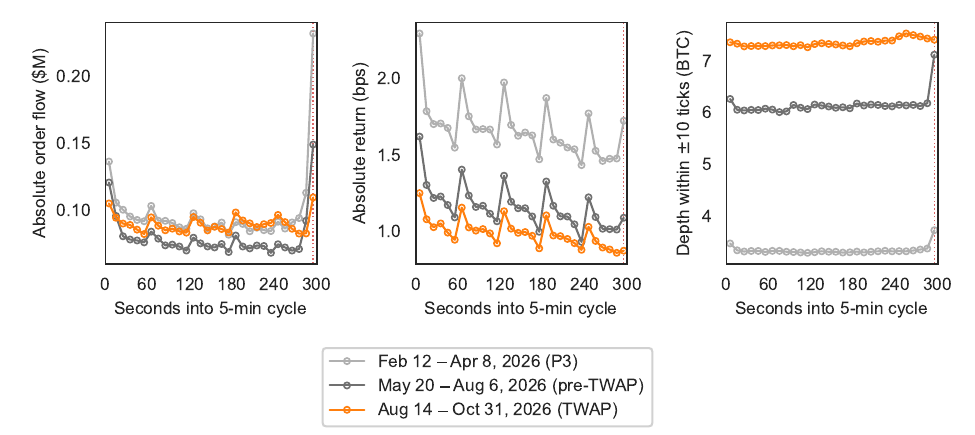}
    \caption{\textbf{Binance BTC intra-cycle profiles by settlement window.} Each panel folds every 10\,s Binance BTC bucket onto its position in the rolling five-minute clock cycle and plots the metric pooled over all cycles in each window: P3 (February 12 to April 8, 2026), the pre-TWAP window (May 20 to August 6, 2026) and the TWAP window (August 14 to October 31, 2026). \emph{Left}: absolute order flow per cycle (\$M), the per-cycle $|$net signed taker dollar flow$|$ in the bin. \emph{Center}: absolute 10\,s return ($|\Delta \log \mathrm{mid}|$, bps). \emph{Right}: time-weighted depth within $\pm$10 ticks of the mid (cumulative resting size, bid plus ask, BTC). The dotted vertical line marks the final 10\,s before the close. October 10, 2025 is excluded from all Binance panels.}
    \label{fig:twap_fold}
\end{figure}

\begin{table}[!t]
    \centering
    \footnotesize
    \caption{\textbf{Near-settlement activity relative to the pre-Polymarket baseline, by settlement window.} This table repeats the regression of Table~\ref{tab:reg_uncond} with each later window in place of P3: for window $w$, cycle-level OLS on the (cycle, bin) panel restricted to bins $\{0\text{--}24, 29\}$ and to the pre-Polymarket period (P1, Jul--Oct 2025) and $w$, $\log y_{c,b} = \gamma_b + \alpha_d + \mu_h + \delta_w\,\mathbf{1}\{b{=}29\}{\cdot}\mathbf{1}\{w\} + \varepsilon_{c,b}$, where $w$ is P3 (February 12 to April 8, 2026), the pre-TWAP window (May 20 to August 6, 2026), or the TWAP window (August 14 to October 31, 2026). Each row is its own regression; the P3 row is Table~\ref{tab:reg_uncond}. $\delta_w$ is the bin-29-versus-body activity increase in window $w$ relative to the pre-Polymarket baseline. The dependent variable is $\log y_{c,b}$, the log of the metric named in the column header: \emph{Order flow}, per-cycle absolute net signed taker dollar flow; \emph{Abs.\ return}, absolute 10\,s log return (bps); \emph{Depth}, time-weighted resting size within $\pm$10 ticks of the mid. Fixed effects are date, hour-of-day, and bin; the specification pools all cycles. Coefficients are log differences. One observation per (cycle, bin) cell; the footer reports each regression's sample. HC1 standard errors in parentheses are clustered at the cycle level. Significance: $^{*}\,p<0.10$, $^{**}\,p<0.05$, $^{***}\,p<0.01$.}
    \label{tab:twap_activity}
    \begin{tabular}{l S[table-format=-1.3] S[table-format=-1.3] S[table-format=-1.3]}
    \toprule
     & {(1)} & {(2)} & {(3)} \\
     & {Order flow} & {Abs. return} & {Depth} \\
    \midrule
    $\delta_{\text{P3}}$: $\mathbf{1}\{b{=}29\}\times\mathbf{1}\{\text{P3}\}$ & 0.406{\sym{***}} & 0.140{\sym{***}} & 0.048{\sym{***}} \\
     & (0.019) & (0.018) & (0.003) \\[4pt]
    $\delta_{\text{pre-TWAP}}$: $\mathbf{1}\{b{=}29\}\times\mathbf{1}\{\text{pre-TWAP}\}$ & 0.265{\sym{***}} & 0.066{\sym{***}} & 0.052{\sym{***}} \\
     & (0.018) & (0.017) & (0.003) \\[4pt]
    $\delta_{\text{TWAP}}$: $\mathbf{1}\{b{=}29\}\times\mathbf{1}\{\text{TWAP}\}$ & 0.123{\sym{***}} & 0.002 & 0.008{\sym{***}} \\
     & (0.016) & (0.018) & (0.003) \\
    \midrule
    Date FE & {Yes} & {Yes} & {Yes} \\
    Hour-of-day FE & {Yes} & {Yes} & {Yes} \\
    Bin FE & {Yes} & {Yes} & {Yes} \\
    Observations, P1 $\cup$ P3 & {1{,}168{,}118} & {738{,}313} & {1{,}168{,}114} \\
    Observations, P1 $\cup$ pre-TWAP & {1{,}340{,}334} & {807{,}516} & {1{,}340{,}310} \\
    Observations, P1 $\cup$ TWAP & {1{,}340{,}336} & {773{,}023} & {1{,}340{,}332} \\
    \bottomrule
\end{tabular}

\end{table}

We next turn to the post-settlement reversal and estimate equation~\eqref{eq:reversal_did} in the same way, separately at each lag $j$ and for each window $w$, so that $\delta_w$ is the change in the reversal slope in window $w$ relative to the pre-launch baseline. Figure~\ref{fig:twap_rev} plots the slopes by window and Table~\ref{tab:twap_rev} reports the estimates. At $j=+1$, relative to a baseline with no reversal, the slope changes by $\delta_{\text{P3}} = -0.15$ in P3, while the change is indistinguishable from zero in the TWAP window: the reversal is absent after the switch.

\begin{figure}[!t]
    \centering
    \includegraphics[width=0.6\textwidth]{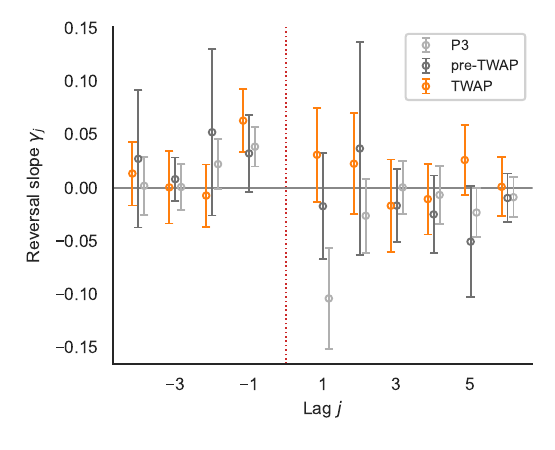}
    \caption{\textbf{Reversal slope $\gamma_j$ by lag and settlement window.} Markers are the slope of the lag-$j$ 10\,s Binance BTC log return on the closing-bin return $r_{29}$, from one regression per lag with $r_{29}$ interacted with the window indicator, date fixed effects, and hour-of-day fixed effects interacted with window, with $\pm 1.96$ date-clustered standard errors; lag 0 is omitted. $j<0$ denotes bin $29+j$ of the same cycle, $j>0$ bin $j-1$ of the next cycle. Series: P3 (February 12 to April 8, 2026), the pre-TWAP window (May 20 to August 6, 2026), and the TWAP window (August 14 to October 31, 2026). The dotted vertical line marks the close.}
    \label{fig:twap_rev}
\end{figure}

\begin{table}[!t]
    \centering
    \footnotesize
    \caption{\textbf{Post-settlement reversal relative to the pre-Polymarket baseline, by settlement window.} This table repeats the regression of Table~\ref{tab:reversal_uncond} with each later window in place of P3: for window $w$ and lag $j \in \{-3,-2,-1,+1,+2,+3\}$, the cycle-level OLS $r_j = \alpha_d + \mu_h + \beta\,r_{29} + \delta_w\,r_{29}\mathbf{1}\{w\} + \varepsilon$ on the pre-Polymarket period (P1, Jul 1 -- Oct 8, 2025) and $w$, where $w$ is P3 (February 12 to April 8, 2026), the pre-TWAP window (May 20 to August 6, 2026), or the TWAP window (August 14 to October 31, 2026). $r_{29}$ is the final-10\,s log return of the cycle and $r_j$ the 10\,s log return at lag $j$: $j<0$ denotes bin $29+j$ of the same cycle, $j>0$ bin $j-1$ of the next cycle within the same UTC date. $\beta$ is the pre-Polymarket reversal slope, reported from the P1 and P3 regression, and each $\delta_w$ the change in the slope in window $w$ from its own regression; the $\beta$ and $\delta_{P3}$ rows are Table~\ref{tab:reversal_uncond}. $\alpha_d$ and $\mu_h$ are date and hour-of-day fixed effects. One observation per cycle; the footer reports each regression's sample at $j=+1$. Standard errors in parentheses are clustered by date. Significance: $^{*}\,p<0.10$, $^{**}\,p<0.05$, $^{***}\,p<0.01$.}
    \label{tab:twap_rev}
    \begin{tabular}{l S[table-format=-1.3] S[table-format=-1.3] S[table-format=-1.3] S[table-format=-1.3] S[table-format=-1.3] S[table-format=-1.3]}
    \toprule
     & \multicolumn{3}{c}{Pre-close} & \multicolumn{3}{c}{Post-close} \\
    \cmidrule(lr){2-4} \cmidrule(lr){5-7}
     & {$j = -3$} & {$j = -2$} & {$j = -1$} & {$j = +1$} & {$j = +2$} & {$j = +3$} \\
    \midrule
    $\beta$ & -0.014 & 0.015 & 0.051{\sym{***}} & 0.049 & -0.033{\sym{*}} & -0.002 \\
     & (0.034) & (0.014) & (0.018) & (0.032) & (0.019) & (0.019) \\[2pt]
    $\delta_{\text{P3}}$ & 0.015 & 0.007 & -0.012 & -0.152{\sym{***}} & 0.007 & 0.002 \\
     & (0.036) & (0.019) & (0.020) & (0.040) & (0.026) & (0.023) \\[2pt]
    $\delta_{\text{pre-TWAP}}$ & 0.023 & 0.037 & -0.018 & -0.066 & 0.070 & -0.015 \\
     & (0.036) & (0.042) & (0.026) & (0.041) & (0.054) & (0.026) \\[2pt]
    $\delta_{\text{TWAP}}$ & 0.015 & -0.023 & 0.013 & -0.017 & 0.056{\sym{*}} & -0.015 \\
     & (0.038) & (0.021) & (0.023) & (0.039) & (0.031) & (0.029) \\
    \midrule
    Date FE & {Yes} & {Yes} & {Yes} & {Yes} & {Yes} & {Yes} \\
    Hour-of-day FE & {Yes} & {Yes} & {Yes} & {Yes} & {Yes} & {Yes} \\
    Observations, P1 $\cup$ P3 & \multicolumn{6}{c}{44{,}772 (156 dates)} \\
    Observations, P1 $\cup$ pre-TWAP & \multicolumn{6}{c}{51{,}371 (179 dates)} \\
    Observations, P1 $\cup$ TWAP & \multicolumn{6}{c}{51{,}373 (179 dates)} \\
    \bottomrule
\end{tabular}

\end{table}

As in Section~\ref{sec:atm}, the sharper test compares \NTM with \FFM cycles within each window, which removes any change common to all cycles in a window and isolates the cycles in which a push can flip the resolution. If averaging removes the manipulation, the concentration of the footprint in \NTM cycles should disappear under the TWAP, and it does. As Appendix Table~\ref{tab:twap_cond} shows, the additional bin-29 order flow in \NTM over \FFM cycles is $1.18$ log points in the pre-TWAP window, about 3.3 times the \FFM increase on the typical cycle, and it is indistinguishable from zero in the TWAP window. Appendix Table~\ref{tab:twap_rev_ntm} shows the same for the reversal: the \NTM increment to the reversal slope at $j=+1$ is $-0.18$ in the pre-TWAP window and insignificant in the TWAP window. Both pre-TWAP estimates are close to their P3 values of $1.37$ and $-0.15$, so the weaker pooled footprint in the pre-TWAP window, with its smaller order-flow increase and insignificant change in the reversal slope, comes from \FFM cycles and not from the cycles in which a push can flip the resolution. Appendix~\ref{app:twap_cond} gives the details.

\paragraph{Outcomes and profits after the switch.}
\label{sec:twap_outcomes}

Do the final-bin pushes that remain still decide outcomes? Table~\ref{tab:twap_push_char} repeats Table~\ref{tab:push_characteristics} for the two windows. We hold the definition of a manipulated cycle fixed at the P3 threshold, $\textit{PushIntensity} \geq 16.11$, which classifies 10.3\% of cycles in the pre-TWAP window and 6.2\% in the TWAP window as manipulated, and we measure the favored-side price $p_{\mathrm{fav}}$ ten seconds before the close under point settlement and 60 seconds before the close, when the settlement window opens, under the TWAP.

Before the switch, the outcome rates resemble those of P3: a push against the favored side wins in 49\% to 66\% of manipulated cycles across the four lower confidence bins, compared to 10\% to 40\% in normal cycles, and in 27\% of the cycles in which the market was at least 90\% confident, compared to 0.7\%. In the TWAP window, a push of the same size no longer raises the underdog's win rate. The underdog wins 24\%, 16\%, 23\%, and 12\% of manipulated cycles in the four lower bins, at or below its normal-cycle rates of 46\%, 33\%, 24\%, and 13\%; in the top bin it wins 2.5\% compared to 2.0\%. Same-direction pushes, which before the switch add to the favored side's win rate near the money as in P3, are now within a few percentage points of the normal-cycle rate in every bin.

\begin{table}[!t]
    \centering
    \footnotesize
    \caption{\textbf{Manipulated-cycle outcome rates by Polymarket-favored confidence, pre-TWAP vs.\ TWAP.} This table repeats Table~\ref{tab:push_characteristics} for the BTC five-minute contract in the pre-TWAP window (May 20 to August 6, 2026, left block) and the TWAP window (August 14 to October 31, 2026, right block). Manipulated cycles are those with $\textit{PushIntensity} \geq 16.11$, the P3 top-decile threshold of Table~\ref{tab:push_characteristics} applied unchanged (10.3\% and 6.2\% of the two windows' cycles); normal cycles are the rest. Cycles are binned by the Polymarket-favored confidence $p_{\mathrm{fav}} = \max(p_{\mathrm{Up}}, 1 - p_{\mathrm{Up}})$, measured at TTM~$=10$\,s under point settlement as in Table~\ref{tab:push_characteristics} and at TTM~$=60$\,s under the TWAP, the last trade at the opening of the 60\,s settlement window (at TTM~$=10$\,s five sixths of the TWAP window has already printed). The push direction is the sign of the bin-29 Binance signed notional; \emph{Same-direction} and \emph{Opposite-direction}, \emph{Win rate}, and \emph{Win rate (normal cycles)} are defined as in Table~\ref{tab:push_characteristics}. The pre-TWAP window has 2,123 manipulated cycles with a confidence label and the TWAP window 1,415.}
    \label{tab:twap_push_char}
    \begin{tabular}{clcccc}
    \toprule
     & & \multicolumn{2}{c}{Pre-TWAP (May 20 -- Aug 6)} & \multicolumn{2}{c}{TWAP (Aug 14 -- Oct 31)} \\
    \cmidrule(lr){3-4}\cmidrule(lr){5-6}
    \makecell{$p_{\mathrm{fav}}$\\range} & \makecell{Manipulation\\direction} & Win rate & \makecell{Win rate\\(normal cycles)} & Win rate & \makecell{Win rate\\(normal cycles)} \\
    \midrule
    \multirow{2}{*}{{[0.50, 0.60)}} & Same-direction & 78.0\% & 59.6\% & 50.0\% & 53.7\% \\
     & Opposite-direction & 66.3\% & 40.4\% & 23.5\% & 46.3\% \\
    \addlinespace
    \multirow{2}{*}{{[0.60, 0.70)}} & Same-direction & 77.5\% & 67.8\% & 72.7\% & 67.1\% \\
     & Opposite-direction & 58.8\% & 32.2\% & 16.0\% & 32.9\% \\
    \addlinespace
    \multirow{2}{*}{{[0.70, 0.80)}} & Same-direction & 79.1\% & 81.5\% & 74.4\% & 75.6\% \\
     & Opposite-direction & 49.4\% & 18.5\% & 22.5\% & 24.4\% \\
    \addlinespace
    \multirow{2}{*}{{[0.80, 0.90)}} & Same-direction & 91.6\% & 89.8\% & 90.0\% & 86.8\% \\
     & Opposite-direction & 49.2\% & 10.2\% & 12.0\% & 13.2\% \\
    \addlinespace
    \multirow{2}{*}{{[0.90, 1.00)}} & Same-direction & 98.1\% & 99.3\% & 96.7\% & 98.0\% \\
     & Opposite-direction & 27.1\% & 0.7\% & 2.5\% & 2.0\% \\
    \bottomrule
\end{tabular}

\end{table}

The remaining pushes are also no longer profitable for the wallets that profited from them before the switch. We apply the manipulator rule of Section~\ref{sec:wallet_cohorts}, at least five manipulated cycles and at least \$2{,}000 of net PnL on them, to the pre-TWAP window with the same fixed threshold, and follow the resulting 921 wallets into the TWAP window. As Table~\ref{tab:twap_cohort_ft} shows, before the switch they resemble the P3 manipulators of Table~\ref{tab:cohort_pnl}: they gain \$7.4 million on manipulated cycles and lose \$1.1 million on normal ones, earn \$3{,}168 per manipulated cycle, and are net positive as a group in 65\% of manipulated cycles compared to 41\% of normal ones. After the switch, 323 of them still trade, in about 1{,}850 cycles each and with 7\% of the contract's volume, but they no longer earn more in manipulated cycles: they make \$8{,}000 on those cycles in total, \$6 per cycle, and are net positive in 45\% of manipulated cycles compared to 47\% of normal ones. Over all cycles the 921 wallets earn \$0.4 million in the TWAP window, compared to \$6.3 million before.

\begin{table}[!t]
    \centering
    \footnotesize
    \caption{\textbf{The pre-TWAP manipulator cohort across the switch.} The 921 wallets labelled manipulators in the pre-TWAP window (at least five manipulated cycles and at least \$2,000 net PnL on them) tracked in the pre-TWAP window (May 20 to August 6, 2026) and the TWAP window (August 14 to October 31, 2026). Manipulated cycles are those with $\textit{PushIntensity} \geq 16.11$ in every window; normal cycles are the rest. ``Wallets active'' counts cohort wallets with at least one fill in the window; ``Cycles per wallet'' is the mean number of cycles traded per active wallet; ``Share of market volume'' is the cohort's USDC notional over all wallets' notional. ``Net PnL'' rows sum per-(cycle, wallet) gross PnL over the cohort (closed-loop per cycle, on-chain maker attribution, fees not deducted) on manipulated and on normal cycles. The last four rows are the quantities of Table~\ref{tab:cohort_pnl} by cycle type, over cycles with at least one cohort fill: ``PnL per cycle'' is the cohort's summed PnL averaged over cycles, and ``\% cycles cohort net positive'' the share of cycles in which that sum is positive.}
    \label{tab:twap_cohort_ft}
    \begin{tabular}{llcc}
    \toprule
     & & {Pre-TWAP} & {TWAP} \\
     & & {May 20 -- Aug 6} & {Aug 14 -- Oct 31} \\
    \midrule
    \multicolumn{2}{l}{Wallets active} & {921} & {323} \\
    \multicolumn{2}{l}{Cycles per wallet} & 1570.7 & 1853.6 \\
    \multicolumn{2}{l}{Share of market volume} & 0.111 & 0.073 \\
    \midrule
    \multirow{2}{*}{Net PnL (\$ thousand)} & Manipulated cycles & 7386.8 & 8.0 \\
     & Normal cycles & -1134.5 & 424.6 \\
    \addlinespace
    \multirow{2}{*}{PnL per cycle (\$)} & Manipulated cycles & 3167.6 & 5.7 \\
     & Normal cycles & -55.8 & 20.1 \\
    \addlinespace
    \multirow{2}{*}{\% cycles cohort net positive} & Manipulated cycles & {64.8\%} & {45.4\%} \\
     & Normal cycles & {41.0\%} & {47.1\%} \\
    \bottomrule
\end{tabular}

\end{table}

We find little evidence that manipulators adapted by spreading the push over the whole 60-second settlement window. As Appendix~\ref{app:twap_adapt} shows, some order flow moves into the earlier bins of the window, at a fraction of the former final-bin increase, but it is not concentrated in \NTM cycles, and the cycles with the largest net order flow over the window transfer no more wealth than normal cycles. In sum, averaging the settlement price over 60 seconds removed the final-ten-second push in the cycles where it could flip the resolution, together with its effect on outcomes and the profits from it.

\section{Conclusion}
\label{sec:conclusion}

Prediction markets increasingly list contracts that settle on an asset price that holders can move by trading the underlying. A model and the on-chain record of Polymarket's five-minute Bitcoin contract agree: such a contract degrades price discovery in the underlying and transfers wealth from ordinary traders to the few who manipulate it, even as it makes the underlying more liquid.

The mechanism is a settlement-time push: a holder of the binary contract trades the underlying in the final seconds to move the settlement reference toward her side. The spot maker cannot separate this uninformed push from ordinary flow, so price discovery falls and wealth passes to the manipulator and, through the rent it earns at a zero spread, to the prediction-market maker. The near-settlement spike appears only once the five-minute contract exists; the price move reverts within ten seconds; and the effect concentrates in the still-live cycles, where even a side priced above ninety percent is overturned about a third of the time, against one in a hundred otherwise. Some \$8.2~million was transferred over two months, mostly from retail.

The harm can be fixed: it tracks the contract horizon, the remedy our model provides. A manipulator pushes only when the price is close enough to the strike for one push to change the outcome, so a longer horizon makes such order histories rarer and drives her profit to zero. The settlement rule supplies parallel levers, from averaging or randomizing the close to capping the binary position or taxing settlement-window spot trades. Polymarket's switch to a 60-second TWAP benchmark in August 2026 shows the averaging lever at work~\citep{polymarket2026twap}: the final-seconds push largely disappeared, and the wallets that had profited from it went from \$3{,}168 per manipulated cycle before the switch to \$6 after. Manipulation here is a settlement-design problem.

\clearpage
\appendix
\setcounter{figure}{0}
\setcounter{table}{0}
\renewcommand{\thefigure}{A\arabic{figure}}
\renewcommand{\thetable}{A\arabic{table}}
\noindent{\Large\bfseries Appendix\par}
\addcontentsline{toc}{section}{Appendix}
\section{Omitted Proofs}
\label{app:proofs}

\subsection{Stage 2 Quotes and Posteriors}
\label{app:stage2-quotes-and-posteriors}

We prove \cref{lem:stage2-quotes} in this appendix.

\paragraph{Outcome probabilities.}
The informed trader arrives with probability \(\pi\). She buys if \(V=+\sigma\)
and sells if \(V=-\sigma\) conditional on arrival. A liquidity trader arrives
with probability \(1-\pi\). She buys or sells with equal probability conditional
on arrival. Therefore,
\[
    \Prb\left(Y_2=\hbuy\mid V=\pm\sigma\right)
    =\frac{1\pm\pi}{2},
    \qquad
    \Prb\left(Y_2=\hsell\mid V=\pm\sigma\right)
    =\frac{1\mp\pi}{2}.
\]
Marginalizing over \(V\) gives
$$
    q_{\hbuy}=q_{\hsell}=\frac{1}{2}.
$$

\paragraph{Posterior and quotes.}
Bayes' rule gives
$$
    \ppost_{\hbuy}=\Prb\left(V=+\sigma\mid Y_2=\hbuy\right)
    =\frac{\Prb\left(Y_2=\hbuy\mid V=+\sigma\right)\Prb\left(V=+\sigma\right)}{q_{\hbuy}}
    =\frac{\frac{1+\pi}{2} \cdot \frac{1}{2}}{\frac{1}{2}}=\frac{1+\pi}{2},
$$
$$
    \ppost_{\hsell}=1 - \ppost_{\hbuy}=\frac{1-\pi}{2}.
$$
Break-even pricing requires
$$
    \frac{1 - \pi}{2} A_2 = \frac{\pi}{2} \left(\sigma - A_2\right)
$$
where the left-hand side is the expected profit from selling to a liquidity trader and the
right-hand side is the expected loss from selling to the informed trader. Solving gives
$$
    A_2=\sigma\pi=\bar V_{\hbuy}.
$$
By symmetry,
$$
    B_2=-\sigma\pi=\bar V_{\hsell}.
$$
The quotes are indeed inside \(\left(-\sigma,\sigma\right)\) for \(\pi\in\left(0, 1\right)\),
confirming that the informed trader will never stay out at stage~2.

\paragraph{Residual variance.}
Direct calculation of \(\rvar_{h_2}=4\sigma^2\ppost_{h_2}\left(1-\ppost_{h_2}\right)\) gives
\[
    \rvar_{\hbuy}
    =\rvar_{\hsell}
    =\sigma^2\left(1-\pi^2\right).
\]

\subsection{Stage 3 Quotes and Posteriors}
\label{app:stage3-quotes-and-posteriors}

We prove \cref{prop:stage3-closed-form-quotes} in this appendix.

\paragraph{State-conditional outcome probabilities.}
The manipulator enters with probability \(p_{\mathrm M}\). She buys
with probability \(\alpha_{h_2}\) and stays out with probability \(1-\alpha_{h_2}\)
conditional on entry. The informed trader arrives with probability
\(\left(1-p_{\mathrm M}\right)\pi\). She buys if \(V=+\sigma\) and sells
if \(V=-\sigma\) conditional on arrival. A liquidity
trader arrives with probability \(\left(1-p_{\mathrm M}\right)\left(1-\pi\right)\). She
buys or sells with equal probability conditional on arrival. Therefore,
\begin{align*}
    \Prb\left(Y_3=\hbuy\mid V=\pm\sigma, Y_2 = h_2\right)  & =
    p_{\mathrm M}\alpha_{h_2}+\left(1-p_{\mathrm M}\right)\frac{1\pm\pi}{2}, \\
    \Prb\left(Y_3=\hzero\mid V=\pm\sigma, Y_2 = h_2\right) & =
    p_{\mathrm M}\left(1-\alpha_{h_2}\right),                                \\
    \Prb\left(Y_3=\hsell\mid V=\pm\sigma, Y_2 = h_2\right) & =
    \left(1-p_{\mathrm M}\right)\frac{1\mp\pi}{2}.
\end{align*}
Marginalizing over $V \mid Y_2 = h_2$ gives
\begin{align*}
    q_{h_2\hbuy}
     & =
    p_{\mathrm M}\alpha_{h_2}+\left(1-p_{\mathrm M}\right)\left[\frac{1-\pi}{2}+\pi\ppost_{h_2}\right], \\
    q_{h_2\hzero}
     & =
    p_{\mathrm M}\left(1-\alpha_{h_2}\right),                                                           \\
    q_{h_2\hsell}
     & =
    \left(1-p_{\mathrm M}\right)\left[\frac{1-\pi}{2}+\pi\left(1-\ppost_{h_2}\right)\right].
\end{align*}

\paragraph{Break-even ask.}
A stage-3 buy comes from the manipulator or a liquidity trader, both uninformed, or from
the informed trader, who buys only when \(V=+\sigma\). Break-even pricing requires
\[
    \left[p_{\mathrm M}\alpha_{h_2}+\left(1-p_{\mathrm M}\right)\frac{1-\pi}{2}\right]
    \left(A_{h_2}-\bar V_{h_2}\right)
    =
    \left(1-p_{\mathrm M}\right)\pi\ppost_{h_2}\left(\sigma-A_{h_2}\right)
\]
where the left-hand side is the expected profit from selling to an uninformed buyer and the
right-hand side is the expected loss from selling to the informed trader. Solving gives
\begin{equation}
    \label{eq:app-stage3-ask-general}
    A_{h_2}
    =
    \sigma\,
    \frac{
        p_{\mathrm M}\alpha_{h_2}\left(2\ppost_{h_2}-1\right)
        +\frac{1-p_{\mathrm M}}{2}
        \left[2\ppost_{h_2}-1+\pi\right]}
    {p_{\mathrm M}\alpha_{h_2}
        +\frac{1-p_{\mathrm M}}{2}
        \left[1+\left(2\ppost_{h_2}-1\right)\pi\right]} .
\end{equation}
Substituting the stage-2 posteriors $\ppost_{\hbuy}=\frac{1+\pi}{2}$ and $\ppost_{\hsell}=\frac{1-\pi}{2}$
gives
\[
    A_{\hbuy}
    =
    \frac{\sigma\pi\left[p_{\mathrm M}\alpha_{\hbuy}+\left(1-p_{\mathrm M}\right)\right]}
    {p_{\mathrm M}\alpha_{\hbuy}+\left(1-p_{\mathrm M}\right)\frac{1+\pi^2}{2}} ,
    \qquad
    A_{\hsell}
    =
    \frac{-\sigma\pi\, p_{\mathrm M}\alpha_{\hsell}}
    {p_{\mathrm M}\alpha_{\hsell}+\left(1-p_{\mathrm M}\right)\frac{1-\pi^2}{2}} .
\]

\paragraph{Break-even bid.}
The manipulator never sells, so a stage-3 sell comes only from a liquidity trader or from
the informed trader, who sells only when \(V=-\sigma\). Break-even pricing requires
\[
    \left(1-p_{\mathrm M}\right)\frac{1-\pi}{2}\left(\bar V_{h_2}-B_{h_2}\right)
    =
    \left(1-p_{\mathrm M}\right)\pi\left(1-\ppost_{h_2}\right)\left(\sigma+B_{h_2}\right)
\]
where the left-hand side is the expected profit from buying from a liquidity trader who sells and the
right-hand side is the expected loss from buying from the informed trader. The common factor
\(\left(1-p_{\mathrm M}\right)\) cancels, so the bid is independent of \(\alpha_{h_2}\) and
\(p_{\mathrm M}\). Solving gives
\begin{equation}
    \label{eq:app-stage3-bid-general}
    B_{h_2}
    =
    \sigma\,
    \frac{\left(2\ppost_{h_2}-1\right)-\pi}
    {1-\left(2\ppost_{h_2}-1\right)\pi}.
\end{equation}
Substituting the stage-2 posteriors $\ppost_{\hbuy}=\frac{1+\pi}{2}$ and $\ppost_{\hsell}=\frac{1-\pi}{2}$ gives
\[
    B_{\hbuy}=0,\quad B_{\hsell}=\frac{-2\sigma\pi}{1+\pi^2} .
\]

\paragraph{No-trade posterior.}
No-trade occurs only when the manipulator is present, and she is uninformed. When it has
positive probability (\(\alpha_{h_2}<1\)) the posterior equals the prior entering stage~3; at
\(\alpha_{h_2}=1\) the event has zero probability and is assigned the same belief by convention.
Hence
$$
    \Prb\left(V=\pm\sigma\mid Y_3=\hzero, Y_2 = h_2\right)=
    \Prb\left(V=\pm\sigma\mid Y_2 = h_2\right)
$$
and
$$
    \bar V_{h_2\hzero}=\bar V_{h_2}=\sigma\left(2\ppost_{h_2}-1\right),
$$
constant in \(\alpha_{h_2}\) and \(p_{\mathrm M}\).

\paragraph{Monotonicities.}
Recall the break-even pricing condition for the ask
\[
    \left[p_{\mathrm M}\alpha_{h_2}+\left(1-p_{\mathrm M}\right)\frac{1-\pi}{2}\right]
    \left(A_{h_2}-\bar V_{h_2}\right)
    =
    \left(1-p_{\mathrm M}\right)\pi\ppost_{h_2}\left(\sigma-A_{h_2}\right).
\]
For \(p_{\mathrm M}\in\left(0,1\right)\) and \(\pi\in\left(0,1\right)\) both
probability coefficients are strictly positive, so the balance forces
\(A_{h_2}\in\left(\bar V_{h_2},\sigma\right)\), which we can also verify
from the closed form above.
The right-hand side does not contain \(\alpha_{h_2}\), while the left-hand side
strictly increases in \(\alpha_{h_2}\). Raising \(\alpha_{h_2}\) therefore
moves the uninformed-profit side above the informed-loss side at the prevailing
quote. Hence \(A_{h_2}\) has to strictly decrease to restore equality.
For monotonicity in \(p_{\mathrm M}\), divide both sides by \(1-p_{\mathrm M}>0\)
to get the following, then apply a similar argument.
\[
    \left[\frac{p_{\mathrm M}}{1-p_{\mathrm M}}\,\alpha_{h_2}+\frac{1-\pi}{2}\right]
    \left(A_{h_2}-\bar V_{h_2}\right)
    =
    \pi\ppost_{h_2}\left(\sigma-A_{h_2}\right).
\]
The bid and the no-trade posterior mean are independent of both $p_{\mathrm M}$ and $\alpha_{h_2}$, as the closed forms above show.
\subsection{Selected Stage-3 Continuation Equilibrium}
\label{app:selected-continuation-existence}

\begin{definition}[Stable stage-3 fixed point]\label{def:stage3-selected-continuation}
    A stage-3 fixed point \(\alpha_{h_2}\) is \emph{stable} if for some
    \(\eta>0\),
    \begin{equation}
        \label{eq:weak-stability}
        \left(\alpha-\alpha_{h_2}\right)\Gamma_{h_2}\left(\alpha\right)\le 0
        \quad
        \text{for all }
        \alpha\in\left(\alpha_{h_2}-\eta,\alpha_{h_2}+\eta\right)\cap\left[0,1\right] .
    \end{equation}
    Equivalently, \(\Gamma_{h_2}\ge 0\) just left and \(\le 0\) just
    right of \(\alpha_{h_2}\). This excludes strict crossings from
    below, where local pressure pushes the conjecture away. Write
    \(\mathcal E_{h_2}\) for the set of stable fixed points.
\end{definition}

\begin{lemma}[Selected stage-3 continuation equilibrium]\label{prop:stage3-selected-existence}
    Fix an interior entry probability \(p_{\mathrm M}\in\left(0,1\right)\). The set of stable fixed points \(\mathcal E_{h_2}\) is nonempty for each \(h_2\), so the selected push probability \(\alpha_{h_2}^{*}=\min\mathcal E_{h_2}\) is well defined.
\end{lemma}

We prove \cref{prop:stage3-selected-existence} in this appendix. The closed-form expression for $A_{h_2}$ shows that it is continuous in \(\alpha_{h_2}\),
which implies that $\Gamma_{h_2}$ is continuous in \(\alpha_{h_2}\) as well
since $H$ is continuous.
Moreover, \(H\left(\cdot/\varepsilon\right)\) is constant outside
\(\left(-\varepsilon,\varepsilon\right)\) and affine inside, and
\(A_{h_2}\left(\alpha_{h_2}\right)\) is continuous and strictly monotone, so
\(\Gamma_{h_2}\) is strictly monotone on each of the at most three intervals on
which the ask lies above, inside, or below \(\left[-\varepsilon,\varepsilon\right]\). \\
Specifically, if the ask lies inside \(\left(-\varepsilon,\varepsilon\right)\),
$$
    \Gamma_{h_2}'\left(\alpha_{h_2}\right) = \left(\frac{1}{2\varepsilon}-1\right)A_{h_2}'\left(\alpha_{h_2}\right) < 0
$$ because \(A_{h_2}'\left(\alpha_{h_2}\right)<0\) and \cref{ass:canonical} implies \(\varepsilon<\frac{1}{2}\).
If the ask lies above \(\varepsilon\) or below \(-\varepsilon\), then
$$
    \Gamma_{h_2}'\left(\alpha_{h_2}\right) = -A_{h_2}'\left(\alpha_{h_2}\right) > 0.
$$
In particular, \(\Gamma_{h_2}\) has finitely many zeros. If \(\Gamma_{h_2}\le 0\) on a right-neighborhood of \(0\),
then \(\alpha_{h_2}=0\) is stable. If \(\Gamma_{h_2}\ge 0\) on a left-neighborhood
of \(1\), then \(\alpha_{h_2}=1\) is stable. Otherwise neither holds, so every
right-neighborhood of \(0\) contains a point with \(\Gamma_{h_2}>0\) and every
left-neighborhood of \(1\) a point with \(\Gamma_{h_2}<0\); among the finitely
many zeros lying between two such points, the first downcrossing has
\(\Gamma_{h_2}\ge 0\) immediately to its left and \(\le 0\) immediately to its
right, hence is stable.
Therefore, \(\mathcal E_{h_2}\) is nonempty and finite, so the smallest stable
fixed point \(\alpha_{h_2}^{*}=\min\mathcal E_{h_2}\) is well defined.

\subsection{Price Discovery}
\label{app:price-discovery}

We prove \cref{prop:price-discovery-sign} in this appendix.
By the law of total variance within each \(h_2\),
\[
    \sum_{h_3}q_{h_2h_3}\,\rvar_{h_2h_3}
    =\rvar_{h_2}- \operatorname{Var}\left(\bar V_{h_2h_3} \mid Y_2=h_2\right),
\]
where
\[
    \operatorname{Var}\left(\bar V_{h_2h_3} \mid Y_2=h_2\right)
    =
    \sum_{h_3}q_{h_2h_3}\left(\bar V_{h_2h_3}-\bar V_{h_2}\right)^2
\]
is the dispersion of the stage-3 posterior mean within each \(h_2\).
Since stage~2 has no manipulator, its outcome probabilities \(q_{h_2}\) and conditional
variances \(\rvar_{h_2}\) are all independent of \(p_{\mathrm M}\). Hence, it
suffices to show, for each \(h_2\), that
\[
    \operatorname{Var}\left(\bar V_{h_2h_3} \mid Y_2=h_2\right)\left(p_{\mathrm M}\right)
    <
    \operatorname{Var}\left(\bar V_{h_2h_3} \mid Y_2=h_2\right)\left(0\right).
\]
By \cref{cor:stage3-pm-buy-quote-monotonicity},
\(\bar V_{h_2\hzero}\left(p_{\mathrm M}\right)=\bar V_{h_2}\),
so the $h_3 = \hzero$ term does not contribute to the within-\(h_2\) dispersion. \\
A sell order never comes from the manipulator, so the posterior mean after a sell is
unaffected by \(p_{\mathrm M}\), but its probability is thinned by the manipulator's
presence:
\[
    \bar V_{h_2\hsell}\left(p_{\mathrm M}\right)=\bar V_{h_2\hsell}\left(0\right),
    \qquad
    q_{h_2\hsell}\left(p_{\mathrm M}\right)=\left(1-p_{\mathrm M}\right)q_{h_2\hsell}\left(0\right),
\]
so the $h_3 = \hsell$ term equals \(\left(1-p_{\mathrm M}\right)\) times its benchmark value. \\
A buy order only moves the posterior mean when it comes from the informed trader:
\begin{align*}
      & q_{h_2\hbuy}\left(p_{\mathrm M}\right)\left(\bar V_{h_2\hbuy}\left(p_{\mathrm M}\right)-\bar V_{h_2}\right)                                                                                                                \\
    = & \E\left[\left(V-\bar V_{h_2}\right)\mathbf 1\{Y_3=\hbuy\}\mid Y_2=h_2\right]                                                                                                                               \\
    = & \underbrace{\E\left[\left(V-\bar V_{h_2}\right)\mathbf 1\{Y_3=\hbuy \text{ from the manipulator}\}\mid Y_2=h_2\right]}_{=0\ \ (\text{buy}\perp V)}                                                         \\
      & +\underbrace{\E\left[\left(V-\bar V_{h_2}\right)\mathbf 1\{Y_3=\hbuy \text{ from a liquidity trader}\}\mid Y_2=h_2\right]}_{=0\ \ (\text{buy}\perp V)}                                                     \\
      & +\underbrace{\E\left[\left(V-\bar V_{h_2}\right)\mathbf 1\{Y_3=\hbuy \text{ from the informed trader}\}\mid Y_2=h_2\right]}_{=\left(1-p_{\mathrm M}\right)\pi\ppost_{h_2}\left(\sigma-\bar V_{h_2}\right)} \\
    = & \left(1-p_{\mathrm M}\right)\pi\ppost_{h_2}\left(\sigma-\bar V_{h_2}\right),
\end{align*}
which equals \(\left(1-p_{\mathrm M}\right)\) times the benchmark quantity
\(q_{h_2\hbuy}\left(0\right)\left(\bar V_{h_2\hbuy}\left(0\right)-\bar V_{h_2}\right)\).
Hence,
\begin{align*}
        & q_{h_2\hbuy}\left(p_{\mathrm M}\right)\left(\bar V_{h_2\hbuy}\left(p_{\mathrm M}\right)-\bar V_{h_2}\right)^2                                                                                  \\
    =   & \left(1-p_{\mathrm M}\right)^2\frac{q_{h_2\hbuy}\left(0\right)}{q_{h_2\hbuy}\left(p_{\mathrm M}\right)}\,q_{h_2\hbuy}\left(0\right)\left(\bar V_{h_2\hbuy}\left(0\right)-\bar V_{h_2}\right)^2 \\
    \le & \left(1 - p_{\mathrm M}\right)q_{h_2\hbuy}\left(0\right)\left(\bar V_{h_2\hbuy}\left(0\right)-\bar V_{h_2}\right)^2
\end{align*}
because $q_{h_2\hbuy}\left(p_{\mathrm M}\right)=
    p_{\mathrm M}\alpha_{h_2}+\left(1-p_{\mathrm M}\right)q_{h_2\hbuy}\left(0\right)
    \ge \left(1 - p_{\mathrm M}\right)q_{h_2\hbuy}\left(0\right)$. \\
Summing the three outcomes,
\[
    \operatorname{Var}\left(\bar V_{h_2h_3} \mid Y_2=h_2\right)\left(p_{\mathrm M}\right)
    \leq
    \left(1 - p_{\mathrm M}\right)\operatorname{Var}\left(\bar V_{h_2h_3} \mid Y_2=h_2\right)\left(0\right).
\]
With \(\pi>0\) the informed trader moves the benchmark posterior, so
\[
    \operatorname{Var}\left(\bar V_{h_2h_3} \mid Y_2=h_2\right)\left(0\right) > 0.
\]
Therefore, for any \(p_{\mathrm M}\in\left(0,1\right)\),
\[
    \operatorname{Var}\left(\bar V_{h_2h_3} \mid Y_2=h_2\right)\left(p_{\mathrm M}\right)
    < \operatorname{Var}\left(\bar V_{h_2h_3} \mid Y_2=h_2\right)\left(0\right).
\]

\subsection{Canonical-Regime Equilibrium Calculations}
\label{app:canonical-regime-equilibrium-calculations}

We prove \cref{prop:canonical-equilibrium} in this appendix, which
focuses on the small-noise region
\[
    \varepsilon < \sigma\pi<\frac{1}{2}.
\]

\subsubsection{Selected stage-3 continuation equilibrium}
\label{app:selected-continuation-equilibrium-by-path}

This subsection proves \cref{prop:push-policy}. We find the selected continuation equilibrium at each order history by a
three-step search for the selected point \(\min\mathcal E_{h_2}\) of \cref{prop:stage3-selected-existence}: \\
\emph{Step 1:} Check whether \(\alpha_{h_2}=0\) is a stable fixed point. \\
\emph{Step 2:} If \(\alpha_{h_2}=0\) fails, locate the interior fixed point by solving
\(\Gamma_{h_2}\left(\alpha_{h_2}\right)=0\). \\
\emph{Step 3:} If no interior fixed point exists either, the continuation equilibrium
can only be \(\alpha_{h_2}=1\).

\paragraph{Favorable order history \(h_2=\hbuy\): pure stay-out equilibrium.}

For any $p_{\mathrm M} \in \left(0, 1\right)$ and any \(\alpha_{\hbuy} \in \left[0, 1\right]\),
\(A_{\hbuy}\left(\alpha_{\hbuy}; p_{\mathrm M}\right) \geq \bar V_{\hbuy} = \sigma\pi\ > \varepsilon\),
so \(H\left(A_{\hbuy}\left(\alpha_{\hbuy}; p_{\mathrm M}\right)/\varepsilon\right)= H\left(\bar V_{\hbuy}/\varepsilon\right)=1\). Hence
\[
    \Gamma_{\hbuy}\left(\alpha_{\hbuy}; p_{\mathrm M}\right)
    =1-1-\left[A_{\hbuy}\left(\alpha_{\hbuy}; p_{\mathrm M}\right)-\sigma\pi\right]
    =-\left[A_{\hbuy}\left(\alpha_{\hbuy}; p_{\mathrm M}\right)-\sigma\pi\right] \leq 0.
\]
Therefore, for any $p_{\mathrm M} \in \left(0, 1\right)$, \(\alpha_{\hbuy}^{*}=0\) is a stable fixed point and is clearly the smallest.
Intuitively, after a favorable stage-2 buy the posterior mean already sits above the band around the strike (\(\bar V_{\hbuy}=\sigma\pi>\varepsilon\)), so the contract pays \(1\) whether or not
the manipulator buys the spot; a stage-3 buy cannot raise her in-the-money probability and only incurs a
push cost, so she stays out.

\paragraph{Unfavorable order history \(h_2=\hsell\): mixed or pure push.}
Fix \(p_{\mathrm M}\in\left(0,1\right)\). At
\(\alpha_{\hsell}=0\), \cref{prop:stage3-closed-form-quotes} gives
\(A_{\hsell}\left(0;p_{\mathrm M}\right)=0\), while
\(\bar V_{\hsell}=-\sigma\pi<-\varepsilon\). Therefore
\[
    \Gamma_{\hsell}\left(0;p_{\mathrm M}\right)
    =
    \frac{1}{2}-0
    -\left[0-\left(-\sigma\pi\right)\right]
    =
    \frac{1}{2}-\sigma\pi
    >0,
\]
so \(\alpha_{\hsell}=0\) is not a fixed point. For an interior fixed point
with \(A_{\hsell}\in\left[-\varepsilon,\varepsilon\right]\), the fixed-point condition is
\[
    \frac{A_{\hsell}+\varepsilon}{2\varepsilon}
    -
    \left(A_{\hsell}+\sigma\pi\right)
    =
    0,
\]
which gives the candidate mixed-branch stage-3 ask
\[
    A_{\hsell}
    =
    \frac{\varepsilon\left(2\sigma\pi-1\right)}{1-2\varepsilon}
    \in\left(-\varepsilon,0\right).
\]
Inverting the closed-form expression in \cref{prop:stage3-closed-form-quotes} gives the candidate mixed-branch
push probability
\[
    \alpha_{\hsell}\left(p_{\mathrm M}\right)
    =
    \frac{1-p_{\mathrm M}}{p_{\mathrm M}}
    \frac{\varepsilon\left(1-\pi^2\right)\left(1-2\sigma\pi\right)}
    {2\left(\sigma\pi-\varepsilon\right)} > 0.
\]
This candidate is admissible if and only if
\[
    \alpha_{\hsell}\left(p_{\mathrm M}\right) < 1
    \quad\Longleftrightarrow\quad
    p_{\mathrm M} > p_0
    \coloneq
    \frac{\varepsilon\left(1-\pi^2\right)\left(1-2\sigma\pi\right)}
    {2\left(\sigma\pi-\varepsilon\right)
        +\varepsilon\left(1-\pi^2\right)\left(1-2\sigma\pi\right)}.
\]
Since $A_\hsell$ strictly decreases in $\alpha_\hsell$ by \cref{cor:stage3-pm-buy-quote-monotonicity},
we can verify that $\Gamma_\hsell$ crosses zero from positive to negative as $\alpha_\hsell$ crosses
this candidate from below to above. Thus, for \(p_{\mathrm M} \in \left(p_0, 1\right)\), this candidate is the
selected continuation equilibrium. For \(p_{\mathrm M} \in \left(0,p_0\right]\), this candidate would be equal to
or greater than $1$, which is not admissible, so the selected continuation equilibrium is the endpoint \(\alpha_{\hsell}^{*}=1\):
\[
    \alpha_{\hsell}^{*}\left(p_{\mathrm M}\right)
    =
    \begin{cases}
        1,
         & 0<p_{\mathrm M}\leq p_0, \\[0.75em]
        \displaystyle
        \frac{1-p_{\mathrm M}}{p_{\mathrm M}}
        \frac{\varepsilon\left(1-\pi^2\right)\left(1-2\sigma\pi\right)}
        {2\left(\sigma\pi-\varepsilon\right)},
         & p_0<p_{\mathrm M}<1.
    \end{cases}
\]
The corresponding equilibrium ask is
\[
    A_{\hsell}^{*}\left(p_{\mathrm M}\right)
    =
    \begin{cases}
        \displaystyle
        -\frac{2\sigma\pi p_{\mathrm M}}
         {\left(1-\pi^2\right)+\left(1+\pi^2\right)p_{\mathrm M}},
         & 0<p_{\mathrm M}\leq p_0, \\[1.25em]
        \displaystyle
        -\frac{\varepsilon\left(1-2\sigma\pi\right)}{1-2\varepsilon},
         & p_0<p_{\mathrm M}<1.
    \end{cases}
\]
On both branches, this ask remains in the band around the strike.\footnote{On the mixed branch,
    the ask lies in $\left(-\varepsilon,0\right)$ by \cref{ass:canonical}. On the pure-push branch, it is
    clearly below zero, and it is easy to verify that the ask is strictly decreasing in $p_{\mathrm M}$. The two branches
    meet at $p_\mathrm{M} = p_0$.}
Therefore the in-the-money probability induced by a stage-3 push after an unfavorable
order history is
\[
    H\left(A_{\hsell}^{*}\left(p_{\mathrm M}\right)/\varepsilon\right)
    =
    \begin{cases}
        \displaystyle
        \frac{1}{2}
        -
        \frac{\sigma\pi p_{\mathrm M}}
        {\varepsilon\left[\left(1-\pi^2\right)+\left(1+\pi^2\right)p_{\mathrm M}\right]},
         & 0<p_{\mathrm M}\leq p_0, \\[1.25em]
        \displaystyle
        \frac{\sigma\pi-\varepsilon}{1-2\varepsilon},
         & p_0<p_{\mathrm M}<1.
    \end{cases}
\]
The push cost after a stage-2 sell is
\[
    A_{\hsell}^{*}\left(p_{\mathrm M}\right)-\bar V_{\hsell}
    =
    \begin{cases}
        \displaystyle
        \frac{\sigma\pi\left(1-\pi^2\right)\left(1-p_{\mathrm M}\right)}
        {\left(1-\pi^2\right)+\left(1+\pi^2\right)p_{\mathrm M}},
         & 0<p_{\mathrm M}\leq p_0, \\[1.25em]
        \displaystyle
        \frac{\sigma\pi-\varepsilon}{1-2\varepsilon},
         & p_0<p_{\mathrm M}<1.
    \end{cases}
\]
The four displayed quantities are continuous at \(p_{\mathrm M}=p_0\), where the
pure-push and mixed branches meet.

\subsubsection{Optimal manipulator entry probability}
\label{app:optimal-manipulator-entry-probability}

Both the pricing of the binary contract and the manipulator's stage-1 entry problem use
the selected continuation equilibrium from \cref{app:selected-continuation-equilibrium-by-path}.
Combining \cref{app:stage2-quotes-and-posteriors}, the general ask and
bid formulas \eqref{eq:app-stage3-ask-general} and
\eqref{eq:app-stage3-bid-general}, and \cref{app:selected-continuation-equilibrium-by-path},
the stage-3 asks, bids, and no-trade posterior means used below are
\[
    A_{\hbuy}^{*}
    =
    \frac{2\sigma\pi}{1+\pi^2},
    \qquad
    A_{\hsell}^{*}\left(p_{\mathrm M}\right)
    =
    \begin{cases}
        \displaystyle
        -\frac{2\sigma\pi p_{\mathrm M}}
         {\left(1-\pi^2\right)+\left(1+\pi^2\right)p_{\mathrm M}},
         & 0<p_{\mathrm M}\leq p_0, \\[1.2em]
        \displaystyle
        -\frac{\varepsilon\left(1-2\sigma\pi\right)}{1-2\varepsilon},
         & p_0<p_{\mathrm M}<1,
    \end{cases}
\]
\[
    B_{\hbuy}
    =
    0,
    \qquad
    B_{\hsell}
    =
    -\frac{2\sigma\pi}{1+\pi^2},
\]
\[
    \bar V_{\hbuy}
    =
    \sigma\pi,
    \qquad
    \bar V_{\hsell}
    =
    -\sigma\pi.
\]

\paragraph{Pure-push branch \(p_{\mathrm M} \in \left(0, p_0\right]\).}
On this branch, \(\alpha_{\hbuy}^{*}=0\) and \(\alpha_{\hsell}^{*}=1\). As a result,
there is no push cost after a stage-2 buy, while the push cost after a stage-2 sell
is $A_{\hsell}^{*}\left(p_{\mathrm M}\right)-\bar V_{\hsell}$.
The per-history expected contract payoffs are:
\begin{align*}
    \Pi_{\mathrm M}\left(p_{\mathrm M},\hbuy\right)
     & =
    \alpha_{\hbuy}^{*}H\left(A_{\hbuy}^{*}/\varepsilon\right)
    +\left(1-\alpha_{\hbuy}^{*}\right)H\left(\bar V_{\hbuy}/\varepsilon\right)
    =
    H\left(\bar V_{\hbuy}/\varepsilon\right)
    =
    1,                                                                  \\
    \Pi_{\mathrm M}\left(p_{\mathrm M},\hsell\right)
     & =
    \alpha_{\hsell}^{*}H\left(A_{\hsell}^{*}\left(p_{\mathrm M}\right)/\varepsilon\right)
    +\left(1-\alpha_{\hsell}^{*}\right)H\left(\bar V_{\hsell}/\varepsilon\right)
    =
    H\left(A_{\hsell}^{*}\left(p_{\mathrm M}\right)/\varepsilon\right), \\
    \Pi_{\mathrm L}\left(p_{\mathrm M},\hbuy\right)
     & =
    q_{\hbuy \hbuy}\left(0\right)H\left(A_{\hbuy}^{*}/\varepsilon\right)
    +
    q_{\hbuy \hsell}\left(0\right)H\left(B_{\hbuy}/\varepsilon\right)
    =
    \frac{1+\pi^2}{2}\cdot 1+\frac{1-\pi^2}{2}\cdot\frac{1}{2}
    =
    \frac{3+\pi^2}{4},                                                  \\
    \Pi_{\mathrm L}\left(p_{\mathrm M},\hsell\right)
     & =
    q_{\hsell \hbuy}\left(0\right)H\left(A_{\hsell}^{*}\left(p_{\mathrm M}\right)/\varepsilon\right)
    +
    q_{\hsell \hsell}\left(0\right)H\left(B_{\hsell}/\varepsilon\right)
    =
    \frac{1-\pi^2}{2}
    H\left(A_{\hsell}^{*}\left(p_{\mathrm M}\right)/\varepsilon\right),
\end{align*}
where $q_{h_2 h_3}\left(0\right)$ denotes the conditional stage-3 outcome probabilities defined in
\cref{sec:setup-order-flow} evaluated at $p_{\mathrm M} = 0$. Since each stage-2 order has probability $1/2$ (\cref{lem:stage2-quotes}), averaging the last two lines gives the expression in \cref{cor:liquidity-payoff}; it is below $1/2$ because $A_{\hsell}^{*}\left(p_{\mathrm M}\right)<0$ for $p_{\mathrm M}>0$. \\
The manipulator's expected contract payoff net of the expected push cost is
\begin{align*}
      & \Pi_{\mathrm M}\left(p_{\mathrm M}\right)-K\left(p_{\mathrm M}\right) \\
    = &
    q_{\hbuy}\Pi_{\mathrm M}\left(p_{\mathrm M},\hbuy\right)
    +
    q_{\hsell}\Pi_{\mathrm M}\left(p_{\mathrm M},\hsell\right)
    -
    q_{\hbuy}\cdot 0
    -
    q_{\hsell}\left[A_{\hsell}^{*}\left(p_{\mathrm M}\right)-\bar V_{\hsell}\right]   \\
    = &
    \frac{1}{2}
    +
    \frac{1}{2}H\left(A_{\hsell}^{*}\left(p_{\mathrm M}\right)/\varepsilon\right)
    -
    \frac{1}{2}\left[A_{\hsell}^{*}\left(p_{\mathrm M}\right)-\bar V_{\hsell}\right]  \\
    = &
    \frac{1}{2}
    +
    \frac{1}{2} \cdot
    \frac{
        \varepsilon\left[\left(1-\pi^2\right)+\left(1+\pi^2\right)p_{\mathrm M}\right]
        -2\sigma\pi p_{\mathrm M}
    }
    {
        2\varepsilon\left[\left(1-\pi^2\right)+\left(1+\pi^2\right)p_{\mathrm M}\right]
    }
    -
    \frac{1}{2} \cdot
    \frac{\sigma\pi\left(1-\pi^2\right)\left(1-p_{\mathrm M}\right)}
    {\left(1-\pi^2\right)+\left(1+\pi^2\right)p_{\mathrm M}}                   \\
    = &
    \frac{1}{2}
    +
    \frac{\left(1-\pi^2\right)\left(1-2\sigma\pi\right)
        \left(1-p_{\mathrm M}/p_0\right)}
    {4\left[\left(1-\pi^2\right)+\left(1+\pi^2\right)p_{\mathrm M}\right]}.
\end{align*}
The break-even stage-1 ask obtained from the prediction-market maker's break-even
condition is
\begin{align*}
      & A_1^{\PM,\mathrm{BE}}\left(p_{\mathrm M}\right)                                 \\
    = &
    p_{\mathrm M}\Pi_{\mathrm M}\left(p_{\mathrm M}\right)
    +
    \left(1-p_{\mathrm M}\right)\Pi_{\mathrm L}\left(p_{\mathrm M}\right)           \\
    = &
    p_{\mathrm M}
    \left[
        q_{\hbuy}\Pi_{\mathrm M}\left(p_{\mathrm M},\hbuy\right)
        +
        q_{\hsell}\Pi_{\mathrm M}\left(p_{\mathrm M},\hsell\right)
        \right]
    +
    \left(1-p_{\mathrm M}\right)
    \left[
        q_{\hbuy}\Pi_{\mathrm L}\left(p_{\mathrm M},\hbuy\right)
        +
        q_{\hsell}\Pi_{\mathrm L}\left(p_{\mathrm M},\hsell\right)
        \right]                   \\
    = &
    p_{\mathrm M}
    \left[
        \frac{1}{2}
        +
        \frac{1}{2}H\left(A_{\hsell}^{*}\left(p_{\mathrm M}\right)/\varepsilon\right)
        \right]
    +
    \left(1-p_{\mathrm M}\right)
    \left[
        \frac{3+\pi^2}{8}
        +
        \frac{1-\pi^2}{4}
        H\left(A_{\hsell}^{*}\left(p_{\mathrm M}\right)/\varepsilon\right)
        \right]               \\
    = &
    \frac{p_{\mathrm M}}{2}
    +
    \left(1-p_{\mathrm M}\right)\frac{3+\pi^2}{8}
    +
    \left[
        \frac{p_{\mathrm M}}{2}
        +
        \frac{\left(1-p_{\mathrm M}\right)\left(1-\pi^2\right)}{4}
        \right]
    \left[
        \frac{1}{2}
        -
        \frac{\sigma\pi p_{\mathrm M}}
        {\varepsilon\left[\left(1-\pi^2\right)+\left(1+\pi^2\right)p_{\mathrm M}\right]}
        \right] \\
    = &
    \frac{1}{2}
    -
    \frac{p_{\mathrm M}\left(\sigma\pi-\varepsilon\right)}{4\varepsilon} < \frac{1}{2}.
\end{align*}

\paragraph{Mixed branch \(p_{\mathrm M} \in \left(p_0, 1\right)\).}
On this branch, \(\alpha_{\hbuy}^{*}=0\) and \(\alpha_{\hsell}^{*}\left(p_{\mathrm M}\right)\in\left(0,1\right)\).
The per-history expected contract payoffs are:
\begin{align*}
    \Pi_{\mathrm M}\left(p_{\mathrm M},\hbuy\right)
     & =
    1 \qquad \text{(same as the pure-push branch)},                     \\
    \Pi_{\mathrm M}\left(p_{\mathrm M},\hsell\right)
     & =
    \alpha_{\hsell}^{*}\left(p_{\mathrm M}\right)
    H\left(A_{\hsell}^{*}\left(p_{\mathrm M}\right)/\varepsilon\right)
    +
    \left(1-\alpha_{\hsell}^{*}\left(p_{\mathrm M}\right)\right)
    H\left(\bar V_{\hsell}/\varepsilon\right)
    =
    \alpha_{\hsell}^{*}\left(p_{\mathrm M}\right)
    H\left(A_{\hsell}^{*}\left(p_{\mathrm M}\right)/\varepsilon\right), \\
    \Pi_{\mathrm L}\left(p_{\mathrm M},\hbuy\right)
     & =
    \frac{3+\pi^2}{4} \qquad \text{(same as the pure-push branch)},     \\
    \Pi_{\mathrm L}\left(p_{\mathrm M},\hsell\right)
     & =
    q_{\hsell \hbuy}\left(0\right)H\left(A_{\hsell}^{*}\left(p_{\mathrm M}\right)/\varepsilon\right)
    +
    q_{\hsell \hsell}\left(0\right)H\left(B_{\hsell}/\varepsilon\right)
    =
    \frac{1-\pi^2}{2}
    H\left(A_{\hsell}^{*}\left(p_{\mathrm M}\right)/\varepsilon\right).
\end{align*}
The manipulator's expected contract payoff net of the expected push cost is
\begin{align*}
      & \Pi_{\mathrm M}\left(p_{\mathrm M}\right)-K\left(p_{\mathrm M}\right) \\
    = &
    q_{\hbuy}\Pi_{\mathrm M}\left(p_{\mathrm M},\hbuy\right)
    +
    q_{\hsell}\Pi_{\mathrm M}\left(p_{\mathrm M},\hsell\right)
    -
    q_{\hbuy}\cdot 0
    -
    q_{\hsell}\alpha_{\hsell}^{*}\left(p_{\mathrm M}\right)
    \left[A_{\hsell}^{*}\left(p_{\mathrm M}\right)-\bar V_{\hsell}\right]             \\
    = &
    \frac{1}{2}
    +
    \frac{\alpha_{\hsell}^{*}\left(p_{\mathrm M}\right)}{2}
    \frac{\sigma\pi-\varepsilon}{1-2\varepsilon}
    -
    \frac{\alpha_{\hsell}^{*}\left(p_{\mathrm M}\right)}{2}
    \frac{\sigma\pi-\varepsilon}{1-2\varepsilon}                                      \\
    = &
    \frac{1}{2}.
\end{align*}
The break-even stage-1 ask obtained from the prediction-market maker's break-even
condition is
\begin{align*}
      & A_1^{\PM,\mathrm{BE}}\left(p_{\mathrm M}\right)                       \\
    = &
    p_{\mathrm M}\Pi_{\mathrm M}\left(p_{\mathrm M}\right)
    +
    \left(1-p_{\mathrm M}\right)\Pi_{\mathrm L}\left(p_{\mathrm M}\right) \\
    = &
    p_{\mathrm M}
    \left[
        q_{\hbuy}\Pi_{\mathrm M}\left(p_{\mathrm M},\hbuy\right)
        +
        q_{\hsell}\Pi_{\mathrm M}\left(p_{\mathrm M},\hsell\right)
        \right]
    +
    \left(1-p_{\mathrm M}\right)
    \left[
        q_{\hbuy}\Pi_{\mathrm L}\left(p_{\mathrm M},\hbuy\right)
        +
        q_{\hsell}\Pi_{\mathrm L}\left(p_{\mathrm M},\hsell\right)
        \right]         \\
    = &
    p_{\mathrm M}
    \left[
        \frac{1}{2}
        +
        \frac{\alpha_{\hsell}^{*}\left(p_{\mathrm M}\right)}{2}
        \frac{\sigma\pi-\varepsilon}{1-2\varepsilon}
        \right]
    +
    \left(1-p_{\mathrm M}\right)
    \left[
        \frac{3+\pi^2}{8}
        +
        \frac{\left(1-\pi^2\right)\left(\sigma\pi-\varepsilon\right)}
        {4\left(1-2\varepsilon\right)}
        \right]          \\
    = &
    \frac{p_{\mathrm M}}{2}
    +
    \left(1-p_{\mathrm M}\right)
    \frac{\varepsilon\left(1-\pi^2\right)\left(1-2\sigma\pi\right)}
    {4\left(1-2\varepsilon\right)}
    +
    \left(1-p_{\mathrm M}\right)
    \left[
        \frac{3+\pi^2}{8}
        +
        \frac{\left(1-\pi^2\right)\left(\sigma\pi-\varepsilon\right)}
        {4\left(1-2\varepsilon\right)}
        \right]          \\
    = &
    \frac{1}{2}
    -
    \frac{\left(1-p_{\mathrm M}\right)\left(1-\pi^2\right)\left(1-2\sigma\pi\right)}{8} < \frac{1}{2}.
\end{align*}

\paragraph{Non-crossing equilibrium ask and optimal entry probability.} From the above calculations,
\(A_1^{\PM,\mathrm{BE}}\left(p_{\mathrm M}\right)<1/2\) for any \(p_{\mathrm M}\in\left(0,1\right)\). Therefore,
the non-crossing equilibrium stage-1 ask is
\[
    A_1^{\PM{*}}\left(p_{\mathrm M}\right)
    =
    \max\left\{
    A_1^{\PM,\mathrm{BE}}\left(p_{\mathrm M}\right),
    \frac{1}{2}
    \right\}
    =
    \frac{1}{2}
\]
for any \(p_{\mathrm M}\in\left(0,1\right)\). The manipulator's expected profit conditional on entering at
stage~1 is
\[
    S_{\mathrm M}\left(p_{\mathrm M}\right)
    =
    \Pi_{\mathrm M}\left(p_{\mathrm M}\right)
    -
    K\left(p_{\mathrm M}\right)
    -
    \frac{1}{2},
\]
which yields
\[
    S_{\mathrm M}\left(p_{\mathrm M}\right)
    =
    \begin{cases}
        \displaystyle
        \frac{\left(1-\pi^2\right)\left(1-2\sigma\pi\right)
            \left(1-p_{\mathrm M}/p_0\right)}
        {4\left[\left(1-\pi^2\right)+\left(1+\pi^2\right)p_{\mathrm M}\right]},
         & 0<p_{\mathrm M}\leq p_0, \\[1.5em]
        0,
         & p_0 < p_{\mathrm M}<1.
    \end{cases}
\]
The pure-push branch is positive when $0 < p_{\mathrm M} < p_0$ and vanishes at
$p_{\mathrm M} = p_0$. Because the market makers observe $p_{\mathrm M}$ (\cref{sec:setup-players}), the stage-3 ask $A_{\hsell}^{*}(p_{\mathrm M})$ above responds to it while the stage-1 ask stays at $1/2$, and the manipulator chooses $p_{\mathrm M}$ anticipating that response (\cref{def:full-equilibrium}). Therefore, she maximizes the ex-ante profit
\begin{equation}
    \label{eq:app-canonical-entry-objective}
    \begin{aligned}
        \bar S_{\mathrm M}\left(p_{\mathrm M}\right)
         & =
        p_{\mathrm M}S_{\mathrm M}\left(p_{\mathrm M}\right)
        =
        \frac{p_{\mathrm M}\left(1-\pi^2\right)\left(1-2\sigma\pi\right)
            \left(1-p_{\mathrm M}/p_0\right)}
        {4\left[\left(1-\pi^2\right)+\left(1+\pi^2\right)p_{\mathrm M}\right]}
    \end{aligned}
\end{equation}
over \(0<p_{\mathrm M} < p_0\). The first-order condition reduces to
\[
    \left(1+\pi^2\right)p_{\mathrm M}^2
    +
    2\left(1-\pi^2\right)p_{\mathrm M}
    -
    \left(1-\pi^2\right)p_0
    =
    0.
\]
The unique positive root is
\[
    p_{\mathrm M}^{*}
    =
    \sqrt{\left(\frac{1-\pi^2}{1+\pi^2}\right)^2
        +\frac{1-\pi^2}{1+\pi^2}p_0}
    -\frac{1-\pi^2}{1+\pi^2},
\]
which we can verify is the unique maximizer over \(\left(0, p_0\right)\). At $p_{\mathrm M}^{*}$,
\begin{align*}
    \alpha_{\hbuy}^{*}\left(p_{\mathrm M}^{*}\right)
     & =
    0,   \\
    \alpha_{\hsell}^{*}\left(p_{\mathrm M}^{*}\right)
     & =
    1,   \\
    S_{\mathrm M}\left(p_{\mathrm M}^{*}\right)
     & =
    \frac{\left(1-2\sigma\pi\right)p_{\mathrm M}^{*}}
    {4p_0},
    \\
    \bar S_{\mathrm M}\left(p_{\mathrm M}^{*}\right)
     & =
    \frac{\left(1-2\sigma\pi\right)p_{\mathrm M}^{*2}}
    {4p_0}.
\end{align*}

\subsubsection{Comparative statics derivations}
\label{app:comparative-statics-derivations}

We verify the signs in \cref{lem:freq-magnitude} under the
equilibrium of \cref{prop:canonical-equilibrium}.

\paragraph{Pure-push cutoff \(p_0\).} We can sign the partial derivatives by inspection
of the closed-form expression for $p_0$
\[
    p_0 = \left(1 + \frac{{2\left(\sigma\pi-\varepsilon\right)}}
    {\varepsilon\left(1-\pi^2\right)\left(1-2\sigma\pi\right)}
    \right)^{-1},
\]
which gives
\[
    \frac{\partial p_0}{\partial\varepsilon} > 0, \qquad
    \frac{\partial p_0}{\partial\sigma} < 0, \qquad
    \frac{\partial p_0}{\partial\pi} < 0.
\]

\paragraph{Optimal entry probability \(p_{\mathrm M}^{*}\).}
For \(z \in \left(0, 1\right)\), $p_0 \in \left(0, 1\right)$, the map
\[
    \left(z, p_0\right)
    \mapsto
    \sqrt{z^2+z p_0}-z
    = \frac{p_0}{\sqrt{1+p_0/z}+1}
\]
strictly increases in both \(z\) and \(p_0\). Since
\[
    p_{\mathrm M}^{*}
    =
    \sqrt{
        \left(\frac{1-\pi^2}{1+\pi^2}\right)^2
        +
        \frac{1-\pi^2}{1+\pi^2}p_0
    }
    -
    \frac{1-\pi^2}{1+\pi^2},
\]
and $\left(1-\pi^2\right)/\left(1+\pi^2\right)$ strictly decreases in $\pi$, we obtain
\[
    \frac{\partial p_{\mathrm M}^{*}}{\partial\varepsilon}>0,
    \qquad
    \frac{\partial p_{\mathrm M}^{*}}{\partial\sigma}<0,
    \qquad
    \frac{\partial p_{\mathrm M}^{*}}{\partial\pi}<0.
\]

\paragraph{Stage-3 ask \(A_{\hsell}^{*}\left(p_{\mathrm M}^{*}\right)\).}
Recall that
\[
    A_{\hsell}^{*}\left(p_{\mathrm M}^{*}\right)
    =
    -\frac{2\sigma\pi p_{\mathrm M}^{*}}
    {\left(1-\pi^2\right)+\left(1+\pi^2\right)p_{\mathrm M}^{*}}
    =
    -\frac{2\sigma\pi}
    {\left(1-\pi^2\right)/p_{\mathrm M}^{*}+\left(1+\pi^2\right)}
    <0.
\]
This quantity depends on \(\varepsilon\) only through \(p_{\mathrm M}^{*}\) and
strictly decreases in \(p_{\mathrm M}^{*}\) holding all other terms fixed. Since
\(\partial p_{\mathrm M}^{*}/\partial\varepsilon>0\), we obtain
\[
    \frac{\partial A_{\hsell}^{*}\left(p_{\mathrm M}^{*}\right)}
    {\partial\varepsilon}
    <0.
\]
For \(\sigma\) and \(\pi\), using the closed-form expression for \(p_{\mathrm M}^{*}\), we have
\[
    \left(1-\pi^2\right)+\left(1+\pi^2\right)p_{\mathrm M}^{*}
    =
    \sqrt{
        \left(1-\pi^2\right)^2
        +
        \left(1-\pi^2\right)\left(1+\pi^2\right)p_0
    }
\]
and therefore
\[
    A_{\hsell}^{*}\left(p_{\mathrm M}^{*}\right)
    =
    -\frac{2\sigma\pi}{1+\pi^2}
    \left[
        1
        -
        \sqrt{
            \frac{1-\pi^2}
            {\left(1-\pi^2\right)+\left(1+\pi^2\right)p_0}
        }
        \right].
\]
Substituting the closed-form expression for \(p_0\) gives
\begin{align*}
    A_{\hsell}^{*}\left(p_{\mathrm M}^{*}\right)
     & =
    -\frac{2\sigma\pi}{1+\pi^2}
    \left[
        1
        -
        \sqrt{
            1
            -
            \frac{\varepsilon\left(1+\pi^2\right)\left(1-2\sigma\pi\right)}
            {2\sigma\pi\left(1-2\varepsilon\right)}
        }
        \right] \\
     & =
    -\frac{\varepsilon\left(1-2\sigma\pi\right)}
     {1-2\varepsilon}
    \left[
        1+
        \sqrt{
            1-
            \frac{\varepsilon\left(1+\pi^2\right)\left(1-2\sigma\pi\right)}
            {2\sigma\pi\left(1-2\varepsilon\right)}
        }
        \right]^{-1} \\
     & =
    -\frac{\varepsilon\left(1-2\sigma\pi\right)}
     {1-2\varepsilon}
    \left[
        1+
        \sqrt{
            1-
            \frac{\varepsilon}{2\sigma\left(1-2\varepsilon\right)}
            \left[
                \pi+\frac{1}{\pi}-2\sigma\left(1+\pi^2\right)
                \right]
        }
        \right]^{-1}.
\end{align*}
By inspection of the last expression, the square-root term strictly increases in both
\(\sigma\) and \(\pi\), while the positive fraction multiplying the leading negative
sign strictly decreases in both \(\sigma\) and \(\pi\). Therefore,
\[
    \frac{\partial A_{\hsell}^{*}\left(p_{\mathrm M}^{*}\right)}
    {\partial\sigma}
    >0,
    \qquad
    \frac{\partial A_{\hsell}^{*}\left(p_{\mathrm M}^{*}\right)}
    {\partial\pi}
    >0.
\]

\paragraph{Push cost at $p_{\mathrm M}^{*}$.}
Since \(\bar V_{\hsell}=-\sigma\pi\), the push cost is
\[
    A_{\hsell}^{*}\left(p_{\mathrm M}^{*}\right)-\bar V_{\hsell}
    =
    A_{\hsell}^{*}\left(p_{\mathrm M}^{*}\right)+\sigma\pi.
\]
Therefore,
\begin{align*}
    \frac{\partial}{\partial\varepsilon}
    \left[A_{\hsell}^{*}\left(p_{\mathrm M}^{*}\right)-\bar V_{\hsell}\right]
     & =
    \frac{\partial A_{\hsell}^{*}\left(p_{\mathrm M}^{*}\right)}
    {\partial\varepsilon}
    <0,  \\
    \frac{\partial}{\partial\sigma}
    \left[A_{\hsell}^{*}\left(p_{\mathrm M}^{*}\right)-\bar V_{\hsell}\right]
     & =
    \frac{\partial A_{\hsell}^{*}\left(p_{\mathrm M}^{*}\right)}
    {\partial\sigma}
    +
    \pi
    >0,  \\
    \frac{\partial}{\partial\pi}
    \left[A_{\hsell}^{*}\left(p_{\mathrm M}^{*}\right)-\bar V_{\hsell}\right]
     & =
    \frac{\partial A_{\hsell}^{*}\left(p_{\mathrm M}^{*}\right)}
    {\partial\pi}
    +
    \sigma
    >0.
\end{align*}

\paragraph{Ex-ante  profit $\bar{S}_{\mathrm M}\left(p_{\mathrm M}^{*}\right)$.}
Using the first-order condition from \cref{app:optimal-manipulator-entry-probability} evaluated at
\(p_{\mathrm M}^{*}\), we can verify that
\[
    \bar S_{\mathrm M}\left(p_{\mathrm M}^{*}\right)
    =
    \frac{1-2\sigma\pi}{4}
    \frac{p_{\mathrm M}^{*2}}{p_0}
    =
    \frac{1-2\sigma\pi}{4}
    \left(
    \frac{1+\pi^2}{1-\pi^2} + \frac{2}{p_\mathrm{M}^*}
    \right)^{-1}.
\]
Combined with the comparative statics results for $p_{\mathrm M}^{*}$, we obtain by inspection of the last
expression that
\[
    \frac{\partial\bar S_{\mathrm M}\left(p_{\mathrm M}^{*}\right)}
    {\partial\varepsilon}
    >0,
    \qquad
    \frac{\partial\bar S_{\mathrm M}\left(p_{\mathrm M}^{*}\right)}
    {\partial\sigma}
    <0,
    \qquad
    \frac{\partial\bar S_{\mathrm M}\left(p_{\mathrm M}^{*}\right)}
    {\partial\pi}
    <0.
\]

\paragraph{Profit per entry $S_{\mathrm M}\left(p_{\mathrm M}^{*}\right)$.}
Let
\[
    T
    \coloneq
    \sqrt{
        1-
        \frac{\varepsilon\left(1+\pi^2\right)\left(1-2\sigma\pi\right)}
        {2\sigma\pi\left(1-2\varepsilon\right)}
    }
    =
    \sqrt{
        \frac{1-\pi^2}
        {\left(1-\pi^2\right)+\left(1+\pi^2\right)p_0}
    }.
\]
In the small-noise region, \(T\in\left(\sqrt{\left(1-\pi^2\right)/2},1\right)\).
Using the closed-form expression for \(p_{\mathrm M}^{*}\), we have
\begin{align*}
    \frac{p_{\mathrm M}^{*}}{p_0}
     & =
    \frac{1-\pi^2}{1+\pi^2}\frac{1}{p_0}\left(
    \sqrt{
        1 + \frac{1+\pi^2}{1-\pi^2} p_0
    } - 1
    \right)
    =
    \left(\sqrt{1 + \frac{1+\pi^2}{1-\pi^2} p_0}+1\right)^{-1}
    =
    \frac{T}{1+T}.
\end{align*}
Therefore, the conditional manipulation profit at the optimal entry probability decomposes as
\[
    S_{\mathrm M}\left(p_{\mathrm M}^{*}\right)
    =
    \underbrace{
        \vphantom{\frac{1-2\sigma\pi}{4}\frac{T}{1+T}}
        \frac{1-2\sigma\pi}{4}
    }_{\vphantom{\text{entry-ratio channel}}\text{prize-compression channel}}
    \cdot
    \underbrace{
        \vphantom{\frac{1-2\sigma\pi}{4}\frac{T}{1+T}}
        \frac{T}{1+T}
    }_{\vphantom{\text{prize-compression channel}}\text{entry-ratio channel}}.
\]
Since \(\partial p_0 / \partial \varepsilon>0\), the expression for \(T\) shows that
\(T\) strictly decreases in \(\varepsilon\). Hence
\[
    \frac{\partial S_{\mathrm M}\left(p_{\mathrm M}^{*}\right)}
    {\partial\varepsilon}
    <0.
\]
For \(\sigma\) and \(\pi\), the two channels move in opposite directions. The
prize-compression channel \(\left(1-2\sigma\pi\right)/4\) strictly decreases in both
\(\sigma\) and \(\pi\). The entry-ratio channel \(T/\left(1+T\right)\) strictly increases
in \(T\), and direct differentiation of \(T\) gives
\begin{align*}
    \frac{\partial T}{\partial\sigma}
     & =
    \frac{1-T^2}{2\sigma T\left(1-2\sigma\pi\right)}
    >0,  \\
    \frac{\partial T}{\partial\pi}
     & =
    \frac{\left(1-T^2\right)\left(1-\pi^2+4\sigma\pi^3\right)}
    {2\pi T\left(1+\pi^2\right)\left(1-2\sigma\pi\right)}
    >0.
\end{align*}
Thus the entry-ratio channel strictly increases in both \(\sigma\) and \(\pi\).
Putting the two channels together,
\begin{align*}
    \frac{\partial S_{\mathrm M}\left(p_{\mathrm M}^{*}\right)}
    {\partial\sigma}
     & =
    \frac{1 - T - 4\sigma\pi T^2}{8\sigma T \left(1+T\right)}, \\
    \frac{\partial S_{\mathrm M}\left(p_{\mathrm M}^{*}\right)}
    {\partial\pi}
     & =
    \frac{\left(1-T\right)\left(1-\pi^2+4\sigma\pi^3\right)
        -
        4\sigma\pi\left(1+\pi^2\right)T^2}
    {8\pi T \left(1+T\right)\left(1+\pi^2\right)}.
\end{align*}
The denominators are strictly positive. Therefore,
\begin{align*}
     & \frac{\partial S_{\mathrm M}\left(p_{\mathrm M}^{*}\right)}
       {\partial\sigma}
    >0
    \quad\Longleftrightarrow\quad
    T
    <
    T_\sigma^{*}
    \coloneq
    \frac{2}{1+\sqrt{1+16\sigma\pi}},                                 \\
     & \frac{\partial S_{\mathrm M}\left(p_{\mathrm M}^{*}\right)}
       {\partial\pi}
    >0
    \quad\Longleftrightarrow\quad
    T
    <
    T_\pi^{*}
    \coloneq
    \frac{2}
    {1+
        \sqrt{
            1+
            \frac{16\sigma\pi\left(1+\pi^2\right)}
            {1-\pi^2+4\sigma\pi^3}
        }}.
\end{align*}
Moreover, \(T_\pi^{*}<T_\sigma^{*}\), because
\[
    1+\pi^2
    >
    1-\pi^2+4\sigma\pi^3
    \quad\Longleftrightarrow\quad
    \sigma\pi < \frac{1}{2}.
\]
To state the sign conditions as settlement-noise cutoffs, note that $1 - T^2$ strictly decreases in
$T$ when $T \geq 0$, so for \(j\in\left\{\sigma,\pi\right\}\),
\[
    \frac{\partial S_{\mathrm M}\left(p_{\mathrm M}^{*}\right)}
    {\partial j}>0
    \quad\Longleftrightarrow\quad
    T < T_j^*
    \quad\Longleftrightarrow\quad
    1-T^2 > 1-T_j^{*2}.
\]
Since
$$
    1 - T^2 = \frac{\varepsilon\left(1+\pi^2\right)\left(1-2\sigma\pi\right)}
    {2\sigma\pi\left(1-2\varepsilon\right)}
$$
strictly increases in $\varepsilon$, the last condition is equivalent to
\[
    \varepsilon > \varepsilon_j^*
    \coloneq
    \frac{2\sigma\pi \left(1 - T_j^{*2}\right)}
    {\left(1+\pi^2\right)\left(1-2\sigma\pi\right)+4\sigma\pi \left(1-T_j^{*2}\right)},
\]
provided \(\varepsilon_j^*<\sigma\pi\). If
\(\varepsilon_j^*\ge\sigma\pi\), the partial derivative remains negative throughout
the small-noise region. Since \(T_\pi^*<T_\sigma^*\), we have
\(\varepsilon_\pi^*>\varepsilon_\sigma^*\). That is, a
positive partial derivative with respect to $\pi$, when it exists, requires higher
settlement noise than one with respect to \(\sigma\). \\
When \(\varepsilon\) is small, both \(\partial S_{\mathrm M}/\partial\sigma\)
and \(\partial S_{\mathrm M}/\partial\pi\) are negative: at low settlement noise the entry-ratio channel $T/\left(1+T\right)$ is already near its ceiling of $1/2$, so increasing $\sigma$ or $\pi$ mainly compresses the prize and conditional profit falls.

\subsubsection{Liquidity-trader loss decomposition}
\label{app:equilibrium-contract-payoffs}

We prove \cref{prop:pm-maker-rent}. At
\(p_{\mathrm M}^{*}\in\left(0,p_0\right)\), the selected continuation
equilibrium after a stage-2 sell is on the pure-push branch.
From \cref{app:optimal-manipulator-entry-probability}, the break-even stage-1 ask is
\[
    A_1^{\PM,\mathrm{BE}}\left(p_{\mathrm M}^{*}\right)
    =
    \frac{1}{2}
    -
    \frac{p_{\mathrm M}^{*}\left(\sigma\pi-\varepsilon\right)}
    {4\varepsilon}.
\]
Since the posted non-crossing ask is \(A_1^{\PM *}=1/2\), the prediction-market maker's expected
rent per buy order is
\[
    R^{*}
    \coloneq
    A_1^{\PM *}-A_1^{\PM,\mathrm{BE}}\left(p_{\mathrm M}^{*}\right)
    =
    \frac{p_{\mathrm M}^{*}\left(\sigma\pi-\varepsilon\right)}
    {4\varepsilon}>0.
\]
The liquidity trader's payoff on the pure-push branch is
\[
    \Pi_{\mathrm L}\left(p_{\mathrm M}^{*}\right)
    =
    \frac{3+\pi^2}{8}
    +
    \frac{1-\pi^2}{4}
    H\left(A_{\hsell}^{*}\left(p_{\mathrm M}^{*}\right)/\varepsilon\right),
\]
with
\[
    H\left(A_{\hsell}^{*}\left(p_{\mathrm M}^{*}\right)/\varepsilon\right)
    =
    \frac{1}{2}
    -
    \frac{\sigma\pi p_{\mathrm M}^{*}}
    {\varepsilon\left[\left(1-\pi^2\right)
            +\left(1+\pi^2\right)p_{\mathrm M}^{*}\right]} .
\]
Therefore the expected loss per buy order for the liquidity trader is
\[
    L\left(p_{\mathrm M}^{*}\right)
    \coloneq
    \frac{1}{2}-\Pi_{\mathrm L}\left(p_{\mathrm M}^{*}\right)
    =
    \frac{\left(1-\pi^2\right)\sigma\pi p_{\mathrm M}^{*}}
    {4\varepsilon\left[\left(1-\pi^2\right)
            +\left(1+\pi^2\right)p_{\mathrm M}^{*}\right]} > 0.
\]
Finally, the accounting identity follows directly from the definitions:
\begin{align*}
      & R^{*}
    +
    p_{\mathrm M}^{*}S_{\mathrm M}\left(p_{\mathrm M}^{*}\right)
    +
    p_{\mathrm M}^{*}K\left(p_{\mathrm M}^{*}\right) \\
    = &
    \frac{1}{2}
    -
    \left[
        p_{\mathrm M}^{*}\Pi_{\mathrm M}\left(p_{\mathrm M}^{*}\right)
        +
        \left(1-p_{\mathrm M}^{*}\right)
        \Pi_{\mathrm L}\left(p_{\mathrm M}^{*}\right)
        \right]
    +
    p_{\mathrm M}^{*}
    \left[
        \Pi_{\mathrm M}\left(p_{\mathrm M}^{*}\right)
        -
        K\left(p_{\mathrm M}^{*}\right)
        -
        \frac{1}{2}
        \right]
    +
    p_{\mathrm M}^{*}K\left(p_{\mathrm M}^{*}\right) \\
    = &
    \left(1-p_{\mathrm M}^{*}\right)
    \left[
        \frac{1}{2}
        -
        \Pi_{\mathrm L}\left(p_{\mathrm M}^{*}\right)
        \right] \\
    = &
    \left(1-p_{\mathrm M}^{*}\right)
    L\left(p_{\mathrm M}^{*}\right).
\end{align*}

\subsection{Two-Ordinary-Round Extension Calculations}
\label{app:horizon-two}

This appendix derives the two-ordinary-round extension used in
\cref{sec:results-horizon-two}. The analysis is conditional on a stage-1
prediction-market buy; the sell branch is symmetric.

\subsubsection{Stage 2 Posteriors}
\label{app:horizon-two-posteriors}

\paragraph{Stage-2 order history probabilities.}
Stages~2.1 and~2.2 are independent ordinary Glosten--Milgrom spot rounds except
that the prior over $V$ entering stage~2.2 is the posterior from stage~2.1.
At each round, the informed trader arrives with probability \(\pi\). She buys if
\(V=+\sigma\) and sells if \(V=-\sigma\) conditional on arrival. A liquidity trader arrives
with probability \(1-\pi\). She buys or sells with equal probability conditional
on arrival. Therefore, for \(j\in\left\{2.1,2.2\right\}\),
\[
    \Prb\left(Y_j=\hbuy\mid V=\pm\sigma\right)
    =\frac{1\pm\pi}{2},
    \qquad
    \Prb\left(Y_j=\hsell\mid V=\pm\sigma\right)
    =\frac{1\mp\pi}{2}.
\]
Conditional on \(V\), the two ordinary spot orders are independent. Hence, for
stage-2 order history \(h_2=h_{2.1}h_{2.2}\),
\[
    \Prb\left(h_2=\hbuy\hbuy\mid V=+\sigma\right)
    =
    \left(\frac{1+\pi}{2}\right)^2,
    \qquad
    \Prb\left(h_2=\hbuy\hbuy\mid V=-\sigma\right)
    =
    \left(\frac{1-\pi}{2}\right)^2,
\]
\[
    \Prb\left(h_2=\hsell\hsell\mid V=+\sigma\right)
    =
    \left(\frac{1-\pi}{2}\right)^2,
    \qquad
    \Prb\left(h_2=\hsell\hsell\mid V=-\sigma\right)
    =
    \left(\frac{1+\pi}{2}\right)^2,
\]
\[
    \Prb\left(h_2=\hbuy\hsell\mid V=+\sigma\right)
    =
    \Prb\left(h_2=\hbuy\hsell\mid V=-\sigma\right)
    =
    \frac{1-\pi^2}{4},
\]
\[
    \Prb\left(h_2=\hsell\hbuy\mid V=+\sigma\right)
    =
    \Prb\left(h_2=\hsell\hbuy\mid V=-\sigma\right)
    =
    \frac{1-\pi^2}{4}.
\]
Marginalizing over \(V\) gives
$$
    q_{\hbuy\hbuy}
    =
    q_{\hsell\hsell}
    =
    \frac{1+\pi^2}{4},
    \qquad
    q_{\hbuy\hsell}
    =
    q_{\hsell\hbuy}
    =
    \frac{1-\pi^2}{4}.
$$

\paragraph{Posteriors.}
Bayes' rule gives
$$
    \ppost_{\hbuy\hbuy}
    =
    \Prb\left(V=+\sigma\mid h_2=\hbuy\hbuy\right)
    =
    \frac{
        \Prb\left(h_2=\hbuy\hbuy\mid V=+\sigma\right)
        \Prb\left(V=+\sigma\right)
    }{q_{\hbuy\hbuy}}
    =
    \frac{
        \left(\frac{1+\pi}{2}\right)^2 \cdot \frac{1}{2}
    }{\frac{1+\pi^2}{4}}
    =
    \frac{\left(1+\pi\right)^2}{2\left(1+\pi^2\right)}
$$
and similarly
$$
    \ppost_{\hbuy\hsell}
    =
    \ppost_{\hsell\hbuy}
    =
    \frac{1}{2},
    \qquad
    \ppost_{\hsell\hsell}
    =
    \frac{\left(1-\pi\right)^2}{2\left(1+\pi^2\right)}.
$$
Using \(\bar V_{h_2}=\sigma\left(2\ppost_{h_2}-1\right)\) gives
$$
    \bar V_{\hbuy\hbuy}
    =
    \frac{2\sigma\pi}{1+\pi^2},
    \qquad
    \bar V_{\hbuy\hsell}
    =
    \bar V_{\hsell\hbuy}
    =
    0,
    \qquad
    \bar V_{\hsell\hsell}
    =
    -\frac{2\sigma\pi}{1+\pi^2}.
$$

\paragraph{Residual variance.}
Direct calculation of
\(\rvar_{h_2}=4\sigma^2\ppost_{h_2}\left(1-\ppost_{h_2}\right)\) gives
\[
    \rvar_{\hbuy\hbuy}
    =
    \rvar_{\hsell\hsell}
    =
    \frac{\sigma^2\left(1-\pi^2\right)^2}{\left(1+\pi^2\right)^2},
    \qquad
    \rvar_{\hbuy\hsell}
    =
    \rvar_{\hsell\hbuy}
    =
    \sigma^2.
\]

\subsubsection{Selected stage-3 continuation equilibrium}
\label{app:horizon-two-continuation}

Recall that the gap between the expected continuation payoff from buying and that
from staying out after stage-2 order history \(h_2\) is
\[
    \Gamma_{h_2}\left(\alpha_{h_2};p_{\mathrm M}\right)
    =
    H\left(\frac{A_{h_2}\left(\alpha_{h_2};p_{\mathrm M}\right)}{\varepsilon}\right)
    -
    H\left(\frac{\bar V_{h_2}}{\varepsilon}\right)
    -
    \left[A_{h_2}\left(\alpha_{h_2};p_{\mathrm M}\right)-\bar V_{h_2}\right].
\]
We find the selected continuation equilibrium at each stage-2 order history by the same
three-step search of \cref{app:selected-continuation-equilibrium-by-path}.
Throughout, stage-3 asks are obtained by substituting the relevant stage-2 posterior
from \cref{app:horizon-two-posteriors} into the general ask formula
\eqref{eq:app-stage3-ask-general}.

\paragraph{Favorable agreeing order history \(h_2=\hbuy\hbuy\): pure stay-out equilibrium.}
For any \(p_{\mathrm M}\in\left(0,1\right)\) and any
\(\alpha_{\hbuy\hbuy}\in\left[0,1\right]\), the stage-3 ask is at least the
posterior mean, so
\[
    A_{\hbuy\hbuy}\left(\alpha_{\hbuy\hbuy};p_{\mathrm M}\right)
    \geq
    \bar V_{\hbuy\hbuy}
    =
    \frac{2\sigma\pi}{1+\pi^2}
    >
    \sigma\pi
    >
    \varepsilon.
\]
Therefore
\(H\left(A_{\hbuy\hbuy}\left(\alpha_{\hbuy\hbuy};p_{\mathrm M}\right)/\varepsilon\right)
=H\left(\bar V_{\hbuy\hbuy}/\varepsilon\right)=1\). Hence
\[
    \Gamma_{\hbuy\hbuy}\left(\alpha_{\hbuy\hbuy};p_{\mathrm M}\right)
    =
    1-1-
    \left[
        A_{\hbuy\hbuy}\left(\alpha_{\hbuy\hbuy};p_{\mathrm M}\right)
        -\bar V_{\hbuy\hbuy}
        \right]
    \leq
    0.
\]
Thus, for any \(p_{\mathrm M}\in\left(0,1\right)\),
\(\alpha_{\hbuy\hbuy}^{*}=0\) is a stable fixed point and is clearly the smallest.

\paragraph{Unfavorable agreeing order history \(h_2=\hsell\hsell\): pure stay-out equilibrium.}
For any \(p_{\mathrm M}\in\left(0,1\right)\),
\[
    \bar V_{\hsell\hsell}
    =
    -\frac{2\sigma\pi}{1+\pi^2}
    <
    -\sigma\pi
    <
    -\varepsilon.
\]
At \(\alpha_{\hsell\hsell}=0\), substituting
\(\ppost_{\hsell\hsell}
=\left(1-\pi\right)^2 / \left(2\left(1+\pi^2\right)\right)\)
into \eqref{eq:app-stage3-ask-general} gives
\[
    A_{\hsell\hsell}\left(0;p_{\mathrm M}\right)
    =
    -\sigma\pi
    <
    -\varepsilon.
\]
Since the ask strictly decreases in \(\alpha_{\hsell\hsell}\) and
is at least the posterior mean,
\[
    \bar V_{\hsell\hsell}
    \le
    A_{\hsell\hsell}\left(\alpha_{\hsell\hsell};p_{\mathrm M}\right)
    \le
    -\sigma\pi
    <
    -\varepsilon.
\]
Consequently
\(H\left(A_{\hsell\hsell}\left(\alpha_{\hsell\hsell};p_{\mathrm M}\right)/\varepsilon\right)
=H\left(\bar V_{\hsell\hsell}/\varepsilon\right)=0\), and
\[
    \Gamma_{\hsell\hsell}\left(\alpha_{\hsell\hsell};p_{\mathrm M}\right)
    =
    0-0
    -\left[
        A_{\hsell\hsell}\left(\alpha_{\hsell\hsell};p_{\mathrm M}\right)
        -\bar V_{\hsell\hsell}
        \right]
    \leq
    0.
\]
Thus, for any \(p_{\mathrm M}\in\left(0,1\right)\),
\(\alpha_{\hsell\hsell}^{*}=0\) is a stable fixed point and is clearly the smallest.

\paragraph{Disagreeing order histories \(h_2\in\left\{\hbuy\hsell,\hsell\hbuy\right\}\): pure push equilibrium.}
Fix \(p_{\mathrm M}\in\left(0,1\right)\). For either disagreeing order history,
\(\bar V_{h_2}=0\). At \(\alpha_{h_2}=0\), substituting \(\ppost_{h_2}=1/2\) into
\eqref{eq:app-stage3-ask-general} gives
\(A_{h_2}\left(0;p_{\mathrm M}\right)=\sigma\pi\), and therefore
\[
    \Gamma_{h_2}\left(0;p_{\mathrm M}\right)
    =
    1
    -
    \frac{1}{2}
    -
    \left[
        \sigma\pi - 0
        \right]
    =
    \frac{1}{2}-\sigma\pi
    >
    0,
\]
so \(\alpha_{h_2}=0\) is not a fixed point. For any
\(\alpha_{h_2}\in\left[0,1\right]\), the ask is
\[
    A_{h_2}\left(\alpha_{h_2};p_{\mathrm M}\right)
    =
    \frac{\sigma\pi\left(1-p_{\mathrm M}\right)}
    {1-p_{\mathrm M}+2p_{\mathrm M}\alpha_{h_2}}
    \in
    \left[
        \frac{\sigma\pi\left(1-p_{\mathrm M}\right)}{1+p_{\mathrm M}},
        \sigma\pi
        \right]
    \subset
    \left(0,\frac{1}{2}\right).
\]
If \(A_{h_2}\left(\alpha_{h_2};p_{\mathrm M}\right)\ge\varepsilon\), then
\[
    \Gamma_{h_2}\left(\alpha_{h_2};p_{\mathrm M}\right)
    =
    \frac{1}{2}-A_{h_2}\left(\alpha_{h_2};p_{\mathrm M}\right)
    \geq
    \frac{1}{2}-\sigma\pi
    >
    0.
\]
If \(0<A_{h_2}\left(\alpha_{h_2};p_{\mathrm M}\right)<\varepsilon\), then
\[
    \Gamma_{h_2}\left(\alpha_{h_2};p_{\mathrm M}\right)
    =
    A_{h_2}\left(\alpha_{h_2};p_{\mathrm M}\right)
    \left(\frac{1}{2\varepsilon}-1\right)
    >
    0
\]
because \(\varepsilon<1/2\). Thus no interior fixed point exists either, and the
selected continuation equilibrium is the endpoint
\[
    \alpha_{\hbuy\hsell}^{*}
    =
    \alpha_{\hsell\hbuy}^{*}
    =
    1.
\]

\subsubsection{Optimal manipulator entry probability}
\label{app:horizon-two-entry}

Both the pricing of the binary contract and the manipulator's stage-1 entry problem use
the selected continuation equilibrium from \cref{app:horizon-two-continuation}.
Combining \cref{app:horizon-two-posteriors}, the general ask and bid formulas
\eqref{eq:app-stage3-ask-general} and \eqref{eq:app-stage3-bid-general}, and
\cref{app:horizon-two-continuation}, the stage-3 asks, bids, and no-trade
posterior means used below are
\[
    A_{\hbuy\hbuy}^{*}\left(p_{\mathrm M}\right)
    =
    \frac{\sigma\pi\left(3+\pi^2\right)}{1+3\pi^2},
    \qquad
    A_{\hsell\hsell}^{*}\left(p_{\mathrm M}\right)
    =
    -\sigma\pi,
    \qquad
    A_{\hbuy\hsell}^{*}\left(p_{\mathrm M}\right)
    =
    A_{\hsell\hbuy}^{*}\left(p_{\mathrm M}\right)
    =
    \frac{\sigma\pi\left(1-p_{\mathrm M}\right)}
    {1+p_{\mathrm M}}.
\]
\[
    B_{\hbuy\hbuy}
    =
    \sigma\pi,
    \qquad
    B_{\hsell\hsell}
    =
    -\frac{\sigma\pi\left(3+\pi^2\right)}{1+3\pi^2},
    \qquad
    B_{\hbuy\hsell}
    =
    B_{\hsell\hbuy}
    =
    -\sigma\pi.
\]
\[
    \bar V_{\hbuy\hbuy}
    =
    \frac{2\sigma\pi}{1+\pi^2},
    \qquad
    \bar V_{\hbuy\hsell}
    =
    \bar V_{\hsell\hbuy}
    =
    0,
    \qquad
    \bar V_{\hsell\hsell}
    =
    -\frac{2\sigma\pi}{1+\pi^2}.
\]

For each stage-2 order history \(h_2\), the per-history expected contract payoffs are
\begin{align*}
    \Pi_{\mathrm M}\left(p_{\mathrm M},h_2\right)
     & =
    \alpha_{h_2}^{*}H\left(A_{h_2}^{*}\left(p_{\mathrm M}\right)/\varepsilon\right)
    +
    \left(1-\alpha_{h_2}^{*}\right)
    H\left(\bar V_{h_2}/\varepsilon\right), \\
    \Pi_{\mathrm L}\left(p_{\mathrm M},h_2\right)
     & =
    q_{h_2\hbuy}\left(0\right)H\left(A_{h_2}^{*}\left(p_{\mathrm M}\right)/\varepsilon\right)
    +
    q_{h_2\hsell}\left(0\right)H\left(B_{h_2}/\varepsilon\right).
\end{align*}
Substituting the selected continuation equilibrium and the stage-3 quotes gives
\begin{align*}
     & \Pi_{\mathrm M}\left(p_{\mathrm M},\hbuy\hbuy\right)
    =
    1, \qquad
    \Pi_{\mathrm M}\left(p_{\mathrm M},\hsell\hsell\right)
    =
    0,                                                           \\
     & \Pi_{\mathrm M}\left(p_{\mathrm M},\hbuy\hsell\right)
    =
    \Pi_{\mathrm M}\left(p_{\mathrm M},\hsell\hbuy\right)
    =
    H\left(
    \frac{\sigma\pi\left(1-p_{\mathrm M}\right)}
    {\varepsilon\left(1+p_{\mathrm M}\right)}
    \right),                                                     \\
     & \Pi_{\mathrm L}\left(p_{\mathrm M},\hbuy\hbuy\right)
    =
    1, \qquad
    \Pi_{\mathrm L}\left(p_{\mathrm M},\hsell\hsell\right)
    =
    0,                                                           \\
     & \Pi_{\mathrm L}\left(p_{\mathrm M},\hbuy\hsell\right)
    =
    \Pi_{\mathrm L}\left(p_{\mathrm M},\hsell\hbuy\right)
    =
    \frac{1}{2}
    H\left(
    \frac{\sigma\pi\left(1-p_{\mathrm M}\right)}
    {\varepsilon\left(1+p_{\mathrm M}\right)}
    \right).
\end{align*}
Aggregating over the four stage-2 order histories gives
\begin{align*}
    \Pi_{\mathrm M}\left(p_{\mathrm M}\right)
     & =
    \frac{1+\pi^2}{4}
    +
    \frac{1-\pi^2}{2}
    H\left(
    \frac{\sigma\pi\left(1-p_{\mathrm M}\right)}
    {\varepsilon\left(1+p_{\mathrm M}\right)}
    \right),
    \\
    \Pi_{\mathrm L}\left(p_{\mathrm M}\right)
     & =
    \frac{1+\pi^2}{4}
    +
    \frac{1-\pi^2}{4}
    H\left(
    \frac{\sigma\pi\left(1-p_{\mathrm M}\right)}
    {\varepsilon\left(1+p_{\mathrm M}\right)}
    \right),
    \\
    K\left(p_{\mathrm M}\right)
     & =
    \frac{\left(1-\pi^2\right)\sigma\pi\left(1-p_{\mathrm M}\right)}
    {2\left(1+p_{\mathrm M}\right)}.
\end{align*}
The break-even stage-1 ask obtained from the prediction-market maker's break-even
condition is
\begin{align*}
      & A_1^{\PM,\mathrm{BE}}\left(p_{\mathrm M}\right)
    \notag                                              \\
    = &
    p_{\mathrm M}\Pi_{\mathrm M}\left(p_{\mathrm M}\right)
    +
    \left(1-p_{\mathrm M}\right)\Pi_{\mathrm L}\left(p_{\mathrm M}\right)
    \notag                                              \\
    = &
    \frac{1+\pi^2}{4}
    +
    \frac{\left(1-\pi^2\right)\left(1+p_{\mathrm M}\right)}{4}
    H\left(
    \frac{\sigma\pi\left(1-p_{\mathrm M}\right)}
    {\varepsilon\left(1+p_{\mathrm M}\right)}
    \right)
    \notag                                              \\
    = &
    \begin{cases}
        \displaystyle
        \frac{1}{2}
        +
        \frac{\left(1-\pi^2\right)p_{\mathrm M}}{4},
         & \displaystyle
        \frac{\sigma\pi\left(1-p_{\mathrm M}\right)}
        {1+p_{\mathrm M}}
        \geq \varepsilon, \\[0.8em]
        \displaystyle
        \frac{1}{2}
        +
        \frac{
            \left(1-\pi^2\right)\left(1-p_{\mathrm M}\right)
            \left(\sigma\pi-\varepsilon\right)
        }{8\varepsilon},
         & \displaystyle
        \frac{\sigma\pi\left(1-p_{\mathrm M}\right)}
        {1+p_{\mathrm M}}
        < \varepsilon,
    \end{cases}
\end{align*}
which is strictly greater than \(1/2\) for any
\(p_{\mathrm M}\in\left(0,1\right)\). Therefore the
non-crossing equilibrium stage-1 ask is just the break-even
stage-1 ask:
\[
    A_1^{\PM *}\left(p_{\mathrm M}\right)
    =
    \max\left\{
    A_1^{\PM,\mathrm{BE}}\left(p_{\mathrm M}\right),
    \frac{1}{2}
    \right\}
    =
    A_1^{\PM,\mathrm{BE}}\left(p_{\mathrm M}\right)
\]
for any \(p_{\mathrm M}\in\left(0,1\right)\), and
\[
    B_1^{\PM *}\left(p_{\mathrm M}\right)
    =
    1-A_1^{\PM *}\left(p_{\mathrm M}\right).
\]

The manipulator's expected profit conditional on entering at stage~1 is
\begin{align*}
      & S_{\mathrm M}\left(p_{\mathrm M}\right)  \\
    = &
    \Pi_{\mathrm M}\left(p_{\mathrm M}\right)
    -
    K\left(p_{\mathrm M}\right)
    -
    A_1^{\PM *}\left(p_{\mathrm M}\right), \notag    \\
    = &
    \frac{1-\pi^2}{4}
    \left[
        \left(1-p_{\mathrm M}\right)
        H\left(
        \frac{\sigma\pi\left(1-p_{\mathrm M}\right)}
        {\varepsilon\left(1+p_{\mathrm M}\right)}
        \right)
        -
        \frac{2\sigma\pi\left(1-p_{\mathrm M}\right)}
        {1+p_{\mathrm M}}
        \right] \notag \\
    = &
    \frac{1-\pi^2}{4}
    \begin{cases}
        \displaystyle
        \frac{\left(1-p_{\mathrm M}\right)\left(1+p_{\mathrm M}-2\sigma\pi\right)}
        {1+p_{\mathrm M}},
         & \displaystyle
        0 < p_{\mathrm M}
        \leq \frac{\sigma\pi-\varepsilon}{\sigma\pi+\varepsilon}, \\[1.2em]
        \displaystyle
        \frac{
            \left(1-p_{\mathrm M}\right)
            \left[
                \varepsilon\left(1+p_{\mathrm M}\right)
                +\sigma\pi\left(1-p_{\mathrm M}\right)
                -4\varepsilon\sigma\pi
                \right]
        }{
            2\varepsilon\left(1+p_{\mathrm M}\right)
        },
         & \displaystyle
        \frac{\sigma\pi-\varepsilon}{\sigma\pi+\varepsilon} < p_{\mathrm M}
        < 1.
    \end{cases}
\end{align*}
As in \cref{app:optimal-manipulator-entry-probability}, the manipulator chooses $p_{\mathrm M}$ anticipating the quotes it induces. Therefore, the manipulator maximizes the ex-ante profit
\[
    \bar S_{\mathrm M}\left(p_{\mathrm M}\right)
    =
    p_{\mathrm M}S_{\mathrm M}\left(p_{\mathrm M}\right)
    =
    \frac{1-\pi^2}{4}\Phi\left(p_{\mathrm M};\sigma,\pi,\varepsilon\right),
\]
where
\begin{equation}
    \Phi\left(p_{\mathrm M};\sigma,\pi,\varepsilon\right)
    \coloneq
    p_{\mathrm M}
    \left[
        \left(1-p_{\mathrm M}\right)
        H\left(
        \frac{\sigma\pi\left(1-p_{\mathrm M}\right)}
        {\varepsilon\left(1+p_{\mathrm M}\right)}
        \right)
        -
        \frac{2\sigma\pi\left(1-p_{\mathrm M}\right)}
        {1+p_{\mathrm M}}
        \right].
    \label{eq:app-horizon-two-entry-phi}
\end{equation}
Equivalently,
\[
    \Phi\left(p_{\mathrm M};\sigma,\pi,\varepsilon\right)
    =
    \begin{cases}
        \displaystyle
        \frac{
            p_{\mathrm M}\left(1-p_{\mathrm M}\right)
            \left(1+p_{\mathrm M}-2\sigma\pi\right)
        }{
            1+p_{\mathrm M}
        },
         & \displaystyle
        0 < p_{\mathrm M}
        \leq \frac{\sigma\pi-\varepsilon}{\sigma\pi+\varepsilon}, \\[1.2em]
        \displaystyle
        \frac{
            p_{\mathrm M}\left(1-p_{\mathrm M}\right)
            \left[
                \varepsilon\left(1+p_{\mathrm M}\right)
                +\sigma\pi\left(1-p_{\mathrm M}\right)
                -4\varepsilon\sigma\pi
                \right]
        }{
            2\varepsilon\left(1+p_{\mathrm M}\right)
        },
         & \displaystyle
        \frac{\sigma\pi-\varepsilon}{\sigma\pi+\varepsilon} < p_{\mathrm M}
        < 1.
    \end{cases}
\]
The two branches agree at the cutoff
\(\left(\sigma\pi-\varepsilon\right)/\left(\sigma\pi+\varepsilon\right)\), so
\(\Phi\) is continuous on \(\left(0,1\right)\). The first branch is strictly
positive because \(1+p_{\mathrm M}-2\sigma\pi>1-2\sigma\pi>0\). On the second
branch, the bracketed factor has slope \(\varepsilon-\sigma\pi<0\) in
\(p_{\mathrm M}\), so it is bounded below by its value at \(p_{\mathrm M}=1\),
namely \(2\varepsilon\left(1-2\sigma\pi\right)>0\). Therefore, $\Phi$ is strictly
positive on \(\left(0,1\right)\). If we define the boundary values
$$
    \Phi\left(0;\sigma,\pi,\varepsilon\right)=\lim_{p_{\mathrm M} \downarrow 0^+} \Phi\left(p_{\mathrm M};\sigma,\pi,\varepsilon\right) = 0,
    \qquad
    \Phi\left(1;\sigma,\pi,\varepsilon\right)=\lim_{p_{\mathrm M} \uparrow 1^-} \Phi\left(p_{\mathrm M};\sigma,\pi,\varepsilon\right) = 0,
$$
consistent with the convention that $\bar{S}_{\mathrm M}\left(0\right)=\bar{S}_{\mathrm M}\left(1\right)=0$,
then $\Phi$ is continuous on $\left[0, 1\right]$. Therefore, a maximizer exists and
must be an interior point in \(\left(0,1\right)\). To show that it is unique, we
can equivalently work with \(\ln \Phi\), whose branchwise first-order derivative is
\[
    \frac{\partial \ln \Phi}{\partial p_{\mathrm M}}
    =
    \begin{cases}
        \displaystyle
        \frac{1}{p_{\mathrm M}}
        -
        \frac{1}{1-p_{\mathrm M}}
        +
        \frac{1}{1+p_{\mathrm M}-2\sigma\pi}
        -
        \frac{1}{1+p_{\mathrm M}},
         & \displaystyle
        0 < p_{\mathrm M}
        < \frac{\sigma\pi-\varepsilon}{\sigma\pi+\varepsilon}, \\[1.2em]
        \displaystyle
        \frac{1}{p_{\mathrm M}}
        -
        \frac{1}{1-p_{\mathrm M}}
        +
        \frac{\varepsilon-\sigma\pi}
        {\varepsilon\left(1+p_{\mathrm M}\right)+\sigma\pi\left(1-p_{\mathrm M}\right)-4\varepsilon\sigma\pi}
        -
        \frac{1}{1+p_{\mathrm M}},
         & \displaystyle
        \frac{\sigma\pi-\varepsilon}{\sigma\pi+\varepsilon} < p_{\mathrm M}
        < 1.
    \end{cases}
\]
The branchwise second-order derivative is
\[
    \frac{\partial^2 \ln \Phi}{\partial p_{\mathrm M}^2}
    =
    \begin{cases}
        \displaystyle
        -\frac{1}{p_{\mathrm M}^2}
        -\frac{1}{\left(1-p_{\mathrm M}\right)^2}
        -\frac{1}{\left(1+p_{\mathrm M}-2\sigma\pi\right)^2}
        +\frac{1}{\left(1+p_{\mathrm M}\right)^2}
        <0,
         & \displaystyle
        0 < p_{\mathrm M}
        < \frac{\sigma\pi-\varepsilon}{\sigma\pi+\varepsilon}, \\[1.2em]
        \displaystyle
        \begin{aligned}
             & -\frac{1}{p_{\mathrm M}^2}
            -\frac{1}{\left(1-p_{\mathrm M}\right)^2}
            -\frac{\left(\sigma\pi-\varepsilon\right)^2}
             {\left[\varepsilon+\sigma\pi-4\varepsilon\sigma\pi-\left(\sigma\pi-\varepsilon\right)p_{\mathrm M}\right]^2} \\
             & +\frac{1}{\left(1+p_{\mathrm M}\right)^2}
            <0,
        \end{aligned}
         & \displaystyle
        \frac{\sigma\pi-\varepsilon}{\sigma\pi+\varepsilon} < p_{\mathrm M}
        < 1.
    \end{cases}
\]
On the first branch, the inequality follows because
\(0 < 1+p_{\mathrm M}-2\sigma\pi<1+p_{\mathrm M}\). On the second branch, the
inequality follows because $0 < 1 - p_{\mathrm M} < 1 + p_{\mathrm M}$.
Thus $\partial \ln \Phi / \partial p_{\mathrm M}$ is strictly decreasing
on both branches. At $\left(\sigma\pi-\varepsilon\right)/\left(\sigma\pi+\varepsilon\right)$,
the left derivative exceeds the right derivative:
\[
    \lim_{p_{\mathrm M}\uparrow\left(\frac{\sigma\pi-\varepsilon}{\sigma\pi+\varepsilon}\right)_-}
    \frac{\partial \ln \Phi}{\partial p_{\mathrm M}}
    -
    \lim_{p_{\mathrm M}\downarrow\left(\frac{\sigma\pi-\varepsilon}{\sigma\pi+\varepsilon}\right)_+}
    \frac{\partial \ln \Phi}{\partial p_{\mathrm M}}
    =
    \frac{
        \left(\sigma\pi+\varepsilon\right)^2
    }{
        4\varepsilon\sigma\pi\left(1-\sigma\pi-\varepsilon\right)
    }
    >
    0.
\]
Therefore, \(\ln \Phi\) (and equivalently \(\Phi\)) has a unique maximizer on $\left(0, 1\right)$.
It is either the unique interior zero of the first-order derivative or the kink
$\left(\sigma\pi-\varepsilon\right)/\left(\sigma\pi+\varepsilon\right)$.
Let \(p_{\mathrm M}^{*\left(2\right)}\) denote this unique maximizer. It remains to show that it
is strictly greater than the corresponding maximizer in the baseline model, which we denote by
\(p_{\mathrm M}^{*\left(1\right)}\).

Bounding the closed-form expression for \(p_{\mathrm M}^{*\left(1\right)}\) in
\cref{prop:canonical-equilibrium}, we have
\[
    0 <
    p_{\mathrm M}^{*\left(1\right)}
    =
    \sqrt{\left(\frac{1-\pi^2}{1+\pi^2}\right)^2
        +\frac{1-\pi^2}{1+\pi^2}p_0}
    -\frac{1-\pi^2}{1+\pi^2}
    <
    \sqrt{
        1+
        \frac{\varepsilon\left(1-2\sigma\pi\right)}
        {2\left(\sigma\pi-\varepsilon\right)
            +\varepsilon\left(1-2\sigma\pi\right)}
    }
    -1
    < \sqrt{2} - 1,
\]
where we used the fact that $p_{\mathrm M}^{*\left(1\right)}$ strictly increases in both
\(\left(1-\pi^2\right)/\left(1+\pi^2\right)\) and \(p_0\) as established in
\cref{app:comparative-statics-derivations}, $0 < \left(1-\pi^2\right)/\left(1+\pi^2\right) < 1$,
and
\[
    0
    <
    p_0
    <
    \frac{\varepsilon\left(1-2\sigma\pi\right)}
    {2\left(\sigma\pi-\varepsilon\right)
        +\varepsilon\left(1-2\sigma\pi\right)}
    <
    1.
\]
Because the first-order derivative of $\ln \Phi$ strictly decreases on each branch
and jumps downward at $\left(\sigma\pi-\varepsilon\right)/\left(\sigma\pi+\varepsilon\right)$,
it suffices to show that it is still positive at the sharper upper bound on
\(p_{\mathrm M}^{*\left(1\right)}\) displayed above.
If that upper bound lies weakly below
\(\left(\sigma\pi-\varepsilon\right)/\left(\sigma\pi+\varepsilon\right)\), only the first
branch is relevant. It is strictly decreasing on $\left(0, 1\right)$ and equals
\(\left(\sqrt{2}-2\sigma\pi\right)^{-1}>0\) at \(p_{\mathrm M}=\sqrt{2}-1\), so it must be
positive at the sharper upper bound.
If instead that upper bound lies strictly above
\(\left(\sigma\pi-\varepsilon\right)/\left(\sigma\pi+\varepsilon\right)\), it suffices to
check the second branch at that value. At the sharper upper-bound
value,
\[
    p_{\mathrm M}^2+2p_{\mathrm M}
    =
    \frac{\varepsilon\left(1-2\sigma\pi\right)}
    {2\left(\sigma\pi-\varepsilon\right)+\varepsilon\left(1-2\sigma\pi\right)}
    \iff
    \varepsilon
    =
    \frac{
        2\sigma\pi \left(p_{\mathrm M}^2 +2p_{\mathrm M}\right)
    }{
        \left(1+2\sigma\pi\right)\left(p_{\mathrm M}^2+2p_{\mathrm M}\right)
        +1-2\sigma\pi
    }.
\]
Substituting this expression into the third term in the second branch gives
\[
    \frac{\varepsilon-\sigma\pi}
    {\varepsilon\left(1+p_{\mathrm M}\right)
        +\sigma\pi\left(1-p_{\mathrm M}\right)-4\varepsilon\sigma\pi}
    =
    \frac{p_{\mathrm M}^2+2p_{\mathrm M}-1}
    {\left(1+p_{\mathrm M}\right)\left(p_{\mathrm M}^2+4p_{\mathrm M}+1\right)}.
\]
The remaining terms in the second branch satisfy
\[
    \frac{1}{p_{\mathrm M}}
    -
    \frac{1}{1-p_{\mathrm M}}
    -
    \frac{1}{1+p_{\mathrm M}}
    =
    \frac{-p_{\mathrm M}^2-2p_{\mathrm M}+1}
    {p_{\mathrm M}\left(1-p_{\mathrm M}\right)\left(1+p_{\mathrm M}\right)}.
\]
Therefore the second branch at the sharper upper-bound value equals
\begin{align*}
     & \frac{p_{\mathrm M}^2+2p_{\mathrm M}-1}{1+p_{\mathrm M}}
    \left[
        \frac{1}{p_{\mathrm M}^2+4p_{\mathrm M}+1} -
        \frac{1}{p_{\mathrm M}\left(1-p_{\mathrm M}\right)}
        \right]
    =
    -\frac{
         \left(2p_{\mathrm M}+1\right)\left(p_{\mathrm M}^2+2p_{\mathrm M}-1\right)
     }{
         p_{\mathrm M}\left(1-p_{\mathrm M}\right)
         \left(p_{\mathrm M}^2+4p_{\mathrm M}+1\right)
     }
    >
    0,
\end{align*}
with \(p_{\mathrm M}\) evaluated at the sharper upper-bound value. The expression is
positive because this sharper upper bound lies in \(\left(0,\sqrt{2}-1\right)\), so
\(p_{\mathrm M}^2+2p_{\mathrm M}-1<0\) whereas all other factors are positive.
Therefore, $\bar S_{\mathrm M}\left(p_{\mathrm M}\right)$ (under the two-ordinary-round
extension) is still increasing at \(p_{\mathrm M}^{*\left(1\right)}\), which implies
\[
    p_{\mathrm M}^{*\left(2\right)}
    >
    p_{\mathrm M}^{*\left(1\right)}.
\]

\subsubsection{Price discovery}
\label{app:horizon-two-price-discovery}
We prove the price-discovery part of
\cref{prop:horizon-two-equilibrium}. By the law of total
variance, for \(n=1,2\),
\[
    \E\left[\rvar_{h_2 h_3}^{\left(n\right)}\right]\left(p_{\mathrm M}\right)
    =
    \rvar_{\emptyset}
    -
    \operatorname{Var}\left(\bar V_{h_2h_3}^{\left(n\right)}\right)
    \left(p_{\mathrm M}\right)
    =
    \sigma^2
    -
    \E\left[
        \left(\bar V_{h_2h_3}^{\left(n\right)}\left(p_{\mathrm M}\right)\right)^2
        \right],
\]
because $\E\left[\bar V_{h_2h_3}^{\left(n\right)}\left(p_{\mathrm M}\right)\right]
    = \E\left[V\right]= 0$ by the law of iterated expectations. As in \cref{app:price-discovery},
lower residual variance is equivalent to greater dispersion of the final posterior mean.

We first compute the no-entry benchmarks. In the baseline model
\(\left(n=1\right)\), when \(p_{\mathrm M}=0\), stages~2 and~3 consist of two ordinary
spot rounds and thus share the same event structure as stage~2 under the two-ordinary-round extension,
which allows us to reuse the results from \cref{app:horizon-two-posteriors}.
The two agreeing order histories ($h_2=h_3=\hbuy$ or $\hsell$) each have probability
\(\left(1+\pi^2\right)/4\) and posterior means
\(\pm 2\sigma\pi/\left(1+\pi^2\right)\), while the two disagreeing order histories
(\(h_2=\hbuy, h_3=\hsell\) or \(h_2=\hsell, h_3=\hbuy\))
have final posterior mean zero. Hence
\begin{align*}
    \E\left[\rvar_{h_2 h_3}^{\left(1\right)}\right]\left(0\right)
     & =
    \sigma^2
    -
    2\cdot
    \frac{1+\pi^2}{4}
    \left(
    \frac{2\sigma\pi}{1+\pi^2}
    \right)^2
    =
    \frac{\sigma^2\left(1-\pi^2\right)}{1+\pi^2}.
\end{align*}
Under the two-ordinary-round extension \(\left(n=2\right)\), when
\(p_{\mathrm M}=0\), stages~2.1, 2.2, and~3 are three consecutive ordinary spot
rounds. We can calculate the order history probabilities and posteriors following the same
approach as in \cref{app:horizon-two} with an additional ordinary spot
round. The all-buy and all-sell order histories each have probability
\(\left(1+3\pi^2\right)/8\) and posterior means
\(\pm\sigma\pi\left(3+\pi^2\right)/\left(1+3\pi^2\right)\). The three ordered
order histories with two buys and one sell have posterior mean \(\sigma\pi\), the
three ordered order histories with one buy and two sells have posterior mean
\(-\sigma\pi\), and each of these six mixed ordered order histories has probability
\(\left(1-\pi^2\right)/8\). Therefore
\begin{align*}
    \E\left[\rvar_{h_2 h_3}^{\left(2\right)}\right]\left(0\right)
     & =
    \sigma^2
    -
    2\cdot
    \frac{1+3\pi^2}{8}
    \left(
    \frac{\sigma\pi\left(3+\pi^2\right)}{1+3\pi^2}
    \right)^2
    -
    6\cdot
    \frac{1-\pi^2}{8}
    \left(\sigma\pi\right)^2
    =
    \frac{\sigma^2\left(1-\pi^2\right)^2\left(1+2\pi^2\right)}{1+3\pi^2}.
\end{align*}
Hence the no-manipulation gain from adding the extra ordinary round is
\begin{align*}
    \E\left[\rvar_{h_2 h_3}^{\left(1\right)}\right]\left(0\right)
    -
    \E\left[\rvar_{h_2 h_3}^{\left(2\right)}\right]\left(0\right)=
    \frac{
        \sigma^2\pi^2\left(1-\pi^2\right)\left(1+\pi^2+2\pi^4\right)
    }{
        \left(1+\pi^2\right)\left(1+3\pi^2\right)
    } > 0.
\end{align*}

We next compute the increase in
\(\E\left[\rvar_{h_2 h_3}^{\left(2\right)}\right]\) caused by the selected
manipulation continuation for a fixed \(p_{\mathrm M}\in\left(0,1\right)\). Applying the law of total variance
to the dispersion of the final posterior mean as in \cref{app:price-discovery}, we have
\[
    \operatorname{Var}\left(\bar V_{h_2h_3}^{\left(2\right)}\right)
    \left(p_{\mathrm M}\right)
    =
    \operatorname{Var}\left(\bar V_{h_2}^{\left(2\right)}\right)
    +
    \sum_{h_2}q_{h_2}
    \operatorname{Var}\left(
    \bar V_{h_2h_3}^{\left(2\right)}
    \mid Y_2=h_2
    \right)\left(p_{\mathrm M}\right).
\]
$\operatorname{Var}\left(\bar V_{h_2}^{\left(2\right)}\right)$ and \(q_{h_2}\) are
independent of $p_{\mathrm M}$. Therefore, the increase in
\(\E\left[\rvar_{h_2 h_3}^{\left(2\right)}\right]\) equals the weighted loss
of posterior-mean dispersion conditional on stage-2 order history $h_2$:
\begin{align*}
      & \E\left[\rvar_{h_2 h_3}^{\left(2\right)}\right]\left(p_{\mathrm M}\right)
    -
    \E\left[\rvar_{h_2 h_3}^{\left(2\right)}\right]\left(0\right)                 \\
    = &
    \sum_{h_2}q_{h_2}
    \left[
        \operatorname{Var}\left(
        \bar V_{h_2h_3}^{\left(2\right)}
        \mid Y_2=h_2
        \right)\left(0\right)
        -
        \operatorname{Var}\left(
        \bar V_{h_2h_3}^{\left(2\right)}
        \mid Y_2=h_2
        \right)\left(p_{\mathrm M}\right)
        \right].
\end{align*}
At the two agreeing order histories $h_2 \in \{\hbuy\hbuy, \hsell\hsell\}$,
\(q_{\hbuy\hbuy}=q_{\hsell\hsell}=\left(1+\pi^2\right)/4\), and, as we will prove
below, the bracket equals
\[
    \frac{
        \sigma^2\pi^2p_{\mathrm M}\left(1-\pi^2\right)^3
    }{
        \left(1+\pi^2\right)^2\left(1+3\pi^2\right)
    }.
\]
At the two disagreeing order histories $h_2 \in \{\hbuy\hsell, \hsell\hbuy\}$,
\(q_{\hbuy\hsell}=q_{\hsell\hbuy}=\left(1-\pi^2\right)/4\), and, as we will prove
below, the bracket equals
\[
    \frac{
        2\sigma^2\pi^2p_{\mathrm M}
    }{
        1+ p_{\mathrm M}
    }.
\]
Therefore,
\begin{align*}
    \E\left[\rvar_{h_2 h_3}^{\left(2\right)}\right]\left(p_{\mathrm M}\right)
    -
    \E\left[\rvar_{h_2 h_3}^{\left(2\right)}\right]\left(0\right)
    =
    \sigma^2\pi^2p_{\mathrm M}\left(1-\pi^2\right)
    \left[
        \frac{\left(1-\pi^2\right)^2}
        {2\left(1+\pi^2\right)\left(1+3\pi^2\right)}
        +
        \frac{1}{1+p_{\mathrm M}}
        \right],
\end{align*}
which strictly increases in \(p_{\mathrm M}\) by inspection.\footnote{
    It is straightforward to verify that the expression's left limit at
    \(p_{\mathrm M}=1\) equals the gain from adding the
    extra ordinary round in the no-entry benchmark $
        \E\left[\rvar_{h_2 h_3}^{\left(1\right)}\right]\left(0\right)
        -
        \E\left[\rvar_{h_2 h_3}^{\left(2\right)}\right]\left(0\right)
    $, as expected, since when $p_{\mathrm M}=1$, stage~3 is fully occupied
    by the manipulator and does not contribute to price discovery at all.
}
Therefore, for every \(p_{\mathrm M}\in\left(0,1\right)\),
\begin{align*}
      & \E\left[\rvar_{h_2 h_3}^{\left(1\right)}\right]\left(0\right)
    -
    \E\left[\rvar_{h_2 h_3}^{\left(2\right)}\right]\left(p_{\mathrm M}\right)    \\
    = &
    \left[
        \E\left[\rvar_{h_2 h_3}^{\left(1\right)}\right]\left(0\right)
        -
        \E\left[\rvar_{h_2 h_3}^{\left(2\right)}\right]\left(0\right)
        \right]
    -
    \left[
        \E\left[\rvar_{h_2 h_3}^{\left(2\right)}\right]\left(p_{\mathrm M}\right)
        -
        \E\left[\rvar_{h_2 h_3}^{\left(2\right)}\right]\left(0\right)
        \right] \\
    = &
    \sigma^2\pi^2\left(1-\pi^2\right)\left(1-p_{\mathrm M}\right)
    \left[
        \frac{\left(1-\pi^2\right)^2}
        {2\left(1+\pi^2\right)\left(1+3\pi^2\right)}
        +
        \frac{1}{2\left(1+p_{\mathrm M}\right)}
        \right]                              \\
    > &
    0.
\end{align*}
Finally, since \(p_{\mathrm M}^{*\left(1\right)}, p_{\mathrm M}^{*\left(2\right)}\in\left(0,1\right)\),
\[
    \E\left[\rvar_{h_2 h_3}^{\left(2\right)}\right]
    \left(p_{\mathrm M}^{*\left(2\right)}\right)
    <
    \E\left[\rvar_{h_2 h_3}^{\left(1\right)}\right]\left(0\right)
    <
    \E\left[\rvar_{h_2 h_3}^{\left(1\right)}\right]
    \left(p_{\mathrm M}^{*\left(1\right)}\right),
\]
where the second inequality follows from \cref{prop:price-discovery-sign}.
This proves the price-discovery part of
\cref{prop:horizon-two-equilibrium}.

We now verify the claims about $
    \operatorname{Var}\left(
    \bar V_{h_2h_3}^{\left(2\right)}
    \mid Y_2=h_2
    \right)\left(0\right)
    -
    \operatorname{Var}\left(
    \bar V_{h_2h_3}^{\left(2\right)}
    \mid Y_2=h_2
    \right)\left(p_{\mathrm M}\right)$ below.

\paragraph{Agreeing order histories $h_2 \in \left\{\hbuy\hbuy, \hsell\hsell\right\}$.}
At \(h_2=\hbuy\hbuy\), the stage-2 posterior mean is $\bar V_{\hbuy\hbuy} = 2\sigma\pi /\left(1+\pi^2\right)$.
When $p_{\mathrm M} =0$, stage~3 is another ordinary spot round, so the conditional probability of
a stage-3 buy is
\begin{align*}
    q_{\hbuy\hbuy\hbuy}\left(0\right)
    =
    \ppost_{\hbuy\hbuy}\frac{1+\pi}{2}
    +
    \left(1-\ppost_{\hbuy\hbuy}\right)\frac{1-\pi}{2}
    =
    \frac{1+3\pi^2}{2\left(1+\pi^2\right)}
\end{align*}
and the conditional probability of a stage-3 sell is
\[
    q_{\hbuy\hbuy\hsell}\left(0\right)
    =
    1 - q_{\hbuy\hbuy\hbuy}\left(0\right)
    =
    \frac{1-\pi^2}{2\left(1+\pi^2\right)}.
\]
Bayes' rule gives the posterior mean after a stage-3 buy
\begin{align*}
      & \bar V_{\hbuy\hbuy\hbuy}^{\left(2\right)}\left(0\right)                             \\
    = &
    \sigma
    \frac{
        \Prb\left(h_2 = \hbuy\hbuy, h_3=\hbuy\mid V=+\sigma\right)\Prb\left(V=+\sigma\right)
        -
        \Prb\left(h_2 = \hbuy\hbuy, h_3=\hbuy\mid V=-\sigma\right)\Prb\left(V=-\sigma\right)
    }{
        \Prb\left(h_2 = \hbuy\hbuy, h_3=\hbuy\mid V=+\sigma\right)\Prb\left(V=+\sigma\right)
        +
        \Prb\left(h_2 = \hbuy\hbuy, h_3=\hbuy\mid V=-\sigma\right)\Prb\left(V=-\sigma\right)
    } \\
    = &
    \sigma
    \frac{
        \left(\frac{1+\pi}{2}\right)^3\frac{1}{2}
        -
        \left(\frac{1-\pi}{2}\right)^3\frac{1}{2}
    }{
        \left(\frac{1+\pi}{2}\right)^3\frac{1}{2}
        +
        \left(\frac{1-\pi}{2}\right)^3\frac{1}{2}
    }                                            \\
    = &
    \sigma
    \frac{
        \left(1+\pi\right)^3-\left(1-\pi\right)^3
    }{
        \left(1+\pi\right)^3+\left(1-\pi\right)^3
    }                                            \\
    = &
    \frac{\sigma\pi\left(3+\pi^2\right)}{1+3\pi^2}
\end{align*}
and similarly the posterior mean after a stage-3 sell
\begin{align*}
    \bar V_{\hbuy\hbuy\hsell}^{\left(2\right)}\left(0\right)
     & =
    \sigma
    \frac{
        \left(1+\pi\right)^2\left(1-\pi\right)
        -
        \left(1-\pi\right)^2\left(1+\pi\right)
    }{
        \left(1+\pi\right)^2\left(1-\pi\right)
        +
        \left(1-\pi\right)^2\left(1+\pi\right)
    }
    =
    \sigma\pi.
\end{align*}
Therefore the posterior-mean dispersion conditional on stage-2 order history $h_2=\hbuy\hbuy$ is
\begin{align*}
      & \operatorname{Var}\left(
    \bar V_{\hbuy\hbuy h_3}^{\left(2\right)}
    \mid Y_2=\hbuy\hbuy
    \right)\left(0\right)         \\
    = &
    \frac{1+3\pi^2}{2\left(1+\pi^2\right)}
    \left[
        \frac{\sigma\pi\left(3+\pi^2\right)}{1+3\pi^2}
        -
        \frac{2\sigma\pi}{1+\pi^2}
        \right]^2
    +
    \frac{1-\pi^2}{2\left(1+\pi^2\right)}
    \left[
        \sigma\pi
        -
        \frac{2\sigma\pi}{1+\pi^2}
        \right]^2 \\
    = &
    \frac{\sigma^2\pi^2\left(1-\pi^2\right)^3}
    {\left(1+\pi^2\right)^2\left(1+3\pi^2\right)}.
\end{align*}
When $p_{\mathrm M} \in \left(0, 1\right)$, the selected continuation equilibrium has
\(\alpha_{\hbuy\hbuy}^{*}=0\), so the manipulator does not pool with ordinary buy or
sell orders. When she occupies stage~3 with probability $p_{\mathrm M}$, she always
produces a no-trade outcome which does not move the posterior and hence contributes zero to
$\operatorname{Var}\left(\bar V_{\hbuy\hbuy h_3}^{\left(2\right)}\mid Y_2=\hbuy\hbuy\right)\left(p_{\mathrm M}\right)$.
Otherwise, with probability \(1-p_{\mathrm M}\), the manipulator is absent from stage~3
and it becomes an ordinary spot round. Thus
\[
    \operatorname{Var}\left(
    \bar V_{\hbuy\hbuy h_3}^{\left(2\right)}
    \mid Y_2=\hbuy\hbuy
    \right)\left(p_{\mathrm M}\right)
    =
    \left(1-p_{\mathrm M}\right)
    \frac{\sigma^2\pi^2\left(1-\pi^2\right)^3}
    {\left(1+\pi^2\right)^2\left(1+3\pi^2\right)}.
\]
The same calculation applies at \(h_2=\hsell\hsell\) by symmetry:
\[
    \operatorname{Var}\left(
    \bar V_{\hsell\hsell h_3}^{\left(2\right)}
    \mid Y_2=\hsell\hsell
    \right)\left(p_{\mathrm M}\right)
    =
    \left(1-p_{\mathrm M}\right)
    \frac{\sigma^2\pi^2\left(1-\pi^2\right)^3}
    {\left(1+\pi^2\right)^2\left(1+3\pi^2\right)}.
\]
Subtracting the two expressions verifies the claim.

\paragraph{Disagreeing order histories $h_2 \in \left\{\hbuy\hsell, \hsell\hbuy\right\}$.}
At \(h_2=\hbuy\hsell\), the stage-2 posterior mean is $\bar{V}_{\hbuy\hsell} = 0$.
When $p_{\mathrm M} =0$, stage~3 is another ordinary spot round, so the conditional probability
of a stage-3 buy is
\begin{align*}
    q_{\hbuy\hsell\hbuy}\left(0\right)
    =
    \ppost_{\hbuy\hsell}\frac{1+\pi}{2}
    +
    \left(1-\ppost_{\hbuy\hsell}\right)\frac{1-\pi}{2}
    =
    \frac{1}{2}
\end{align*}
and the conditional probability of a stage-3 sell is
\[
    q_{\hbuy\hsell\hsell}\left(0\right)
    =
    1 - q_{\hbuy\hsell\hbuy}\left(0\right)
    =
    \frac{1}{2}.
\]
Bayes' rule gives the posterior mean after a stage-3 buy
\begin{align*}
      & \bar V_{\hbuy\hsell\hbuy}^{\left(2\right)}\left(0\right)                             \\
    = &
    \sigma
    \frac{
        \Prb\left(h_2 = \hbuy\hsell, h_3=\hbuy\mid V=+\sigma\right)\Prb\left(V=+\sigma\right)
        -
        \Prb\left(h_2 = \hbuy\hsell, h_3=\hbuy\mid V=-\sigma\right)\Prb\left(V=-\sigma\right)
    }{
        \Prb\left(h_2 = \hbuy\hsell, h_3=\hbuy\mid V=+\sigma\right)\Prb\left(V=+\sigma\right)
        +
        \Prb\left(h_2 = \hbuy\hsell, h_3=\hbuy\mid V=-\sigma\right)\Prb\left(V=-\sigma\right)
    } \\
    = &
    \sigma
    \frac{
        \left(\frac{1+\pi}{2}\right)^2\frac{1-\pi}{2}\frac{1}{2}
        -
        \left(\frac{1-\pi}{2}\right)^2\frac{1+\pi}{2}\frac{1}{2}
    }{
        \left(\frac{1+\pi}{2}\right)^2\frac{1-\pi}{2}\frac{1}{2}
        +
        \left(\frac{1-\pi}{2}\right)^2\frac{1+\pi}{2}\frac{1}{2}
    }                              \\
    = &
    \sigma
    \frac{
        \left(1+\pi\right)^2\left(1-\pi\right)
        -
        \left(1-\pi\right)^2\left(1+\pi\right)
    }{
        \left(1+\pi\right)^2\left(1-\pi\right)
        +
        \left(1-\pi\right)^2\left(1+\pi\right)
    }                                                \\
    = &
    \sigma\pi,
\end{align*}
and similarly the posterior mean after a stage-3 sell
\begin{align*}
    \bar V_{\hbuy\hsell\hsell}^{\left(2\right)}\left(0\right)
     & =
    \sigma
    \frac{
        \left(1+\pi\right)\left(1-\pi\right)^2
        -
        \left(1-\pi\right)\left(1+\pi\right)^2
    }{
        \left(1+\pi\right)\left(1-\pi\right)^2
        +
        \left(1-\pi\right)\left(1+\pi\right)^2
    }
    =
    -\sigma\pi.
\end{align*}
Therefore the posterior-mean dispersion conditional on stage-2 order history $h_2=\hbuy\hsell$ is
\begin{align*}
    \operatorname{Var}\left(
    \bar V_{\hbuy\hsell h_3}^{\left(2\right)}
    \mid Y_2=\hbuy\hsell
    \right)\left(0\right)
    =
    \frac{1}{2}
    \left[
        \sigma\pi
        -
        0
        \right]^2
    +
    \frac{1}{2}
    \left[
        -\sigma\pi
        -
        0
        \right]^2
    =
    \sigma^2\pi^2.
\end{align*}
When $p_{\mathrm M} \in \left(0, 1\right)$, the selected continuation equilibrium has \(\alpha_{\hbuy\hsell}^{*}=1\)
and equilibrium ask and bid as given in \cref{app:horizon-two-entry}:
\[
    A_{\hbuy\hsell}^{*}\left(p_{\mathrm M}\right)
    =
    \frac{\sigma\pi\left(1-p_{\mathrm M}\right)}{1+p_{\mathrm M}},
    \qquad
    B_{\hbuy\hsell}
    =
    -\sigma\pi.
\]
The conditional probability of a stage-3 buy is
\begin{align*}
    q_{\hbuy\hsell\hbuy}\left(p_{\mathrm M}\right)
    =
    p_{\mathrm M}
    +
    \left(1-p_{\mathrm M}\right)
    q_{\hbuy\hsell\hbuy}\left(0\right)
    =
    \frac{1+p_{\mathrm M}}{2}
\end{align*}
and the conditional probability of a stage-3 sell is
\begin{align*}
    q_{\hbuy\hsell\hsell}\left(p_{\mathrm M}\right)
    =
    1 - q_{\hbuy\hsell\hbuy}\left(p_{\mathrm M}\right)
    =
    \frac{1-p_{\mathrm M}}{2}.
\end{align*}
A no-trade outcome at stage-3 must come from the manipulator, so it does not move the posterior and
contributes zero to $\operatorname{Var}\left(
    \bar V_{\hbuy\hsell h_3}^{\left(2\right)}
    \mid Y_2=\hbuy\hsell
    \right)\left(p_{\mathrm M}\right)$. Therefore,
\begin{align*}
    \operatorname{Var}\left(
    \bar V_{\hbuy\hsell h_3}^{\left(2\right)}
    \mid Y_2=\hbuy\hsell
    \right)\left(p_{\mathrm M}\right)
    =
    \frac{1+p_{\mathrm M}}{2}
    \left[
        \frac{\sigma\pi\left(1-p_{\mathrm M}\right)}
        {1+p_{\mathrm M}} - 0
        \right]^2
    +
    \frac{1-p_{\mathrm M}}{2}\left[-\sigma\pi - 0\right]^2
    =
    \sigma^2\pi^2\frac{1-p_{\mathrm M}}{1+p_{\mathrm M}}.
\end{align*}
The same calculation applies at \(h_2=\hsell\hbuy\) by symmetry:
\begin{align*}
    \operatorname{Var}\left(
    \bar V_{\hsell\hbuy h_3}^{\left(2\right)}
    \mid Y_2=\hsell\hbuy
    \right)\left(p_{\mathrm M}\right)
    =
    \sigma^2\pi^2\frac{1-p_{\mathrm M}}{1+p_{\mathrm M}}.
\end{align*}
Subtracting the two expressions verifies the claim.

\subsection{\texorpdfstring{$n$}{n}-Ordinary-Round Extension Calculations}
\label{app:horizon-n}

We prove \cref{prop:horizon-n-equilibrium,cor:horizon-n-manipulation-prone-mass}
in this appendix. The analysis is
conditional on a stage-1 prediction-market buy; the sell branch is symmetric.

\subsubsection{Stage 2 Posteriors}
\label{app:horizon-n-posteriors}

Fix \(n \geq 1\). A stage-2 order history is a sequence of $n$ outcomes:
\[
    h_2=h_{2.1}\cdots h_{2.n}\in\left\{\hbuy,\hsell\right\}^n.
\]
As we will see, the relevant sufficient statistic for the stage-2 posterior that enters
stage~3 is the order imbalance
\[
    d\left(h_2\right)
    =
    \sum_{j=1}^{n}\one\left\{h_{2.j}=\hbuy\right\}
    -
    \sum_{j=1}^{n}\one\left\{h_{2.j}=\hsell\right\}
    =
    2\sum_{j=1}^{n}\one\left\{h_{2.j}=\hbuy\right\}
    -n,
\]
which takes values in the set
\[
    \mathcal D_n
    =
    \left\{-n,-n+2,\ldots,n-2,n\right\}.
\]
For \(k\in\mathcal D_n\), let
\[
    b\left(k\right)=\frac{n+k}{2},
    \qquad
    s\left(k\right)=\frac{n-k}{2},
\]
so that stage-2 order histories with imbalance \(k\) have \(b\left(k\right)\) buys and
\(s\left(k\right)\) sells. Conditional on \(V=+\sigma\), any individual
stage-2 order history with imbalance \(k\) has probability
\[
    \left(\frac{1+\pi}{2}\right)^{b\left(k\right)}
    \left(\frac{1-\pi}{2}\right)^{s\left(k\right)},
\]
whereas conditional on \(V=-\sigma\), it has probability
\[
    \left(\frac{1-\pi}{2}\right)^{b\left(k\right)}
    \left(\frac{1+\pi}{2}\right)^{s\left(k\right)}.
\]
Averaging over the symmetric prior on \(V\), any individual stage-2 order history
\(h_2\) with \(d\left(h_2\right)=k\) has probability
\[
    q_{h_2}
    =
    \frac{
        \left(1+\pi\right)^{b\left(k\right)}
        \left(1-\pi\right)^{s\left(k\right)}
        +
        \left(1-\pi\right)^{b\left(k\right)}
        \left(1+\pi\right)^{s\left(k\right)}
    }{2^{n+1}}.
\]
Thus the total probability of all stage-2 order histories with imbalance \(k\) is
\begin{equation}
    \label{eq:app-horizon-n-imbalance-cell-mass}
    q_{d=k}
    \coloneq
    \sum_{h_2:d\left(h_2\right)=k}q_{h_2}
    =
    \binom{n}{b\left(k\right)}
    \frac{
        \left(1+\pi\right)^{b\left(k\right)}
        \left(1-\pi\right)^{s\left(k\right)}
        +
        \left(1-\pi\right)^{b\left(k\right)}
        \left(1+\pi\right)^{s\left(k\right)}
    }{2^{n+1}}.
\end{equation}
Since $b\left(-k\right)=s\left(k\right)=n-b\left(k\right)$ and similarly
$s\left(-k\right)=b\left(k\right)=n-s\left(k\right)$, we can verify that
\begin{equation}
    \label{eq:app-horizon-n-imbalance-symmetry}
    q_{d=-k} = q_{d=k}.
\end{equation}
Let
\[
    r \coloneq \frac{1+\pi}{1-\pi}
\]
denote the likelihood ratio of a buy order in an ordinary spot round:
a buy is \(r\) times as likely when \(V=+\sigma\) as when \(V=-\sigma\). By symmetry,
the likelihood ratio of a sell order is $r^{-1}$.
Since the $n$ ordinary spot rounds are conditionally independent given $V$,
Bayes' rule gives directly, for any stage-2 order history \(h_2\) with
\(d\left(h_2\right)=k\),
\[
    \begin{aligned}
          & \ppost_{h_2}                                \\
        = &
        \Prb\left(V=+\sigma\mid Y_2=h_2\right)    \\
        = &
        \frac{
            \Prb\left(Y_2=h_2\mid V=+\sigma\right)
            \Prb\left(V=+\sigma\right)
        }{
            \Prb\left(Y_2=h_2\mid V=+\sigma\right)
            \Prb\left(V=+\sigma\right)
            +
            \Prb\left(Y_2=h_2\mid V=-\sigma\right)
            \Prb\left(V=-\sigma\right)
        } \\
        = &
        \frac{
            \left(\frac{1+\pi}{2}\right)^{b\left(k\right)}
            \left(\frac{1-\pi}{2}\right)^{s\left(k\right)}
            \frac{1}{2}
        }{
            \left(\frac{1+\pi}{2}\right)^{b\left(k\right)}
            \left(\frac{1-\pi}{2}\right)^{s\left(k\right)}
            \frac{1}{2}
            +
            \left(\frac{1-\pi}{2}\right)^{b\left(k\right)}
            \left(\frac{1+\pi}{2}\right)^{s\left(k\right)}
            \frac{1}{2}
        }
        =
        \frac{r^k}{1+r^k},
    \end{aligned}
\]
where the last equality cancels common factors and uses \(b\left(k\right)-s\left(k\right)=k\).
This shows that \(\ppost_{h_2}\) depends on \(h_2\) only through \(k=d\left(h_2\right)\), which proves
that \(d\left(h_2\right)\) is indeed a sufficient statistic for the stage-2 posterior.
We write this common value as
\[
    \ppost_{d=k}
    =
    \frac{r^k}{1+r^k}.
\]
The posterior mean and variance are therefore
\[
    \bar V_{d=k}
    =
    \sigma\left(2\ppost_{d=k}-1\right)
    =
    \sigma\frac{r^k-1}{r^k+1},
\]
and
\[
    \rvar_{d=k}
    =
    4\sigma^2\ppost_{d=k}\left(1-\ppost_{d=k}\right)
    =
    \frac{4\sigma^2r^k}{\left(1+r^k\right)^2}.
\]
Thus every stage-2 order history with the same imbalance induces the same posterior:
\[
    \ppost_{h_2}
    =
    \ppost_{d=d\left(h_2\right)},
    \qquad
    \bar V_{h_2}
    =
    \bar V_{d=d\left(h_2\right)},
    \qquad
    \rvar_{h_2}
    =
    \rvar_{d=d\left(h_2\right)}.
\]

\subsubsection{Selected stage-3 continuation equilibrium}
\label{app:horizon-n-continuation}

Substituting \(\ppost_{d=k}\) into the general ask and bid formulas
\eqref{eq:app-stage3-ask-general} and \eqref{eq:app-stage3-bid-general} gives
\[
    A_{d=k}\left(\alpha_{d=k};p_{\mathrm M}\right)
    =
    \sigma
    \frac{
        p_{\mathrm M}\alpha_{d=k}\left(\frac{r^k-1}{r^k+1}\right)
        +
        \frac{1-p_{\mathrm M}}{2}\left(\frac{r^k-1}{r^k+1}+\pi\right)
    }{
        p_{\mathrm M}\alpha_{d=k}
        +
        \frac{1-p_{\mathrm M}}{2}\left(1+\frac{r^k-1}{r^k+1}\pi\right)
    }, \qquad
    B_{d=k}
    =
    \sigma\frac{r^{k-1}-1}{r^{k-1}+1}.
\]
At imbalance $k$, the gap
between the expected continuation payoff from buying and that from staying out is
\[
    \Gamma_{d=k}\left(\alpha_{d=k};p_{\mathrm M}\right)
    =
    H\left(\frac{A_{d=k}\left(\alpha_{d=k};p_{\mathrm M}\right)}{\varepsilon}\right)
    -
    H\left(\frac{\bar V_{d=k}}{\varepsilon}\right)
    -
    \left[
        A_{d=k}\left(\alpha_{d=k};p_{\mathrm M}\right)-\bar V_{d=k}
        \right].
\]
\paragraph{Favorable imbalances \(k\geq 1\): pure stay-out equilibrium.}
For any \(p_{\mathrm M}\in\left(0,1\right)\), any
\(\alpha_{d=k}\in\left[0,1\right]\), and any \(k\geq 1\),
the stage-3 ask is at least the posterior mean, so
\[
    A_{d=k}\left(\alpha_{d=k};p_{\mathrm M}\right)
    \geq
    \bar V_{d=k}
    \geq
    \bar V_{d=1}
    =
    \sigma\pi
    >
    \varepsilon.
\]
Therefore
\(H\left(A_{d=k}\left(\alpha_{d=k};p_{\mathrm M}\right)/\varepsilon\right)
=H\left(\bar V_{d=k}/\varepsilon\right)=1\). Hence
\[
    \Gamma_{d=k}\left(\alpha_{d=k};p_{\mathrm M}\right)
    =
    1-1-
    \left[
        A_{d=k}\left(\alpha_{d=k};p_{\mathrm M}\right)-\bar V_{d=k}
        \right]
    \leq
    0.
\]
Thus, for any \(p_{\mathrm M}\in\left(0,1\right)\) and any \(k\geq 1\),
\(\alpha_{d=k}^{*}=0\) is a stable fixed point and is clearly the smallest.

\paragraph{Hopeless imbalances \(k\leq -2\): pure stay-out equilibrium.}
For any \(p_{\mathrm M}\in\left(0,1\right)\) and any \(k\leq -2\),
\[
    A_{d=k}\left(0;p_{\mathrm M}\right)
    =
    \sigma\frac{r^{k+1}-1}{r^{k+1}+1}
    \leq
    \sigma\frac{r^{-2+1}-1}{r^{-2+1}+1}
    =
    \bar V_{d=-1}
    =
    -\sigma\pi
    <
    -\varepsilon.
\]
Since the ask strictly decreases in \(\alpha_{d=k}\) and is at least the posterior mean,
\[
    \bar V_{d=k}
    \leq
    A_{d=k}\left(\alpha_{d=k};p_{\mathrm M}\right)
    \leq
    A_{d=k}\left(0;p_{\mathrm M}\right)
    <
    -\varepsilon.
\]
Consequently
\(H\left(A_{d=k}\left(\alpha_{d=k};p_{\mathrm M}\right)/\varepsilon\right)
=H\left(\bar V_{d=k}/\varepsilon\right)=0\), and
\[
    \Gamma_{d=k}\left(\alpha_{d=k};p_{\mathrm M}\right)
    =
    0-0
    -\left[
        A_{d=k}\left(\alpha_{d=k};p_{\mathrm M}\right)-\bar V_{d=k}
        \right]
    \leq
    0.
\]
Thus, for any \(p_{\mathrm M}\in\left(0,1\right)\) and any \(k\leq -2\),
\(\alpha_{d=k}^{*}=0\) is a stable fixed point and is clearly the smallest.

\paragraph{Boundary imbalance \(k=0\): pure push equilibrium.}
$k=0$ can occur only when \(n\) is even. Fix
\(p_{\mathrm M}\in\left(0,1\right)\). Here
\[
    \bar V_{d=0}=0,
\]
and the stage-3 ask is
\[
    A_{d=0}\left(\alpha_{d=0};p_{\mathrm M}\right)
    =
    \frac{\sigma\pi\left(1-p_{\mathrm M}\right)}
    {1-p_{\mathrm M}+2p_{\mathrm M}\alpha_{d=0}}.
\]
At \(\alpha_{d=0}=0\),
\[
    A_{d=0}\left(0;p_{\mathrm M}\right)=\sigma\pi>\varepsilon,
\]
so
\[
    \Gamma_{d=0}\left(0;p_{\mathrm M}\right)
    =
    1-\frac{1}{2}-\sigma\pi
    =
    \frac{1}{2}-\sigma\pi
    >
    0.
\]
Thus \(\alpha_{d=0}=0\) is not a fixed point. For any
\(\alpha_{d=0}\in\left[0,1\right]\), the ask is
\[
    A_{d=0}\left(\alpha_{d=0};p_{\mathrm M}\right)
    \in
    \left[
        \frac{\sigma\pi\left(1-p_{\mathrm M}\right)}{1+p_{\mathrm M}},
        \sigma\pi
        \right]
    \subset
    \left(0,\frac{1}{2}\right).
\]
If \(A_{d=0}\left(\alpha_{d=0};p_{\mathrm M}\right)\geq\varepsilon\), then
\[
    \Gamma_{d=0}\left(\alpha_{d=0};p_{\mathrm M}\right)
    =
    \frac{1}{2}-A_{d=0}\left(\alpha_{d=0};p_{\mathrm M}\right)
    \geq
    \frac{1}{2}-\sigma\pi
    >
    0.
\]
If \(0<A_{d=0}\left(\alpha_{d=0};p_{\mathrm M}\right)<\varepsilon\), then
\[
    \Gamma_{d=0}\left(\alpha_{d=0};p_{\mathrm M}\right)
    =
    A_{d=0}\left(\alpha_{d=0};p_{\mathrm M}\right)
    \left(\frac{1}{2\varepsilon}-1\right)
    >
    0
\]
because $\varepsilon < 1/2$. Thus no interior fixed point exists either,
and the selected continuation equilibrium is the endpoint
\[
    \alpha_{d=0}^{*}=1.
\]

\paragraph{Boundary imbalance \(k=-1\): mixed or pure push.}
$k=-1$ can occur only when \(n\) is odd. Fix
\(p_{\mathrm M}\in\left(0,1\right)\). Here
\[
    \bar V_{d=-1}=-\sigma\pi,
\]
and
\[
    A_{d=-1}\left(\alpha_{d=-1};p_{\mathrm M}\right)
    =
    \frac{-\sigma\pi p_{\mathrm M}\alpha_{d=-1}}
    {p_{\mathrm M}\alpha_{d=-1}+\left(1-p_{\mathrm M}\right)\frac{1-\pi^2}{2}}.
\]
This is exactly the unfavorable-history continuation problem from
\cref{app:selected-continuation-equilibrium-by-path}. Therefore, with
\[
    p_0
    =
    \frac{\varepsilon\left(1-\pi^2\right)\left(1-2\sigma\pi\right)}
    {2\left(\sigma\pi-\varepsilon\right)
        +\varepsilon\left(1-\pi^2\right)\left(1-2\sigma\pi\right)},
\]
the selected continuation equilibrium is
\[
    \alpha_{d=-1}^{*}\left(p_{\mathrm M}\right)
    =
    \begin{cases}
        1,
         & 0<p_{\mathrm M}\leq p_0, \\[0.8em]
        \displaystyle
        \frac{1-p_{\mathrm M}}{p_{\mathrm M}}
        \frac{\varepsilon\left(1-\pi^2\right)\left(1-2\sigma\pi\right)}
        {2\left(\sigma\pi-\varepsilon\right)},
         & p_0<p_{\mathrm M}<1.
    \end{cases}
\]
The corresponding equilibrium ask is
\[
    A_{d=-1}^{*}\left(p_{\mathrm M}\right)
    =
    \begin{cases}
        \displaystyle
        -\frac{2\sigma\pi p_{\mathrm M}}
         {\left(1-\pi^2\right)+\left(1+\pi^2\right)p_{\mathrm M}},
         & 0<p_{\mathrm M}\leq p_0, \\[1.2em]
        \displaystyle
        -\frac{\varepsilon\left(1-2\sigma\pi\right)}{1-2\varepsilon},
         & p_0<p_{\mathrm M}<1.
    \end{cases}
\]
Thus \(d=-1\) is the unique pivotal imbalance when \(n\) is odd.

\subsubsection{Optimal manipulator entry probability}
\label{app:horizon-n-entry}

Both the pricing of the binary contract and the manipulator's stage-1 entry problem use
the selected continuation equilibrium from \cref{app:horizon-n-continuation}.
\Cref{app:horizon-n-posteriors} shows that the stage-2 posterior depends on
\(h_2\) only through the imbalance \(d\left(h_2\right)\), so we can aggregate over
imbalances rather than individual order histories. Combining \cref{app:horizon-n-posteriors},
the general ask and bid formulas
\eqref{eq:app-stage3-ask-general} and \eqref{eq:app-stage3-bid-general}, and
\cref{app:horizon-n-continuation}, we obtain the following stage-3 asks, bids, and no-trade
posterior means for a stage-2 order history \(h_2\) with \(d\left(h_2\right)=k\):
\[
    A_{d=k}^{*}\left(p_{\mathrm M}\right)
    =
    \begin{cases}
        \displaystyle
        \sigma\frac{r^{k+1}-1}{r^{k+1}+1}=\bar{V}_{d=k+1},
         & k\geq 1 \text{ or } k\leq -2,              \\[1.2em]
        \displaystyle
        \frac{\sigma\pi\left(1-p_{\mathrm M}\right)}
        {1+p_{\mathrm M}},
         & k=0,                                       \\[1.2em]
        \displaystyle
        -\frac{2\sigma\pi p_{\mathrm M}}
         {\left(1-\pi^2\right)+\left(1+\pi^2\right)p_{\mathrm M}},
         & k=-1 \text{ and } 0<p_{\mathrm M}\leq p_0, \\[1.2em]
        \displaystyle
        -\frac{\varepsilon\left(1-2\sigma\pi\right)}{1-2\varepsilon},
         & k=-1 \text{ and } p_0<p_{\mathrm M}<1,
    \end{cases}
\]
\[
    B_{d=k}
    =
    \sigma\frac{r^{k-1}-1}{r^{k-1}+1}=\bar{V}_{d=k-1}.
\]
\[
    \bar V_{d=k}
    =
    \sigma\frac{r^k-1}{r^k+1}.
\]

For each stage-2 order history \(h_2\) with \(d\left(h_2\right)=k\), the expected contract payoffs are
\begin{align*}
    \Pi_{\mathrm M}\left(p_{\mathrm M}, k\right)
     & =
    \alpha_{d=k}^{*}\left(p_{\mathrm M}\right)
    H\left(A_{d=k}^{*}\left(p_{\mathrm M}\right)/\varepsilon\right)
    +
    \left(1-\alpha_{d=k}^{*}\left(p_{\mathrm M}\right)\right)
    H\left(\bar V_{d=k}/\varepsilon\right), \\
    \Pi_{\mathrm L}\left(p_{\mathrm M}, k\right)
     & =
    q_{d=k,\hbuy}\left(0\right)
    H\left(A_{d=k}^{*}\left(p_{\mathrm M}\right)/\varepsilon\right)
    +
    q_{d=k,\hsell}\left(0\right)H\left(B_{d=k}/\varepsilon\right).
\end{align*}

\paragraph{Even \(n\).}
Suppose \(n\) is even. Substituting the selected continuation equilibrium and the
stage-3 quotes gives
\begin{align*}
    \Pi_{\mathrm M}\left(p_{\mathrm M},k\right)
     & =
    \begin{cases}
        1,
         & k\geq 2,  \\[0.6em]
        \displaystyle
        H\left(
        \frac{\sigma\pi\left(1-p_{\mathrm M}\right)}
        {\varepsilon\left(1+p_{\mathrm M}\right)}
        \right),
         & k=0,      \\[1.2em]
        0,
         & k\leq -2,
    \end{cases} \\
    \Pi_{\mathrm L}\left(p_{\mathrm M},k\right)
     & =
    \begin{cases}
        1,
         & k\geq 2,  \\[0.6em]
        \displaystyle
        \frac{1}{2}
        H\left(
        \frac{\sigma\pi\left(1-p_{\mathrm M}\right)}
        {\varepsilon\left(1+p_{\mathrm M}\right)}
        \right),
         & k=0,      \\[1.2em]
        0,
         & k\leq -2.
    \end{cases}
\end{align*}
The probability of imbalance \(k=0\) is \(q_{d=0}\). By
\eqref{eq:app-horizon-n-imbalance-symmetry}, the positive and negative tails have
equal probability, so
\[
    \sum_{k\geq 2}
    q_{d=k}
    =
    \sum_{k \leq -2}
    q_{d=k}
    =
    \frac{1-q_{d=0}}{2}.
\]
Aggregating over $k$ gives
\begin{align*}
    \Pi_{\mathrm M}\left(p_{\mathrm M}\right)
     & =
    \frac{1-q_{d=0}}{2}
    +
    q_{d=0}
    H\left(
    \frac{\sigma\pi\left(1-p_{\mathrm M}\right)}
    {\varepsilon\left(1+p_{\mathrm M}\right)}
    \right), \\
    \Pi_{\mathrm L}\left(p_{\mathrm M}\right)
     & =
    \frac{1-q_{d=0}}{2}
    +
    \frac{q_{d=0}}{2}
    H\left(
    \frac{\sigma\pi\left(1-p_{\mathrm M}\right)}
    {\varepsilon\left(1+p_{\mathrm M}\right)}
    \right), \\
    K\left(p_{\mathrm M}\right)
     & =
    q_{d=0}
    \frac{\sigma\pi\left(1-p_{\mathrm M}\right)}
    {1+p_{\mathrm M}}.
\end{align*}
The break-even stage-1 ask obtained from the prediction-market maker's break-even
condition is
\begin{align*}
      & A_1^{\PM,\mathrm{BE}}\left(p_{\mathrm M}\right) \\
    = &
    p_{\mathrm M}\Pi_{\mathrm M}\left(p_{\mathrm M}\right)
    +
    \left(1-p_{\mathrm M}\right)
    \Pi_{\mathrm L}\left(p_{\mathrm M}\right)       \\
    = &
    \frac{1-q_{d=0}}{2}
    +
    \frac{q_{d=0}\left(1+p_{\mathrm M}\right)}{2}
    H\left(
    \frac{\sigma\pi\left(1-p_{\mathrm M}\right)}
    {\varepsilon\left(1+p_{\mathrm M}\right)}
    \right)                                             \\
    = &
    \begin{cases}
        \displaystyle
        \frac{1}{2}
        +
        \frac{q_{d=0}p_{\mathrm M}}{2},
         & \displaystyle
        \frac{\sigma\pi\left(1-p_{\mathrm M}\right)}
        {1+p_{\mathrm M}}
        \geq \varepsilon, \\[0.8em]
        \displaystyle
        \frac{1}{2}
        +
        \frac{
            q_{d=0}\left(1-p_{\mathrm M}\right)
            \left(\sigma\pi-\varepsilon\right)
        }{4\varepsilon},
         & \displaystyle
        \frac{\sigma\pi\left(1-p_{\mathrm M}\right)}
        {1+p_{\mathrm M}}
        < \varepsilon.
    \end{cases}
\end{align*}
which is strictly greater than \(1/2\) for any
\(p_{\mathrm M}\in\left(0,1\right)\). Therefore the non-crossing equilibrium stage-1
ask is just the break-even stage-1 ask:
\[
    A_1^{\PM *}\left(p_{\mathrm M}\right)
    =
    \max\left\{
    A_1^{\PM,\mathrm{BE}}\left(p_{\mathrm M}\right),
    \frac{1}{2}
    \right\}
    =
    A_1^{\PM,\mathrm{BE}}\left(p_{\mathrm M}\right),
\]
and
\[
    B_1^{\PM *}\left(p_{\mathrm M}\right) = 1 - A_1^{\PM *}\left(p_{\mathrm M}\right)
\]
for any \(p_{\mathrm M}\in\left(0,1\right)\).

The manipulator's expected profit conditional on entering at stage~1 is
\begin{align*}
      & S_{\mathrm M}\left(p_{\mathrm M}\right) \\
    = &
    \Pi_{\mathrm M}\left(p_{\mathrm M}\right)
    -
    K\left(p_{\mathrm M}\right)
    -
    A_1^{\PM *}\left(p_{\mathrm M}\right)           \\
    = &
    \frac{q_{d=0}}{2}
    \left[
        \left(1-p_{\mathrm M}\right)
        H\left(
        \frac{\sigma\pi\left(1-p_{\mathrm M}\right)}
        {\varepsilon\left(1+p_{\mathrm M}\right)}
        \right)
        -
        \frac{2\sigma\pi\left(1-p_{\mathrm M}\right)}
        {1+p_{\mathrm M}}
        \right].
\end{align*}
Hence
\[
    \bar S_{\mathrm M}\left(p_{\mathrm M}\right)
    =
    p_{\mathrm M}S_{\mathrm M}\left(p_{\mathrm M}\right)
    =
    \frac{q_{d=0}}{2}
    \Phi\left(p_{\mathrm M};\sigma,\pi,\varepsilon\right),
\]
where \(\Phi\) is exactly the same function as defined in \eqref{eq:app-horizon-two-entry-phi}
for the manipulator's optimization problem under the two-ordinary-round extension.
Thus every even \(n\) has the same optimal entry probability as the two-ordinary-round extension:
\[
    p_{\mathrm M}^{*\left(n\right)}
    =
    p_{\mathrm M}^{*\left(2\right)}
    \quad
    \text{if }n\text{ is even}.
\]
Restoring horizon superscripts, since \(q_{d=0}=m^{\left(n\right)}\) for even \(n\)
and \(m^{\left(2\right)}=\left(1-\pi^2\right)/2\), where $m^{\left(n\right)}$ is defined in \cref{cor:horizon-n-manipulation-prone-mass},
the preceding formulas imply
\begin{align*}
    S_{\mathrm M}^{\left(n\right)}
    \left(p_{\mathrm M}^{*\left(n\right)}\right)
     & =
    \frac{m^{\left(n\right)}}{m^{\left(2\right)}}
    S_{\mathrm M}^{\left(2\right)}
    \left(p_{\mathrm M}^{*\left(2\right)}\right), \\
    \bar S_{\mathrm M}^{\left(n\right)}
    \left(p_{\mathrm M}^{*\left(n\right)}\right)
     & =
    \frac{m^{\left(n\right)}}{m^{\left(2\right)}}
    \bar S_{\mathrm M}^{\left(2\right)}
    \left(p_{\mathrm M}^{*\left(2\right)}\right)
\end{align*}
for even \(n\).

\paragraph{Odd \(n\).}
Suppose \(n\) is odd. Substituting the selected continuation equilibrium and the
stage-3 quotes gives
\begin{align*}
    \Pi_{\mathrm M}\left(p_{\mathrm M},k\right)
     & =
    \begin{cases}
        1,
         & k\geq 1,  \\[0.6em]
        \displaystyle
        \alpha_{d=-1}^{*}\left(p_{\mathrm M}\right)
        H\left(A_{d=-1}^{*}\left(p_{\mathrm M}\right)/\varepsilon\right),
         & k=-1,     \\[1.2em]
        0,
         & k\leq -3,
    \end{cases} \\
    \Pi_{\mathrm L}\left(p_{\mathrm M},k\right)
     & =
    \begin{cases}
        1,
         & k\geq 3,  \\[0.6em]
        \displaystyle
        \frac{3+\pi^2}{4},
         & k=1,      \\[1.2em]
        \displaystyle
        \frac{1-\pi^2}{2}
        H\left(A_{d=-1}^{*}\left(p_{\mathrm M}\right)/\varepsilon\right),
         & k=-1,     \\[1.2em]
        0,
         & k\leq -3.
    \end{cases}
\end{align*}
The imbalances \(k=\pm 1\) have probabilities \(q_{d=\pm 1}\). By \eqref{eq:app-horizon-n-imbalance-symmetry},
$q_{d=1}=q_{d=-1}$ and the positive and negative tails have equal
probability, so
\[
    \sum_{k\geq 3}
    q_{d=k}
    =
    \sum_{k\leq -3}
    q_{d=k}
    =
    \frac{1}{2}-q_{d=-1}.
\]
Aggregating over $k$ gives
\begin{align*}
    \Pi_{\mathrm M}\left(p_{\mathrm M}\right)
     & =
    \frac{1}{2}
    +
    q_{d=-1}
    \alpha_{d=-1}^{*}\left(p_{\mathrm M}\right)
    H\left(A_{d=-1}^{*}\left(p_{\mathrm M}\right)/\varepsilon\right),              \\
    \Pi_{\mathrm L}\left(p_{\mathrm M}\right)
     & =
    \frac{1}{2}
    +
    q_{d=-1}\left(1-\pi^2\right)
    \left[
        \frac{1}{2}H\left(A_{d=-1}^{*}\left(p_{\mathrm M}\right)/\varepsilon\right)
        -
        \frac{1}{4}
        \right], \\
    K\left(p_{\mathrm M}\right)
     & =
    q_{d=-1}
    \alpha_{d=-1}^{*}\left(p_{\mathrm M}\right)
    \left[
        A_{d=-1}^{*}\left(p_{\mathrm M}\right)+\sigma\pi
        \right].
\end{align*}
The break-even stage-1 ask obtained from the prediction-market maker's break-even
condition is
\begin{align*}
      & A_1^{\PM,\mathrm{BE}}\left(p_{\mathrm M}\right)                                \\
    = &
    p_{\mathrm M}\Pi_{\mathrm M}\left(p_{\mathrm M}\right)
    +
    \left(1-p_{\mathrm M}\right)
    \Pi_{\mathrm L}\left(p_{\mathrm M}\right)                                      \\
    = &
    \frac{1}{2}
    +
    q_{d=-1}
    \left[
        p_{\mathrm M}\alpha_{d=-1}^{*}\left(p_{\mathrm M}\right)
        H\left(A_{d=-1}^{*}\left(p_{\mathrm M}\right)/\varepsilon\right)
        +
        \left(1-p_{\mathrm M}\right)\left(1-\pi^2\right)
        \left[
            \frac{1}{2}H\left(A_{d=-1}^{*}\left(p_{\mathrm M}\right)/\varepsilon\right)
            -
            \frac{1}{4}
            \right]
        \right] \\
    = &
    \begin{cases}
        \displaystyle
        \frac{1}{2}
        -
        q_{d=-1}
        \frac{p_{\mathrm M}\left(\sigma\pi-\varepsilon\right)}
        {2\varepsilon},
         & 0<p_{\mathrm M}\leq p_0, \\[1.2em]
        \displaystyle
        \frac{1}{2}
        -
        q_{d=-1}
        \frac{\left(1-p_{\mathrm M}\right)\left(1-\pi^2\right)\left(1-2\sigma\pi\right)}
        {4},
         & p_0<p_{\mathrm M}<1.
    \end{cases}
\end{align*}
which is strictly less than \(1/2\) for any
\(p_{\mathrm M}\in\left(0,1\right)\). Therefore the non-crossing equilibrium stage-1
ask and bid are respectively
\[
    A_1^{\PM *}\left(p_{\mathrm M}\right)
    =
    \max\left\{
    A_1^{\PM,\mathrm{BE}}\left(p_{\mathrm M}\right),
    \frac{1}{2}
    \right\}
    =
    \frac{1}{2},
    \qquad
    B_1^{\PM *}\left(p_{\mathrm M}\right)
    =
    1 - A_1^{\PM *}\left(p_{\mathrm M}\right)
    =
    \frac{1}{2}
\]
for any \(p_{\mathrm M}\in\left(0,1\right)\).

The manipulator's expected profit conditional on entering at stage~1 is
\[
    S_{\mathrm M}\left(p_{\mathrm M}\right)
    =
    \Pi_{\mathrm M}\left(p_{\mathrm M}\right)
    -
    K\left(p_{\mathrm M}\right)
    -
    \frac{1}{2},
\]
which yields
\[
    S_{\mathrm M}\left(p_{\mathrm M}\right)
    =
    \begin{cases}
        \displaystyle
        q_{d=-1}
        \frac{
            \left(1-\pi^2\right)\left(1-2\sigma\pi\right)
            \left(1-p_{\mathrm M}/p_0\right)
        }{
            2\left[
                \left(1-\pi^2\right)
                +
                \left(1+\pi^2\right)p_{\mathrm M}
                \right]
        },
         & 0<p_{\mathrm M}\leq p_0, \\[1.5em]
        0,
         & p_0<p_{\mathrm M}<1.
    \end{cases}
\]
The pure-push branch is positive when \(0<p_{\mathrm M}<p_0\) and vanishes at
\(p_{\mathrm M}=p_0\). As in \cref{app:optimal-manipulator-entry-probability}, the manipulator chooses $p_{\mathrm M}$ anticipating the quotes it induces. Therefore, the manipulator maximizes the ex-ante profit
\begin{align*}
    \bar S_{\mathrm M}\left(p_{\mathrm M}\right)
    =
    p_{\mathrm M}S_{\mathrm M}\left(p_{\mathrm M}\right)
    =
    q_{d=-1}
    \frac{
        p_{\mathrm M}
        \left(1-\pi^2\right)\left(1-2\sigma\pi\right)
        \left(1-p_{\mathrm M}/p_0\right)
    }{
        2\left[
            \left(1-\pi^2\right)
            +
            \left(1+\pi^2\right)p_{\mathrm M}
            \right]
    }
\end{align*}
over \(0<p_{\mathrm M}<p_0\). This equals \(2q_{d=-1}\) times the manipulator's ex-ante profit
\eqref{eq:app-canonical-entry-objective} in the baseline model. Thus every odd $n$ has the
same optimal entry probability as the baseline model:
\[
    p_{\mathrm M}^{*\left(n\right)}
    =
    p_{\mathrm M}^{*\left(1\right)}
    \quad
    \text{if }n\text{ is odd}.
\]
Since \(p_{\mathrm M}^{*\left(n\right)} = p_{\mathrm M}^{*\left(1\right)} \in \left(0,p_0\right)\),
the optimum lies on the pure-push branch. Restoring horizon superscripts, since
\(q_{d=-1}=m^{\left(n\right)}\) for odd \(n\) and \(m^{\left(1\right)}=1/2\), where $m^{\left(n\right)}$ is defined in \cref{cor:horizon-n-manipulation-prone-mass},
the preceding formulas imply
\begin{align*}
    S_{\mathrm M}^{\left(n\right)}
    \left(p_{\mathrm M}^{*\left(n\right)}\right)
     & =
    \frac{m^{\left(n\right)}}{m^{\left(1\right)}}
    S_{\mathrm M}^{\left(1\right)}
    \left(p_{\mathrm M}^{*\left(1\right)}\right), \\
    \bar S_{\mathrm M}^{\left(n\right)}
    \left(p_{\mathrm M}^{*\left(n\right)}\right)
     & =
    \frac{m^{\left(n\right)}}{m^{\left(1\right)}}
    \bar S_{\mathrm M}^{\left(1\right)}
    \left(p_{\mathrm M}^{*\left(1\right)}\right)
\end{align*}
for odd \(n\). The same observation pins down
$$
    \alpha_{d=-1}^{*\left(n\right)}\left(p_{\mathrm M}^{*\left(n\right)}\right) = 1 \quad \text{if }n\text{ is odd}.
$$

\subsubsection{Long-horizon limit}
\label{app:horizon-n-vanishing-manipulation}

Since \(1-\pi^2 \in \left(0, 1\right)\) and
\(\binom{n}{\left\lfloor n/2\right\rfloor}\leq 2^n\),
\[
    0
    \leq
    m^{\left(n\right)}
    \leq
    \left(1-\pi^2\right)^{\left\lfloor n/2\right\rfloor}.
\]
The limit of $m^{\left(n\right)}$ then follows from the squeeze theorem:
$$
    \lim_{n \uparrow \infty} \left(1-\pi^2\right)^{\left\lfloor n/2\right\rfloor} = 0.
$$
\Cref{app:horizon-n-entry} shows that the optimal profit objects are
the corresponding shorter-horizon objects scaled by the pivotal probabilities:
\[
    S_{\mathrm M}^{\left(n\right)}
    \left(p_{\mathrm M}^{*\left(n\right)}\right)
    =
    \begin{cases}
        \displaystyle
        \frac{m^{\left(n\right)}}{m^{\left(1\right)}}
        S_{\mathrm M}^{\left(1\right)}
        \left(p_{\mathrm M}^{*\left(1\right)}\right),
         & n \text{ is odd},  \\[1.2em]
        \displaystyle
        \frac{m^{\left(n\right)}}{m^{\left(2\right)}}
        S_{\mathrm M}^{\left(2\right)}
        \left(p_{\mathrm M}^{*\left(2\right)}\right),
         & n \text{ is even},
    \end{cases}
\]
and similarly
\[
    \bar S_{\mathrm M}^{\left(n\right)}
    \left(p_{\mathrm M}^{*\left(n\right)}\right)
    =
    \begin{cases}
        \displaystyle
        \frac{m^{\left(n\right)}}{m^{\left(1\right)}}
        \bar S_{\mathrm M}^{\left(1\right)}
        \left(p_{\mathrm M}^{*\left(1\right)}\right),
         & n \text{ is odd},  \\[1.2em]
        \displaystyle
        \frac{m^{\left(n\right)}}{m^{\left(2\right)}}
        \bar S_{\mathrm M}^{\left(2\right)}
        \left(p_{\mathrm M}^{*\left(2\right)}\right),
         & n \text{ is even}.
    \end{cases}
\]
Clearly, as \(n \uparrow \infty\), both profit objects converge to zero along with $m^{\left(n\right)}$.

\subsubsection{Price discovery}
\label{app:horizon-n-price-discovery}

Fix \(n > 2\). In this subsection only, write \(q_{d=k}^{\left(n\right)}\) for
the imbalance-$k$ probability in \eqref{eq:app-horizon-n-imbalance-cell-mass}
computed under the $n$-ordinary-round extension to emphasize the dependence on
$n$. For any stage-2 order history \(h_2\) with \(d\left(h_2\right)=k\) and any
\(p_{\mathrm M}^{\left(n\right)}\in\left(0,1\right)\), also write
\(q_{d=k,h_3}\left(p_{\mathrm M}^{\left(n\right)}\right)\),
\(\bar V_{d=k,h_3}\left(p_{\mathrm M}^{\left(n\right)}\right)\), and
\(\rvar_{d=k,h_3}\left(p_{\mathrm M}^{\left(n\right)}\right)\) for the corresponding
conditional stage-3 outcome probability, posterior mean, and posterior variance.

Then
\[
    \E\left[\rvar_{h_2 h_3}^{\left(n\right)}\right]\left(p_{\mathrm M}^{\left(n\right)}\right)
    =
    \sum_{k\in\mathcal D_n}
    q_{d=k}^{\left(n\right)}
    \sum_{h_3\in\left\{\hbuy,\hzero,\hsell\right\}}
    q_{d=k,h_3}\left(p_{\mathrm M}^{\left(n\right)}\right)
    \rvar_{d=k,h_3}\left(p_{\mathrm M}^{\left(n\right)}\right)
\]
by the law of total expectation. To prove the price-discovery part of
\cref{prop:horizon-n-equilibrium}, it suffices to compare horizon \(n\)
to horizon \(n-1\) through the no-entry \(\left(n-1\right)\)-ordinary-round
benchmark. We show, for every
\(p_{\mathrm M}^{\left(n\right)}, p_{\mathrm M}^{\left(n-1\right)}\in\left(0,1\right)\),
\[
    \E\left[\rvar_{h_2 h_3}^{\left(n\right)}\right]\left(p_{\mathrm M}^{\left(n\right)}\right)
    <
    \E\left[\rvar_{h_2 h_3}^{\left(n-1\right)}\right]\left(0\right)
    <
    \E\left[\rvar_{h_2 h_3}^{\left(n-1\right)}\right]\left(p_{\mathrm M}^{\left(n-1\right)}\right).
\]
At the no-entry benchmark $p_{\mathrm M}^{\left(n-1\right)}=0$, the \(\left(n-1\right)\)-ordinary-round model has
\(n-1\) ordinary spot rounds in stage~2 and one ordinary spot round in stage~3, and thus
stages~2 and~3 together share the same event structure as stage~2 under the $n$-ordinary-round
extension. Hence
\[
    \E\left[\rvar_{h_2 h_3}^{\left(n-1\right)}\right]\left(0\right)
    =
    \sum_{k\in\mathcal D_n}
    q_{d=k}^{\left(n\right)}
    \rvar_{d=k}.
\]
For each \(k\in\mathcal D_n\), the law of total variance gives
\begin{align*}
    \sum_{h_3\in\left\{\hbuy,\hzero,\hsell\right\}}
    q_{d=k,h_3}\left(p_{\mathrm M}^{\left(n\right)}\right)
    \rvar_{d=k,h_3}\left(p_{\mathrm M}^{\left(n\right)}\right)
    =
    \rvar_{d=k}
    -
    \operatorname{Var}\left(
    \bar V_{d=k,h_3}
    \mid d\left(h_2\right)=k
    \right)\left(p_{\mathrm M}^{\left(n\right)}\right).
\end{align*}
Since \(\ppost_{d=k}\in\left(0,1\right)\),
\[
    q_{d=k,\hsell}\left(p_{\mathrm M}^{\left(n\right)}\right)
    =
    \left(1-p_{\mathrm M}^{\left(n\right)}\right)
    \left[
        \frac{1-\pi}{2}
        +
        \pi\left(1-\ppost_{d=k}\right)
        \right]
    >
    0,
\]
\[
    \bar V_{d=k,\hsell}\left(p_{\mathrm M}^{\left(n\right)}\right)
    -
    \bar V_{d=k}
    =
    B_{d=k} - \bar V_{d=k}
    =
    -
    \frac{
        \sigma\pi
        \left[
            1-\left(2\ppost_{d=k}-1\right)^2
            \right]
    }{
        1-\left(2\ppost_{d=k}-1\right)\pi
    }
    <
    0.
\]
Therefore
\[
    \operatorname{Var}\left(
    \bar V_{d=k,h_3}
    \mid d\left(h_2\right)=k
    \right)\left(p_{\mathrm M}^{\left(n\right)}\right)
    \geq
    q_{d=k,\hsell}\left(p_{\mathrm M}^{\left(n\right)}\right)
    \left[
        \bar V_{d=k,\hsell}\left(p_{\mathrm M}^{\left(n\right)}\right)
        -
        \bar V_{d=k}
        \right]^2
    >
    0.
\]
Taking the \(q_{d=k}^{\left(n\right)}\)-weighted average of the law-of-total-variance
identity over \(k\in\mathcal D_n\) gives
\begin{align*}
    \E\left[\rvar_{h_2 h_3}^{\left(n\right)}\right]
    \left(p_{\mathrm M}^{\left(n\right)}\right)
    =
    \sum_{k\in\mathcal D_n}
    q_{d=k}^{\left(n\right)}
    \left[
        \rvar_{d=k}
        -
        \operatorname{Var}\left(
        \bar V_{d=k,h_3}
        \mid d\left(h_2\right)=k
        \right)\left(p_{\mathrm M}^{\left(n\right)}\right)
        \right].
\end{align*}
Each \(q_{d=k}^{\left(n\right)}\) is strictly positive, and the previous display
shows that each subtracted posterior-mean dispersion term is strictly positive,
so
\begin{align*}
    \E\left[\rvar_{h_2 h_3}^{\left(n\right)}\right]\left(p_{\mathrm M}^{\left(n\right)}\right)
    <
    \sum_{k\in\mathcal D_n}
    q_{d=k}^{\left(n\right)}
    \rvar_{d=k}
    =
    \E\left[\rvar_{h_2 h_3}^{\left(n-1\right)}\right]\left(0\right).
\end{align*}
Finally, although \cref{prop:price-discovery-sign} is stated for the baseline model, its
proof in \cref{app:price-discovery} conditions only on the stage-2 posterior
and thus still applies to the \(\left(n-1\right)\)-ordinary-round extension, giving,
for every \(p_{\mathrm M}^{\left(n-1\right)}\in\left(0,1\right)\),
\[
    \E\left[\rvar_{h_2 h_3}^{\left(n-1\right)}\right]\left(p_{\mathrm M}^{\left(n-1\right)}\right)
    >
    \E\left[\rvar_{h_2 h_3}^{\left(n-1\right)}\right]\left(0\right).
\]
Chaining the inequalities gives the stated result.

\clearpage
\section{Polymarket Contract Trading Volume}
\label{app:volume}

This appendix documents the trading volume of the Polymarket crypto up/down contracts that the main analysis studies, across the six coins (BTC, ETH, SOL, XRP, DOGE, BNB) and the main contract durations (five-minute, fifteen-minute, four-hour). The contracts carry tens of millions of dollars of daily volume; the five-minute contract is the most heavily traded and BTC dominates the cross-section.

\paragraph{By coin.} Figure~\ref{fig:app_vol_stacked} stacks daily Polymarket crypto up/down USDC volume by coin, summed across all contract durations, from the earliest Polymarket trade. Volume grows from a few million dollars a day in October 2025 to roughly \$30--40M by March 2026, with a visible step up when the five-minute contract launches in February. BTC is the large majority of the total throughout, ETH a distant second, and SOL/XRP/DOGE/BNB thin slivers, consistent with the manipulation footprint being concentrated in the coin whose pool is deepest.

\paragraph{By contract duration.} Figure~\ref{fig:app_vol_duration} stacks daily BTC volume by contract duration. The fifteen-minute contract carries most of the volume through the pre-five-minute period and the four-hour contract a thin sliver; the five-minute layer appears at its February launch and quickly becomes the largest, raising total daily BTC volume to roughly \$30M.

\begin{figure}[!t]
    \centering
    \begin{subfigure}{\textwidth}
        \centering
        \includegraphics[width=0.92\textwidth]{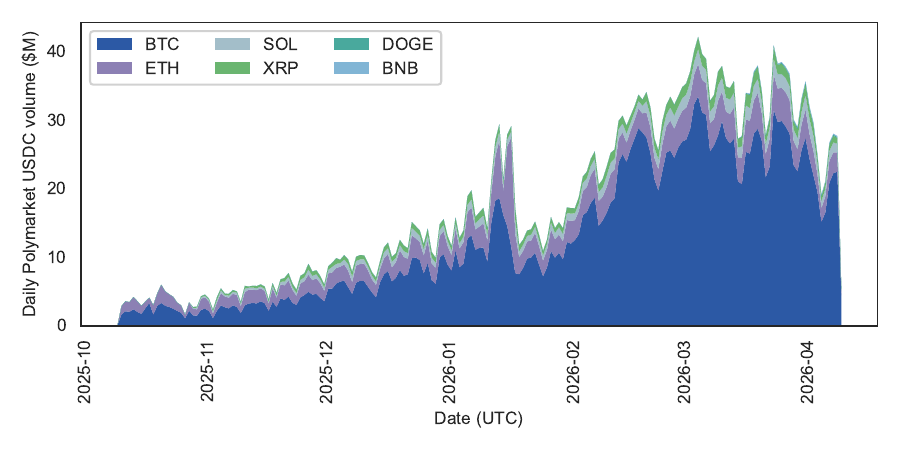}
        \caption{All coins, by coin.}
        \label{fig:app_vol_stacked}
    \end{subfigure}

    \begin{subfigure}{\textwidth}
        \centering
        \includegraphics[width=0.92\textwidth]{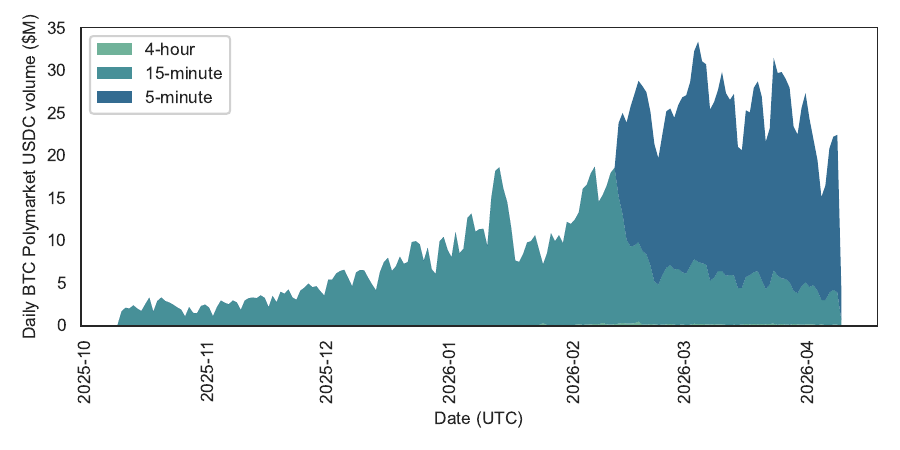}
        \caption{BTC, by contract duration.}
        \label{fig:app_vol_duration}
    \end{subfigure}
    \caption{\textbf{Daily Polymarket crypto up/down USDC volume.} This figure plots stacked daily taker USDC notional on the Polymarket crypto up/down contracts, from the earliest Polymarket trade (the October 2025 fifteen-minute / four-hour launch). \textbf{(a)}~All six coins, summed across the three durations (five-minute, fifteen-minute, four-hour), stacked by coin. \textbf{(b)}~BTC only, stacked by contract duration (four-hour, fifteen-minute, five-minute, from bottom to top).}
    \label{fig:app_vol}
\end{figure}

\clearpage
\section{Overlap and Moneyness Refinements}
\label{app:boundary_atm_zooms}

The two subsections below refine the cycle-level partition under which the main-body activity and reversal results are estimated. The first refines the within-P3 activity regression (Table~\ref{tab:reg_cond}) and reversal regression (Table~\ref{tab:reversal_cond}) into the full 15-min-overlap $\times$ per-contract-NTM groups. The second applies the moneyness cut to the fifteen-minute contract alone in P2, before the five-minute launch, asking whether the NTM concentration already appears at the longer horizon.

\subsection{Activity and reversal by overlap and moneyness}
\label{app:activity_overlap}
\label{app:reversal_overlap}

This subsection refines the within-P3 NTM results, the activity regression of Table~\ref{tab:reg_cond} and the reversal regression of Table~\ref{tab:reversal_cond}, by whether the cycle's close also coincides with a 15-minute mark, where the 15m Polymarket contract settles alongside the 5m contract. The cycle groups are written in set notation: $5\mn \cap 15\mn$ denotes overlapping cycles, where both the 5m and 15m Polymarket contracts settle at the close (four of twelve per hour), and $5\mn \setminus 15\mn$ non-overlapping cycles, where only the 5m contract settles. NTM is per contract from the Polymarket Up-token price near time-to-maturity TTM~$\in[10,12]$\,s: $\NTM_{5\mn}^{\poly}$ iff $|p_{Up,5\mn} - 0.5| < 0.10$, $\NTM_{15\mn}^{\poly}$ iff $|p_{Up,15\mn} - 0.5| < 0.10$ (defined only for overlapping cycles). Together these yield six mutually exclusive groups: $5\mn \setminus 15\mn$ split by 5m-contract NTM, and $5\mn \cap 15\mn$ split into neither-NTM, $\NTM_{5\mn}$-only, $\NTM_{15\mn}$-only, and both-NTM.

\paragraph{Activity.} Restricting to P3 and to bins in $\{0\text{--}24, 29\}$, the regression is
\[
    \begin{aligned}
        \log y_{c,b} \;=\; {} & \alpha_d + \gamma_b + \mu_h
        + \delta_{5\mn\,\cap\,15\mn}\,\mathbf{1}\{b{=}29\}{\cdot}D_{5\mn\,\cap\,15\mn,c}                                                                                                                                                   \\
                              & + \delta_{\NTM_{5\mn},\, 5\mn\setminus 15\mn}\,\mathbf{1}\{b{=}29\}{\cdot}D_{5\mn\setminus 15\mn,c}{\cdot}\mathbf{1}\{\NTM_{5\mn,c}^{\poly}\}                                                              \\
                              & + \delta_{\NTM_{5\mn},\, 5\mn\,\cap\,15\mn}\,\mathbf{1}\{b{=}29\}{\cdot}D_{5\mn\,\cap\,15\mn,c}{\cdot}\mathbf{1}\{\NTM_{5\mn,c}^{\poly}\}                                                                  \\
                              & + \delta_{\NTM_{15\mn},\, 5\mn\,\cap\,15\mn}\,\mathbf{1}\{b{=}29\}{\cdot}D_{5\mn\,\cap\,15\mn,c}{\cdot}\mathbf{1}\{\NTM_{15\mn,c}^{\poly}\}                                                                \\
                              & + \delta_{\NTM_{5\mn}\!\times\!\NTM_{15\mn},\, 5\mn\,\cap\,15\mn}\,\mathbf{1}\{b{=}29\}{\cdot}D_{5\mn\,\cap\,15\mn,c}{\cdot}\mathbf{1}\{\NTM_{5\mn,c}^{\poly}\}{\cdot}\mathbf{1}\{\NTM_{15\mn,c}^{\poly}\} \\
                              & + (\text{level effects}) + \varepsilon_{c,b},
    \end{aligned}
\]
estimated with date, hour-of-day, and bin fixed effects and cycle-clustered standard errors. The baseline, absorbed into the bin fixed effect $\gamma_{29}$, is the $5\mn \setminus 15\mn \times \FFM$ group; each $\delta$ is the bin-29-vs-body increase in its group over that baseline, $\delta_{5\mn \cap 15\mn}$ in overlapping cycles where neither contract is NTM, $\delta_{\NTM_{5\mn},\, 5\mn\setminus 15\mn}$ the 5m-contract NTM effect in non-overlapping cycles, $\delta_{\NTM_{5\mn},\, 5\mn \cap 15\mn}$ the 5m-contract NTM effect in overlapping cycles, and $\delta_{\NTM_{15\mn},\, 5\mn \cap 15\mn}$ the incremental 15m-contract NTM effect in overlapping cycles. The final term, $\delta_{\NTM_{5\mn}\!\times\!\NTM_{15\mn},\, 5\mn \cap 15\mn}$, is the both-NTM interaction: it saturates the four overlap-moneyness cells, so the both-NTM cell is a free parameter rather than the additive sum of the two single-contract effects.

In the order-flow column of Table~\ref{tab:reg_cond_full}, the overlap-without-NTM effect is small ($\hat\delta_{5\mn \cap 15\mn} = 0.21$), while the 5m-contract NTM effect is large in both the non-overlap group ($\hat\delta_{\NTM_{5\mn},\, 5\mn\setminus 15\mn} = 1.32$) and the overlap group ($\hat\delta_{\NTM_{5\mn},\, 5\mn \cap 15\mn} = 1.39$), and the 15m-contract NTM coefficient adds a further $\hat\delta_{\NTM_{15\mn},\, 5\mn \cap 15\mn} = 1.12$ inside the overlap group. The both-NTM interaction is positive on top of that: in cycles NTM by \emph{both} contracts, the bin-29 spike exceeds the sum of the two single-contract effects by a further $\hat\delta_{\NTM_{5\mn}\!\times\!\NTM_{15\mn}} = 1.14$ log points (significant at the 5\% level despite resting on only 11 such cycles), so the both-NTM overlap close sums all four terms, $0.21 + 1.39 + 1.12 + 1.14 = 3.87$ log points ($e^{3.87} \approx 48\times$ on the typical cycle) above the $5\mn \setminus 15\mn$ FFM baseline. This is what the model predicts: with two pools winnable at one close, a push is worth more than the sum of its parts. The same ordering holds for absolute return, and in the Polymarket columns the depth coefficients are again negative (prediction-market makers withdraw). The activity surge thus compounds with both moneyness and contract overlap: it is largest exactly where two contracts settle at one close and the outcome is still live by both.

\begin{table}[!t]
    \centering
    \footnotesize
    \caption{\textbf{Within-P3 bin-29 activity increase by 15-min overlap and per-contract Polymarket NTM, Binance and Polymarket.} This table reports cycle-level OLS estimates on the (cycle, bin) panel restricted to bins $\{0\text{--}24, 29\}$ and P3 cycles, with date, hour-of-day, and bin fixed effects; dependent variable $\log y_{c,b}$, the log of the column metric. Coefficients are log differences ($\delta$ is a multiplicative $e^{\delta}$ effect on the typical geometric-mean cycle). Set notation for the cycle group: $5\mn \cap 15\mn$ = 5m cycles whose close coincides with a 15-min boundary (both the 5m and 15m contracts settle); $5\mn \setminus 15\mn$ = the other eight of twelve per hour (only the 5m contract settles). NTM is per contract from the Polymarket Up-token price near TTM~$\in[10,12]$\,s: $\NTM_{5\mn}^{\poly}$ iff $|p_{Up,5\mn} - 0.5| < 0.10$, $\NTM_{15\mn}^{\poly}$ iff $|p_{Up,15\mn} - 0.5| < 0.10$ (overlapping cycles only). The five reported coefficients are $\mathbf{1}\{b{=}29\}$ interactions; the $5\mn \setminus 15\mn \times$ FFM bin-29 increase is the baseline, absorbed in the bin FE. The $\delta_{\NTM_{5\mn}\!\times\!\NTM_{15\mn}}$ interaction saturates the four overlap-moneyness cells, so each cell's bin-29 increase is a free parameter rather than an additive constraint; on the Binance order-flow sample the cells hold 586 ($5\mn\setminus15\mn \times \NTM_{5\mn}$), 211 ($\NTM_{5\mn}$-only), 132 ($\NTM_{15\mn}$-only), and 11 (both-NTM) cycles, so the interaction is estimated from few cycles. One observation per (cycle, bin) cell. HC1 standard errors in parentheses are clustered at the cycle level. Significance: $^{*}\,p<0.10$, $^{**}\,p<0.05$, $^{***}\,p<0.01$.}
    \label{tab:reg_cond_full}
    \begin{tabular}{l S[table-format=-1.3] S[table-format=-1.3] S[table-format=-1.3] S[table-format=-1.3] S[table-format=-1.3] S[table-format=-1.3]}
    \toprule
     & \multicolumn{3}{c}{Binance} & \multicolumn{3}{c}{Polymarket} \\
    \cmidrule(lr){2-4}\cmidrule(lr){5-7}
     & {Order flow} & {Abs. return} & {Depth} & {Order flow} & {Abs. price change} & {Depth} \\
    \midrule
    $\delta_{5\mn\,\cap\,15\mn}$ & 0.212{\sym{***}} & 0.109{\sym{***}} & 0.013{\sym{*}} & 0.163{\sym{***}} & 0.768{\sym{***}} & 0.026 \\
     & (0.043) & (0.034) & (0.008) & (0.032) & (0.153) & (0.047) \\[4pt]
    $\delta_{\NTM_{5\mn},\, 5\mn\setminus 15\mn}$ & 1.324{\sym{***}} & 0.555{\sym{***}} & 0.154{\sym{***}} & -0.090{\sym{*}} & 2.896{\sym{***}} & -0.980{\sym{***}} \\
     & (0.105) & (0.063) & (0.026) & (0.052) & (0.110) & (0.064) \\[4pt]
    $\delta_{\NTM_{5\mn},\, 5\mn\,\cap\,15\mn}$ & 1.393{\sym{***}} & 0.479{\sym{***}} & 0.168{\sym{***}} & -0.307{\sym{***}} & 1.879{\sym{***}} & -1.146{\sym{***}} \\
     & (0.185) & (0.111) & (0.041) & (0.099) & (0.176) & (0.113) \\[4pt]
    $\delta_{\NTM_{15\mn},\, 5\mn\,\cap\,15\mn}$ & 1.123{\sym{***}} & 0.435{\sym{***}} & 0.205{\sym{***}} & 0.167 & 0.251 & -0.234 \\
     & (0.228) & (0.145) & (0.062) & (0.134) & (0.517) & (0.211) \\[4pt]
    $\delta_{\NTM_{5\mn}\!\times\!\NTM_{15\mn},\, 5\mn\,\cap\,15\mn}$ & 1.143{\sym{**}} & 0.457{\sym{*}} & 0.259 & 0.290 & -1.331 & -0.060 \\
     & (0.583) & (0.277) & (0.262) & (0.341) & (1.097) & (0.324) \\
    \midrule
    Date FE & {Yes} & {Yes} & {Yes} & {Yes} & {Yes} & {Yes} \\
    Hour-of-day FE & {Yes} & {Yes} & {Yes} & {Yes} & {Yes} & {Yes} \\
    Bin FE & {Yes} & {Yes} & {Yes} & {Yes} & {Yes} & {Yes} \\
    Observations & {352{,}370} & {272{,}785} & {352{,}370} & {351{,}791} & {122{,}958} & {126{,}552} \\
    Clusters (cycle) & {13{,}553} & {13{,}551} & {13{,}553} & {13{,}553} & {5{,}001} & {5{,}001} \\
    \bottomrule
\end{tabular}

\end{table}

\paragraph{Reversal.} The reversal counterpart estimates, at each lag, a baseline slope plus four additive increments over the same six groups: the $5\mn \setminus 15\mn \times \FFM$ group is the baseline slope $\gamma_j$ on $r_{29}$, and we add $\delta_{5\mn \cap 15\mn}$ (overlapping close, neither contract NTM), $\delta_{\NTM_{5\mn},\, 5\mn\setminus 15\mn}$ (5m-contract NTM in non-overlapping cycles), $\delta_{\NTM_{5\mn},\, 5\mn \cap 15\mn}$ (5m-contract NTM in overlapping cycles), and $\delta_{\NTM_{15\mn},\, 5\mn \cap 15\mn}$ (incremental 15m-contract NTM in overlapping cycles), each as a slope increment on $r_{29}$. The nine overlapping cycles NTM by both contracts load on both overlap-NTM increments; unlike the saturated activity regression, nine cycles cannot support a free interaction slope.

Figure~\ref{fig:reversal_overlap} plots these slopes by group. At lag $j = +1$ in Table~\ref{tab:reversal_overlap}, the moneyness increments carry the result. Against the non-overlapping FFM baseline ($\gamma_{+1} = -0.08$), a non-overlapping cycle that is $\NTM_{5\mn}$ adds a further $-0.16$, about tripling the baseline reversal, and within overlapping cycles, NTM by either contract adds about $-0.22$ ($\delta_{\NTM_{5\mn},\, 5\mn \cap 15\mn} = -0.23$ and $\delta_{\NTM_{15\mn},\, 5\mn \cap 15\mn} = -0.22$), all significant; the deepest-reverting groups, overlap $\times$ NTM, revert about five times as hard as the baseline. That the $\NTM_{5\mn}$ and $\NTM_{15\mn}$ increments are almost identical says each contract's moneyness is on its own enough to mark a high-reversal cycle, consistent with the manipulator-selection reading: the payoff per dollar of push is larger when two contracts settle at the same close. The overlap level increment ($\delta_{5\mn \cap 15\mn} = -0.08$) we treat as a descriptive control rather than a settlement effect: 15-minute marks attract round-time algorithmic flow of their own (Section~\ref{sec:fifteen_min}). The manipulation reading therefore rests on the moneyness increments, which are defined by the live contracts' own prices and have no clock analogue.

\begin{figure}[t]
    \centering
    \includegraphics[width=0.83\textwidth]{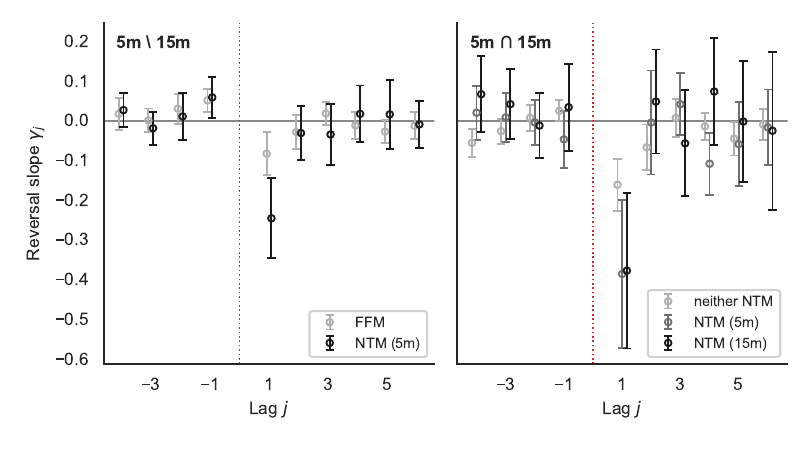}
    \caption{\textbf{Within-P3 post-settlement reversal slope $\gamma_j$, by cycle group.} This figure plots $\gamma_j$ for each cycle group (the underlying regression is Table~\ref{tab:reversal_overlap}): left panel non-overlapping ($5\mn \setminus 15\mn$) cycles, right panel overlapping ($5\mn \cap 15\mn$) cycles, each split by per-contract Polymarket NTM. Markers are $\gamma_j$; error bars are $\pm 1.96 \times$ cluster-robust SE (clustered by date); $j = 0$ is omitted. The both-NTM overlap group (nine cycles) is omitted.}
    \label{fig:reversal_overlap}
\end{figure}

\begin{table}[!t]
    \centering
    \footnotesize
    \caption{\textbf{Within-P3 post-settlement reversal slope by 15-min overlap and per-contract moneyness.} This table reports, at each lag $j \in \{-3,-2,-1,+1,+2,+3\}$, the coefficients from a single cycle-level OLS on P3 cycles (Feb 12 -- Apr 8, 2026): the $5\mn \setminus 15\mn \times \FFM$ group is the baseline reversal slope $\gamma_j$ on $r_{29}$, and the four reported $\delta$ are additive slope increments over it, $\delta_{5\mn \cap 15\mn}$ (overlapping close), $\delta_{\NTM_{5\mn},\, 5\mn\setminus 15\mn}$ and $\delta_{\NTM_{5\mn},\, 5\mn \cap 15\mn}$ (5m-contract NTM in non-overlapping and overlapping cycles), and $\delta_{\NTM_{15\mn},\, 5\mn \cap 15\mn}$ (incremental 15m-contract NTM in overlapping cycles); the stars test each coefficient against zero. Groups use set notation: $5\mn \setminus 15\mn$ = non-overlapping cycles (only the 5m contract settles), $5\mn \cap 15\mn$ = overlapping cycles whose close coincides with a 15-min mark (the 5m and 15m contracts settle simultaneously). Within each, cycles are split by per-contract Polymarket NTM, defined by the Up-token price near TTM~$\in[10,12]$\,s: $\NTM_{5\mn}^{\poly}$ iff $|p_{Up,5\mn} - 0.5| < 0.10$, $\NTM_{15\mn}^{\poly}$ iff $|p_{Up,15\mn} - 0.5| < 0.10$. The nine overlapping cycles NTM by \emph{both} contracts load on both the $\NTM_{5\mn}$ and $\NTM_{15\mn}$ overlap increments --- their slope is constrained to the additive sum (unlike the saturated activity Table~\ref{tab:reg_cond_full}; nine cycles cannot support a free interaction slope). Date, hour-of-day, and cycle-group fixed effects; standard errors in parentheses are clustered by date. Sample is P3 cycles with the relevant Polymarket classifications available: $5\mn \setminus 15\mn$ cycles need the 5m classifier, $5\mn \cap 15\mn$ cycles need both. Negative $j$ = same-cycle bin $(29 + j)$; positive $j$ = bin $(j-1)$ of the next cycle. Significance: $^{*}\,p<0.10$, $^{**}\,p<0.05$, $^{***}\,p<0.01$.}
    \label{tab:reversal_overlap}
    \begin{tabular}{l S[table-format=-1.3] S[table-format=-1.3] S[table-format=-1.3] S[table-format=-1.3] S[table-format=-1.3] S[table-format=-1.3]}
    \toprule
     & \multicolumn{3}{c}{Pre-close} & \multicolumn{3}{c}{Post-close} \\
    \cmidrule(lr){2-4} \cmidrule(lr){5-7}
     & {$j = -3$} & {$j = -2$} & {$j = -1$} & {$j = +1$} & {$j = +2$} & {$j = +3$} \\
    \midrule
    $5\mn\setminus15\mn$, FFM (baseline) & 0.001 & 0.031 & 0.052{\sym{***}} & -0.083{\sym{***}} & -0.028 & 0.019 \\
     & (0.015) & (0.019) & (0.015) & (0.028) & (0.022) & (0.014) \\[2pt]
    $\delta_{5\mn\,\cap\,15\mn}$ & -0.027 & -0.023 & -0.026 & -0.078{\sym{**}} & -0.039 & -0.011 \\
     & (0.022) & (0.025) & (0.021) & (0.033) & (0.037) & (0.028) \\[2pt]
    $\delta_{\NTM_{5\mn},\, 5\mn\setminus 15\mn}$ & -0.020 & -0.019 & 0.007 & -0.163{\sym{***}} & -0.003 & -0.053 \\
     & (0.022) & (0.032) & (0.028) & (0.044) & (0.039) & (0.039) \\[2pt]
    $\delta_{\NTM_{5\mn},\, 5\mn\,\cap\,15\mn}$ & 0.035 & -0.012 & -0.072{\sym{**}} & -0.225{\sym{**}} & 0.062 & 0.035 \\
     & (0.032) & (0.032) & (0.035) & (0.097) & (0.069) & (0.040) \\[2pt]
    $\delta_{\NTM_{15\mn},\, 5\mn\,\cap\,15\mn}$ & 0.068 & -0.020 & 0.009 & -0.216{\sym{**}} & 0.115 & -0.064 \\
     & (0.045) & (0.043) & (0.054) & (0.106) & (0.071) & (0.072) \\
    \midrule
    Date FE & {Yes} & {Yes} & {Yes} & {Yes} & {Yes} & {Yes} \\
    Hour-of-day FE & {Yes} & {Yes} & {Yes} & {Yes} & {Yes} & {Yes} \\
    Cycle-group FE & {Yes} & {Yes} & {Yes} & {Yes} & {Yes} & {Yes} \\
    Observations & \multicolumn{6}{c}{13{,}507} \\
    Clusters (UTC date) & \multicolumn{6}{c}{56} \\
    \bottomrule
\end{tabular}

\end{table}

\subsection{Moneyness for the fifteen-minute contract alone}
\label{app:p2_placebo}

Section~\ref{sec:fifteen_min} finds only a strongly attenuated fifteen-minute footprint, but that test is unconditional: it averages over all overlapping closes, most of which are far from the money by the time the contract settles. Within P3, the footprint concentrates in NTM cycles (Section~\ref{sec:atm}), so a genuine fifteen-minute effect could hide in the unconditional average, diluted by the FFM majority. This appendix therefore applies the moneyness cut to the fifteen-minute contract on its own. We take P2, where the fifteen-minute (and four-hour) contract is live but the five-minute contract is not, restrict to overlapping cycles, those whose close lands on a 15-min mark, where the fifteen-minute contract settles, and classify them NTM vs FFM by the \emph{fifteen-minute} contract's Polymarket Up-token price within $0.10$ of $0.5$ near TTM~$\in[10,12]$\,s, the same moneyness rule applied to the contract that is actually live. If the fifteen-minute horizon supported manipulation, its footprint should surface here, in exactly the cycles where a push could flip its resolution.

Table~\ref{tab:p2_placebo} reports the result, and it is largely a null. Panel~A shows a modest bin-29 activity increase in NTM cycles --- about $0.24$ in order flow and $0.30$ in absolute return, roughly $27\%$ and $35\%$ more close-time activity --- both significant at the 5\% level, but no effect on depth. These magnitudes are a small fraction of the five-minute effect in P3: differencing the log coefficients, the P3 NTM order-flow effect of Table~\ref{tab:reg_cond} is about $e^{1.37 - 0.24} \approx 3.1\times$ larger, and P3 carries the depth and reversal signatures that P2 lacks. Panel~B shows the reversal channel is absent altogether: the FFM baseline reverts at the first post-close lag (the general overlap reversal, present once the fifteen-minute contract is live), but the NTM increment $\delta_{\NTM,+1}$ is small, wrong-signed, and far from significant. The thin NTM sample (about 230 cycles) is the binding constraint on Panel~B. Taken together, the fifteen-minute contract alone generates at most a modest moneyness-graded activity increase, with neither the liquidity withdrawal-versus-provision pattern nor the deeper NTM reversal that define the P3 footprint. Even in the cycles where a push could flip the fifteen-minute resolution, the footprint barely appears: the fifteen-minute window is too long for the push to be reliably pivotal, and the near-settlement effect switches on with the five-minute contract.

\begin{table}[!t]
    \centering
    \footnotesize
    \caption{\textbf{Fifteen-minute-contract moneyness in P2 (BTC).} This table reports near-settlement activity (Panel~A) and post-settlement reversal (Panel~B) in P2 overlapping cycles, split by fifteen-minute-contract Polymarket NTM. Sample: P2 cycles (Oct 9, 2025 -- Feb 11, 2026; fifteen-minute and four-hour contracts live, five-minute not yet) whose close coincides with a 15-min mark, classified NTM iff the fifteen-minute-contract Up-token price is within $0.10$ of $0.5$ near TTM~$\in[10,12]$\,s, FFM otherwise. \emph{Panel~A:} the bin-29 NTM activity increase $\delta_{\NTM}$ (the coefficient on $\mathbf{1}\{b{=}29\}\cdot\mathbf{1}\{\NTM\}$) for each metric, log-OLS on the (cycle, bin) panel restricted to bins $\{0\text{--}24,29\}$ with date, hour-of-day, and bin fixed effects; standard errors clustered at the cycle level. Metrics as in Table~\ref{tab:reg_cond}; coefficients are log differences, so $\delta$ is a multiplicative $e^{\delta}$ effect. \emph{Panel~B:} the post-settlement reversal slope by lag $j$, the FFM baseline slope $\gamma_j$ on $r_{29}$ and the NTM increment $\delta_{\NTM,j}$, with date, hour-of-day, and group fixed effects; standard errors clustered by date; a negative slope is reversal. Negative $j$ = same-cycle bin $(29+j)$; positive $j$ = bin $(j-1)$ of the next cycle. Significance: $^{*}\,p<0.10$, $^{**}\,p<0.05$, $^{***}\,p<0.01$.}
    \label{tab:p2_placebo}
    \textit{Panel A: Near-settlement activity} \\[2pt]
    \begin{tabular}{l S[table-format=-1.3] S[table-format=-1.3] S[table-format=-1.3]}
    \toprule
     & {Order flow} & {Abs.\ return} & {Depth} \\
    \midrule
    $\delta_{\NTM}$ & 0.243{\sym{**}} & 0.303{\sym{**}} & -0.015 \\
     & (0.116) & (0.144) & (0.017) \\[2pt]
    \midrule
    Date FE & {Yes} & {Yes} & {Yes} \\
    Hour-of-day FE & {Yes} & {Yes} & {Yes} \\
    Bin FE & {Yes} & {Yes} & {Yes} \\
    Observations & {158{,}570} & {102{,}329} & {158{,}570} \\
    NTM cycles & \multicolumn{3}{c}{229} \\
    FFM cycles & \multicolumn{3}{c}{5{,}870} \\
    \bottomrule
\end{tabular}

    \\[6pt]
    \textit{Panel B: Post-settlement reversal} \\[2pt]
    \begin{tabular}{l S[table-format=-1.3] S[table-format=-1.3] S[table-format=-1.3] S[table-format=-1.3] S[table-format=-1.3] S[table-format=-1.3]}
    \toprule
     & \multicolumn{3}{c}{Pre-close} & \multicolumn{3}{c}{Post-close} \\
    \cmidrule(lr){2-4} \cmidrule(lr){5-7}
     & {$j = -3$} & {$j = -2$} & {$j = -1$} & {$j = +1$} & {$j = +2$} & {$j = +3$} \\
    \midrule
    FFM $\gamma_j$ & 0.121 & -0.026 & 0.015 & -0.137{\sym{***}} & -0.011 & -0.029 \\
     & (0.108) & (0.036) & (0.028) & (0.045) & (0.051) & (0.037) \\[2pt]
    $\delta_{\NTM,\,j}$ & -0.082 & -0.129{\sym{*}} & -0.004 & 0.076 & -0.184 & 0.011 \\
     & (0.115) & (0.077) & (0.061) & (0.234) & (0.115) & (0.071) \\[2pt]
    \midrule
    Date FE & {Yes} & {Yes} & {Yes} & {Yes} & {Yes} & {Yes} \\
    Hour-of-day FE & {Yes} & {Yes} & {Yes} & {Yes} & {Yes} & {Yes} \\
    Cycle-group FE & {Yes} & {Yes} & {Yes} & {Yes} & {Yes} & {Yes} \\
    Observations & \multicolumn{6}{c}{6{,}069} \\
    Clusters (UTC date) & \multicolumn{6}{c}{125} \\
    \bottomrule
\end{tabular}

\end{table}

\clearpage
\section{Binance Spot Tracks the Chainlink Resolution Price}
\label{app:chainlink_validation}

The 5m / 15m / 4h crypto up/down contracts resolve on a Chainlink VWAP price, not on Binance. Every Binance-based result in the main body therefore rests on the Binance midquote being a faithful proxy for the price that actually settles each market. We check this directly on the subset of BTC markets for which the Chainlink strike (fixed at window open) and the Chainlink settlement print (at window close) are both known; coverage is concentrated in the Mar~16 -- Apr~10, 2026 window. For each such market we read the Binance midquote at the window open and close from the 10-second midquote series via a strictly causal match: each bucket stores the \emph{last} book mid in $[t, t+10\,\text{s})$, so its effective timestamp is the bucket end, and we match each window boundary to the last book mid of the 10-second bucket ending at or before that instant.

Define, per market, the Chainlink--Binance basis at time $t$ (with $\text{mid}_t$ the Binance midquote):
\[
    \textit{Basis}_t = (\text{Chainlink}_t - \text{mid}_t)/\text{mid}_t.
\]
The market's realized Up/Down outcome is $\operatorname{sign}(\text{settlement} - \text{strike})$; the directional call a manipulator targets is whether the Binance mid finishes above or below the strike, $\operatorname{sign}(\text{mid}_{\text{close}} - \text{strike})$, the side a final-10s push tries to land on.

\begin{table}[!t]
    \centering
    \footnotesize
    \caption{\textbf{Binance spot midquote vs.\ the Chainlink resolution price, BTC.} This table compares the Binance midquote with the Chainlink price that settles each market, one observation per Polymarket BTC market with both the Chainlink strike (fixed at window open) and the Chainlink settlement print (at window close) known. \emph{Median $\textit{Basis}_t$} is the median per-market $(\text{Chainlink} - \text{Binance mid})/\text{Binance mid}$ at window open (bps). \emph{SD $\Delta\textit{Basis}_t$} is the standard deviation of the open-to-close change $\Delta\textit{Basis}_t \equiv \textit{Basis}_{\text{close}} - \textit{Basis}_{\text{open}}$ (bps). \emph{Strike-side agreement} is the fraction of markets on which the Binance mid finishes on the same side of the strike as the realized resolution, $\operatorname{sign}(\text{mid}_{\text{close}} - \text{strike}) = \operatorname{sign}(\text{settlement} - \text{strike})$.}
    \label{tab:binance_chainlink_basis}
    \begin{tabular}{lrrrr}
    \toprule
    Interval & $n$ & Median $\textit{Basis}_t$ & SD $\Delta\textit{Basis}_t$ & \makecell{Strike-side\\agreement} \\
    \midrule
    5m & 5{,}553 & $-2.45$ & $1.44$ & $84.8\%$ \\
    15m & 1{,}855 & $-2.49$ & $2.24$ & $91.3\%$ \\
    \bottomrule
\end{tabular}

\end{table}

Table~\ref{tab:binance_chainlink_basis} reports the comparison. The basis is a stable level offset of $\approx -2.5$\,bps and is essentially flat across horizons, so it cancels in the open-to-close return comparison: the within-window change in the basis has standard deviation of only $\approx 1.4$--$2.2$\,bps, small relative to a typical window move ($5.8$\,bps at 5m, $10.8$\,bps at 15m), so within each window the Binance mid and the Chainlink price move essentially together.

Where Binance and Chainlink \emph{can} part company is right at the strike. Because Binance sits $\approx 2.5$\,bps above Chainlink, the Binance mid can finish just above the strike while Chainlink, and hence the resolution, is still below it. The strike-side agreement in the last column is therefore below one: $84.8\%$ at 5m and $91.3\%$ at 15m, rising with horizon as the typical move grows relative to the fixed $\approx 2.5$\,bps basis. The residual disagreement is confined to genuine coin-toss windows: sorting the BTC 5m markets into quintiles by the size of the Binance window move, agreement is $61\%$ in the smallest quintile ($|\text{move}| < 1.8$\,bps), $87\%$ in the middle, and $100\%$ once the move exceeds $\approx 13$\,bps. That $\approx 2.5$\,bps band is precisely the margin a spot push must clear: to flip the resolution a manipulator must push Binance \emph{past} the strike by more than the basis, not merely to it; a push that only reaches the strike still fails on the $\sim 15\%$ of 5m windows where the basis leaves Chainlink on the other side. A push well past the strike, however, reliably carries the resolution, which makes spot-market pressure a coherent channel for moving the prediction-market outcome.

\clearpage
\section{Characteristics of Manipulated Cycles}
\label{app:manipulated_cycle_characteristics}

In Table~\ref{tab:push_cycle_summary}, we contrast manipulated cycles with the normal ones along four dimensions: Binance activity, Polymarket activity, timing, and moneyness.

Because manipulated cycles are selected on $\textit{PushIntensity}$, a large Binance spike among them is in part mechanical; the informative contrast is across venues. On Binance the bin-29 order flow averages \$1.7M against a typical body bin of just \$18k; on Polymarket the bin-29 order flow (\$2.8k) is much closer to its own body bin (\$0.4k) and barely exceeds its normal-cycle level. The near-settlement push moves Binance spot without a comparable increase in Polymarket's own order flow, which indicates that the push is executed on the spot venue, consistent with the manipulation footprint.

Manipulated cycles are also \emph{quieter} away from the close, not busier. The typical body-bin order flow is lower than in normal cycles on both venues, and the cycles fall disproportionately in thin-liquidity windows: 56\% fall in the Asia/overnight hours (vs 40\% of normal cycles) and 44\% on weekends (vs 27\%), against only 23\% in the deep European--US overlap (vs 35\%). Despite this lower activity, the absolute Binance near-settlement order flow is an order of magnitude larger than in a normal cycle (\$1.7M vs \$68k). Manipulation thus concentrates in the overnight and weekend hours, when the order book is thinnest and a given dollar of near-settlement order flow moves the spot price, and with it the Chainlink resolution price, the most.

\begin{table}[!t]
    \centering
    \footnotesize
    \caption{\textbf{Summary statistics for manipulated cycles vs normal cycles, P3 BTC 5\mn.} This table reports summary statistics for manipulated cycles, the top decile by $\textit{PushIntensity}$ ($\textit{PushIntensity} \geq 16.11$), against normal cycles, the remaining 90\%. The \emph{Binance activity} and \emph{Polymarket activity} panels report the same two flow rows in thousands of dollars (USDC on Polymarket): the absolute net order flow in bin~29 (the final 10\,s), and the per-cycle median of the absolute net order flow across the body bins (bins 0--24, the pre-ramp window, the same body denominator that enters $\textit{PushIntensity}$). Bin~29 is a single 10\,s bin and the body figure is the cycle's typical body bin, so the two rows are on the same per-bin scale within each venue and across the two venues. \emph{Binance activity} additionally reports the cross-cycle mean and median of $\textit{PushIntensity}$. \emph{Timing}: hour-of-day bucket (UTC) and weekday share. \emph{Moneyness}: the share of cycles that are NTM (5m-contract Up-token price within $0.10$ of $0.5$ near TTM~$\in[10,12]$\,s), computed over the cycles classifiable by that rule (99\% of manipulated and 94\% of normal cycles). All flow values are cross-cycle means.}
    \label{tab:push_cycle_summary}
    \begin{tabular}{lrr}
    \toprule
     & Manipulated cycles & Normal cycles \\
    \midrule
    Number of cycles & 1{,}613 & 14{,}460 \\
    \addlinespace
    \multicolumn{3}{l}{\textit{Binance activity}} \\
    $\textit{PushIntensity}$, mean & 157.4 & 2.0 \\
    $\textit{PushIntensity}$, median & 57.8 & 0.9 \\
    Bin-29 $|$signed notional$|$ (\$k) & \$1{,}714 & \$68 \\
    Body $|$signed notional$|$, median/bin (\$k) & \$18 & \$43 \\
    \addlinespace
    \multicolumn{3}{l}{\textit{Polymarket activity}} \\
    Bin-29 $|$signed notional$|$ (\$k) & \$2.83 & \$2.69 \\
    Body $|$signed notional$|$, median/bin (\$k) & \$0.39 & \$0.51 \\
    \addlinespace
    \multicolumn{3}{l}{\textit{Timing}} \\
    Asia / overnight (21:00--07:00 UTC) & 55.7\% & 39.9\% \\
    Europe, pre-US (07:00--13:00 UTC) & 21.3\% & 25.5\% \\
    US \& EU--US overlap (13:00--21:00 UTC) & 23.1\% & 34.6\% \\
    Weekend (Sat--Sun) & 43.6\% & 27.0\% \\
    \addlinespace
    \multicolumn{3}{l}{\textit{Moneyness}} \\
    NTM share & 18.7\% & 4.4\% \\
    \bottomrule
\end{tabular}

\end{table}

If the near-settlement push is manipulation, it should concentrate where the outcome is still undecided, that is, in the cycles in which moving the spot price across the strike can flip the resolution. The \emph{Moneyness} panel of Table~\ref{tab:push_cycle_summary} shows it does: among manipulated cycles 18.7\% are NTM, against only 4.4\% of normal cycles. Equivalently, a cycle that is NTM near the close is manipulated 33.0\% of the time versus 9.1\% for FFM cycles, a 3.6-fold higher rate. The two classifications are independent by construction: $\textit{PushIntensity}$ is built entirely from Binance signed flow and the moneyness cut entirely from the Polymarket Up-token price. The concentration is therefore not mechanical. It is consistent with the NTM activity and reversal results: manipulated cycles are more common among the cycles in which a push can flip the resolution.

\clearpage
\section{Manipulator Cohort Threshold Sensitivity}
\label{app:cohort_sensitivity}

The manipulator cohort definition in Section~\ref{subsec:wallet_cohorts_id} uses two ad hoc thresholds: a minimum manipulated-cycle participation count of 5 and a minimum aggregate manipulated-cycle PnL of \$2{,}000. Table~\ref{tab:cohort_sensitivity} re-derives the \emph{full} trader-type attribution at each cell of a $4\times5$ grid of $(n_{\text{cycle,min}}, \$_{\min})$ thresholds: the manipulator cohort is recomputed, the MM cohort is re-derived as the behavioral residual (same passive-rate, fills-per-active-cycle, and inventory criteria as the main text), and retail/other is the final residual, so the entire partition responds to the threshold. For each cell the three panels report the three cohort metrics of Table~\ref{tab:cohort_pnl}: net PnL (Panel A), the fraction of cycles in which each cohort is net negative (Panel B), and the pool-weighted loss share (Panel C), each split across manipulated and normal cycles. Within each regime the three cohort PnLs sum to zero and the three loss shares to 100\% by construction.

All three metrics are stable across the twenty cells. The manipulated-cycle wealth transfer runs directly from retail to manipulators: manipulators net between $+$\$4.1\,M and $+$\$10.0\,M and retail is the near-exact mirror ($-$\$3.6\,M to $-$\$9.3\,M); retail is on the losing side of 64--66\% of manipulated cycles and bears 62--74\% of the manipulated-cycle loss pool at every threshold. The MM cohort is a small bystander in manipulated cycles: its net PnL holds near $-$\$0.6\,M throughout and its loss share stays modest ($8$--$16\%$). Its behavioral profile --- passive, delta-neutral quoting --- is orthogonal to the manipulator's directional-PnL threshold, so moving the threshold barely reclassifies it. In normal cycles the manipulator cohort's PnL is near zero at moderate cells and only mildly negative at the tightest activity cuts ($n_{\text{cycle,min}}\geq 10$), confirming that manipulators do not systematically profit in normal cycles. The core finding --- a large, retail-funded manipulated-cycle profit for the manipulator cohort --- is robust to both threshold choices.

\begin{table}[p]
    \centering
    \footnotesize
    \setlength{\tabcolsep}{4pt}
    \caption{\textbf{Trader-type attribution sensitivity to manipulator thresholds.} This table re-derives the full trader-type partition at each cell of a $4\times5$ threshold grid: the manipulator cohort is the set of wallets with $n_{\text{cycle,p}}\geq n_{\mathrm{cyc,min}}$ and net manipulated-cycle PnL $\geq \$_{\min}$; the MM cohort is the behavioral residual passing the same passive-rate, fills-per-active-cycle, and inventory criteria as the main text; retail/other is the final residual. Each panel reports one trader-type metric of Table~\ref{tab:cohort_pnl} for all three cohorts, split across manipulated and normal cycles: \emph{Panel A}, aggregate net PnL (\$M), plus the manipulator cohort size (Wallets); \emph{Panel B}, the fraction of cycles in which the cohort is net negative; \emph{Panel C}, the cohort's share of the gross loss pool, aggregated across cycles with pool-size weights. Within each regime the three cohort PnLs sum to zero and the three loss shares to 100\%.}
    \label{tab:cohort_sensitivity}

    \textit{Panel A. Net PnL (\$M)}\\[2pt]
    \begin{tabular}{rrrrrrrrr}
    \toprule
     &  &  & \multicolumn{3}{c}{Manipulated cycles (\$M)} & \multicolumn{3}{c}{Normal cycles (\$M)} \\
    \cmidrule(lr){4-6}\cmidrule(lr){7-9}
    $n_{\mathrm{cyc,min}}$ & $\$_{\min}$ & Wallets & Manip. & MM & Retail & Manip. & MM & Retail \\
    \midrule
    3 & \$500 & 2{,}276 & $+$9.98 & $-$0.65 & $-$9.33 & $-$0.92 & $+$3.01 & $-$2.08 \\
    3 & \$1{,}000 & 1{,}456 & $+$9.41 & $-$0.64 & $-$8.76 & $-$0.19 & $+$3.02 & $-$2.83 \\
    3 & \$2{,}000 & 857 & $+$8.56 & $-$0.62 & $-$7.94 & $+$0.28 & $+$3.11 & $-$3.39 \\
    3 & \$5{,}000 & 414 & $+$7.27 & $-$0.59 & $-$6.68 & $+$0.68 & $+$3.21 & $-$3.90 \\
    3 & \$10{,}000 & 221 & $+$5.89 & $-$0.55 & $-$5.34 & $+$0.31 & $+$3.40 & $-$3.70 \\
    5 & \$500 & 2{,}158 & $+$9.56 & $-$0.65 & $-$8.91 & $-$0.94 & $+$3.01 & $-$2.07 \\
    5 & \$1{,}000 & 1{,}389 & $+$9.03 & $-$0.64 & $-$8.39 & $-$0.30 & $+$3.02 & $-$2.72 \\
    5 & \$2{,}000 & 821 & $+$8.22 & $-$0.62 & $-$7.61 & $+$0.09 & $+$3.11 & $-$3.20 \\
    5 & \$5{,}000 & 397 & $+$6.99 & $-$0.59 & $-$6.40 & $+$0.52 & $+$3.21 & $-$3.73 \\
    5 & \$10{,}000 & 210 & $+$5.65 & $-$0.55 & $-$5.10 & $+$0.16 & $+$3.40 & $-$3.56 \\
    10 & \$500 & 1{,}949 & $+$8.63 & $-$0.65 & $-$7.98 & $-$1.43 & $+$3.01 & $-$1.58 \\
    10 & \$1{,}000 & 1{,}275 & $+$8.16 & $-$0.64 & $-$7.52 & $-$0.89 & $+$3.02 & $-$2.13 \\
    10 & \$2{,}000 & 751 & $+$7.41 & $-$0.62 & $-$6.79 & $-$0.53 & $+$3.11 & $-$2.58 \\
    10 & \$5{,}000 & 356 & $+$6.27 & $-$0.59 & $-$5.68 & $-$0.04 & $+$3.21 & $-$3.18 \\
    10 & \$10{,}000 & 181 & $+$5.01 & $-$0.55 & $-$4.46 & $-$0.31 & $+$3.40 & $-$3.09 \\
    20 & \$500 & 1{,}673 & $+$7.37 & $-$0.65 & $-$6.72 & $-$1.43 & $+$3.01 & $-$1.58 \\
    20 & \$1{,}000 & 1{,}130 & $+$6.99 & $-$0.64 & $-$6.35 & $-$1.05 & $+$3.02 & $-$1.97 \\
    20 & \$2{,}000 & 672 & $+$6.33 & $-$0.62 & $-$5.71 & $-$0.75 & $+$3.11 & $-$2.36 \\
    20 & \$5{,}000 & 315 & $+$5.30 & $-$0.59 & $-$4.71 & $-$0.33 & $+$3.21 & $-$2.88 \\
    20 & \$10{,}000 & 151 & $+$4.12 & $-$0.55 & $-$3.57 & $-$0.54 & $+$3.40 & $-$2.86 \\
    \bottomrule
\end{tabular}

\end{table}

\begin{table}[p]\ContinuedFloat
    \centering
    \footnotesize
    \setlength{\tabcolsep}{4pt}
    \caption{\textbf{Trader-type attribution sensitivity to manipulator thresholds} \emph{(continued)}.}

    \textit{Panel B. \% losing cycles}\\[2pt]
    \begin{tabular}{rrrrrrrr}
    \toprule
     &  & \multicolumn{3}{c}{Manipulated cycles} & \multicolumn{3}{c}{Normal cycles} \\
    \cmidrule(lr){3-5}\cmidrule(lr){6-8}
    $n_{\mathrm{cyc,min}}$ & $\$_{\min}$ & Manip. & MM & Retail & Manip. & MM & Retail \\
    \midrule
    3 & \$500 & 35.8\% & 58.8\% & 66.3\% & 62.2\% & 37.6\% & 43.8\% \\
    3 & \$1{,}000 & 36.8\% & 58.8\% & 66.0\% & 61.8\% & 37.7\% & 45.4\% \\
    3 & \$2{,}000 & 38.1\% & 58.6\% & 65.4\% & 62.3\% & 37.7\% & 47.7\% \\
    3 & \$5{,}000 & 38.2\% & 58.7\% & 65.7\% & 60.5\% & 37.5\% & 53.4\% \\
    3 & \$10{,}000 & 40.2\% & 57.6\% & 65.6\% & 66.6\% & 37.2\% & 55.7\% \\
    5 & \$500 & 36.0\% & 58.8\% & 65.8\% & 62.3\% & 37.6\% & 43.7\% \\
    5 & \$1{,}000 & 36.9\% & 58.8\% & 65.7\% & 61.9\% & 37.7\% & 45.4\% \\
    5 & \$2{,}000 & 38.0\% & 58.6\% & 65.3\% & 62.4\% & 37.7\% & 47.7\% \\
    5 & \$5{,}000 & 38.2\% & 58.7\% & 65.5\% & 60.6\% & 37.5\% & 53.4\% \\
    5 & \$10{,}000 & 40.3\% & 57.6\% & 65.2\% & 66.6\% & 37.2\% & 55.7\% \\
    10 & \$500 & 36.7\% & 58.8\% & 65.0\% & 62.3\% & 37.6\% & 43.8\% \\
    10 & \$1{,}000 & 37.2\% & 58.8\% & 65.1\% & 62.0\% & 37.7\% & 45.4\% \\
    10 & \$2{,}000 & 38.4\% & 58.6\% & 64.8\% & 62.6\% & 37.7\% & 47.9\% \\
    10 & \$5{,}000 & 38.6\% & 58.7\% & 65.2\% & 60.7\% & 37.5\% & 53.3\% \\
    10 & \$10{,}000 & 40.7\% & 57.6\% & 64.8\% & 66.9\% & 37.2\% & 55.5\% \\
    20 & \$500 & 37.2\% & 58.8\% & 64.4\% & 62.3\% & 37.6\% & 44.7\% \\
    20 & \$1{,}000 & 37.3\% & 58.8\% & 64.4\% & 61.9\% & 37.7\% & 46.3\% \\
    20 & \$2{,}000 & 38.7\% & 58.6\% & 64.4\% & 62.5\% & 37.7\% & 48.4\% \\
    20 & \$5{,}000 & 38.5\% & 58.7\% & 65.0\% & 60.9\% & 37.5\% & 53.5\% \\
    20 & \$10{,}000 & 40.7\% & 57.6\% & 64.3\% & 67.0\% & 37.2\% & 55.6\% \\
    \bottomrule
\end{tabular}

    \vspace{1.2em}
    \textit{Panel C. Loss share}\\[2pt]
    \begin{tabular}{rrrrrrrr}
    \toprule
     &  & \multicolumn{3}{c}{Manipulated cycles} & \multicolumn{3}{c}{Normal cycles} \\
    \cmidrule(lr){3-5}\cmidrule(lr){6-8}
    $n_{\mathrm{cyc,min}}$ & $\$_{\min}$ & Manip. & MM & Retail & Manip. & MM & Retail \\
    \midrule
    3 & \$500 & 17.8\% & 8.4\% & 73.8\% & 46.3\% & 9.1\% & 44.6\% \\
    3 & \$1{,}000 & 18.7\% & 8.6\% & 72.6\% & 44.9\% & 9.6\% & 45.5\% \\
    3 & \$2{,}000 & 19.7\% & 9.1\% & 71.2\% & 43.1\% & 10.6\% & 46.3\% \\
    3 & \$5{,}000 & 20.3\% & 10.3\% & 69.4\% & 40.0\% & 12.4\% & 47.5\% \\
    3 & \$10{,}000 & 20.1\% & 12.5\% & 67.3\% & 37.0\% & 15.3\% & 47.7\% \\
    5 & \$500 & 18.4\% & 8.6\% & 73.0\% & 46.2\% & 9.3\% & 44.5\% \\
    5 & \$1{,}000 & 19.3\% & 8.8\% & 71.9\% & 44.8\% & 9.8\% & 45.4\% \\
    5 & \$2{,}000 & 20.2\% & 9.3\% & 70.5\% & 43.2\% & 10.8\% & 46.0\% \\
    5 & \$5{,}000 & 20.7\% & 10.6\% & 68.7\% & 40.1\% & 12.7\% & 47.3\% \\
    5 & \$10{,}000 & 20.6\% & 12.8\% & 66.6\% & 37.1\% & 15.5\% & 47.4\% \\
    10 & \$500 & 19.1\% & 9.2\% & 71.6\% & 46.5\% & 9.9\% & 43.6\% \\
    10 & \$1{,}000 & 19.9\% & 9.5\% & 70.5\% & 45.3\% & 10.4\% & 44.3\% \\
    10 & \$2{,}000 & 20.8\% & 10.1\% & 69.1\% & 43.6\% & 11.6\% & 44.8\% \\
    10 & \$5{,}000 & 21.5\% & 11.4\% & 67.1\% & 40.4\% & 13.5\% & 46.1\% \\
    10 & \$10{,}000 & 21.5\% & 13.9\% & 64.7\% & 37.3\% & 16.3\% & 46.3\% \\
    20 & \$500 & 19.8\% & 10.5\% & 69.7\% & 45.9\% & 10.7\% & 43.4\% \\
    20 & \$1{,}000 & 20.5\% & 10.8\% & 68.7\% & 44.9\% & 11.2\% & 43.9\% \\
    20 & \$2{,}000 & 21.3\% & 11.5\% & 67.2\% & 43.3\% & 12.4\% & 44.3\% \\
    20 & \$5{,}000 & 22.2\% & 13.0\% & 64.8\% & 40.2\% & 14.3\% & 45.4\% \\
    20 & \$10{,}000 & 22.2\% & 16.1\% & 61.7\% & 37.2\% & 17.1\% & 45.7\% \\
    \bottomrule
\end{tabular}

\end{table}

\clearpage
\section{Hedging: Position Trajectories and Magnitudes}
\label{app:pushdown_flow}

\paragraph{The timing of the spot flow is wrong for a hedge.} The timing test concerns the cycles where there \emph{is} something to hedge, and asks when the spot trading arrives. Under dynamic hedging, the standard practice in liquid markets, a maker hedges each marginal share as the short accumulates: if the cohort's short builds steadily through the body, Binance buying should increase in step. Figure~\ref{fig:cohort_flow_trajectory} shows it does not. The Polymarket positions build continuously: the manipulator cohort's net Up holdings climb to roughly $8{,}800$ shares by the close, and the MM cohort absorbs a short that grows to roughly $1{,}350$ shares as the final ten seconds begin, a near-linear accumulation, not a flat position unloaded at the end. The matching spot hedge never appears through the body. Cumulative Binance order flow stands at just $-\$46$\,k at second $215$ of the $300$-second cycle, essentially zero and, if anything, slightly against the eventual push direction; the buying arrives only in the final minute, and the final ten seconds alone carry roughly \$1.7\,M, essentially the cycle's entire net order flow in a single step. A spot position built only in the final seconds is inconsistent with a hedge of an exposure that accumulated over the cycle, and consistent with a trade timed to move the settlement price.

\paragraph{The magnitudes rule out a deferred hedge.} One variant survives the timing test: a maker who holds the naked delta through the body and hedges only as gamma rises at the very close would also trade in a single late burst. The magnitudes rule it out. A short binary position can lose at most its notional, so the MM cohort's roughly $1{,}350$-share short puts about \$1{,}400 per cycle at risk, three orders of magnitude below the \$1.7\,M spot trade the channel would attribute to it. No maker, let alone all $28$ active in the average manipulated cycle, hedges a four-figure maximum loss with a seven-figure spot position. The more credible reading is that the bin-29 push is direct trading on the spot market, coordinated with the manipulator cohort's body-period positioning on Polymarket, with at most a residual hedging component in the bin-29 buying.

\begin{figure}[t]
    \centering
    \includegraphics[width=\textwidth]{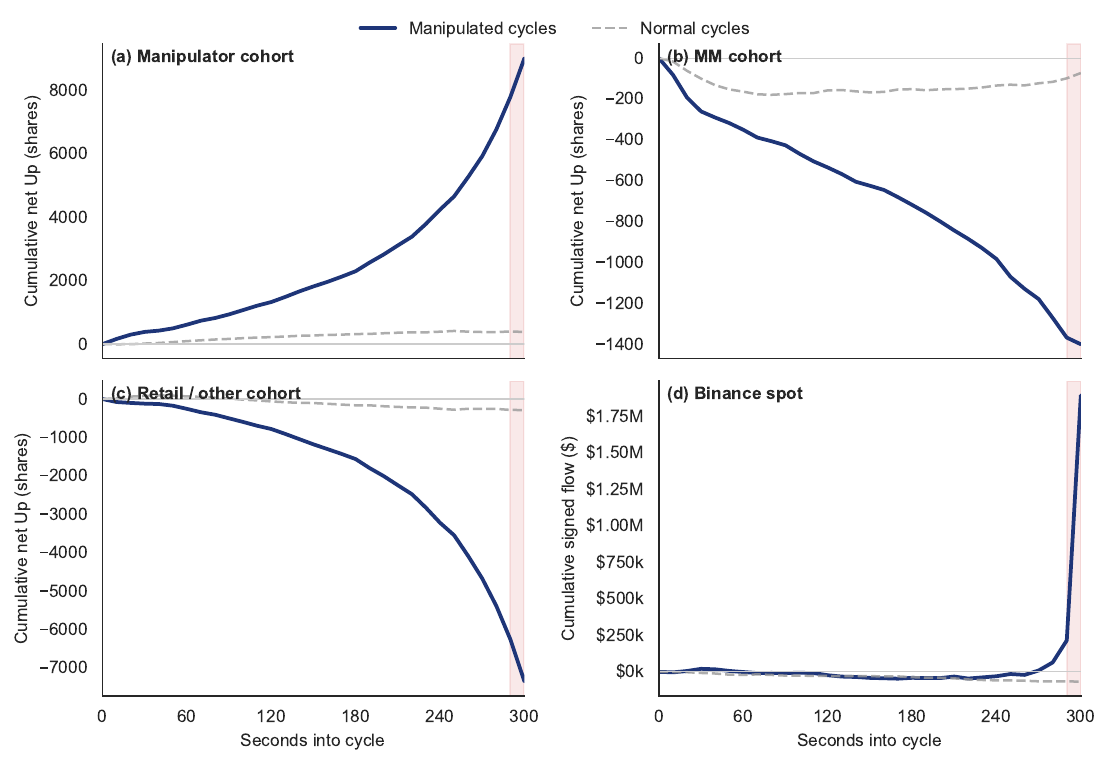}
    \caption{\textbf{Cumulative net position trajectory across the cycle, push-up vs.\ normal cycles, P3 BTC 5\mn.} This figure plots, in each panel, the mean per-cycle cumulative position (or order flow) against seconds into the cycle, starting from zero at the open. The solid line conditions on \emph{push-up} cycles, manipulated cycles whose near-settlement push is upward (labeled ``manipulated cycles'' in the legend), so that positions are not averaged across the two push directions; push-down cycles are reported in Figure~\ref{fig:cohort_flow_trajectory_pushdown}. Panels \textbf{(a)}--\textbf{(c)} plot the net Up-token holdings (shares) of the manipulator, MM, and retail/other cohorts; panel \textbf{(d)} plots the cumulative Binance spot order flow (\$). Dashed gray lines: normal cycles as baselines. Pink shading: the final-10\,s region (seconds 290--300).}
    \label{fig:cohort_flow_trajectory}
\end{figure}

\paragraph{Push-down cycles.} 
Figure~\ref{fig:cohort_flow_trajectory} conditions on \emph{push-up} cycles, manipulated cycles whose final-10s spot push is upward, so that the directional accumulation reads as a positive ramp. This appendix repeats the same construction on the \emph{push-down} mirror (manipulated cycles whose push is downward) in Figure~\ref{fig:cohort_flow_trajectory_pushdown}.

The pattern is the sign-flipped mirror with the identical timing signature. The manipulator cohort ramps smoothly into a net \emph{short} Up position (buying the Down token), reaching $\sim -9{,}500$ shares by close; MMs and retail take the long-Up side; and the Binance order flow is flat through the body and then steps down by $\sim\$1.75$\,M in the final 10\,s (cumulative $\sim -\$2$\,M by close), the down-side counterpart of the $+\$1.7$\,M up-push. The mirror is not exact in magnitude: both the final-10s Binance step ($\sim -\$1.75$\,M vs $\sim +\$1.7$\,M) and the manipulator cohort's cumulative directional position ($\sim -9{,}500$ vs $\sim +8{,}800$ shares) are marginally larger on the down side. But the linear-ramp-then-late-step timing is the same in both directions, confirming that the timing argument against the gamma-hedging channel (Section~\ref{subsec:wallet_cohorts_hedging}) is not an artifact of conditioning on the up direction.

\begin{figure}[t]
    \centering
    \includegraphics[width=\textwidth]{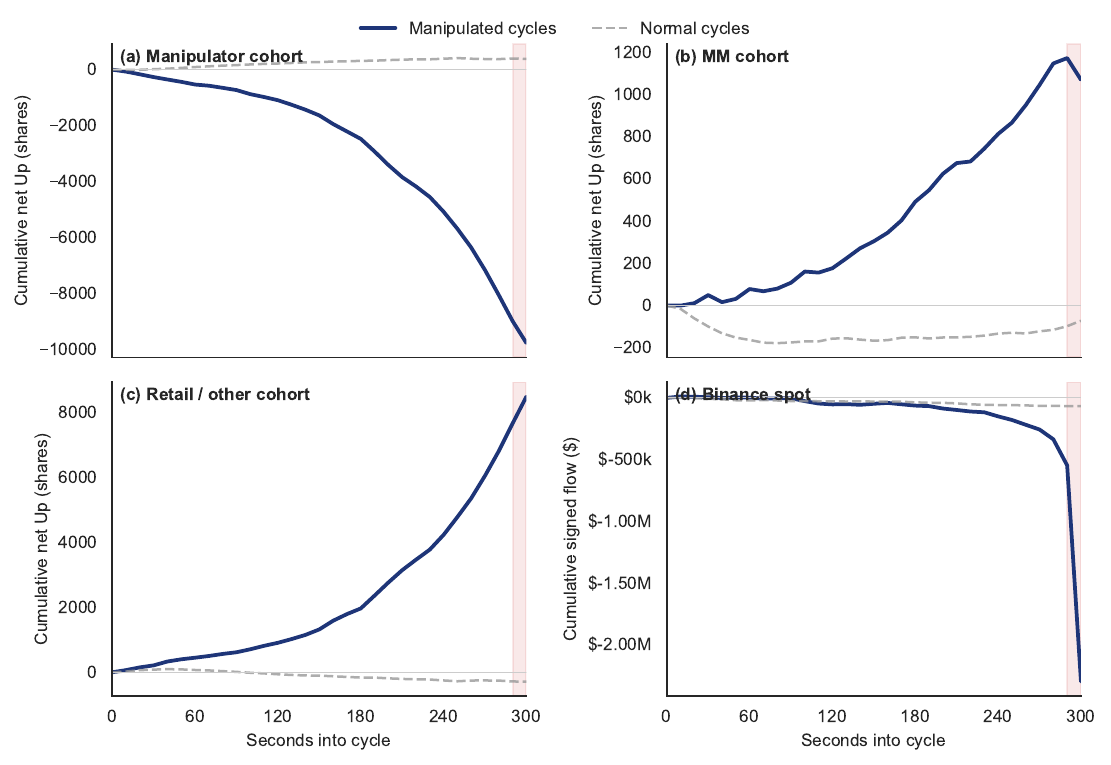}
    \caption{\textbf{Cumulative net position trajectory across the cycle, push-down vs.\ normal cycles, P3 BTC 5\mn.} This figure plots the push-down counterpart of Figure~\ref{fig:cohort_flow_trajectory}; the solid line conditions on push-down cycles (labeled ``manipulated cycles'' in the legend), construction otherwise identical. Panels \textbf{(a)}--\textbf{(c)} plot the net Up-token holdings (shares) of the manipulator, MM, and retail/other cohorts; panel \textbf{(d)} plots the cumulative Binance spot order flow (\$). Dashed gray lines: normal cycles. Pink shading: the final-10\,s region (seconds 290--300).}
    \label{fig:cohort_flow_trajectory_pushdown}
\end{figure}

\section{The Settlement-Rule Change: Additional Exhibits}
\label{app:twap}

This appendix reports the remaining exhibits of Section~\ref{sec:twap}: summary statistics of the two windows on both venues, the near-the-money regressions and profiles of Section~\ref{sec:atm} in each window, the reversal by moneyness, and tests of whether manipulation moved into the 60-second settlement window. The moneyness label follows one rule in every window: a cycle is classified at the opening of the interval over which its settlement price forms. Under point settlement that interval is the final ten seconds, so P3 and the pre-TWAP window use the label of Section~\ref{sec:atm}, the latest trade with time to maturity between 10 and 12 seconds. Under the TWAP it is the 60-second settlement window, so the label is the latest trade with time to maturity between 60 and 62 seconds. The P3 estimates are therefore those of Tables~\ref{tab:reg_cond} and~\ref{tab:reversal_cond}.

\subsection{Sample and summary statistics}
\label{app:twap_summary}

Table~\ref{tab:twap_poly} describes the pre-TWAP and TWAP windows on both venues: the Binance per-cycle measures the study is built on and the prediction market's own late-cycle trading.

Trading on the prediction market moves away from the final ten seconds under the TWAP. As Table~\ref{tab:twap_poly} shows, the share of volume traded in the final ten seconds falls from 6.7\% in the pre-TWAP window to 1.0\% in the TWAP window, and the share of cycles whose Polymarket price crosses even money in the final ten seconds falls from 6.4\% to 0.7\%. Few cycles are still undecided that late: the share of cycles within 0.10 of even money ten seconds before the close falls from 4.0\% to 0.4\%.

\begin{table}[!t]
    \centering
    \footnotesize
    \caption{\textbf{The pre-TWAP and TWAP windows: summary statistics.} Per-cycle statistics for the BTC five-minute contract and its Binance BTC/USDT spot counterpart in the pre-TWAP window (May 20 to August 6, 2026) and the TWAP window (August 14 to October 31, 2026). \emph{Binance}: cycles with a valid push measure; absolute net signed taker notional per cycle, averaged over the body bins 0 to 24 (the median body bin of each cycle) and in the final 10\,s bin; the standard deviation of the final-10\,s log return; and the share of cycles with $\textit{PushIntensity} \geq 16.11$, the P3 top-decile cut of Table~\ref{tab:push_characteristics}. \emph{Polymarket}: cycles with at least one taker-focused on-chain fill; USDC notional per cycle; ``share of volume in last $k$\,s'' is the window-level ratio of notional traded with time to maturity at most $k$ seconds to total notional; ``cycles with a trade in last 10\,s'' is the share of cycles with a taker-focused fill in the final 10 seconds; ``\NTM share at TTM $t$\,s'' is the share of cycles with a taker-focused trade whose time to maturity lies in $[t, t+2]$\,s and whose Up price satisfies $|p_{\text{Up}} - 0.5| < 0.10$, among cycles with such a trade (10\,s is the label time under point settlement and 60\,s under the TWAP); ``belief flip in last 10\,s'' is the share of cycles with a taker-focused trade at time to maturity in $[10, 12]$\,s whose Up price at that trade and at the cycle's last trade lie on opposite sides of 0.5.}
    \label{tab:twap_poly}
    \begin{tabular}{l S[table-format=5.3] S[table-format=5.3]}
    \toprule
     & {Pre-TWAP} & {TWAP} \\
     & {May 20 -- Aug 6} & {Aug 14 -- Oct 31} \\
    \midrule
    \multicolumn{3}{l}{\emph{Binance BTC/USDT}} \\
    Binance cycles & {22{,}751} & {22{,}752} \\
    |order flow| per cycle, body bins (\$ thousand) & 32.3 & 42.1 \\
    |order flow| per cycle, final bin (\$ thousand) & 149.3 & 109.7 \\
    Sd of closing-bin return (bps) & 2.10 & 1.70 \\
    Share of cycles with PushIntensity $\geq$ 16.11 & 0.103 & 0.062 \\
    \addlinespace
    \multicolumn{3}{l}{\emph{Polymarket BTC 5m}} \\
    Cycles & {22{,}668} & {22{,}444} \\
    Volume per cycle (\$ thousand) & 35.6 & 22.6 \\
    Share of volume in last 10 s & 0.067 & 0.010 \\
    Share of volume in last 30 s & 0.156 & 0.055 \\
    Share of volume in last 60 s & 0.281 & 0.194 \\
    Cycles with a trade in last 10 s & 0.956 & 0.734 \\
    NTM share at TTM 10 s & 0.040 & 0.004 \\
    NTM share at TTM 60 s & 0.110 & 0.050 \\
    Belief flip in last 10 s (share of classified) & 0.064 & 0.007 \\
    \bottomrule
\end{tabular}

\end{table}

\subsection{Near-the-money exhibits by settlement window}
\label{app:twap_cond}
\label{app:twap_reversal}

Table~\ref{tab:twap_cond} and Figure~\ref{fig:twap_fold_ntm} repeat the near-the-money regressions and profiles of Section~\ref{sec:atm} in each window, and Figure~\ref{fig:twap_rev_ntm} and Table~\ref{tab:twap_rev_ntm} the reversal by moneyness of Section~\ref{sec:atm}.

In \NTM cycles the footprint is as strong in the pre-TWAP window as in P3 and disappears in the TWAP window. As Table~\ref{tab:twap_cond} shows, the additional bin-29 order flow in \NTM over \FFM cycles is $\hat\delta_{\NTM} = 1.37$ log points in P3 and $1.18$ in the pre-TWAP window, about 3.9 and 3.3 times the \FFM increase on the typical cycle, and the additional absolute return is $0.56$ and $0.54$. In the TWAP window both are indistinguishable from zero, the additional depth falls from $0.17$ to $0.03$, and Figure~\ref{fig:twap_fold_ntm} shows no final-bin spike in the \NTM profile.

The reversal shows the same pattern. As Table~\ref{tab:twap_rev_ntm} shows, the \NTM increment to the reversal slope at $j=+1$ is $-0.15$ in P3 and $-0.18$ in the pre-TWAP window, both significant at the 1\% level, and is insignificant in the TWAP window. The weaker pooled reversal in the pre-TWAP window comes from \FFM cycles, whose reversal slope is $-0.10$ in P3 and indistinguishable from zero in the pre-TWAP window.

\begin{table}[!t]
    \centering
    \footnotesize
    \caption{\textbf{Near-settlement activity in \NTM cycles, by settlement window.} Cycle-level OLS on the (cycle, bin) panel of Binance BTC 10\,s buckets restricted to bins $\{0\text{--}24, 29\}$ and to classified cycles, one regression per window: $\log y_{c,b} = \alpha_d + \mu_h + \gamma_b + \beta_w\,\mathbf{1}\{\NTM_c\} + \delta_{\NTM,w}\,\mathbf{1}\{b{=}29\}\mathbf{1}\{\NTM_c\} + \varepsilon_{c,b}$, with date, hour-of-day, and bin fixed effects. A cycle is \NTM if the Polymarket Up-token price of the latest taker-focused on-chain trade at the regime's label time (time to maturity in $[10, 12]$\,s in P3 and the pre-TWAP window, $[60, 62]$\,s in the TWAP window) satisfies $|p_{\text{Up}} - 0.5| < 0.10$, and \FFM otherwise. Windows and dependent variables as in Table~\ref{tab:twap_activity}. Observations and clusters are summed over the three window regressions. Coefficients are log differences. HC1 standard errors in parentheses are clustered at the cycle level. Significance: $^{*}\,p<0.10$, $^{**}\,p<0.05$, $^{***}\,p<0.01$.}
    \label{tab:twap_cond}
    \begin{tabular}{l S[table-format=-1.3] S[table-format=-1.3] S[table-format=-1.3]}
    \toprule
     & {(1)} & {(2)} & {(3)} \\
     & {Order flow} & {Abs. return} & {Depth} \\
    \midrule
    $\delta_{\NTM,\text{P3}}$ & 1.366{\sym{***}} & 0.555{\sym{***}} & 0.173{\sym{***}} \\
     & (0.086) & (0.052) & (0.022) \\[4pt]
    $\delta_{\NTM,\text{pre-TWAP}}$ & 1.183{\sym{***}} & 0.542{\sym{***}} & 0.166{\sym{***}} \\
     & (0.104) & (0.060) & (0.026) \\[4pt]
    $\delta_{\NTM,\text{TWAP}}$ & 0.062 & 0.009 & 0.028{\sym{***}} \\
     & (0.065) & (0.064) & (0.011) \\
    \midrule
    Date FE & {Yes} & {Yes} & {Yes} \\
    Hour-of-day FE & {Yes} & {Yes} & {Yes} \\
    Bin FE & {Yes} & {Yes} & {Yes} \\
    Observations & {1{,}360{,}439} & {910{,}449} & {1{,}360{,}419} \\
    Clusters (cycle) & {52{,}326} & {52{,}053} & {52{,}325} \\
    \bottomrule
\end{tabular}

\end{table}

\begin{figure}[!t]
    \centering
    \includegraphics[width=\textwidth]{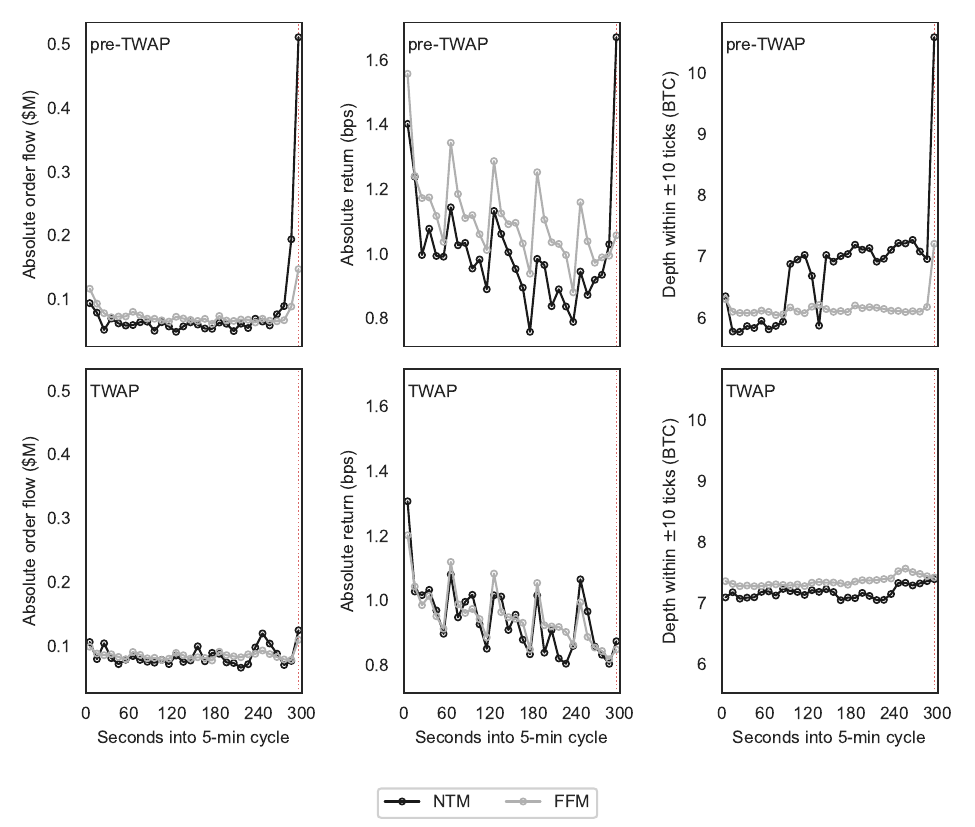}
    \caption{\textbf{Binance BTC intra-cycle profiles, \NTM vs.\ \FFM cycles, pre-TWAP vs.\ TWAP.} Rows are the pre-TWAP window (May 20 to August 6, 2026) and the TWAP window (August 14 to October 31, 2026); columns are absolute order flow per cycle (\$M), absolute 10\,s return (bps), and time-weighted depth within $\pm$10 ticks of the mid (BTC), folded as in Figure~\ref{fig:twap_fold}; the two rows of each column share one vertical scale. Each point is an arithmetic mean over cycles; for the right-skewed order-flow measure a few large cycles can lift a bin's mean (in the TWAP window's \NTM cycles, nine cycles account for a quarter of the bin-24 total), and Table~\ref{tab:twap_cond_late} gives the corresponding log comparison. A cycle is \NTM if the Polymarket Up-token price of the latest taker-focused on-chain trade at the regime's label time (time to maturity in $[10, 12]$\,s in P3 and the pre-TWAP window, $[60, 62]$\,s in the TWAP window) satisfies $|p_{\text{Up}} - 0.5| < 0.10$, and \FFM otherwise; cycles without such a trade are excluded. The dotted vertical line marks the final 10\,s before the close.}
    \label{fig:twap_fold_ntm}
\end{figure}

\begin{figure}[!t]
    \centering
    \includegraphics[width=0.85\textwidth]{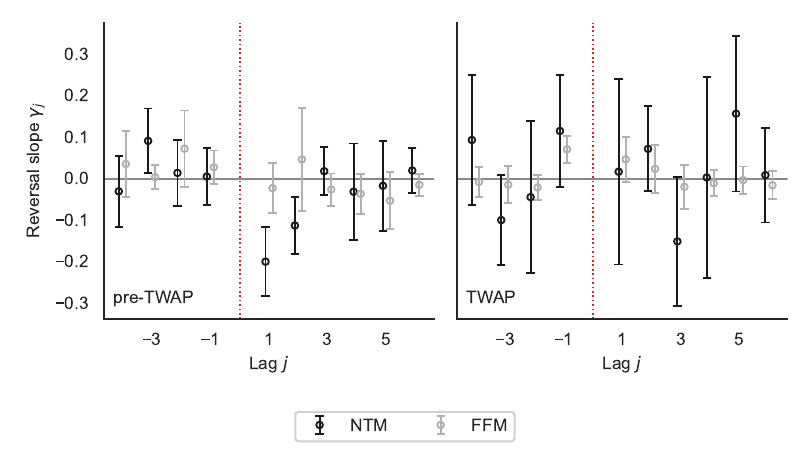}
    \caption{\textbf{Reversal slope $\gamma_j$, \NTM vs.\ \FFM cycles, pre-TWAP vs.\ TWAP.} Markers are the slope of the lag-$j$ 10\,s Binance BTC log return on the closing-bin return $r_{29}$, from one pooled regression per lag with $r_{29}$ interacted with window $\times$ \NTM cell, date fixed effects, hour-of-day fixed effects interacted with window, and an \NTM $\times$ window fixed effect; lag 0 is omitted. $j<0$ denotes bin $29+j$ of the same cycle, $j>0$ bin $j-1$ of the next cycle. A cycle is \NTM if the Polymarket Up-token price of the latest taker-focused on-chain trade at the regime's label time (time to maturity in $[10, 12]$\,s in P3 and the pre-TWAP window, $[60, 62]$\,s in the TWAP window) satisfies $|p_{\text{Up}} - 0.5| < 0.10$, \FFM otherwise. \emph{Left}: pre-TWAP window (May 20 to August 6, 2026). \emph{Right}: TWAP window (August 14 to October 31, 2026). Error bars are $\pm 1.96$ date-clustered standard errors; the dotted vertical line marks the close.}
    \label{fig:twap_rev_ntm}
\end{figure}

\begin{table}[!t]
    \centering
    \footnotesize
    \caption{\textbf{Post-settlement reversal by moneyness, by settlement window.} This table repeats the specification of Table~\ref{tab:reversal_cond} with the windows side by side: cycle-level OLS of the 10\,s Binance BTC log return at lag $j \in \{-3,-2,-1,+1,+2,+3\}$ on the closing-bin return $r_{29}$, one pooled regression per lag, $r_j = \alpha_d + \mu_{h,w} + \nu_{g,w} + \sum_w \big(\gamma_w + \delta_{\NTM,w}\,\mathbf{1}\{\NTM\}\big)\,r_{29}\,\mathbf{1}\{w\} + \varepsilon$, with date fixed effects, hour-of-day fixed effects interacted with window, and a moneyness-group fixed effect interacted with window. $\gamma_w$ is the \FFM baseline slope in window $w$ and $\delta_{\NTM,w}$ the \NTM increment over it, so the \NTM-cycle slope is $\gamma_w + \delta_{\NTM,w}$; the stars test each coefficient against zero. $j<0$ denotes bin $29+j$ of the same cycle, $j>0$ bin $j-1$ of the next cycle within the same UTC date. Windows as in Table~\ref{tab:twap_activity}; a cycle is \NTM if the Polymarket Up-token price of the latest taker-focused on-chain trade at the regime's label time (time to maturity in $[10, 12]$\,s in P3 and the pre-TWAP window, $[60, 62]$\,s in the TWAP window) satisfies $|p_{\text{Up}} - 0.5| < 0.10$, \FFM otherwise; unclassified cycles are excluded. The P3 rows are those of Table~\ref{tab:reversal_cond}. Observations are cycle-lag pairs in the pooled regression. Standard errors in parentheses are clustered by date. Significance: $^{*}\,p<0.10$, $^{**}\,p<0.05$, $^{***}\,p<0.01$.}
    \label{tab:twap_rev_ntm}
    \begin{tabular}{l S[table-format=-1.3] S[table-format=-1.3] S[table-format=-1.3] S[table-format=-1.3] S[table-format=-1.3] S[table-format=-1.3]}
    \toprule
     & \multicolumn{3}{c}{Pre-close} & \multicolumn{3}{c}{Post-close} \\
    \cmidrule(lr){2-4} \cmidrule(lr){5-7}
     & {$j = -3$} & {$j = -2$} & {$j = -1$} & {$j = +1$} & {$j = +2$} & {$j = +3$} \\
    \midrule
    $\gamma_{\text{P3}}$ & -0.002 & 0.019 & 0.041{\sym{***}} & -0.101{\sym{***}} & -0.040{\sym{**}} & 0.007 \\
     & (0.011) & (0.014) & (0.010) & (0.023) & (0.018) & (0.013) \\[2pt]
    $\delta_{\NTM,\text{P3}}$ & 0.008 & -0.013 & -0.004 & -0.154{\sym{***}} & 0.028 & -0.025 \\
     & (0.017) & (0.024) & (0.021) & (0.038) & (0.037) & (0.026) \\[2pt]
    $\gamma_{\text{pre-TWAP}}$ & 0.004 & 0.073 & 0.028 & -0.023 & 0.047 & -0.026 \\
     & (0.015) & (0.047) & (0.021) & (0.031) & (0.063) & (0.020) \\[2pt]
    $\delta_{\NTM,\text{pre-TWAP}}$ & 0.088{\sym{*}} & -0.059 & -0.022 & -0.177{\sym{***}} & -0.160{\sym{**}} & 0.044 \\
     & (0.050) & (0.047) & (0.038) & (0.052) & (0.070) & (0.032) \\[2pt]
    $\gamma_{\text{TWAP}}$ & -0.014 & -0.021 & 0.071{\sym{***}} & 0.047{\sym{*}} & 0.024 & -0.019 \\
     & (0.023) & (0.015) & (0.016) & (0.028) & (0.029) & (0.027) \\[2pt]
    $\delta_{\NTM,\text{TWAP}}$ & -0.085 & -0.023 & 0.044 & -0.030 & 0.048 & -0.132 \\
     & (0.066) & (0.095) & (0.068) & (0.120) & (0.064) & (0.082) \\
    \midrule
    Date FE & {Yes} & {Yes} & {Yes} & {Yes} & {Yes} & {Yes} \\
    Hour-of-day FE & {Yes} & {Yes} & {Yes} & {Yes} & {Yes} & {Yes} \\
    Observations (cycles) & \multicolumn{6}{c}{54{,}001 (221 dates)} \\
    \bottomrule
\end{tabular}

\end{table}

\subsection{Adaptation to the settlement window}
\label{app:twap_adapt}

A manipulator facing a 60-second average could spread the push over the whole settlement window instead of abandoning it. We ask three questions in turn: whether activity moved earlier in the cycle, whether cycles with a large push over the whole window resolve differently, and whether anyone profits in them.

\paragraph{Activity in the settlement window.} We extend the per-window regression of Section~\ref{sec:twap_footprint} from the final bin to every bin of the settlement window, bins 24 to 29:
\[
    \log y_{c,b} \;=\; \gamma_b + \alpha_d + \mu_h + \sum_{k=24}^{29} \delta_{k,w}\,\mathbf{1}\{b{=}k\}{\cdot}\mathbf{1}\{w\} + \varepsilon_{c,b},
\]
on all 30 bins of the cycles of P1 and $w$, so that $\delta_{k,w}$ is the increase of bin $k$ over the body, here bins 0 to 23, in window $w$ relative to the pre-launch baseline. Table~\ref{tab:twap_late} reports the estimates. In the TWAP window, order flow is higher by 0.09, 0.05, and 0.03 log points in bins 25 to 27, unchanged in bin 28, and higher by 0.12 in bin 29; bin 24 is lower by 0.05, as it is in P3 and the pre-TWAP window. In P3, by contrast, the increase is insignificant through bin 26 and rises from 0.07 in bin 27 through 0.14 to 0.40 in bin 29. Some order flow has thus moved into the earlier bins of the window, at a fraction of the former final-bin increase, and the return increases in those bins are at most marginally significant.

The order flow that moved earlier is not concentrated in \NTM cycles. Table~\ref{tab:twap_cond_late} adds the moneyness interaction bin by bin, estimated on the classified cycles of each window:
\[
    \log y_{c,b} \;=\; \gamma_b + \alpha_d + \mu_h + \beta_{w}\,\mathbf{1}\{\NTM_c\} + \sum_{k=24}^{29} \delta_{\NTM,k,w}\,\mathbf{1}\{b{=}k\}{\cdot}\mathbf{1}\{\NTM_c\} + \varepsilon_{c,b},
\]
where $\delta_{\NTM,k,w}$ is the additional increase of bin $k$ over the body in \NTM relative to \FFM cycles. In the TWAP window the \NTM order-flow increment is insignificant in every bin of the settlement window, whereas in the pre-TWAP window it is 0.53 in bin 28 and 1.18 in bin 29. Absolute returns in \NTM cycles are higher in the first two bins of the window, by 0.13 and 0.11 log points, the second only marginally significant, without a matching order-flow increment, and the difference is gone by the final two bins.

\paragraph{Outcomes of window-wide pushes.} We next adapt $\textit{PushIntensity}$ to the new rule: the numerator becomes the net signed Binance order flow over the whole settlement window, bins 24 to 29, and the denominator the median absolute order flow of body bins 0 to 23. Manipulated cycles are the top decile of this measure within the TWAP window, as they are the top decile of the final-bin measure within P3. Table~\ref{tab:twap_push_char_last60} reports their outcomes. These cycles do resolve against the favored side more often than normal cycles: the underdog wins 31\% of them compared to 12\% at confidence between 0.80 and 0.90, 8.8\% compared to 1.8\% above 0.90, and 52\% to 53\% compared to 23\% to 46\% in the three lower bins. The excess cannot by itself be read as manipulation: a large net flow over the window moves a 60-second average whatever its source, and large flows accompany news. The two readings differ in who profits.

\paragraph{Profits under the window-wide definition.} Table~\ref{tab:twap_window_rule} reports wallet gains in the TWAP window with manipulated cycles so defined. Manipulated cycles see no more wealth change hands than normal cycles: gross gains per cycle are \$5{,}700 against \$5{,}600, or 13 against 12 cents per dollar traded. Applying the manipulator rule to these cycles labels 184 wallets with \$1.1 million of PnL on them, a figure that is positive by construction because the rule selects on it. The same wallets earn \$3.5 million on normal cycles and win in 52\% of their manipulated cycles against 51\% of their normal ones. Unlike the pre-TWAP manipulators, whose gains come entirely from manipulated cycles (Table~\ref{tab:twap_cohort_ft}), these wallets have no advantage specific to the cycles with the largest window-wide flows.

Taken together, the switch had the effect one would expect if averaging raises the cost of a push. The final-bin push disappeared from Binance order flow, returns, and depth, from the post-settlement reversal, from trading on the prediction market, and from the manipulators' profits. What remains in the 60-second window is small, spread across bins, absent from the order flow of near-the-money cycles, and not a source of concentrated gains.

\begin{landscape}
\begin{table}[p]
    \centering
    \footnotesize
    \caption{\textbf{Settlement-window bins relative to the pre-Polymarket baseline, by settlement window.} This table extends Table~\ref{tab:twap_activity} from the final bin to every bin of the 60\,s settlement window: for window $w$, cycle-level OLS on the (cycle, bin) panel of Binance BTC 10\,s buckets over all 30 bins, on the pre-Polymarket period (P1, Jul--Oct 2025) and $w$, $\log y_{c,b} = \gamma_b + \alpha_d + \mu_h + \sum_{k=24}^{29} \delta_{k,w}\,\mathbf{1}\{b{=}k\}{\cdot}\mathbf{1}\{w\} + \varepsilon_{c,b}$, with bin, date, and hour-of-day fixed effects, where $w$ is P3 (February 12 to April 8, 2026), the pre-TWAP window (May 20 to August 6, 2026), or the TWAP window (August 14 to October 31, 2026). Each window's column within a metric comes from its own regression; $\delta_{k,w}$ is bin $k$'s increase over the body, bins 0 to 23, in window $w$ relative to the pre-Polymarket baseline. Bins 24 to 29 are the 60\,s settlement window: bin 24 covers seconds 240 to 250 of the cycle (time to maturity 60 to 50\,s) and bin 29 the final 10\,s; the bin-29 rows differ from Table~\ref{tab:twap_activity} only in the body (bins 0 to 23 here, 0 to 24 there). Dependent variables as in Table~\ref{tab:twap_activity}; the footer reports each regression's sample. Coefficients are log differences. HC1 standard errors in parentheses are clustered at the cycle level. Significance: $^{*}\,p<0.10$, $^{**}\,p<0.05$, $^{***}\,p<0.01$.}
    \label{tab:twap_late}
    \begin{tabular}{l S[table-format=-1.3] S[table-format=-1.3] S[table-format=-1.3] S[table-format=-1.3] S[table-format=-1.3] S[table-format=-1.3] S[table-format=-1.3] S[table-format=-1.3] S[table-format=-1.3]}
    \toprule
     & \multicolumn{3}{c}{Order flow} & \multicolumn{3}{c}{Abs. return} & \multicolumn{3}{c}{Depth} \\
    \cmidrule(lr){2-4} \cmidrule(lr){5-7} \cmidrule(lr){8-10}
     & {P3} & {Pre-TWAP} & {TWAP} & {P3} & {Pre-TWAP} & {TWAP} & {P3} & {Pre-TWAP} & {TWAP} \\
    \midrule
    Bin 24 & -0.041{\sym{**}} & -0.047{\sym{***}} & -0.049{\sym{***}} & 0.007 & -0.014 & -0.023 & 0.003 & 0.002 & 0.010{\sym{***}} \\
     & (0.016) & (0.016) & (0.016) & (0.016) & (0.015) & (0.016) & (0.003) & (0.002) & (0.002) \\[2pt]
    Bin 25 & 0.017 & 0.038{\sym{**}} & 0.090{\sym{***}} & -0.011 & 0.031{\sym{*}} & 0.032{\sym{*}} & 0.003 & 0.003 & 0.013{\sym{***}} \\
     & (0.017) & (0.016) & (0.016) & (0.017) & (0.016) & (0.017) & (0.003) & (0.003) & (0.003) \\[2pt]
    Bin 26 & 0.020 & 0.013 & 0.051{\sym{***}} & 0.000 & -0.006 & 0.021 & 0.006{\sym{**}} & 0.004 & 0.011{\sym{***}} \\
     & (0.017) & (0.016) & (0.016) & (0.017) & (0.017) & (0.017) & (0.003) & (0.003) & (0.003) \\[2pt]
    Bin 27 & 0.068{\sym{***}} & 0.024 & 0.033{\sym{**}} & -0.004 & -0.016 & 0.027 & 0.011{\sym{***}} & 0.006{\sym{**}} & 0.011{\sym{***}} \\
     & (0.017) & (0.016) & (0.016) & (0.017) & (0.017) & (0.017) & (0.003) & (0.003) & (0.003) \\[2pt]
    Bin 28 & 0.136{\sym{***}} & 0.064{\sym{***}} & 0.015 & 0.005 & 0.002 & 0.003 & 0.010{\sym{***}} & 0.007{\sym{**}} & 0.006{\sym{**}} \\
     & (0.018) & (0.017) & (0.016) & (0.018) & (0.017) & (0.018) & (0.003) & (0.003) & (0.003) \\[2pt]
    Bin 29 & 0.404{\sym{***}} & 0.263{\sym{***}} & 0.121{\sym{***}} & 0.140{\sym{***}} & 0.066{\sym{***}} & 0.000 & 0.049{\sym{***}} & 0.052{\sym{***}} & 0.009{\sym{***}} \\
     & (0.019) & (0.018) & (0.017) & (0.018) & (0.017) & (0.018) & (0.003) & (0.003) & (0.003) \\
    \midrule
    Bin, date, hour FE & {Yes} & {Yes} & {Yes} & {Yes} & {Yes} & {Yes} & {Yes} & {Yes} & {Yes} \\
    Observations, P1 $\cup$ P3 & \multicolumn{9}{c}{1{,}347{,}830 (44{,}928 cycles)} \\
    Observations, P1 $\cup$ pre-TWAP & \multicolumn{9}{c}{1{,}546{,}538 (51{,}552 cycles)} \\
    Observations, P1 $\cup$ TWAP & \multicolumn{9}{c}{1{,}546{,}544 (51{,}552 cycles)} \\
    \bottomrule
\end{tabular}

\end{table}
\end{landscape}

\begin{landscape}
\begin{table}[p]
    \centering
    \footnotesize
    \caption{\textbf{Settlement-window bins in \NTM cycles, by settlement window.} This table extends Table~\ref{tab:twap_cond} from the final bin to every bin of the 60\,s settlement window: for window $w$, cycle-level OLS on the (cycle, bin) panel of Binance BTC 10\,s buckets over all 30 bins and classified cycles, $\log y_{c,b} = \gamma_b + \alpha_d + \mu_h + \beta_w\,\mathbf{1}\{\NTM_c\} + \sum_{k=24}^{29} \delta_{\NTM,k,w}\,\mathbf{1}\{b{=}k\}{\cdot}\mathbf{1}\{\NTM_c\} + \varepsilon_{c,b}$, with bin, date, and hour-of-day fixed effects, where $w$ is P3 (February 12 to April 8, 2026), the pre-TWAP window (May 20 to August 6, 2026), or the TWAP window (August 14 to October 31, 2026). Each window's column within a metric comes from its own regression; $\delta_{\NTM,k,w}$ is the extra increase of bin $k$ over the body, bins 0 to 23, in \NTM relative to \FFM cycles in window $w$. Bins 24 to 29 are the 60\,s settlement window (bin 24 covers time to maturity 60 to 50\,s, bin 29 the final 10\,s); the bin-29 rows differ from Table~\ref{tab:twap_cond} only in the body (bins 0 to 23 here, 0 to 24 there). The \NTM label is the regime's, as in Table~\ref{tab:twap_cond}. Dependent variables as in Table~\ref{tab:twap_activity}; the footer reports each regression's sample. Coefficients are log differences. HC1 standard errors in parentheses are clustered at the cycle level. Significance: $^{*}\,p<0.10$, $^{**}\,p<0.05$, $^{***}\,p<0.01$.}
    \label{tab:twap_cond_late}
    \begin{tabular}{l S[table-format=-1.3] S[table-format=-1.3] S[table-format=-1.3] S[table-format=-1.3] S[table-format=-1.3] S[table-format=-1.3] S[table-format=-1.3] S[table-format=-1.3] S[table-format=-1.3]}
    \toprule
     & \multicolumn{3}{c}{Order flow} & \multicolumn{3}{c}{Abs. return} & \multicolumn{3}{c}{Depth} \\
    \cmidrule(lr){2-4} \cmidrule(lr){5-7} \cmidrule(lr){8-10}
     & {P3} & {Pre-TWAP} & {TWAP} & {P3} & {Pre-TWAP} & {TWAP} & {P3} & {Pre-TWAP} & {TWAP} \\
    \midrule
    Bin 24 & 0.044 & -0.021 & 0.037 & 0.074 & -0.127{\sym{*}} & 0.130{\sym{**}} & 0.001 & -0.001 & 0.004 \\
     & (0.055) & (0.063) & (0.062) & (0.050) & (0.066) & (0.053) & (0.010) & (0.012) & (0.011) \\[2pt]
    Bin 25 & -0.024 & 0.025 & 0.098 & 0.094{\sym{*}} & -0.140{\sym{**}} & 0.106{\sym{*}} & 0.022{\sym{**}} & -0.000 & 0.006 \\
     & (0.058) & (0.066) & (0.061) & (0.049) & (0.069) & (0.057) & (0.010) & (0.012) & (0.011) \\[2pt]
    Bin 26 & 0.114{\sym{*}} & 0.126{\sym{*}} & 0.037 & 0.088{\sym{*}} & 0.156{\sym{***}} & 0.057 & 0.014 & -0.010 & 0.008 \\
     & (0.059) & (0.069) & (0.060) & (0.052) & (0.059) & (0.063) & (0.011) & (0.012) & (0.011) \\[2pt]
    Bin 27 & 0.090 & 0.047 & -0.053 & -0.094 & -0.029 & 0.040 & 0.008 & 0.002 & 0.008 \\
     & (0.062) & (0.074) & (0.059) & (0.062) & (0.070) & (0.061) & (0.012) & (0.013) & (0.011) \\[2pt]
    Bin 28 & 0.366{\sym{***}} & 0.530{\sym{***}} & 0.024 & 0.059 & 0.063 & 0.004 & 0.015 & 0.013 & 0.020{\sym{*}} \\
     & (0.073) & (0.082) & (0.060) & (0.057) & (0.066) & (0.062) & (0.012) & (0.013) & (0.011) \\[2pt]
    Bin 29 & 1.367{\sym{***}} & 1.182{\sym{***}} & 0.064 & 0.558{\sym{***}} & 0.536{\sym{***}} & 0.014 & 0.173{\sym{***}} & 0.166{\sym{***}} & 0.028{\sym{***}} \\
     & (0.086) & (0.104) & (0.065) & (0.052) & (0.060) & (0.064) & (0.022) & (0.026) & (0.011) \\
    \midrule
    Bin, date, hour FE; NTM main effect & {Yes} & {Yes} & {Yes} & {Yes} & {Yes} & {Yes} & {Yes} & {Yes} & {Yes} \\
    Observations, P3 & \multicolumn{9}{c}{456{,}682 (15{,}223 cycles)} \\
    Observations, pre-TWAP & \multicolumn{9}{c}{567{,}852 (18{,}929 cycles)} \\
    Observations, TWAP & \multicolumn{9}{c}{545{,}206 (18{,}174 cycles)} \\
    \bottomrule
\end{tabular}

\end{table}
\end{landscape}

\begin{table}[!t]
    \centering
    \footnotesize
    \caption{\textbf{Outcome rates by Polymarket-favored confidence with the push measured over the 60\,s settlement window.} As Table~\ref{tab:twap_push_char} with a settlement-window push measure in place of the bin-29 measure: the push is the net signed Binance taker notional over bins 24 to 29 (the last 60 seconds of the cycle) divided by the median absolute signed notional of body bins 0 to 23, and its direction is the sign of that net flow. This is the construction of Table~\ref{tab:push_characteristics} applied to the measure that matches the settlement rule: under the TWAP the settlement price is formed over the whole window, so manipulated cycles are the top decile of the window-wide measure within the TWAP window (August 14 to October 31, 2026; threshold 37.84), as they are the top decile of the final-bin measure within P3 under point settlement, where only the final push matters. Normal cycles are the rest. The confidence $p_{\mathrm{fav}}$ is measured at TTM~$=60$\,s, the opening of the settlement window, as in Table~\ref{tab:twap_push_char}. Other definitions follow Table~\ref{tab:push_characteristics}.}
    \label{tab:twap_push_char_last60}
    \begin{tabular}{clrrrr}
    \toprule
    \makecell{$p_{\mathrm{fav}}$\\range} & \makecell{Manipulation\\direction} & \# cycles & \% cycles & Win rate & \makecell{Win rate\\(normal cycles)} \\
    \midrule
    \multirow{2}{*}{{[0.50, 0.60)}} & Same-direction & 46 & 49\% & 69.6\% & 54.0\% \\
     & Opposite-direction & 47 & 51\% & 53.2\% & 46.0\% \\
    \addlinespace
    \multirow{2}{*}{{[0.60, 0.70)}} & Same-direction & 59 & 48\% & 88.1\% & 68.2\% \\
     & Opposite-direction & 65 & 52\% & 52.3\% & 31.8\% \\
    \addlinespace
    \multirow{2}{*}{{[0.70, 0.80)}} & Same-direction & 73 & 48\% & 90.4\% & 76.6\% \\
     & Opposite-direction & 78 & 52\% & 52.6\% & 23.4\% \\
    \addlinespace
    \multirow{2}{*}{{[0.80, 0.90)}} & Same-direction & 106 & 43\% & 93.4\% & 87.9\% \\
     & Opposite-direction & 141 & 57\% & 31.2\% & 12.1\% \\
    \addlinespace
    \multirow{2}{*}{{[0.90, 1.00)}} & Same-direction & 646 & 48\% & 99.4\% & 98.2\% \\
     & Opposite-direction & 694 & 52\% & 8.8\% & 1.8\% \\
    \bottomrule
\end{tabular}

\end{table}

\begin{table}[!t]
    \centering
    \footnotesize
    \caption{\textbf{Wallet gains in manipulated and normal cycles under the window-wide rule, TWAP window.} Per-(cycle, wallet) gross PnL on the Polymarket BTC five-minute contract in the TWAP window (August 14 to October 31, 2026; closed-loop per cycle, on-chain maker attribution, fees not deducted). A manipulated cycle is defined by the window-wide rule of Table~\ref{tab:twap_push_char_last60}: net signed Binance flow over the 60\,s settlement window relative to the body median, top decile within the window. ``Gross transfer per cycle'' is the sum over wallets of positive PnL in the cycle, averaged over cycles of the given type, in dollars and in cents per dollar of USDC notional traded. ``Wallets meeting the manipulator rule'' applies the rule of Section~\ref{sec:wallet_cohorts}, at least five manipulated cycles and at least \$2,000 of net PnL on them, with manipulated cycles so defined; the last two rows report those wallets' summed PnL and the share of positive (cycle, wallet) rows in each cycle type. Because the rule selects wallets on their PnL in the manipulated cycles, that PnL is positive by construction; the comparison to read is across the two columns.}
    \label{tab:twap_window_rule}
    \begin{tabular}{lcc}
    \toprule
     & {Manipulated cycles} & {Normal cycles} \\
    \midrule
    Cycles & {2{,}271} & {20{,}348} \\
    Gross transfer per cycle (\$) & 5{,}655 & 5{,}574 \\
    Gross transfer per \$ traded (cents) & 13.1 & 11.7 \\
    \addlinespace
    \multicolumn{3}{l}{Wallets meeting the manipulator rule: 184} \\
    Their net PnL (\$ thousand) & 1{,}142 & 3{,}494 \\
    Their win rate (wallet-cycles) & 51.7\% & 51.2\% \\
    \bottomrule
\end{tabular}

\end{table}

\clearpage
\printbibliography

\end{document}